\pdfoutput=1
\documentclass[12pt,a4paper]{amsart} %will be restored later 
\usepackage[T1]{fontenc}
\usepackage[margin=2cm]{geometry}
\usepackage[colorlinks,allcolors=blue,pdfusetitle]{hyperref}
\usepackage{microtype} 

\usepackage{amscd,amsmath,amssymb,amsfonts,xspace,mathrsfs,amsthm}
\usepackage{xcolor}

\usepackage{bbold}
\usepackage{latexsym}
\usepackage{graphicx}
\usepackage{dsfont}
\usepackage{longtable} 
\usepackage{tikz-cd}
\usetikzlibrary{bending}

\newtheorem{theorem}{Theorem}
\newtheorem{proposition}[theorem]{Proposition}%[section]
\newtheorem{lemma}[theorem]{Lemma}

\theoremstyle{definition}
\newtheorem{definition}[theorem]{Definition}%[section]
\numberwithin{equation}{section}

\def\eg{{e.g.\ }}
\def\ie{{i.e.\ }}

\def\det{\,{\rm det}\, }

\DeclareMathOperator\Exp{Exp}

\def\({\left(}
\def\){\right)}
\def\<{\left\langle}
\def\>{\right\rangle}

\newcommand{\Omstar}{\Omega_\star}
\newcommand{\bOmstar}{\bOm_\star}

\newcommand{\de}{\mathrm{d}}

\newcommand{\I}{\mathrm{i}}

\newcommand{\cA}{\mathcal{A}}

\newcommand{\cC}{\mathcal{C}}
\newcommand{\cD}{\mathcal{D}}
\newcommand{\cE}{\mathcal{E}}
\newcommand{\cF}{\mathcal{F}}
\newcommand{\cV}{\mathcal{V}}
\newcommand{\cG}{\mathcal{G}}
\newcommand{\cH}{\mathcal{H}}
\newcommand{\cI}{\mathcal{I}}

\newcommand{\cM}{\mathcal{M}}
\newcommand{\cW}{\mathcal{W}}
\newcommand{\cN}{\mathcal{N}}

\newcommand{\cX}{\mathcal{X}}

\newcommand{\cO}{\mathcal{O}}
\newcommand{\cP}{\mathcal{P}}
\newcommand{\cR}{\mathcal{R}}
\newcommand{\cS}{\mathcal{S}}
\newcommand{\cT}{\mathcal{T}}
\newcommand{\cU}{\mathcal{U}}

\DeclareSymbolFont{AMSa}{U}{msa}{m}{n}
\DeclareSymbolFont{AMSb}{U}{msb}{m}{n}
\DeclareMathSymbol{\fieldR}{\mathalpha}{AMSb}{"52}

\newcommand{\IR}{\mathds{R}}
\newcommand{\IC}{\mathds{C}}
\newcommand{\IZ}{\mathds{Z}}
\newcommand{\IQ}{\mathds{Q}}
\newcommand{\IN}{\mathds{N}}
\newcommand{\IH}{\mathds{H}}

\newcommand{\IP}{\mathds{P}}
\newcommand{\IF}{\mathds{F}}
\newcommand{\IL}{\mathds{L}}
\newcommand{\Tr}{\mbox{Tr}}

\def\bea{\begin{eqnarray}}
\def\eea{\end{eqnarray}}
\def\be{\begin{equation}}
\def\ee{\end{equation}}
\def\ba{\begin{align}}
\def\ea{\end{align}}
\def\bse{\begin{subequations}}
\def\ese{\end{subequations}}

\def\CY{\mathfrak{Y}}

\DeclareMathOperator{\sign}{sign}

\DeclareMathOperator{\rk}{rk}

\DeclareMathOperator{\Li}{Li}
\DeclareMathOperator{\Hom}{Hom}

\DeclareMathOperator{\Rep}{Rep}
\DeclareMathOperator{\Ext}{Ext}
\DeclareMathOperator{\Stab}{Stab}

\DeclareMathOperator{\Coh}{Coh}
\DeclareMathOperator{\Cone}{Cone}
\DeclareMathOperator{\virdim}{virdim}
\DeclareMathOperator{\ST}{ST}
\DeclareMathOperator{\Aut}{Aut}
\DeclareMathOperator{\JH}{JH}
\newcommand{\GLt}{\widetilde{GL^+}(2,\IR)}
\newcommand{\DCohK}{$D^b\Coh(K_{\IP^2})$\,}

\makeatletter
\DeclareRobustCommand{\ch}{\@ifnextchar _{\textnormal{ch}}{\operatorname{ch}}}
\makeatother

\def\bOm{\bar\Omega}

\newcommand{\deltaLP}{\delta_{\rm LP}}
\newcommand{\deltaLPh}{\hat\delta_{\rm LP}}

\def\psicr#1{\psi^{\rm cr}_{#1}}

\newcommand{\mathematica}[4]{\vspace{0.35cm}\noindent\boxed{\begin{minipage}{#2\textwidth}\begin{tabular}{lp{15cm}}{\color{paper_blue}{\scriptsize{\tt In[#1]:}}\raisebox{-0.65pt}{{\scriptsize{\tt=}}}}&{\tt #3}\\{\color{paper_blue}{\scriptsize {\tt Out[#1]:}}\raisebox{-0.65pt}{{\scriptsize{\tt=}}}}&{\tt #4}\end{tabular}\end{minipage}}\vspace{0.35cm}}

\definecolor{varcolor}{rgb}{0.1,0.55,0.25}
\definecolor{functioncolor}{rgb}{0.1,0.35,0.75}
\definecolor{paper_blue}{rgb}{0.3,0.2,0.75}
\definecolor{paper_red}{rgb}{0.65,0.1,0.15}
\definecolor{paper_green}{rgb}{0.05,0.35,0.125}
\definecolor{paper_grey}{gray}{0.375}
\definecolor{perm}{rgb}{0.1,0.45,0.85}
\definecolor{deemph}{rgb}{0.7,0.7,0.7}
\newcommand{\fun}[1]{{\color{functioncolor}{\tt #1}}}

\usepackage{mathtools}
\newcommand{\mathtikz}[2][]{\begin{tikzpicture}[baseline=\the\dimexpr-\fontdimen22\textfont2\relax,#1]#2\end{tikzpicture}}

\newcommand{\Rgeo}{\cR^{\rm geo}}
\newcommand{\Reff}{\cR^{\rm eff}}
\newcommand{\Ract}{\cR}
\newcommand{\Uact}{\cU}
\newcommand{\Ueff}{\cU^{\rm eff}}
\newcommand{\AF}{{\rm AF}}
\newcommand{\abs}[1]{\lvert #1\rvert}

\ifdefined\pdfsuppresswarningpagegroup\pdfsuppresswarningpagegroup=1 \fi

\title{BPS Dendroscopy on Local \texorpdfstring{$\IP^2$}{P2}}

\author{Pierrick Bousseau}
\address{Department of Mathematics, University of Georgia, Athens GA, 30602, USA}
\email{Pierrick.Bousseau@uga.edu}

\author{Pierre Descombes}
\address{Sorbonne Universit\'e, CNRS, 
Laboratoire de Physique Th\'eorique et Hautes Energies (LPTHE, 
UMR 7589)
Campus Pierre et Marie Curie, 4 place Jussieu, F-75005, Paris, France}
\email{descombes,blefloch,pioline@lpthe.jussieu.fr}

\author{Bruno Le Floch}

\author{Boris Pioline}

\date{January 17, 2024}

\begin{document}

\begin{abstract}
The spectrum of BPS states in type IIA string theory compactified on a Calabi-Yau threefold
famously jumps across codimension-one walls in complexified K\"ahler moduli space, leading to an intricate chamber structure. The  Split Attractor Flow Conjecture posits that the BPS index $\Omega_z(\gamma)$ for given charge $\gamma$ and moduli $z$ can be reconstructed from the attractor indices $\Omstar(\gamma_i)$ counting BPS states of charge $\gamma_i$ in their respective attractor chamber, by summing over  a finite set of decorated rooted flow trees known as attractor flow trees. If correct, this provides a classification (or dendroscopy) of the BPS spectrum into different topologies of nested BPS bound states, each having a simple chamber structure. 
Here we investigate this conjecture for the simplest, albeit non-compact, Calabi-Yau threefold, namely the canonical bundle over $\IP^2$. Since the K\"ahler moduli space has complex dimension one and the attractor flow preserves the argument of the central charge, attractor flow trees coincide with scattering sequences of rays in a two-dimensional slice of the scattering diagram $\cD_\psi$ in the space of stability conditions on the derived category of compactly supported coherent sheaves on $K_{\IP^2}$. We combine previous results on the scattering diagram of $K_{\IP^2}$ in the large volume slice with an analysis of the scattering diagram for the 
three-node quiver valid in the vicinity of the orbifold point $\IC^3/\IZ_3$, and prove that the Split Attractor Flow Conjecture holds true on the physical slice of $\Pi$-stability conditions. 
In particular, while there is an infinite set of initial rays related by the group $\Gamma_1(3)$ of auto-equivalences, only a finite number of possible decompositions $\gamma=\sum_i \gamma_i$ contribute to the index $\Omega_z(\gamma)$ for any $\gamma$ and $z$, with constituents $\gamma_i$ related by spectral flow to the fractional branes at the orbifold point. 
We further explain the absence of jumps in the index between the orbifold and large volume points for normalized torsion free sheaves, and uncover new `fake walls' across which the dendroscopic structure changes but the total index remains constant. 
\end{abstract}

\maketitle
\tableofcontents

\section{Introduction and summary}

Determining the spectrum of BPS states at generic points in the moduli space in string theory models with $\cN=2$ supersymmetry in four dimensions is an important problem, with far reaching implications both for physics and mathematics. On the physics side, it challenges our understanding of black holes at the microscopic level; on the mathematics side, it connects to deep questions in algebraic and symplectic geometry.

 In the context of type IIA string theory compactified on a Calabi-Yau (CY) threefold $\CY$, BPS states correspond to objects in the bounded derived category of coherent sheaves $\cC=D^b\Coh(\CY)$ 
which are stable for a particular Bridgeland stability condition determined by the complexified K\"ahler 
moduli, known as $\Pi$-stability  \cite{MR1403918,Douglas:2000gi,Douglas:2000ah,MR2373143}. For a general compact CY threefold, the construction of the space of Bridgeland stability conditions  
$\Stab(\cC)$  is a difficult mathematical problem, and the identification of the submanifold $\Pi \subset \Stab(\CY)$ corresponding to $\Pi$-stability depends on the  symplectic geometry of $\CY$ (namely, its genus-zero Gromov-Witten invariants). 
For fixed charge $\gamma\in K(\CY)$ and central charge\footnote{As we recall in \S\ref{sec_Brid}, 
a stability condition
$\sigma=(Z,\cA)$ on $\cC$ also involves a choice of Abelian subcategory $\cA\subset\cC$ (the heart).
We omit it here for brevity since it is locally determined by the central charge $Z$.} 
$Z$
(determined by the K\"ahler moduli), stable objects are counted by the generalized 
Donaldson-Thomas (DT) invariant $\Omega_Z(\gamma)$ \cite{thomas1998holomorphic,ks,Joyce:2008pc}. The latter, being integer valued, is locally constant but discontinuous across real-codimension one walls in $\Stab(\CY)$  (hence also on $\Pi$), due to the (dis)appearance of destabilizing sub-objects, leading to an intricate chamber structure. While the jump is determined in terms of the invariants on one side of the wall by a universal wall-crossing formula~\cite{ks,Joyce:2008pc}, it is desirable to develop a global understanding of the BPS spectrum which allows to identify stable objects  at any point $z$ in $\Pi\subset  \Stab(\CY)$.

\subsection{The Split Attractor Flow Conjecture}

The physical picture of BPS states as multi-centered black holes suggests one way to achieve this goal, namely to decompose stable BPS states of charge $\gamma$ into bound states of 
elementary constituents of charge $\gamma_i$, with a hierarchical structure determined by 
attractor flow trees~\cite{Denef:2000nb,Denef:2001xn}. As we review in more detail
in \S\ref{sec_SAF}, the latter are rooted trees  decorated
with charges $\gamma_e$ along the edges, 
embedded in $\Pi\subset  \Stab(\CY)$  such that the root vertex is mapped to the desired point $z\in\Pi$,
edges follow the gradient flow (also known as attractor flow \cite{Ferrara:1995ih}) for the modulus of the central charge $|Z(\gamma_e)|^2$ along the slice $\Pi$, and split at vertices on
walls of marginal stability where the central charges of the incoming and descending charges
become aligned. The aforementioned constituents $\gamma_i$ arise as the end points (or leaves)
of the tree, where the central charge is attracted to a local minimum of  $|Z(\gamma_i)|$ along $\Pi$,
or to a conifold point where $Z(\gamma_i)=0$. We denote by $Z_{\gamma_i}(\gamma)$ the central charge at this local minimum, and by 
$\Omstar(\gamma_i)\coloneqq\Omega_{Z_{\gamma_i}}(\gamma_i)$ the corresponding value of
the DT invariant, known as attractor index. 

The Split Attractor Flow Conjecture (SAFC), 
originally proposed in \cite{Denef:2000nb,Denef:2001xn} and sharpened in \cite{Denef:2007vg}, 
posits that for any $\gamma\in \Gamma$ and $z\in\Pi$,  the BPS index $\Omega_Z(\gamma)$ can be computed by summing over 
a finite number of such attractor flow trees, weighted by the
 product of attractor indices  $\Omstar(\gamma_i)$ and by some combinatorial factor
 obtained by applying the wall-crossing formula at each vertex.\footnote{The original
 formulation of the conjecture 
 relied on the primitive wall-crossing formula and 
  overlooked issues arising  when some of the constituents
 carry non-primitive or identical charges. In \S\ref{sec_SAF},
 using insights from \cite{Denef:2007vg,Manschot:2010xp,Manschot:2010qz,Alexandrov:2018iao} 
 we give a more precise
 version of the conjecture in terms of the rational DT invariants 
$\bOm(\gamma)$ defined in \eqref{defOmb}.}
If correct, this picture
provides a categorization (which we like to call dendroscopy) 
of the BPS spectrum at $z\in \Pi$ into different types, each having
a simple region of stability delimited by the first splitting
at the root of the tree,
and reduces the determination of the BPS spectrum to the computation of the attractor invariants 
$\Omstar(\gamma_i)$. Unfortunately, for a compact CY threefold $\CY$,  the computation of these invariants seems very difficult and  the Split Attractor Flow Conjecture is still wide open, 
despite some encouraging results \cite{Denef:2007vg,Collinucci:2008ht,VanHerck:2009ww,Manschot:2010xp,Kontsevich:2013rda,Gaddam:2016xum,Alexandrov:2022pgd}.

The problem however becomes more tractable for certain non-compact CY threefolds, such that the category  $D^b\Coh(\CY)$ is isomorphic to the derived category $D^b\Rep(Q,W)$ of representations  of a certain quiver with potential $(Q,W)$. In particular, in the vicinity of an orbifold point where the central charges associated to the nodes of the quiver all lie in a common half-plane, the heart of the stability condition reduces to the category of quiver representations (or some tilt of it) and 
the notion of  attractor index has a simple definition using
King stability for the (suitably perturbed) `self-stability' parameter 
$\theta_\star(\gamma)=\langle -, \gamma \rangle$ \cite{Mozgovoy:2020has}.  In that context, the enumeration of attractor flow trees becomes straightforward and precise versions of the Split Attractor Flow Conjecture
have been 
proposed \cite{Alexandrov:2018iao,Mozgovoy:2020has} and then established rigorously \cite{Mozgovoy:2021iwz,Arguz:2021zpx} using the mathematical framework of operads and scattering diagrams, respectively. As already anticipated in \cite{Kontsevich:2013rda} and as
will become apparent shortly, scattering diagrams turn out to 
be the mathematical incarnation of split attractor flows (at least for non-compact CY threefolds), while the physical interpretation of the trees in the operadic approach of \cite{Mozgovoy:2021iwz} remains obscure at present. 

 \subsection{The Attractor Conjecture}
 
As for  the attractor indices 
which enter these formulae, 
 it was conjectured in  \cite{Beaujard:2020sgs},
 that for quivers $(Q,W)$ associated to a non-compact CY threefolds of the form $\CY=K_{S}$ (namely, the total space of the canonical bundle over a Fano surface $S$), 
 the attractor invariants $\Omstar(\gamma)$ take a very simple form: 
$\Omstar(\gamma)=0$ except when $\gamma$  corresponds to a dimension vector supported on one node of the quiver (in which case  $\Omstar(\gamma)=1$), or $\gamma=k\delta$ with $k\geq 1$ and $\delta$ the charge vector for the skyscraper sheaf (in which case $\Omstar(k\delta)=-\chi_{\CY}$, the Euler number of $\CY$). This Attractor Conjecture (AC) was arrived at by comparing the quiver indices with 
the counting of Gieseker-semi-stable sheaves on $S$, and supported by an analysis of the expected
dimension of the moduli space of quiver representations in the self-stability chamber. 
Further evidence and an extension of AC to all toric CY three-folds was presented in \cite{Mozgovoy:2020has,Descombes:2021snc}. In this work, we focus on the simplest case $\CY=K_{\IP^2}$ (also known as local $\IP^2$), which is a crepant resolution of the orbifold singularity $\IC^3/\IZ_3$
and whose derived category of (compactly supported) coherent sheaves is isomorphic to the
derived category of a three-node quiver $(Q,W)$ shown in Figure \ref{figQuiver}. 
Combining ideas  from  \cite{Beaujard:2020sgs,Descombes:2021egc}, we prove
Theorem~\ref{thm:attractor-conj}, which states
that the Attractor Conjecture holds for this quiver, thereby 
providing the attractor invariants relevant in the vicinity of the orbifold point.

Our main 
goal in this work is to extend this picture away from the orbifold point, and connect it to 
the scattering diagram for the derived category of sheaves on $\IP^2$ constructed
by one of the authors in \cite{Bousseau:2019ift}. 
Before presenting our results in more detail however, we need to pause and explain the 
relation between scattering diagrams and flow trees. 

\subsection{Scattering diagrams and attractor flow trees}

Scattering diagrams were first introduced in the context of the Strominger-Yau-Zaslow approach to mirror symmetry \cite{kontsevich2006affine,gross2011real}, and applied to DT invariants of  quivers with potential in \cite{bridgeland2016scattering}. More generally, as we explain in \S\ref{sec_Scattdiag}, for a given triangulated category $\cC$ and phase\footnote{The scattering diagram $\cD_\psi$
is invariant under $(\psi,\gamma,Z)\mapsto(\psi+\pi,-\gamma,Z)$ and $(-\psi,\gamma^\vee,Z^\vee)$
where $\gamma^\vee$ is the image of $\gamma$ under derived duality, and $Z^\vee(\gamma):=-\overline{Z(\gamma^\vee)}$. For most of this work
we restrict to the  interval $(-\frac{\pi}{2},\frac{\pi}{2}]$.} 
$\psi\in\IR/2\pi\IZ$,
the scattering diagram $\cD_\psi$ is supported on a set of real-codimension one loci (or active rays) 
$\Ract_\psi(\gamma)$ in the space of stability conditions $\Stab\cC$ where
the central charge has fixed argument $\arg Z(\gamma)=\psi+\frac{\pi}{2}$ and
supports semi-stable objects, in the sense that the rational DT index 
\be
\label{defOmb}
\bOm_Z(\gamma) \coloneqq \sum_{m|\gamma} \frac{y-y^{-1}}{m(y^m-y^{-m})} \Omega_Z(\gamma/m)\vert_{y\to y^m}
\ee
is non-zero.
Each point along $\Ract_\psi(\gamma)$ is 
equipped with an automorphism 
\be
\Uact_Z(\gamma)=\exp\bigl(\bOm_Z(\gamma)  \cX_\gamma/(y^{-1}-y)\bigr)
\ee
of the quantum torus  algebra spanned by formal variables $\cX_\gamma$ satisfying 
$\cX_\gamma\, \cX_{\gamma'} =(-y)^{\langle \gamma,\gamma'\rangle} \cX_{\gamma+\gamma'}$.
The set $\cD_\psi$ of all active rays $\Ract_\psi(\gamma)$ equipped with 
$\Uact_Z(\gamma)$ then forms a consistent scattering diagram, which informally means that the product of the 
automorphisms $\Uact_Z(\gamma)$ around each codimension-two intersection must equal one. This property uniquely
specifies the  invariants $\bOm_Z(\gamma)$ on  outgoing\footnote{We postpone the definition of incoming and outgoing rays to \S\ref{sec_Scattdiag}.
 For the present discussion, it suffices to orient the restriction of the rays along a transverse
 plane in the vicinity of a codimension-two intersection, according to the gradient of the central
 charge $|Z(\gamma)|$.
 }
 rays in terms of those on incoming rays. 
In the context of quivers with potential, one can further restrict the scattering diagram from
the space $(\IH_B)^{n}$ of Bridgeland stability conditions (where $\IH_B$ is the upper
half-plane $\{ z\in \IC: \Im z> 0 \ \mbox{or}\  (\Im z=0 \ \mbox{and}\ \Re z<0) \}$ and
$n$ denotes the number of nodes in the quiver)
to the 
space $\IR^n$ of  King stability conditions, 
such that $\cD_\psi$ becomes a complex of convex rational polyhedral cones \cite{bridgeland2016scattering}. 

In order to understand the relation between scattering diagrams and attractor flow trees, the key observation (elaborated upon in \S\ref{sec_flow}) is that for a local CY threefold, the central charge
$Z(\gamma)$ is a holomorphic function of complexified K\"ahler moduli $z\in \Pi$, which implies that 
\begin{enumerate} 
\item The argument of $Z(\gamma)$ is constant along the gradient flow of $|Z(\gamma)|$
\item Local minima of $|Z(\gamma)|$ can only occur on the boundary of $\Pi$ or at points $z\in\Pi$ 
where  $|Z(\gamma)|=0$
\end{enumerate}
When $\Pi$ has complex dimension 1 (as is the case for $K_{\IP^2}$), 
the first observation implies that lines of gradient flow of 
$|Z(\gamma)|$ must lie along active rays $\Ract_\psi(\gamma)$, for a suitable value of $\psi$ determined by the initial value of the flow. Since vertices in the attractor flow tree have to 
lie on walls of marginal stability where the central charges of the parent edge $Z(\gamma_v)$ 
and descendant edges $Z(\gamma_e), e\in {\rm ch}(v)$ become aligned, they must also 
 must lie at the intersection of the corresponding rays $\Ract_\psi(\gamma_v)$ and
 $\Ract_\psi(\gamma_e), e\in {\rm ch}(v)$. 
Since stable BPS states of charge $\gamma$ are ruled out at stability conditions where their central charge vanishes  (a consequence of the support property for stability conditions), 
the second observation shows that the attractor points can only occur at the boundary of $\Pi$, corresponding to the initial rays of the scattering diagram. 
Starting from the leaves and going up towards the root, 
one can therefore view a split attractor flow as a sequence of scatterings 
of a set of initial 
rays $\Ract_\psi(\gamma_i)$, 
such that the final ray carries the desired charge $\gamma=\sum \gamma_i$
and passes through the desired point $z\in \Pi$ in the space of $\Pi$-stability conditions. 
When $\dim_\IC \Pi>1$, the connection between attractor flow trees and scattering diagrams is less
direct, since the edges are real-dimension one trajectories embedded in real-codimension one rays.
Nonetheless, in the vicinity of real-codimension two loci where active rays intersect, one can always take a two-dimensional transverse section such that the previous picture applies.

\subsection{The physical slice of \texorpdfstring{$\Pi$}{Pi}-stability conditions}

\begin{figure}[ht]
\begin{center}
\includegraphics[width=8cm]{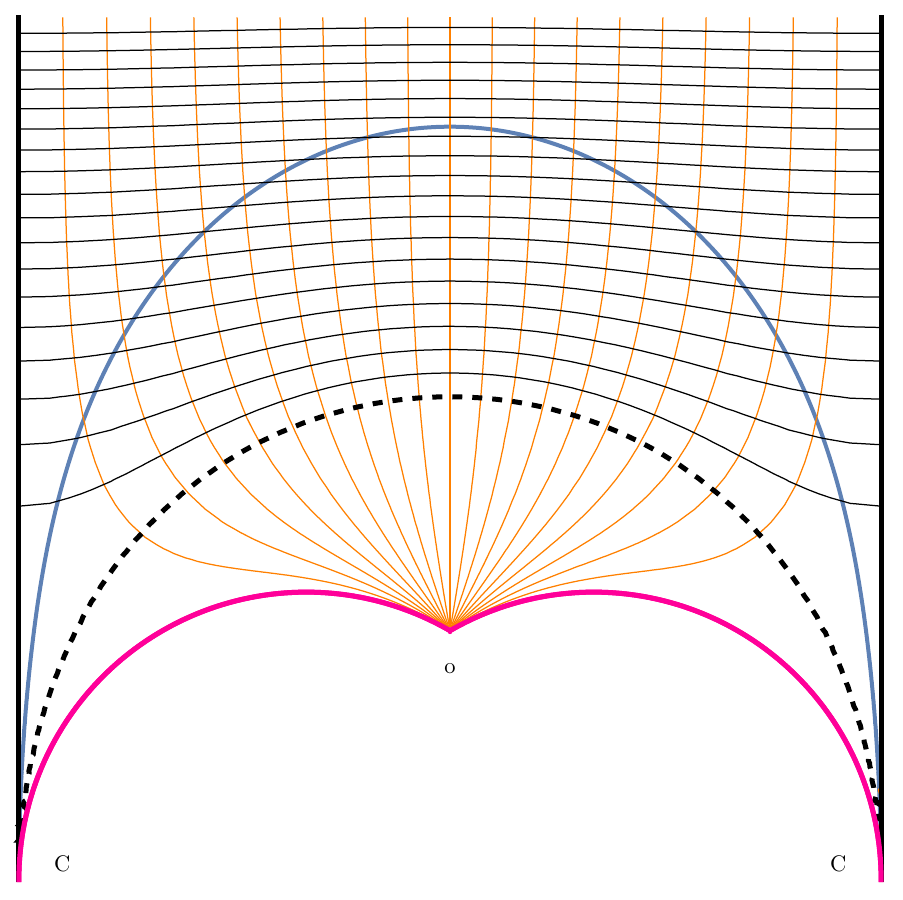}
\end{center}
\caption{The fundamental domain $\cF_o$ \eqref{defFo} 
is the region between the two vertical lines and above the two red arcs, centered around 
the orbifold point $\tau_o=\frac{1}{\sqrt3} e^{5\pi\I/6}=-\frac12+\frac{\I}{2\sqrt3}$. 
The domain $\cF_C$ \eqref{defFC} consists of
the right half of $\cF_o$, and the $\tau\mapsto\tau+1$ translate of the
left half, hence is centered around the conifold point $\tau=0$. 
The orange and black thin lines
correspond to the contours of constant $s$ and $t=\sqrt{2w-s^2}$, respectively, where 
$(s,w)$ are defined in \eqref{defsw}.
 In the region $\IH^{\rm LV}$ above the dashed line (corresponding to $t>0$), the exact central charge \eqref{defZ} along the  $\Pi$-stability slice  is related to the large volume central charge \eqref{defZLV} by a $\GLt$ action. In the region  $\IH^o$ below the blue line (corresponding to $w+\frac12 s<0$), but still in the fundamental domain $\cF_o$, it is instead related to the 
quiver stability condition.
\label{figfundO}}
\end{figure}

Returning to the special case of the local projective plane, the space
$\Stab(K_{\IP^2})$ of stability conditions on  $D^b \Coh(K_{\IP^2})$ was analyzed in detail in \cite{bridgeland2006stability,Bayer:2009brq}. Using the $\IC^\times$ subgroup of the $\GLt$ action 
on the real and imaginary part of the 
central charge function, there is no loss of generality in assuming that $Z(\delta)=1$, where $\delta=[0,0,-1]$ is the Chern vector of the (anti)D0-brane, corresponding to the skyscraper sheaf $\cO_x[1]$ (the homological shift $[1]$ is for later convenience). The central charge
for a general Chern vector $\gamma=[r,d,\ch_2]$ therefore takes the form 
\be
\label{defZ}
Z(\gamma) = - r T_D + d T -\ch_2
\ee
where $r$ is the rank (or D4-brane charge), $d$ the first Chern class (or D2-brane charge) and
$\ch_2$ the second Chern class (or D0-brane charge), and  $(T,T_D)\in \IC^2$ parametrize the quotient
$\Stab(K_{\IP^2})/\IC^\times$. Mirror symmetry
selects a particular complex one-dimensional slice 
\be
\label{defZPi}
Z_\tau(\gamma) = - r T_D(\tau) + d T(\tau) -\ch_2
\ee
parametrized by $\tau$ in the Poincar\'e upper half-plane $\IH$,
such that 
$T$ and $T_D$ are given by periods on a family of elliptic curves with $\Gamma_1(3)$ level structure. The modulus
$\tau$ of the elliptic curve parametrizes the universal cover of the 
modular curve $X_1(3)=\IH/\Gamma_1(3)$, 
and $\Gamma_1(3)$ is the index 4 congruence subgroup of $SL(2,\IZ)$ generated by $T\colon\tau\mapsto \tau+1$ and $V\colon\tau\mapsto \tau/(1-3\tau)$ such that $(VT)^3=1$.
The modular curve has two cusps at images of $\tau=\I\infty$ and $\tau=0$, corresponding
to the large volume and conifold points, respectively, and one elliptic point of order 3 at images of
$\tau_o=\frac{1}{\sqrt3} e^{5\pi\I/6}=-\frac12+\frac{\I}{2\sqrt3}$, corresponding to the orbifold point 
$\IC^3/\IZ_3$. A fundamental domain $\cF_o$ \eqref{defFo} centered around the orbifold point $\tau_o$ is shown in Figure \ref{figfundO}. Due to monodromies around these singular points, the periods $T(\tau),T_D(\tau)$
are not modular functions of $\Gamma_1(3)$. Rather, we show in Appendix \ref{sec_Eichler} 
that they are given by Eichler-type integrals 
\be
\label{Eichler0}
 \begin{pmatrix} T \\ T_D  \end{pmatrix}
= \begin{pmatrix} -\frac12 \\ \frac13  \end{pmatrix} 
+  \int_{\tau_o}^{\tau} \begin{pmatrix} 1 \\u \end{pmatrix} \, 
\, C(u) \de u
\ee
where  $C (\tau) = \frac{\eta(\tau)^9}{\eta(3\tau)^3}$ is a weight 3 Eisenstein series for $\Gamma_1(3)$, which has neither poles nor zeros in the Poincar\'e upper half plane. This representation will play a central role in this work, as it gives a global and  numerically efficient\footnote{This formula is implemented in the Mathematica package {\tt P2Scattering.m} along with many other routines for plotting scattering diagrams, scanning possible flow trees, etc, see Appendix \ref{sec_mathematica} for details.}  formula 
for the analytic continuation of $Z_\tau(\gamma)$ to 
the universal cover $\IH$ of  the complexified K\"ahler moduli space $X_1(3)$. Near the large
volume limit $\Im\tau\gg 1$, one finds that the central charge function reduces to a quadratic polynomial,
\be
\label{defZLV}
Z_{(s,t)}^{\rm LV}(\gamma) \coloneqq - \frac{r}2  (s+\I t)^2+ d (s+\I t) -\ch_2
\ee  
with $ \tau\simeq s+\I t$. 
In fact, observing that the variables $(s,w)\in\IR^2$ defined by
\be
\label{defsw}
s\coloneqq\frac{\Im T_D}{\Im T}\ ,\quad 
w\coloneqq-\Re T_D + \frac{\Im T_D}{\Im T} \Re T=-\frac{\Im (T \bar T_D)}{\Im T}\ , \quad 
\ee
are invariant under  the action of $\GLt$ on $\Stab(K_{\IP^2})$ (after fixing $Z(\delta)=1$),
one easily checks that in the domain $\IH^{\rm LV}$ defined 
by the condition $w > \frac12 s^2$ (keeping only the connected component containing
the cusp at $\tau=\I\infty$), the central charge charge function \eqref{defZ} can be brought to the
large volume form \eqref{defZLV} with $t=\sqrt{2w-s^2}$.  As shown in  Figure \ref{figfundO},
the domain $\IH^{\rm LV}$ only covers a proper subset of the fundamental 
domain $\cF_o$, in particular it does not include a neighborhood of the orbifold point.
 
As explained in \cite{Bayer:2009brq} and reviewed in \S\ref{sec_Pistab} below, for any point $\tau\in \IH$ there exists a stability condition on \DCohK with central charge function given by the mirror symmetry prescription \eqref{defZ}. In the fundamental domain $\cF_o$ 
and in its translates,  the
heart $\cA(\tau)$ is constructed using the usual tilting pair construction 
(built from the subcategories of sheaves with slope $\mu=\frac{d}{r}$  less or greater than $s=\frac{\Im T_D}{\Im T}$). This construction is then extended to the full Poincar\'e upper half-plane using
the group  $\Gamma_1(3)$ of auto-equivalences of the derived category \DCohK generated by tensor product with $\cO_{\CY}(1)$ (corresponding to $T\colon\tau\mapsto \tau+1$) and by the spherical twist
$\ST_{\cO}$ 
with respect to the structure sheaf $\cO$ of the zero section (corresponding to 
$V\colon\tau\mapsto \tau/(1-3\tau)$). The interior of the 
fundamental domain $\cF_C$ 
and its translates are singled out by the condition that the stability condition is geometric, 
\ie the skyscraper sheaves $\cO_x$ are stable with fixed phase.
\looseness=-1

\subsection{The scattering diagram at large volume}

\begin{figure}[ht]
\begin{center}
\includegraphics[width=12cm]{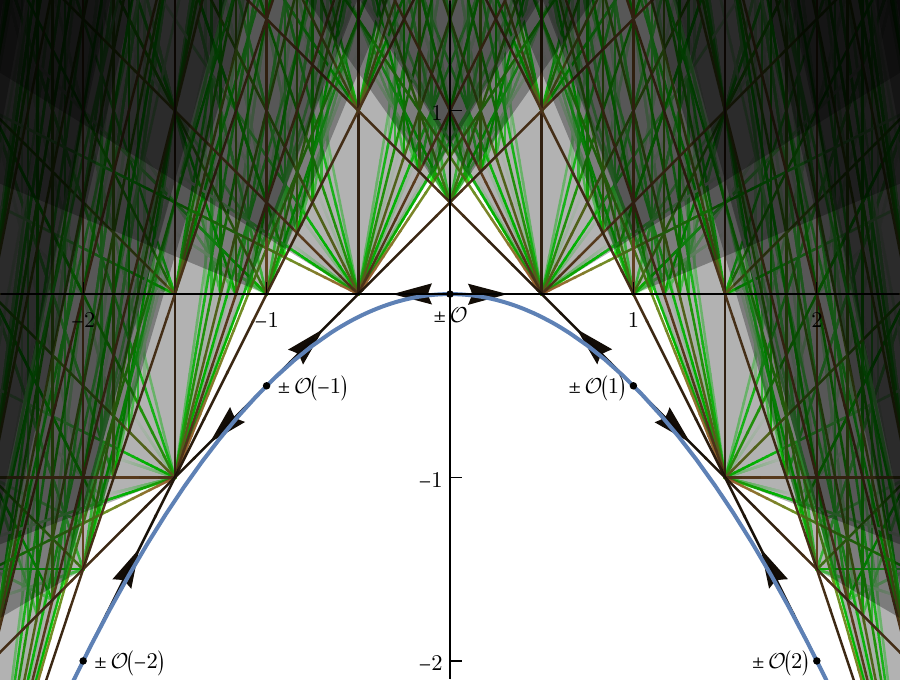}
\end{center}
\caption{Large volume scattering diagram in $(x,y)$ plane. The initial rays $\Ract_0(\cO(m))$
and $\Ract_0(\cO(m)[1])$ are tangent  to the parabola $y=-\frac12 x^2$ at $(x,y)=(m,-\frac12 m^2)$ and move away from this point leftward and rightward, respectively. 
The gray areas indicate regions containing a dense set of rays.
Ray colors encode the electric potential $2(d-rx)$, with lighter colors corresponding to larger values.
\label{figDLVxy}}
\end{figure}

 In \cite{Bousseau:2019ift},  the stability scattering diagram for $\cC=D^b\Coh_c(\IP^2)$ 
 was constructed for the one-parameter family of stability conditions of the form \eqref{defZLV}, 
 assuming the special value  $\psi=0$ for the phase. The construction was performed 
using a different set of coordinates $(x,y)=(s,\frac12(t^2-s^2))$ with $y>-\frac12 x^2$, such that the rays $\Ract_0(\gamma)$
become segments of 
 straight lines $ry+dx-\ch_2=0$, similar to standard
 affine  scattering diagrams
in the context of mirror symmetry \cite{kontsevich2006affine,gross2011real}.
The main result of this analysis is that the initial rays of the resulting diagram,
which we denote by 
$\cD^{\rm LV}_0$,   
consist of a  pair of rays $\Ract_0(\cO(m))$ and $\Ract_0(\cO(m)[1])$
emitted from every integer points $(x,y)=(m,-\frac12 m^2)\in \IZ$ tangent to the parabola $y=-\frac12 x^2$, 
where the central charge of the coherent sheaf $\cO(m)$ vanishes\footnote{Indeed,
 $Z^{\rm LV}_{(s,t)}(\gamma)=-\frac12(s+\I t-m)^2$ for $\gamma=[1,m,\frac12 m^2]$.},
 see Figure~\ref{figDLVxy}. 
 
 In \S\ref{sec_LV}, we recast this construction in $(s,t)$ coordinates, explain it in more physical terms and demonstrate its usefulness for computing the BPS indices. 
  In particular, we observe that in these coordinates, the ray $\Ract_0(\gamma)$ 
for  $\gamma=[r,d,\ch_2]$ is  either included in a vertical straight line (when $r=0$), in a branch of hyperbola asymptoting to  the `light-cone' $|t-s|={\rm cst}$ (when $r\neq 0$ and $\Delta:=\frac12 d^2-\frac{\ch_2}{r}\neq 0$)
 or in a branch of the said light-cone  (when $r\neq 0$ and $\Delta=0$).  Thus, a useful 
 analogy is to view the ray $\Ract(\gamma)$ as the worldline of a particle of global charge 
 $\gamma$  and electric charge $r$, propagating in the two-dimensional (half) Minkowski space spanned by the space and time coordinates $(s,t)$, immersed in a constant electric field. The objects 
 $\cO(m)$ and $\cO(m)[1]$ correspond physically to a D4-brane with $m$ units of flux, and its anti-particle,
 carrying electric  charge $\pm1$. 
 These objects  are pair-produced at $t=0$ and $s\in \IZ$, scatter against each other 
 and produce  an infinite set of outgoing rays, that in turn collide {\it ad infinitum}, producing the complete
 set of BPS states at large volume (i.e. late time). The resulting diagram is shown in Figure~\ref{figDLV}.
 The attractor flow trees can be thought of as sequences of scatterings producing a  particle of desired charge $\gamma$ and going through a desired point $(s,t)$ in Minkowski spacetime. 
 
 \begin{figure}[ht]
\begin{center}
\includegraphics[width=14cm]{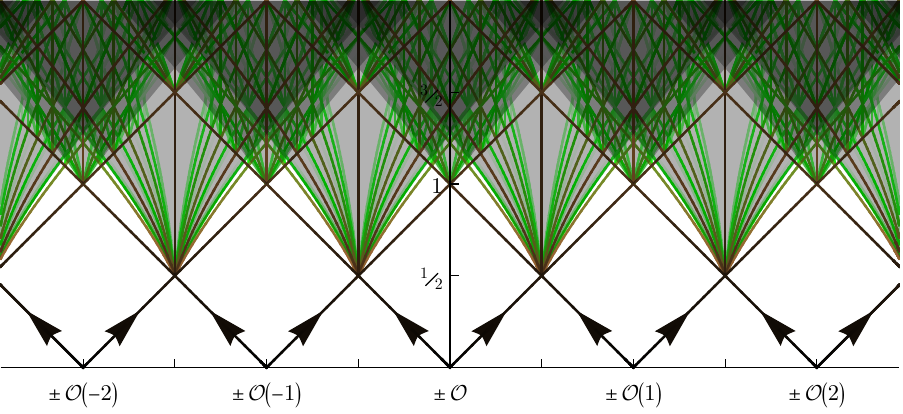}
\end{center}
\caption{Large volume scattering diagram in $(s,t)$ plane. The initial rays $\Ract_0(\cO(m))$
and $\Ract_0(\cO(m)[1])$ are emitted at $(s,t)=(m,0)$ and moving leftward and rightward, respectively. 
\label{figDLV}}
\end{figure}

 While this electromagnetic analogy has some peculiarities, e.g. the fact that pair production only takes place at $t=0$ and integer spatial positions (unlike standard  Schwinger pair production), it does provide valuable insight.  In particular, it makes it obvious that rays can only propagate inside the forward light-cone  (a property which we refer to as {\it causality})\footnote{In $(x,y)$ coordinates, rays are contained in a cone tangent to the parabola $y=-\frac12x^2$.}, and that the (conveniently normalized) 
 electric potential  $\varphi_s(\gamma)=2(d -r s)$ can only increase along a trajectory (a property
 which played an important role in the construction of \cite{Bousseau:2019ift})
 In \S\ref{sec_inidata}, we combine these two properties to derive a bound on the number and charges of possible constituents $\cO(m_i)$ and $\cO(m'_j)[1]$ that contribute to the index  $\Omega_{(s,t)}(\gamma)$
 at any point $(s,t)$ such that $\Re[Z^{\rm LV}_{(s,t)}(\gamma)]=0$. This 
 bound  shows that the SAFC holds along the large volume slice for trees rooted on such loci,
 and gives an effective (if not particularly efficient) algorithm for determining the finite list of 
 attractor flow trees (or scattering sequences) contributing to the index 
$\Omega_{(s,t)}(\gamma)$.

Using this algorithm, we reproduce the 
well known chamber structure for DT invariants along the large volume slice
~\cite{arcara2013minimal,bertram2014birational,maciocia2014computing},
consisting of a finite nested sequence of walls of marginal stability,
 such that 
the index
vanishes inside the innermost wall and is equal to the  index $\Omega_\infty(\infty)$ counting Gieseker-semistable sheaves outside the outermost wall.  
For illustration, in \S\ref{sec_Hilb} and \S\ref{sec_D2D0} we determine the trees contributing to $\Omega_{\infty}(\gamma)$
 for $\gamma=[1,0,1-n]$ and $\gamma=[0,d,\ch_2]$ for low values of $n$ and $d$.
 In the first case,  the moduli space of Gieseker semi-stable sheaves coincides with the 
 Hilbert scheme of $n$ points on $\IP^2$, the index of which is well-known
 \cite{Gottsche:1990} (higher rank examples are considered in Appendix \S\ref{sec_higherk}).
  In the second case
 $r=0$,  we recover  the genus zero Gopakumar-Vafa invariants $N_d^{(0)}$  in the unrefined limit\footnote{More generally, the refined Gieseker index  computes the character $\sum_{j_L,j_R} N^{(j_L, j_R)}_d \chi_{j_L}(y_L)  \chi_{j_R}(y_R)$ on the diagonal $y_L=y_R=y$, where $N^{(j_L, j_R)}_d$ are the refined BPS invariants \cite{Katz:1999xq,Choi:2018xgr}.  It is an interesting open question
 to generalize the scattering diagram away from the Nekrasov-Shatashvilii limit $y_L=y_R$.} $y\to 1$.
 We further match the contributing trees with the known stratification of the moduli space of Gieseker stable sheaves, extending the observations in \cite{Bousseau:2019ift}.

Although the choice $\psi=0$ has the advantage (exploited in \cite{Bousseau:2019ift}) 
that the geometric rays $\Rgeo_\psi(\gamma)$ become
 straight lines in $(x,y)$ coordinates, it does not give access
to the index $\Omega_{(s,t)}(\gamma)$  away from loci 
where $Z^{\rm LV}_{(s,t)}(\gamma)$ is purely imaginary. In \S\ref{sec_LVpsi}, we generalize the 
scattering diagram $\cD^{\rm LV}_0$  
to a diagram $\cD^{\rm LV}_\psi$  valid for any
 $\psi\in(-\frac{\pi}{2},\frac{\pi}2)$. 
While the  walls of marginal stability are by construction independent of $\psi$,  it turns out that the $\psi$-dependence of the rays $\Ract_\psi(\gamma)$  can be absorbed by a linear coordinate transformation $(s,t)\mapsto (s+t\tan\psi,t/\cos\psi)$ which preserves the walls and the boundary
at $t=0$. Thus, the topology of the trees contributing to the index at large volume for any charge $\gamma$ is independent of $\psi$, the only change being in the location of the 
vertices along the walls of marginal stability.  As we shall see momentarily, this is no longer true
for the exact scattering diagram involving the exact central charge function \eqref{defZ}.

\subsection{The orbifold scattering diagram}

\begin{figure}[ht]
\begin{center}
\begin{tikzpicture}[inner sep=2mm,scale=2]
  \node (a) at ( -1,0) [circle,draw] {$n_1$};
  \node (b) at ( 0,1.7) [circle,draw] {$n_2$};
  \node (c)  at ( 1,0) [circle,draw] {$n_3$};
 \draw [->>>] (b) to node[auto] {$a_i$} (a);
 \draw [->>>] (c) to node[auto] {$b_j$} (b);
 \draw [->>>] (a) to node[auto] {$c_k$} (c);
\end{tikzpicture}
\end{center}
\caption{Quiver describing the BPS spectrum around the orbifold point $\tau_o$. The potential is 
$W=\sum_{i,j,k} \epsilon_{ijk} \Tr(a_i b_j c_k)$, and the dimension vector 
$(n_1,n_2,n_3)$ is related to the Chern vector $[r,d,\ch_2]$ via 
$(n_1,n_2,n_3) = (-\frac32 d-\ch_2-r,  -\frac12 d-\ch_2, \frac12 d-\ch_2)$. 
\label{figQuiver}}
\end{figure}
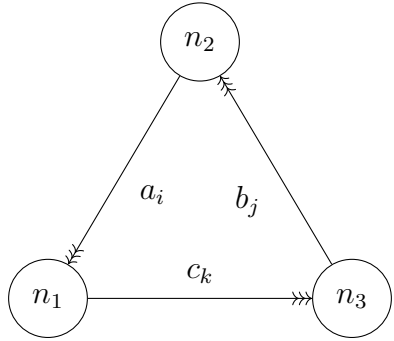

In the vicinity of the orbifold point, the geometry of $K_{\IP^2}$ degenerates into the 
orbifold singularity $\IC^3/\IZ_3$, and the BPS spectrum is instead described by stable objects
in the derived category of representations $D^b\Rep(Q,W)$ of the quiver with potential shown in Figure \ref{figQuiver}, which we refer to as the orbifold quiver. This quiver arises from the tilting sequence in $D^b\Coh_c(K_{\IP^2})$
obtained from  the Ext-exceptional collection
\be
\label{excepcoll}
E_1=i_*(\cO)[-1], \quad E_2=i_*(\Omega(1)), \quad E_3 = i_*(\cO(-1))[1]
\ee
where $\Omega$ is the cotangent bundle of $\IP^2$, $i_*$ denotes the lift from 
$\IP^2$ to $K_{\IP^2}$ and $[k]$ denotes the cohomological shift by $k$ units. 
Importantly, the central charges $Z_\tau(E_i)$  of
the three objects are aligned at the orbifold point $\tau_o$, and they remain in a common half-plane
in an open region $\IH^o$ around  $\tau_o$ 
defined by the inequality $w<-\frac{1}{2}s$ in the
fundamental domain $\cF_o$ \eqref{defFo}, along with the images of that region under the $\IZ_3$ 
symmetry around  $\tau_o$ (see Figure \ref{figfundO}). This ensures that the heart of the 
stability condition coincides with the Abelian category of quiver representations in the
region $\IH^o$, up to a $\GLt$ transformation.  
 
Following Bridgeland \cite{bridgeland2016scattering}, the DT invariants for the quiver $(Q,W)$ 
are determined by a
scattering diagram $\cD_Q$ defined in the affine space $\IR^3$ spanned by King stability (also
known as Fayet-Iliopoulos) parameters $(\theta_1,\theta_2,\theta_3)$. For any dimension vector
$\gamma=(n_1,n_2,n_3)\in \IZ^3$, the active
 ray $\Ract_o(\gamma)$ is defined as the locus where 
$\bOm_\theta(\gamma)\neq 0$ inside the hyperplane  
$\{\theta\in\IR^3: n_1\theta_1+n_2\theta_2+n_3\theta_3=0\}$,
where $\bOm_\theta(\gamma)$ is the rational DT invariant associated to the moduli space of $\theta$-semistable representations of $(Q,W)$ with dimension vector $\gamma$ (in particular, 
$\Ract_o(\gamma)$ is empty unless the $n_i$'s are all positive). In~\S\ref{sec_P2proof}, building on earlier arguments \cite{Beaujard:2020sgs,Mozgovoy:2020has,Descombes:2021egc},
 we prove that 
the only initial rays are those for $\gamma\in\{\gamma_1,\gamma_2,\gamma_3,k\delta\}$ with 
$\gamma_1=(1,0,0),\gamma_2=(0,1,0),\gamma_3=(0,0,1),
\delta=(1,1,1), k\geq 1$.
% , with DT invariants 
% \be
% \label{attindexP2}
% \Omega_\star(k\gamma_i)=\delta_{k,1} \ ,\quad \Omega_\star(k\delta)=-y^3-y-1/y
% \ee

\begin{theorem}[Attractor Conjecture for the $\IC^3/\IZ_3$ orbifold quiver]\label{thm:attractor-conj}
For the quiver with potential $(Q,W)$ 
shown in Figure \ref{figQuiver}, 
the attractor invariant $\Omstar(\gamma)$ vanishes for all dimension vectors $\gamma=(n_1,n_2,n_3)$ except for
\be
\label{attindexP2}
\Omega_\star(k\gamma_i)=\delta_{k,1} \ ,\quad \Omega_\star(k\delta)=-y^3-y-1/y
\ee
\end{theorem}

The complete scattering diagram $\cD_Q$ is then determined from this initial data by consistency using the flow tree formula of \cite{Alexandrov:2018iao,Arguz:2021zpx}. By scaling invariance, the scattering diagram can be
restricted to the hyperplane $\theta_1+\theta_2+\theta_3=1$
with no loss of information, except for the rays $\Ract_o(\delta)$ associated to D0-branes
which are no longer visible. The resulting 
two-dimensional scattering diagram $\cD_o$ is shown in Figure \ref{figMcKayScat},
including only the initial rays and a few secondary rays. 

\begin{figure}[ht]
\begin{center}
\includegraphics[width=10cm]{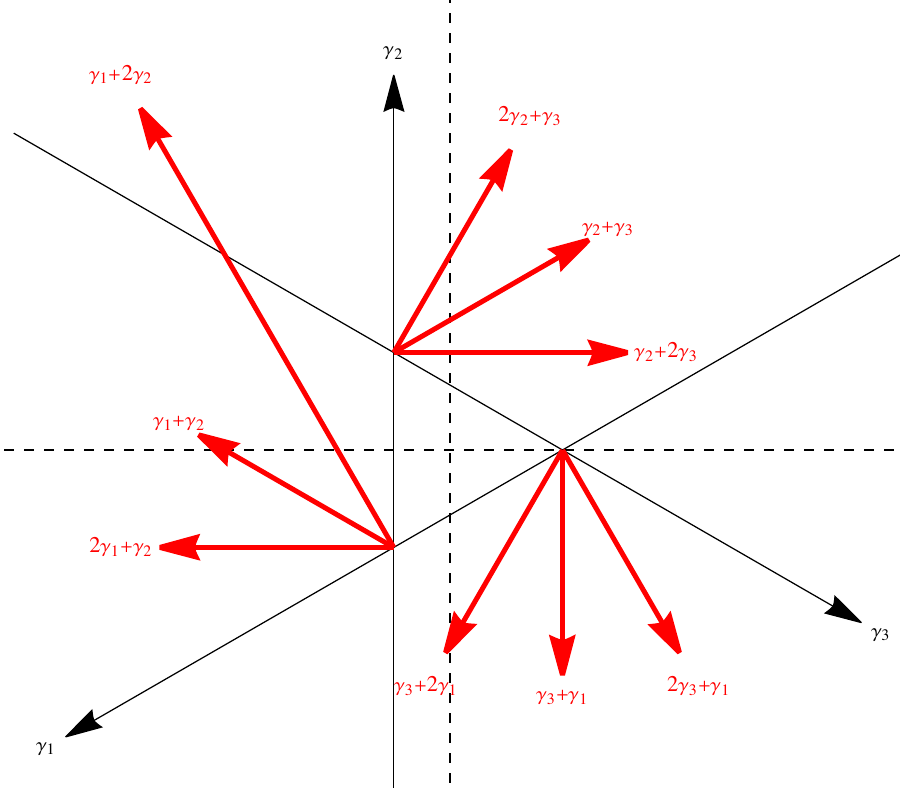}
\end{center}
\caption{Two-dimensional section $\cD_o$ of the orbifold scattering diagram $\cD_Q$ along the hyperplane 
$\theta_1+\theta_2+\theta_3=1$. 
The initial rays associated to $\gamma_1,\gamma_2,\gamma_3$ are drawn in black, 
a few secondary rays are plotted in red. 
This diagram
is embedded in the exact scattering diagram $\cD^\Pi_\psi$ around $\tau=\tau_o$  by identifying the horizontal and vertical axis with the coordinates $(u,v)$ defined in \eqref{defuvxy}. 
\label{figMcKayScat}}
\end{figure}

Since the dimension vectors 
of the initial rays lie in the positive octant of $\IZ^3$, the enumeration of all possible scattering 
trees for fixed total dimension vector $\gamma=(n_1,n_2,n_3)$ is straightforward, unlike for the large volume scattering
diagram discussed previously. This gives an efficient algorithm to determine the quiver indices
$\Omega_{\theta}(\gamma)$ for arbitrary dimension vector $\gamma$
and stability parameters $\theta$. The latter are in turn equal to the DT invariants 
$\Omega_\tau(\gamma)$ in the region $\IH^o$ around the orbifold point $\tau_o$,
upon relating the Chern vectors and dimension vectors, as in the caption of Figure \ref{figQuiver}, and  equating the King stability parameters $\theta_i$ with $\Re[e^{-\I\psi} Z_\tau(\gamma_i)]$ (up to overall rescaling). In~\S\ref{sec_orbifold}, we show that the restriction of the orbifold scattering diagram $\cD_Q$ to the hyperplane $\theta_1+\theta_2+\theta_3=1$ agrees with the exact scattering diagram
$\cD^\Pi_\psi$ (to be defined below)  in a region around the
orbifold point~$\tau_o$.

\subsection{The exact scattering diagram}

For the exact central charge function \eqref{defZ} and associated Bridgeland stability conditions, one can likewise define  the scattering diagram $\cD^{\Pi}_\psi$ as the set of active rays 
$\Ract_\psi(\gamma)$ in the Poincar\'e upper half-plane such that  $Z_\tau(\gamma)$ has fixed argument $\psi+\frac{\pi}{2}$ and $\bOm_\tau(\gamma)\neq 0$. Since the conifold points $\tau=m\in \IZ$ lie on the boundary of the domain $\IH^{\rm LV}$ (defined below \eqref{defsw})
covered by the large volume scattering diagram $\cD^{\rm LV}_\psi$, the initial data must include
 the rays associated to $\cO(m)$ and $\cO(m)[1]$,  along with their images under $\Gamma_1(3)$. 
 In particular, since the spherical twist $\ST_\cO$ maps 
$\cO(0)[n]\mapsto \cO(0)[n+2]$, there are now an infinite set of rays emitted from each
point $\tau=m$, as shown in Figure \ref{fighomshifted}. Similarly, there is  an infinite set of rays emitted from every rational $\tau=\frac{p}{q}$ with $q\neq 0\mod 3$  which are in the same orbit under $\Gamma_1(3)$, where objects of charge $\pm\gamma_C$ become massless (the relevant
objects are computed in \S\ref{sec_massless} and shown in Table~\ref{Conifoldtab} for 
$0\leq p<q\leq 5$).
In particular, this includes the initial rays associated to the exceptional objects $E_i$
in \eqref{excepcoll}, emanating from $\tau=0,-\frac12,-1$, as well as translates of 
those. \looseness=-1

\begin{figure}[ht]
\begin{center}
\includegraphics[height=8cm]{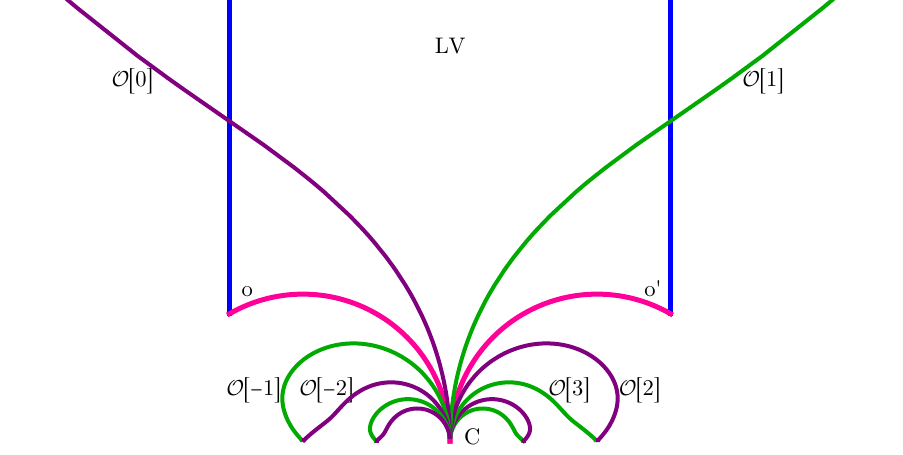} 
\end{center}
\caption{Rays $\Ract_\psi(\cO[n])$ emitted from $\tau=0$ for $\psi=0$. The same picture holds 
near all conifold points at $\tau=\frac{p}{q}$ with $q\neq 0 \mod 3$. For each $k$, the rays corresponding to the D4-brane
$\cO[2k]$ (in purple) and anti-brane $\cO[2k+1]$ (in green) end at the
same large-volume point: $\tau=1/(3k)$ for $k\neq 0$ and $\tau=\I\infty$ for $k=0$. As $\psi$ increases, the rays $\Ract_\psi(\cO[n])$ move counterclockwise.}
\label{fighomshifted}
\end{figure}

\begin{table}[ht]
\caption{The object $E$ of charge $\gamma_C$ which becomes massless at $\tau=\frac{p}{q}$ with $q\neq 0\mod 3$ is obtained by acting on $\cO$ by an auto-equivalence $g$ mapping $\tau=0$ to 
$\tau=\frac{p}{q}$. 
Here $T,U,V$ denote the generators $\tau\mapsto \tau+1, \tau\mapsto \frac{\tau}{1+3\tau}$ and $\tau\mapsto\frac{\tau}{1-3\tau}$. 
\label{Conifoldtab}}\vspace{-.5\baselineskip}
$
\begin{array}{|c|r|r|r|r|} \hline
\tau & g &\gamma_C & \Delta(\gamma_C)   & E \\ \hline 
0 &  \mathds{1} &  [1,0,1) & 0  & \cO \\
1/5 & U^2 T^{-1} &- [5,1,6) & 3/{25}  & E\rightarrow\Omega(2)[-1] \rightarrow \cO^{\oplus 3}[2] \xrightarrow{+1} \\
1/4 & U T  & [4,1,6) & -3/32 &  E\rightarrow\cO(1) \rightarrow \cO^{\oplus 3}[3] \xrightarrow{+1} \\
2/5 & U T^{-2}  &- [5,2,6) & 12/25 &  E\rightarrow\cO(-2) \rightarrow \cO^{\oplus 6} \xrightarrow{+1} \\
%3/7 & U T^{-1} V T & [7,3,10) & 15/49  & \Omega \to \cO^{\oplus 9} \to E\xrightarrow{+1}
% E\rightarrow\Omega[1] \rightarrow \cO^{\oplus 9}[1] \xrightarrow{+1}  \\
1/2 & T V T  & -[2,1,3) & 3/8 & \Omega(2)[1]\\
%4/7  & T V T U T^{-1}  & [7,4,12) & 15/49&
%E \to \cO(1)^{\oplus 9} \to \Omega(4) \xrightarrow{+1}
% \cO(1)^{\oplus 9}[-1]\rightarrow \Omega(4)[-1] 
% \rightarrow E \xrightarrow{+1} \\
3/5  & T V T^2 & -[5,3,8) & 12/25 &   \cO(1)^{\oplus 6}\rightarrow \cO(3) \rightarrow E \xrightarrow{+1} \\
3/4 & T V T^{-1} & [4,3,10) &-3/32  & \cO(1)^{\oplus 3}[-3]\rightarrow \cO \rightarrow E  \xrightarrow{+1} \\
4/5 & T V^2 T  & -[5,4,12) & 3/25 & \cO(1)^{\oplus 3}[-2]\rightarrow \Omega(2)[1] \rightarrow E \xrightarrow{+1} \\
1 & T  & [1,1,3) & 0  & \cO(1) \\ \hline
\end{array}
$
\end{table}

In order to analyze the structure of the resulting scattering diagram, it is convenient to introduce 
affine coordinates\footnote{Note that the map $\tau\mapsto(x,y)$ is not injective on $\IH$, but its restriction to
the fundamental domain $\cF_C$ and its translates is, see Figure \ref{figGrTTD}.}
\be
\label{defxysgen0}
x\coloneqq \frac{\Re\left( e^{-\I \psi} T\right)}{\cos\psi}\ ,\quad 
y\coloneqq  -\frac{\Re\left( e^{-\I \psi} T_D\right)}{\cos\psi}
\ee
such that the geometric rays in the $(x,y)$-plane are contained in straight lines $\{ry+dx=\ch_2\}$, oriented along  the vector $(-r,d)$. In these coordinates, the conifold point $\tau=m$ is mapped to 
 \be
(x_{\cO(m)},y_{\cO(m)}) = \left(m + \cV \tan\psi,  -\frac12{m^2} - m \cV \tan \psi\right) 
\ee
where $\cV$  is the {\it quantum volume}\footnote{This quantum volume was first computed in
\cite[(4.1)]{Klemm:1999gm} in terms of Barnes' G-function, and turns out to be a special
value of the L-function associated to the Eisenstein series $C(\tau)$, as noted independently
in \cite{Bonisch:2022mgw}, see \eqref{defV} and
\eqref{f12L0} for the explicit relations.}
\be
\label{defV0}
\cV\coloneqq\Im T(0) =  \frac{27}{4\pi^2} \Im\left[ \Li_2\left(e^{2\pi\I/3}\right) \right]
\simeq 0.462758
\ee
In particular, just as in the large volume scattering diagram
of \cite{Bousseau:2019ift}, the initial rays associated to $\cO(m)$ and $\cO(m)[1]$ are straight 
lines tangent to the parabola $y=-\frac12 x^2$ at $x=m$, but their starting point is 
displaced by a horizontal distance $\cV_\psi\coloneqq\cV\tan\psi$ along that tangent. 

\begin{figure}[htb]
\begin{center}
\includegraphics[height=8cm]{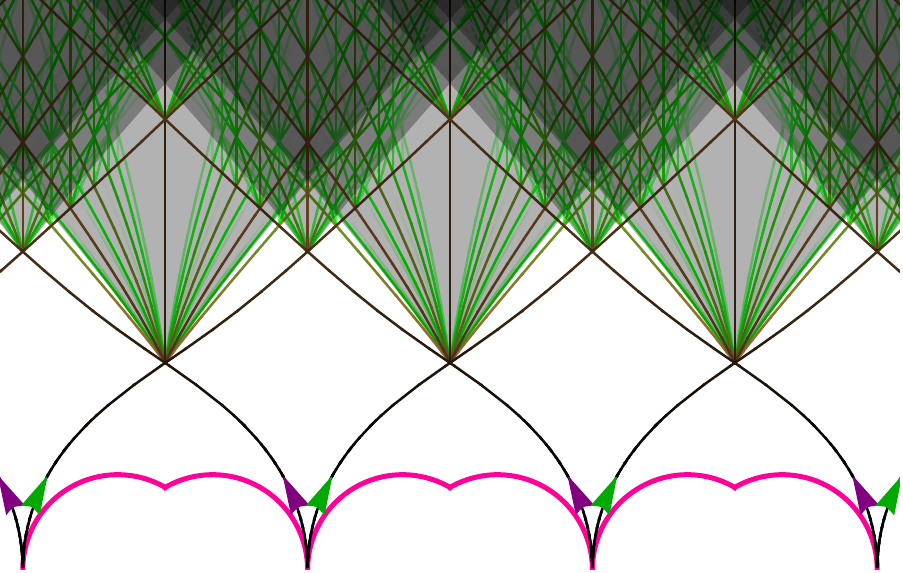}
\end{center}
\caption{Scattering diagram $\cD_\psi^{\Pi}$ for $\psi=0$, restricted to the fundamental domain and its translates.
  The initial rays are $\Ract_\psi(\cO(m))$ and $\Ract_\psi(\cO(m)[1])$. 
  The gray areas indicate regions where the scattering diagram is dense.\label{figscattPi0}}
\end{figure}

\begin{figure}[ht]
\begin{center}
\includegraphics[height=8cm]{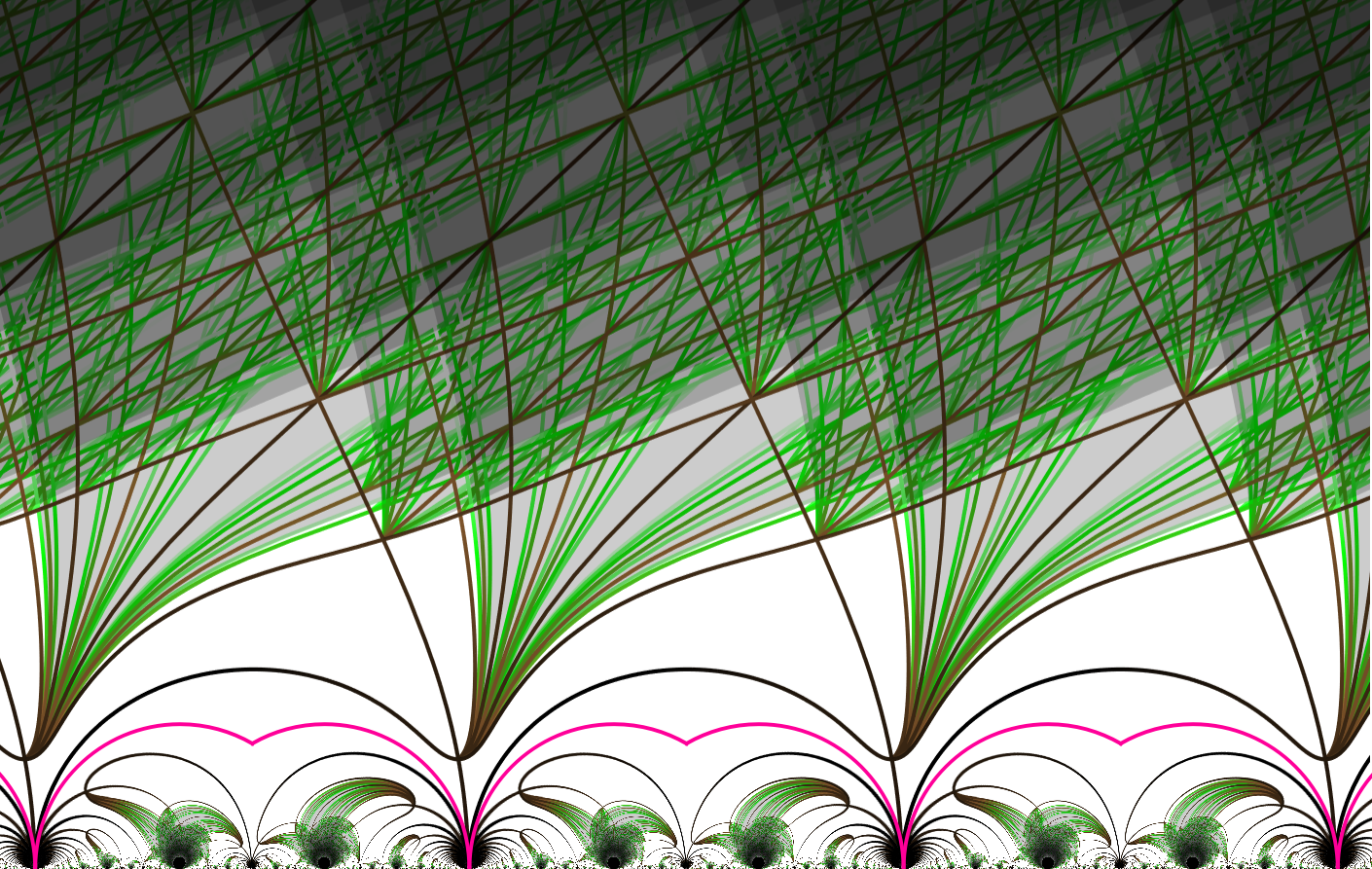} \\[3mm]
\includegraphics[height=8cm]{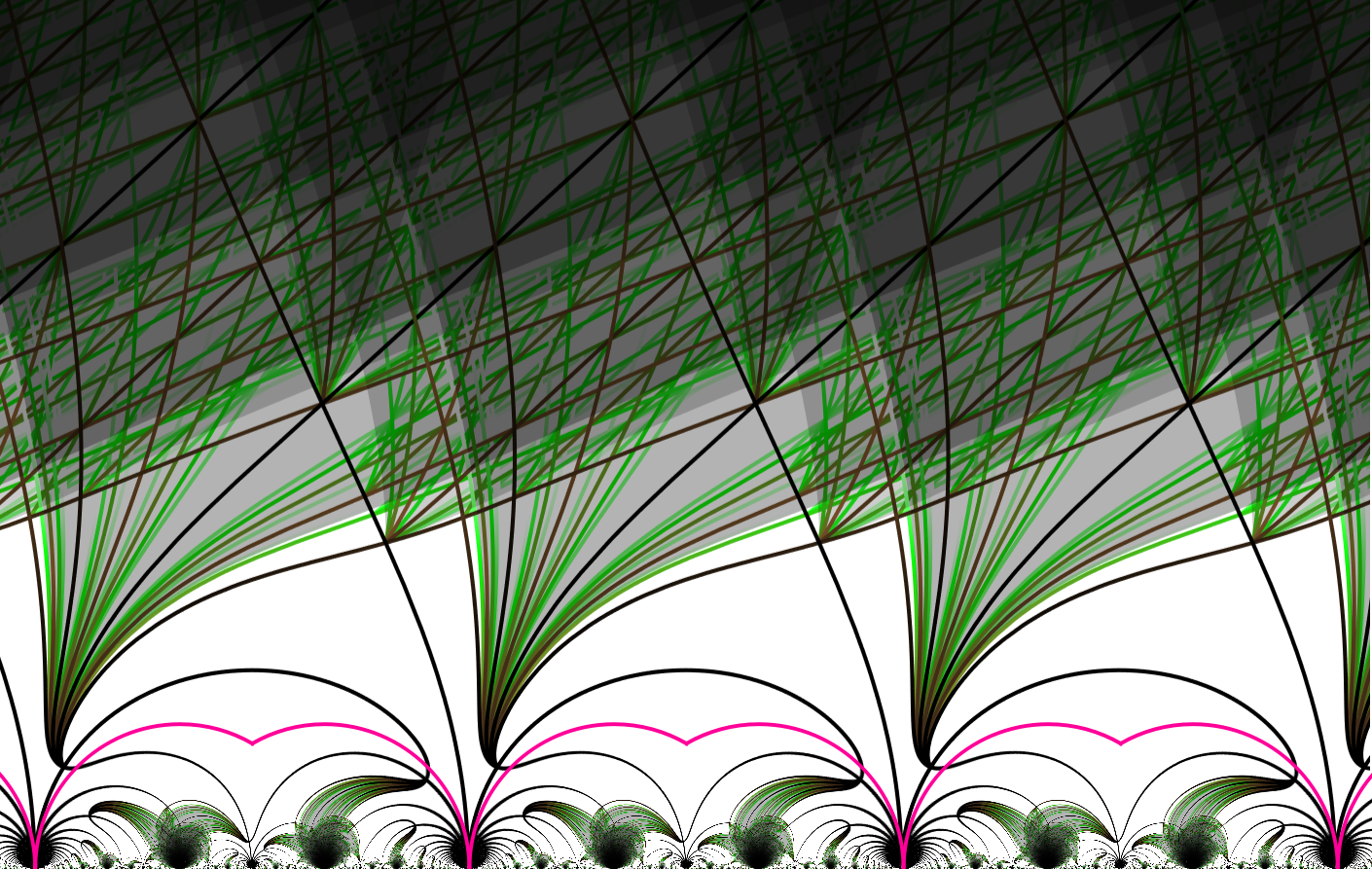} 
\end{center}
\caption{Scattering diagrams $\cD_\psi^{\Pi}$ for $\psi=-0.820$ (top)
and $\psi=-0.828$ (bottom),
on either side of the critical phase  $\psicr{-1/2}\simeq -0.82406$.
As $\psi\to\psicr{-1/2}$ from above, the rays $\Ract(\cO(m-1)[1])$
emitted at $\tau=m-1$
approach arbitrarily close to the conifold point $\tau=m$, before
escaping to $\tau=\I\infty$.
Beyond this critical value, these rays approach $\tau=m$ before escaping
to a different large
volume point $\tau=m-1/3$.  The rays emerging from the first
intersection above $\tau=m$
describe bound states of $\cO(m-1)[1]$ and $\cO(m)$ for
$\psi>\psicr{-1/2}$ and bound states
of $\cO(m)[1]$ and $\Omega(m+2)$ for phases slightly below $\psicr{-1/2}$.
The first diagram splits into $\Gamma_1(3)$ images of a connected
component that lies in the
union of translates of $\cF_C$, while the second diagram is connected. 
\label{figscattPicr}}
\end{figure}

For  small enough $\psi$, namely $|\cV_\psi|<\frac12$, this displacement does not affect the
structure of the scattering diagram, so that the exact scattering diagram $\cD_\psi^{\Pi}$ 
coincides with the large volume scattering diagram $\cD^{\rm LV}_0$ in the region above 
the parabola in the $(x,y)$ plane, up to shifting the starting points of the initial rays. 
In the original coordinate $\tau$, 
the only initial rays which escape towards the large volume region $\tau=\I\infty$  are those associated to $\cO(m)$ and $\cO(m)[1]$ (see Figure \ref{figscattPi0}), and their intersections patterns are identical to those of the large volume scattering diagram $\cD^{\rm LV}_0$ in $(s,t)$ plane, up to a change of variable $\tau\mapsto (s,t)$ obtained by equating the coordinates $(x,y)$ on both sides. 
In particular, the topology of the trees contributing to the 
index $\Omega_\infty(\gamma)$ along the rays $\Ract_0(\gamma)$ is unchanged, and the SAFC 
for $\cD_\psi^{\Pi}$ follows from the SAFC for $\cD^{\rm LV}_0$.

In contrast, for  $|\cV_\psi|>\frac{1}{2}$, the displacement of the starting points of the initial 
rays associated to $\cO(m)$ and $\cO(m)[1]$ is large enough  that the first collision no longer
involves two consecutive rays $\Ract(\cO(m-1)[1])$ and $\Ract(\cO(m))$.
Taking $\cV_\psi<-\frac{1}{2}$ for definiteness, the ray $\Ract(\cO(m-1)[1])$ interacts with two
``new''\footnote{For small phases these rays escape towards other large volume limits.} rays
$\Ract(\cO(m)[-1])$ and $\Ract(\Omega(m+1))$ in a region near the orbifold point $\tau_o+m$, in such
a way that these three initial rays generate a portion (which grows with~$|\cV_\psi|$) of the
orbifold scattering diagram $\cD_o$ corresponding to the exceptional collection~\eqref{excepcoll}
tensored with~$\cO(m)$. The resulting outgoing rays escape towards the large volume points
$\tau=\I\infty,-1/3,-2/3$, and those that escape towards $\I\infty$ collide further with the
initial ray $\Ract(\cO(m))$ and with rays for different values of $m\in\IZ$.
In fact, as $\psi$ approaches the critical value $\cV_\psi=-\frac12$,
the ray $\Ract(\cO(m)[1])$ emitted at $\tau=m$ approaches arbitrary close to the conifold point 
$\tau=m+1$, and it escapes to the large volume point $\tau=m+\frac23$
(respectively, $\tau=\I\infty$) as $\psi$ approaches the critical value $\cV_\psi=-\frac12$
from below (respectively, from above as in Figure~\ref{figscattPicr}).
More generally, we find that 
the topology of the scattering diagram jumps at
a countable set of critical phases values 
where some ray $\cR_\psi(\gamma)$ ends up at a conifold point.

\begin{definition}\label{def:critical}
  The phase $\psi\in(-\pi/2,\pi/2)$ is \textbf{critical} if any of four equivalent conditions holds:
  \begin{enumerate}
  \item\label{def-critical-end} an active ray $\cR_\psi(\gamma)$ ends at $\tau=0$ (or any other conifold point);
  \item\label{def-critical-start} an (initial) active ray $\cR_\psi(\gamma)$ with 
  $\gamma\notin [1,0,0) \IZ$ starts at $\tau=0$;
  \item\label{def-critical-LV} the point $(x,y)=(\cV_\psi,0)$ is the intersection of $\Ract^{\rm LV}_0(\cO)$ and another active ray of~$\cD^{\rm LV}_0$\!;
  \item\label{def-critical-orbi} % the point
    $\theta=(0,\frac{1}{2}+|\cV_\psi|,\frac{1}{2}-|\cV_\psi|)$ is the intersection of $\Ract^o(\gamma_1)$ and another active ray of~$\cD_o$.
  \end{enumerate}
\end{definition}

The equivalence of the four characterizations is proven by mapping DT invariants along rays of the exact diagram near $\tau=0$ to rays of the large volume diagram~$\cD^{\rm LV}_0$ in~\S\ref{sec_init} and orbifold diagram~$\cD_o$ in \S\ref{sec_crit} and~\S\ref{sec_start}.  As ray intersections in~$\cD^{\rm LV}_0$ (or~$\cD_o$) have rational coordinates~$(x,y)$ (or rational $\theta$ up to scaling, respectively), critical phases have rational values of~$\cV_\psi$.  In detail,
critical phases take the form 
\be
\label{defpsicr}
\psicr{\alpha}=\arctan(\alpha/ \cV)
\ee
or equivalently $\cV_\psi=\alpha$,  where $|\alpha|$ belongs to a dense set of rational values 
in the range $(\frac12 \sqrt{5},\infty)$, or to the discrete series
\be
\label{discseries}
\Bigl\{ \frac{F_{2k}+F_{2k+2}}{2F_{2k+1}} , k\geq 0 \Bigr\} = \{\tfrac{1}{2}, 1, \tfrac{11}{10} , \tfrac{29}{26} , 
\tfrac{19}{17} , \ldots \}  
\ee
with $F_p$  the $p$-th Fibonacci number (with $F_0=0$, $F_1=1$),  converging to 
$\frac12 \sqrt{5} \simeq 1.11803$.
This discrete series and dense set can be read off along the $y=0$ line in the large volume diagram of Figure~\ref{figDLVxy}.
  In Figures \ref{fig011psi}
and~\ref{fig101psi}, we show examples of trees contributing to $\gamma=\ch(\cO)$ and $\gamma=\ch(\cO_C)$, with discontinuities occuring only on a subset of critical phases, specifically at half-integer values of~$\cV_\psi$.

\begin{figure}
\begin{center}
\includegraphics[width=8cm]{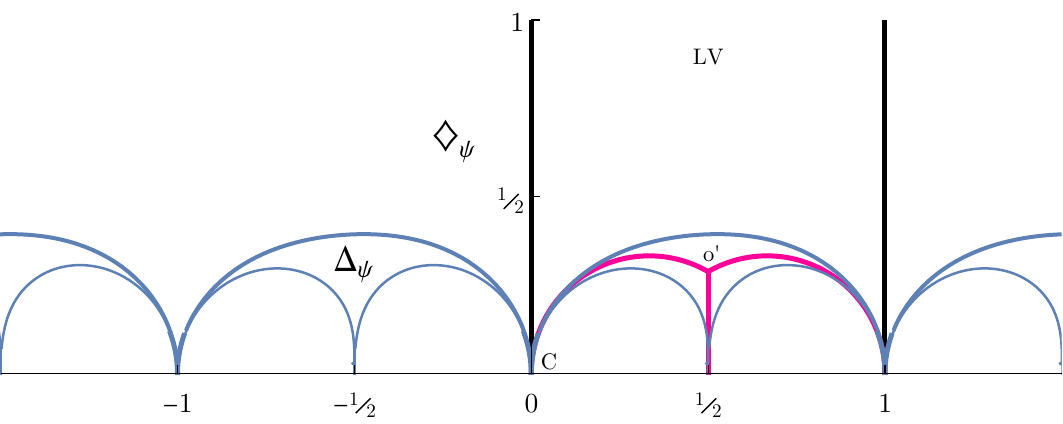}\hspace*{5mm}
\includegraphics[width=7cm]{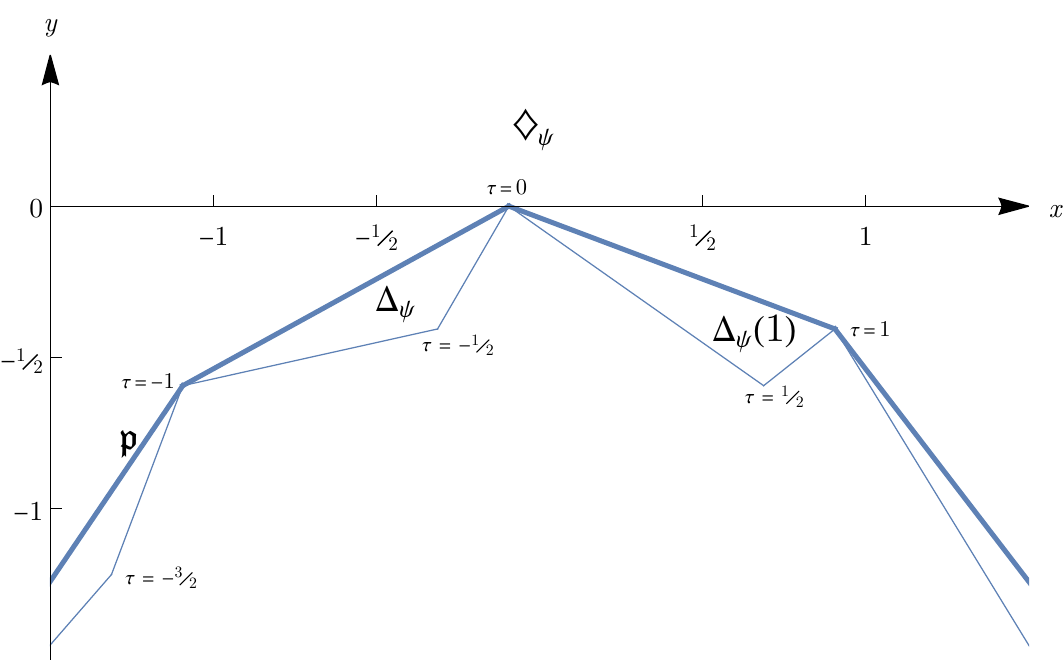}
\end{center}
\caption{The upper half plane $\IH$ is partitioned into regions~$g\cdot \diamondsuit_\psi$
around images of the large volume point $\I\infty$ and $g\cdot \Delta_\psi$ around images
of the orbifold point $\tau_o$. The region $\Delta_\psi$ and its translates $\Delta_\psi(k)$
map to triangles in the $(x,y)$ plane, while the region $\diamondsuit_\psi$ above all blue curves 
maps to the region
above the ``jagged parabola''~$\mathfrak{p}$ defined in \eqref{eqjagged}.  The figures are drawn for $\psi=-0.2$.  While the figure in  the $\tau$~plane is rather insensitive to the value of 
the phase $\psi$, the images of these regions in the $(x,y)$ plane depend significantly on it: for instance the images of $\diamondsuit_\psi$ and~$\Delta_\psi$ overlap for $\psi<-\psicr{1/2}$.\label{fig:Delta-diamondsuit}}
\end{figure}

\begin{figure}
\begin{center}
\includegraphics[width=8cm]{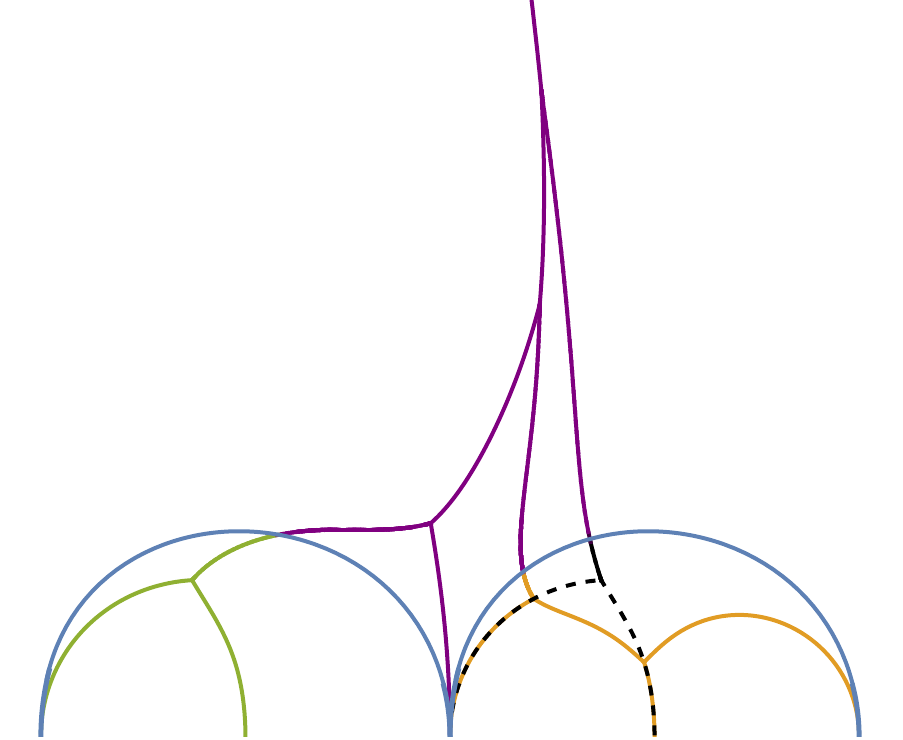}\qquad
\includegraphics[width=8cm]{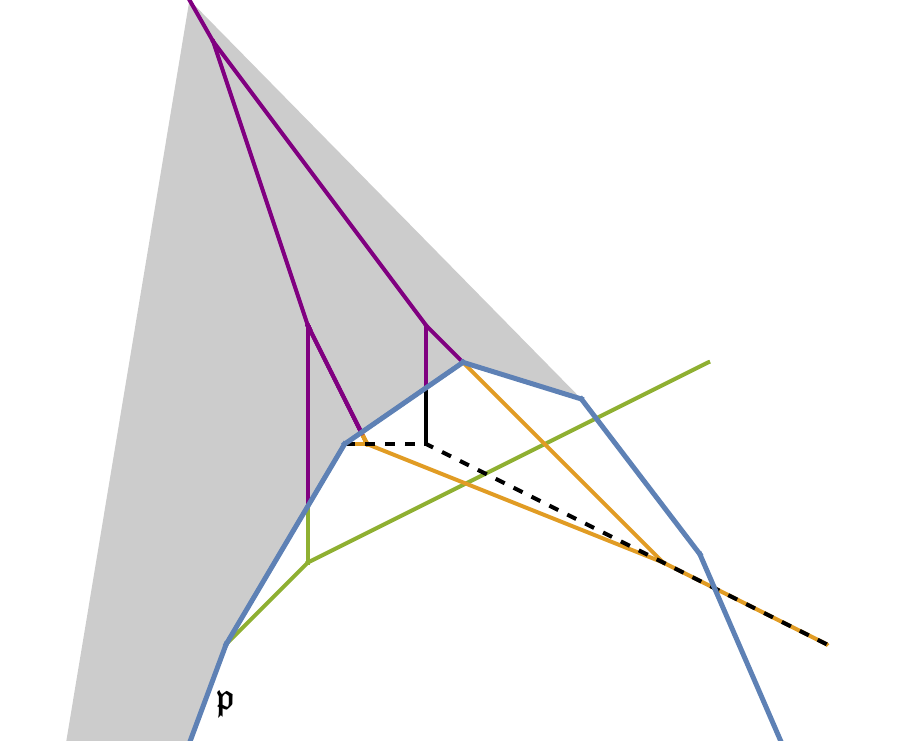}
\end{center}
\caption{Left: Example of attractor flow tree $\cT$ 
 in the $\tau$~plane, for $|\cV_\psi|>1/2$ (specifically, 
  $\cT=\{\{\{2 \cO(-1)[1], \Omega(1)\},
         \{4 \cO[1], \{3 \Omega(2), \cO(1)[-1]\}\}\},\allowbreak
       \{3 \cO(1), \{\Omega(2), 2 \cO[1]\}\}\}$ and $\psi=-1.2$),~in which the various subtrees (or shrubs)~$\cT_I$ corresponding to exceptional collections are given different colors (green, orange, black, partly dashed to make overlaps visible).
Together with further initial rays, the rays entering the region~$\diamondsuit_\psi$ (above the blue curves) participate in an attractor flow tree in that region, consisting only of outbound rays (solid purple curves).  
Right: same tree in the $(x,y)$ plane; some initial rays in
     the orbifold regions appear to start in~$\diamondsuit_\psi$,
     because $\tau\mapsto(x,y)$ is not injective.  The intersection
     of~$\diamondsuit_\psi$ with the convex hull of~$\mathfrak{p}$
     and of the tree's root, shaded in gray, is a bounded region.
     \label{fig:composite-tree}}
\end{figure}

Away from  critical values of $\cV_\psi$, we show that for any total charge $\gamma$,
flow trees rooted in the large volume region admit a two-stage structure, with a `trunk' inside 
the region $\diamondsuit_\psi$ lying above a certain piecewise linear region $y \geq \mathsf{p}(x)$
in the $(x,y)$ plane (defined in \eqref{eqjagged} and shown in Figure \ref{fig:Delta-diamondsuit}), and subtrees (or  `shrubs') inside triangular regions 
$\Delta_\psi(m)$ in the $(x,y)$ plane containing the image of the orbifold point $\tau_o+m$ (see the example in Figure~\ref{fig:composite-tree}). Within each triangular region, the shrubs
reduce to attractor flow trees for the orbifold quiver,  with leaves given by initial rays of the exceptional collection \eqref{excepcoll} tensored with $\cO(m)$. In addition to these shrubs, the trunk can also have leaves of type $\Ract(\cO(m))$
for $\cV_\psi<0$ (or $\Ract(\cO(m)[1])$ for $\cV_\psi>0$). 
We give effective bounds on
the possible constituents which show that the SAFC holds for $\cD^\Pi_\psi$ for any non-critical value of~$\cV_\psi$.

Finally, for $\psi=\pm\frac{\pi}{2}$, we find that the scattering diagram drastically simplifies. Indeed, 
the geometric rays $\Im[Z_\tau(\gamma)]=0$ for $r\neq 0$ reduce to the contour lines 
$s=\frac{d}{r}$ of the function $s=\frac{\Im T_D}{\Im T}$ defined in \eqref{defsw} (for $r=0$, 
the geometric rays $\Rgeo_{\pi/2}(\gamma)$ are empty). Hence, 
scattering can only take place at the orbifold point $\tau_o$ and its images under $\Gamma_1(3)$, where the function $s$ is ill-defined. At each orbifold point, there are three incoming rays associated to the objects of corresponding exceptional collection, which scatter all at once according to the orbifold scattering diagram $\cD_o$. In particular, a ray emitted from
the orbifold point $\tau_o+m$ into the fundamental domain $\cF_o(m)$
will escape to $\tau=\I\infty$ in the range $m-1<\tau_1<m$ without encountering any wall of marginal stability.
For $m=0$, this
explains why the  index for the orbifold quiver in the anti-attractor chamber $\Omega_c(\gamma)$ 
agrees with the Gieseker index $\Omega_\infty(\gamma)$ for normalized torsion free 
sheaves, as observed in \cite{Douglas:2000qw,Beaujard:2020sgs}.
We also prove the SAFC for that phase, which leads altogether to the following theorem:

\begin{theorem}[Split Attractor Flow Conjecture for local $\IP^2$]\label{thm:SAFC}
  For any $\Pi$-stability condition $z\in\Pi\simeq\IH$ of $K_{\IP^2}$ and any charge vector~$\gamma$ such that $\psi=\arg(-\I Z_z(\gamma))$ is a non-critical phase in $(-\pi/2,\pi/2]$ in the sense of Definition~\ref{def:critical}, there are finitely many (maximally extended) split attractor flows starting from~$z$ whose leaves are active rays.
  All leaves are $\Gamma_1(3)$ images of the rays $\Ract(\cO)$ and $\Ract(\cO[1])$ that emanate from the conifold point $\tau=0$.
  Moreover, if $|\psi|<\psicr{1/2}$, then such flows do not exist for $z\in\Delta_\psi$, while for $z\in\diamondsuit_\psi$ the leaves are $\Ract(\cO(m))$ and~$\Ract(\cO(m)[1])$, $m\in\IZ$.
  If $\psicr{1/2}<|\psi|<\pi/2$, then for $z\in\Delta_\psi$ the leaves are rays $\Ract(E_{j=1,2,3})$ corresponding to the exceptional collection~\eqref{excepcoll}, while for $z\in\diamondsuit_\psi$ the leaves are rays $\Ract(\cO(m))$ and $\Ract(E_j(m))$ (for $j=1,2,3$ and $m\in\IZ$), emanating from conifold points at integer and half-integer values of~$\tau$.
  Finally, if $\psi=\pi/2$, then for $z$ in the interior of~$\cF_o$ the leaves are $\Ract(E_{j=1,2,3})$.
\end{theorem}

\begin{figure}
\begin{center}
\includegraphics[height=9cm]{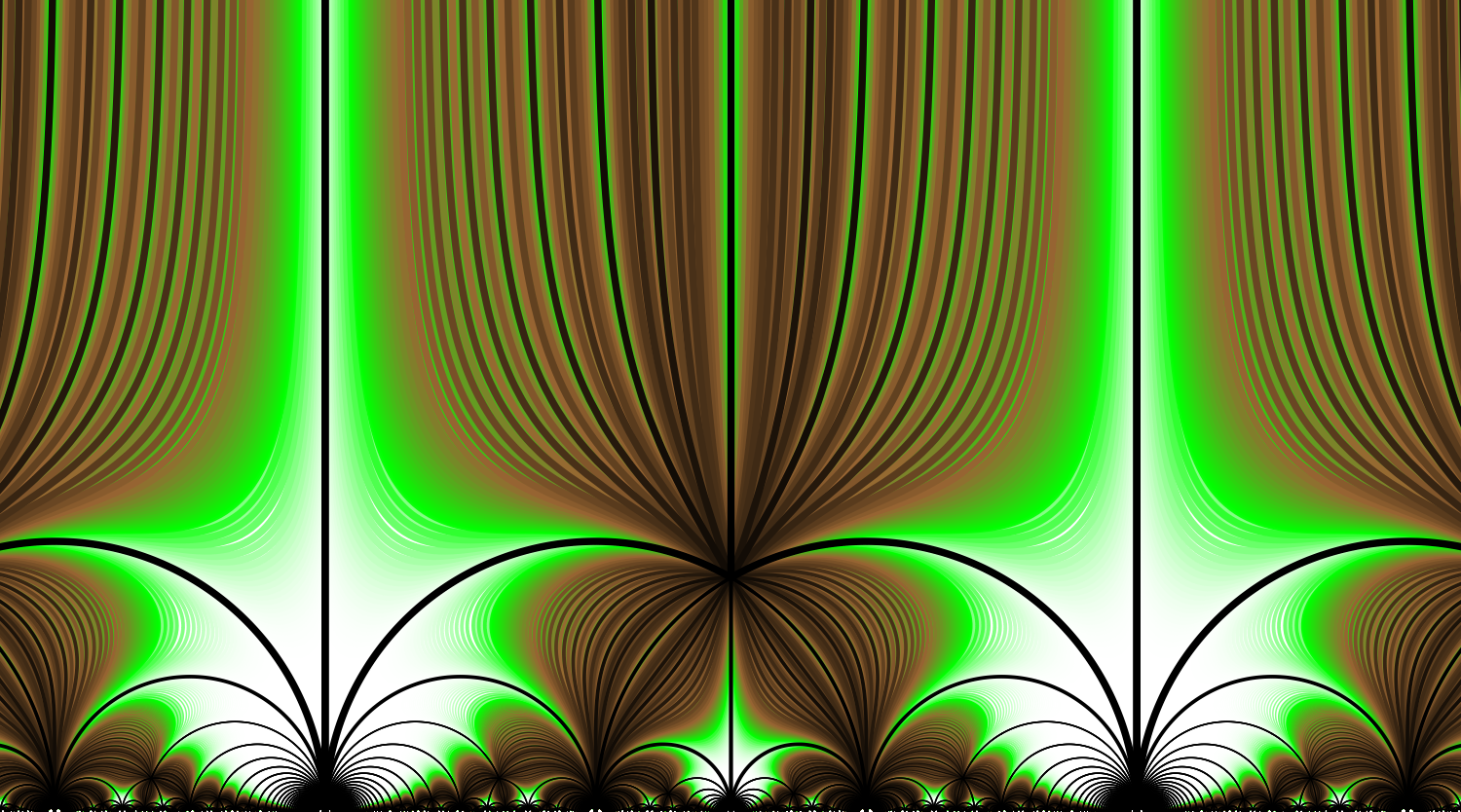}
\end{center}
\caption{Scattering diagram for $\psi=\pm\frac{\pi}{2}$. Geometric rays coincides with contours of the function
$s=\frac{\Im T_D}{\Im T}$. The orbifold scattering diagram is embedded in a infinitesimal neighborhood
of all $\Gamma_1(3)$ images of $\tau_o$, which shrinks as $\psi\to\pm\frac{\pi}{2}$. 
\label{figslopecontour}}
\end{figure}

\subsection{Outline}

This work is organized as follows. In \S\ref{sec_gen} we recall some general facts about
moduli spaces of sheaves on $\IP^2$, the structure of the derived category $D^b\Coh(K_{\IP^2})$
and the space of Bridgeland stability conditions associated to it. 
In \S\ref{sec_diag} we recall the definition and main properties 
of the scattering diagram for quivers with potentials, extend this notion to general triangulated categories, and explain the relation between scattering rays and attractor flow trees.
In \S\ref{sec_LV}, we revisit the scattering diagram $\cD_\psi^{\rm LV}$
for the large volume central charge constructed in \cite{Bousseau:2019ift}, generalize it to arbitrary 
$\psi$, give an effective algorithm for determining the possible initial rays contributing to the index,  
and illustrate this procedure by computing the Gieseker index for various Chern vectors of rank 0 or~1. In~\S\ref{sec_orbifold}, we prove the Attractor Conjecture for the orbifold quiver, construct the corresponding  scattering diagram $\cD_o$ associated to the orbifold quiver, and determine  the scattering
 sequences contributing to the quiver index for various dimension vectors. \
 In \S\ref{sec_global}, we 
 determine the scattering diagram $\cD^{\Pi}_\psi$ on the slice of  $\Pi$-stability conditions, and show the SAFC in this case (Theorem~\ref{thm:SAFC}).
 In  Appendix \ref{sec_Eichler} we derive the Eichler integral representation \eqref{Eichler0}
 of  the periods $T,T_D$ provided by local mirror symmetry, and use it to obtain expansions
 around the large volume, conifold and orbifold points. In \S\ref{sec_massless} we determine
 the object in $D^b \Coh_c K_{\IP^2}$ which becomes massless at the orbifold point $\tau=p/q$
 for low values of $p,q$. In \S\ref{app:bounds} we determine the possible end points of the flow to deduce the initial data of scattering diagrams for all phases~$\psi$.
 In \S\ref{sec_defDT}, we provide some details 
 on the mathematical definition of DT invariants.
  In \S\ref{sec_higherk} we give further examples of scattering sequences 
 for higher rank sheaves at large volume, complementing the examples in~\S\ref{sec_rank01}. 
 Finally, in \S\ref{sec_mathematica} we describe the main features of a
 Mathematica package which we have developed in the course of this investigation, which is 
freely available for further explorations.

 \subsection{Acknowledgments}
 The authors are grateful to  Tom Bridgeland, Miranda Cheng, David Jaramillo Duque, Amir Kashani-Poor, Albrecht Klemm, Emmanuel Macr\`i, Jan Manschot, Sergey Mozgovoy and Thorsten Schimannek for useful discussions. We especially thank Thorsten Schimannek for his
help with the material in \S\ref{sec_Eichler}.
 The research of PD and BP is supported 
by Agence Nationale de la Recherche under contract number ANR-21-CE31-0021.
The research of PB is partially supported by the NSF grant DMS-2302117.

%%%%%%%%%%%%%%%%%%%

\section{Generalities\label{sec_gen}}
In this section, we first collect some basic facts about coherent sheaves on $\IP^2$ and 
$K_{\IP^2}$, Bridgeland stability conditions $\Stab\cC$
on the derived category $\cC=D^b(\Coh_c K_{\IP^2})$ of  compactly supported 
sheaves, and identify the slice $\Pi\subset \Stab\cC$ of 
physical stability conditions.

\subsection{Gieseker-stable sheaves on \texorpdfstring{$\IP^2$}{P2}}

Given a coherent sheaf $E$ on $\IP^2$, we denote its rank  by $r(E)$, 
its degree by $d(E)=\int_{\IP^2} c_1(E) \cdot H$ (where $H$ is the hyperplane section generating $H^2(\IP^2,\IZ)$),  its second Chern character by $\ch_2(E)$, and by $\gamma(E)$ the 
Chern vector $[r,d,\ch_2]$ valued in $\IZ\oplus \IZ\oplus \frac12 \IZ$. We denote by $\cO=\cO_{\IP^2}(0)$ the structure sheaf, with Chern vector $[1,0,0]$, and by 
$\cO_C$ the structure sheaf of the rational curve $C$ in the  hyperplane class,  
with Chern vector $[0,1,-\frac12]$.

For any pair of coherent sheaves, 
the Euler form $\chi(E,E')$ is given by the Riemann--Roch formula
\be\begin{aligned}
\chi(E,E') & \coloneqq \dim \Hom(E,E')-\dim\Ext^1(E,E') + \dim\Ext^2(E,E') \\
& \phantom{:} = r r'+\frac32(rd'-r'd) + r \ch'_2 + r' \ch_2 - d d' 
\end{aligned}
\ee
In particular for $E=\cO$, the Euler characteristic
\be
\chi(E)\coloneqq\chi(\cO,E) =  \dim H^0(E)-\dim H^1(E) + \dim H^2(E) = r(E)+\frac32 d(E) + \ch_2(E) 
\ee
is an integer. With some abuse of notation, we also denote by $\gamma(E)$  the vector 
$[r,d,\chi)$, valued in $\IZ^3$ (note the round closing bracket, to distinguish it from the vector
$[r,d,\ch_2]$). We denote the antisymmetrized Euler form
also known as Dirac-Schwinger-Zwanzinger pairing)  by\footnote{Our convention for the 
antisymmetrized Euler form  is  consistent with 
 \cite{bridgeland2016scattering,Mozgovoy:2020has} and opposite to that in
\cite{Manschot:2010qz,Alexandrov:2018iao,Beaujard:2020sgs,Arguz:2021zpx}}
\be
\label{DSZ}
\langle \gamma, \gamma' \rangle \coloneqq \chi(\gamma,\gamma') - \chi(\gamma',\gamma) =
3 \left( r d'-r' d \right)
\ee

For a coherent sheaf $E$ with $\rk(E)\neq 0$, we define the slope $\mu(E)$
and discriminant $\Delta(E)$ by 
\be
\label{defmuDelta}
\mu(E) \coloneqq \frac{d(E)}{r(E)}, \quad \Delta(E)\coloneqq\frac12 \mu(E)^2 - \frac{\ch_2(E)}{r(E)}
\ee
and denote them by $\mu_\gamma$ and $\Delta_\gamma$ 
(when $r(E)=0$ and $d(E)\neq 0$, we set $\mu(E)=+\infty$). 
Under tensoring with the $m$-th power of the line bundle $\cO_{C}$, the Chern vector transforms
as 
\be
\label{specflow}
\gamma \mapsto \gamma(m)\coloneqq[r, d+mr, \ch_2+m d+\frac{r}{2} m^2]
\ee
such that $\mu\mapsto \mu+m$ while $\Delta$ is invariant. In particular, the `fluxed D4-brane'
$\cO(m)$ has  Chern vector $[1,m,\frac12 m^2 ]$, slope $m$ and vanishing discriminant.
A sheaf with $r\neq 0$ is said to be normalized if its slope $\mu$ lies in the interval $(-1,0]$.

A coherent sheaf $E$ on $\IP^2$ is said to be of pure dimension $n$ if the dimension of the 
support of any non-zero subsheaf (including $E$ itself) is of complex dimension $n$. 
A torsion-free sheaf $E$ is said to be slope-semistable if it is of pure dimension 2 and if for any 
subsheaf $F\subset E$ one has $\mu(F)\leq \mu(E)$. It is Gieseker-semistable 
if it is of pure dimension 2 and if for any 
subsheaf $F\subset E$ one has $\mu(F)\leq \mu(E)$, with $\Delta(F)\geq \Delta(E)$ in 
case of equality. Gieseker stability is defined by requiring $\Delta(F)>\Delta(E)$ in 
case $\mu(F)=\mu(E)$, and slope stability by requiring $\mu(F)<\mu(E)$ for any proper subsheaf. In particular, slope-stability implies Gieseker stability, which implies Gieseker semistability, which implies slope-semistability.

Let $\cM_\infty(\gamma)$ be the moduli space of Gieseker-semistable sheaves 
with Chern vector $\gamma=[r,d,\ch_2]=[r,d,\chi)$ (the rationale for the notation $\infty$ will become apparent in \S\ref{sec_raywall}). If $\cM_\infty(\gamma)$ is not empty, then it
is a normal, irreducible, factorial projective variety of dimension 
\be
\label{dimMGieseker}
\dim_\IC \cM_\infty(\gamma) = \dim\Ext^1(E,E) = r^2 (2\Delta -1) + 1
\ee
Moreover, it is smooth whenever $[r,d,\chi)$ is a primitive vector in $\IZ^3$, such that semi-stable sheaves are automatically stable.

In order to state the condition for $\cM_\infty(\gamma)$ to be non-empty, we consider
exceptional stable sheaves, defined as those for which $\Hom(E,E)=\IC,
 \Ext^1(E,E)=\Ext^2(E,E)=0$.
Such sheaves are then necessarily homogenous stable vector bundles, and have
a trivial moduli space. 
They are entirely specified by their slope $\mu$, which can take value in an infinite set $\cE\subset \IQ$, called the set of exceptional slopes. For $\mu=\frac{p}{r}\in \cE$ with $r>0$ and $(p,r)$ coprime, the exceptional bundle of slope $\mu$ has rank $r$ and discriminant $\Delta_\mu=\frac12(1-\frac{1}{r^2})$. The set 
$\cE$ is the union of an increasing family $\cE_n\subset \cE_{n+1}$  obtained by the following recursive construction: $\cE_0=\IZ$ and $\cE_{n+1}$ is obtained from $\cE_{n}$ by adjoining the slopes 
\be
\mu = \frac12 (\mu_1+\mu_2) + \frac{\Delta_{\mu_2}-\Delta_{\mu_1}}{3+\mu_1-\mu_2}
\ee
between any consecutive slopes $(\mu_1,\mu_2)$ in $\cE_n$. 
For example, after four steps one gets 
\be
\cE \supset \cE_4\cap [0,1] = \left\{0,\frac{13}{34},\frac{5}{13},\frac{75}{194},\frac{2}{5},\frac{179}{433},\frac{12}{29},\frac{70}{169},\frac{1}{2},\frac{99}{169}
   ,\frac{17}{29},\frac{254}{433},\frac{3}{5},\frac{119}{194},\frac{8}{13},\frac{21}{34},1\right\}
\ee
The exceptional bundles of integer slope are the structure sheaves $\cO(m)$, while the 
exceptional bundle of  half-integer slope $m-\frac32$ is 
the twisted cotangent bundle $\Omega(m)$, defined by the exact sequence 
\be
\label{defOm1}
0\rightarrow \Omega(m) \rightarrow \cO(m-1)^{\oplus 3}\rightarrow \cO(m) \rightarrow 0
\ee
with Chern vector $\gamma=[2,2m-3,m^2-3m+\frac32]$. 

For any Chern vector $\gamma=[r,d,\chi)$ with $r>0$, $d\in \IZ$, $\chi\in \IZ$, 
the condition for $\cM_\infty(\gamma)\neq \emptyset$ is then \cite{drezet1985fibres}
\be
\label{MGiesekerNotEmpty}
\Delta(\gamma)\geq \deltaLP(\mu(\gamma)) \qquad \mbox{or} \quad 
( \mu(\gamma) \in \cE \ \mbox{and}\   \Delta(\gamma) = \Delta_{\mu(\gamma)})
\ee
where $\deltaLP(\mu)$ is the `Dr\'ezet--Le Potier curve'  
\be
\label{LPcurve}
\deltaLP(\mu) \coloneqq 
\sup_{\mu' \in \cE ,\, |\mu'-\mu|<3}
\left[ P( - |\mu'-\mu| ) - \Delta_{\mu'} \right]
\ee
where $P(x)=\frac12(x^2+3x+2)$ (see Figure \ref{figLP}). In particular, $\cM_\infty(\gamma)$ 
is empty unless
the Bogomolov bound $\Delta(\gamma)\geq 0$ 
is satisfied. 

\begin{figure}[ht]
\begin{center}
\includegraphics[height=6cm]{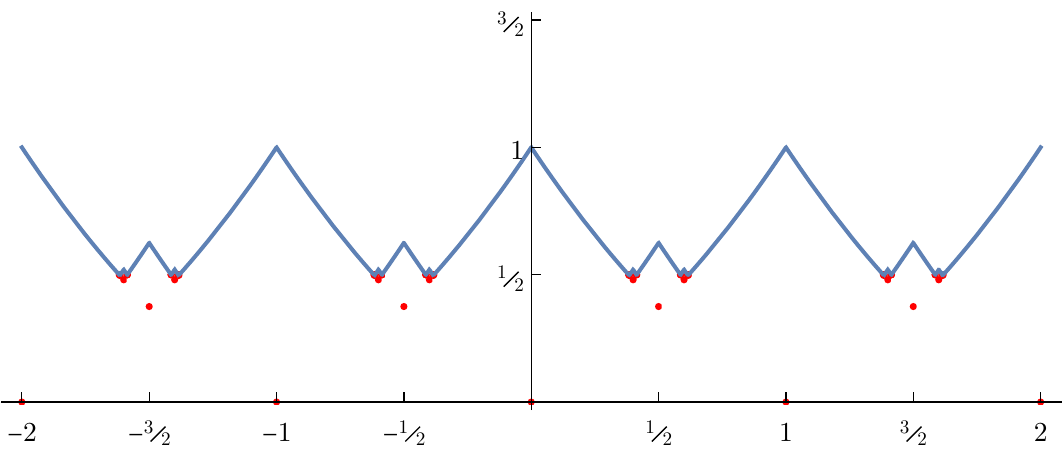} 
\end{center}
\caption{Dr\'ezet--Le Potier curve $\deltaLP(\mu)$ (in blue) and discriminants $\Delta_\mu$ of exceptional sheaves (in red).
\label{figLP}}
\end{figure}

We define the refined Gieseker index  $\Omega_\infty(\gamma)$ as the Poincar\'e-Laurent polynomial for the moduli space $\cM_\infty(\gamma)$,
\be
\label{OmGieseker}
\Omega_\infty(\gamma) = \sum_{p\geq 0} b_{p}(\cM_\infty(\gamma)) (-y)^{p-\dim_{\IC}\cM_\infty(\gamma)}
\ee
In the limit $y\to 1$, $\Omega_\infty(\gamma)$ reduces to the signed Euler characteristic 
$(-1)^{\dim_{\IC}\cM_\infty(\gamma)} e(\cM_\infty(\gamma))$.
When the inequality \eqref{MGiesekerNotEmpty} is saturated, $\cM_\infty(\gamma)$ 
has Picard rank $b_2(\cM_\infty(\gamma))=1$, or $2$ when the inequality is strict. 
The Gieseker index for sheaves on $\IP^2$ with arbitrary rank was determined in \cite{Manschot:2014cca}, by relating it to the Gieseker index for sheaves on the Hirzebruch surface
$\IF_1$ (which coincides with the blow-up of $\IP^2$ at one point) and using wall-crossing arguments.
Quite remarkably, the generating series of Gieseker indices with fixed rank $r$ and degree $d$ 
are conjectured to transform as mock Jacobi forms of the parameters $(\tau,w)$ conjugate to the 
second Chern class $\ch_2$ and Betti degree $p$ in \eqref{OmGieseker} \cite{Alexandrov:2019rth}.  The simplest case is for
rank~$1$, where the moduli space $\cM_\infty([1,0,1-n))$ coincides with the  Hilbert scheme of $n$ points on $\IP^{2}$, and the   
 generating series is an actual Jacobi form \cite{Gottsche:1990},
\be
\label{genHilb}
\begin{aligned}
& \sum_{n\geq 0}\Omega_\infty([1,0,1-n)) \, q^{n} =  \frac{\I q^{1/8} (y-1/y)}{\theta_1(q,y^2)} \\
& \quad = 1 + (y^2+1+1/y^2) q + 
(y^2+1+1/y^2)^2 q^2 + (y^6+2y^4+5y^2+6+\dots) q^3 + \dots
\\
& \quad \stackrel{y\to 1}{\to} 1 + 3 q + 9 q^2 + 22 q^3 + 51 q^4 + 108 q^5 + 221 q^6 + 429 q^7 + 
 810 q^8 + 1479 q^9 + \dots
\end{aligned}
\ee
where $\theta_1(q,y)=\I \sum_{r\in\IZ+\frac12}(-1)^{r-\frac12} q^{r^2/2} y^r$ is the Jacobi theta series (the dots in the middle line indicate the obvious additional terms required for invariance under $y\mapsto 1/y$). 
Further examples of generating series of Gieseker indices of rank $r>1$ can be found in \cite[\S A]{Beaujard:2020sgs}.

The notion of Gieseker semi-stable sheaves extends to vanishing rank as follows. 
A torsion sheaf $E$ is Gieseker-semistable if it is of pure dimension 1 (hence $d(E)\neq 0$) and if for any proper subsheaf $F\subset E$ one has $\nu(F)\leq \nu(E)$, with  $\nu(E)
=\frac{\chi(E)}{d(E)} = \frac{\ch_2(E)}{d(E)} + \frac32$.
The moduli space $\cM_\infty(\gamma)$ for $\gamma=[0,d,\chi)$ is non-empty for any $d>0$, and is invariant under $(d,\chi)\mapsto (d,d+\chi)$ and  $(d,\chi)\mapsto (d, -\chi)$. As in the torsion-free case, it is a normal, irreducible, factorial projective variety of dimension 
\be
\label{dimMGieseker0}
\dim_\IC \cM_\infty(\gamma) = \dim\Ext^1(E,E) = d^2+1
\ee
and it is smooth whenever $(d,\chi)$ are coprime. 
For $\gamma=[0,1,1)$, corresponding to a D2-brane wrapped on a curve $C$ in the  linear system 
$|H|$, one has  $\cM_\infty(\gamma)=\IP^2$. The  Gieseker index, defined in the same way as 
in \eqref{OmGieseker}, turns out to be completely independent of $\chi$  \cite{maulik2020cohomological}, and 
related to the refined Gopakumar-Vafa invariants 
$N_d^{(j_L,j_R)}$ \cite{Katz:1999xq,Choi:2018xgr} via
\be
\label{OmD2GV}
\Omega_\infty( [0,d,\chi) )= \sum_{j_L, j_R} (-1)^{2j_L+2j_R}  \chi_{j_L}(y) \chi_{j_R}(y) 
N_d^{(j_L,j_R)}
\ee
where $\chi_j(y)=\sum_{m=-j}^j y^{-2m} = \frac{y^{2j+1}-y^{-2j-1}}{y-1/y}$ is the  character of the spin-$j$ representation of $SU(2)$. The identification of the two fugacities $y_L$ and $y_R$ conjugate to 
$j_L$ and $j_R$ is natural in the Nekrasov-Shatashvilii limit of the refined topological string amplitude.
Using the refined BPS invariants in \cite[Table 2 p.61]{Huang:2011qx}, we find for degree up to 6,
\be\label{GVtab}
\begin{aligned}
  \Omega_\infty( [0,1,\chi) ) & = y^2+1+1/y^2 \\
  \Omega_\infty( [0,2,\chi) ) & = - y^5-y^3-y-1/y-1/y^3-1/y^5 \\
  \Omega_\infty( [0,3,\chi) ) & = y^{10}+2 y^8+3 y^6+3 y^4+3 y^2+3+ \dots \\
  \Omega_\infty( [0,4,\chi) ) & = - y^{17}-2 y^{15}-6 y^{13}-10 y^{11}-14 y^9-15 y^7-16 y^5-16 y^3-16 y-\dots \\
  \Omega_\infty( [0,5,\chi) ) & = y^{26}+2 y^{24}+6 y^{22}+13 y^{20}+26 y^{18}+45 y^{16}+68 y^{14}+87 y^{12} \\
  & \quad +100 y^{10}+107 y^8 +111 y^6+112 y^4+113 y^2+113+\dots \\
  \Omega_\infty( [0,6,\chi) ) & = -y^{37}-2 y^{35}-6 y^{33}-13 y^{31}-29 y^{29}-54 y^{27}-101 y^{25} \\
  & \quad -169 y^{23}-273 y^{21} -401  y^{19} -547 y^{17}-675 y^{15}-779 y^{13} \\
  & \quad -847 y^{11}-894 y^9-919 y^7-935 y^5 -942 y^3-945 y+\dots
\end{aligned}
\ee
In \S\ref{sec_Hilb}, \S\ref{sec_D2D0} and \S\ref{sec_rank01} we reproduce some of the invariants \eqref{genHilb} and \eqref{GVtab} from attractor flow trees in the large volume and orbifold scattering diagrams.

\subsection{Derived category of coherent sheaves on  \texorpdfstring{$K_{\IP^2}$}{KP2}\label{sec_dkp2}}

The total space $X=K_{\IP^2}$ of the canonical bundle over $\IP^2$ is a non-compact, smooth 
Calabi-Yau threefold, obtained as a crepant resolution of the quotient $\IC^3/\IZ_3$ with diagonal action $z_i \mapsto e^{2\pi\I/3} z_i$, $i=1,2,3$. We denote by $\Coh_c X$ the Abelian category
of compactly supported coherent sheaves on $X$. 
A  compactly supported coherent sheaf $E$ 
is equivalent to a pair $(F,\phi)$ where $F=\pi_*(E) \in\Coh \IP^2$
is the push forward of $E$, and $\phi\colon F\to F\otimes K_{\IP^2}$ is a morphism
(see \eg \cite[\S 8.1]{Mozgovoy:2022oli}). We abuse notation
and denote by $\ch(E)=\ch(F)=[r,d,\chi)$ the Chern character of $F$.
Physically, $F$ describes the gauge field on a stack of $r$ D4-branes wrapped on~$\IP^2$, while $\phi$ is the Higgs field describing  fluctuations in the fiber direction $K_X$. Given a sheaf $F$
on~$\IP^2$, we denote by $i_\star(F)=(F,0)$ its embedding along the zero section.

As explained for example in \cite{Aspinwall:2004jr},
the category of BPS states in type IIA string theory 
on $\IR^{3,1}\times X$ coincides with the bounded derived category $\cC=D^b(\Coh_c X)$ of coherent sheaves on $X$ with compact support. 
An object
$E\in \cC$  is a complex of coherent sheaves $\dots \to E^{-1} \to E^0 \to E^1 \to \dots$ of arbitrary (but finite) length, where the component $E^k$ in cohomological degree $k$ is a coherent sheaf 
with compact support on $X$. 
We denote by $\cH^k(E)$ the
cohomology of the complex at the $k$-th place, 
and by $\ch(E)\coloneqq\sum_k (-1)^k \ch(E^k)$ the Chern character of the complex. 
The  homological shift $E\mapsto E[1]$ taking the complex $(E^k)_{k\in\IZ}$ to $(E^{k-1})_{k\in\IZ}$ maps D-branes with charge $\gamma=\ch(E)=[r,d,\chi)$ 
to anti-D-branes with opposite charge $-\gamma$. 

By Serre duality, the Euler form on 
$X$ coincides with the antisymmetrized Euler form \eqref{DSZ} on $\IP^2$ 
\cite[Corollary 8.2]{Mozgovoy:2022oli},
\be
\chi_X( E, E') \coloneqq \sum_{k=0}^3 (-1)^k \dim\Ext^k_X( E, E') = 3(r d' - r' d)
\ee
where $\ch(E)=[r,d,\chi)$ and $\ch(E')=[r',d',\chi')]$. 
For $E=i_*(F)$ and $E'=i_*(F')$, the extension groups on $X$ can be computed 
from those on $\IP^2$ via \eqref{lemma46},
\be
\label{lemma46}
\Ext_X^k( i_*F, i_*F') = \Ext^k_{\IP^2}(F,F') \oplus  \Ext^{3-k}_{\IP^2}(F',F) 
\ee
In particular, an exceptional sheaf $F$ on $\IP^2$ lifts to a spherical object $S=i_*(F)$ in $\cC$, \ie
an object such that  $\Ext_X^k(S,S)\simeq \IC$ for $k=0,3$ and zero otherwise. 

\medskip

Starting from the 
Ext-exceptional collection on $\IP^2$ given by 
\be
F_1= \cO[-1], \quad F_2 = \Omega(1), \quad F_3=\cO(-1)[1]
\ee
and embedding it along the zero section, 
one obtains a tilting sequence $\cS=\bigoplus_{i=1}^3 E_i$ of objects $E_i=i_*(F_i)$ on $X$ 
which generate the category $\cC$ and such that $\Ext^k(\cS,\cS)=0$ for $k\neq 0$. 
Physically, these objects correspond to the fractional branes of the superconformal field theory
on the orbifold $\IC^3/\IZ_3$. 
The category $\cC$ is then equivalent to the derived category of representations of the Jacobian algebra\footnote{Recall that the Jacobian algebra  $J(Q,W)$  is the quotient of the path algebra
by the ideal generated by relations $\{ \partial_a W : a\in Q_1\}$.}
 $J(Q,W)$ for a quiver with potential $(Q, W)$  associated to $\cS$. 
The corresponding quiver  (shown in Figure \ref{figQuiver}) in has 3 nodes corresponding to 
each fractional brane, with $3$ arrows $
a_i\colon E_2\to E_1$, $b_j\colon E_3\to E_2$, $c_k\colon E_1\to E_3$ in agreement with 
$\Ext^1(E_i,E_{i-1})=\IZ^3$  (with index~$i$ identified modulo~$3$).
The  potential $W=\sum_{i,j,k} \epsilon_{ijk} \Tr(a_i b_j c_k)$ can be determined by studying 
the $A_\infty$ structure on the derived category of coherent sheaves \cite{Aspinwall:2004bs}, or
read off from the associated brane tiling \cite{Franco:2005rj}. 

The derived category $\cC$ admits a large group of auto-equivalences $\Aut(\cC)$, generated by the 
homological shift $E\mapsto E[1]$, by the translation $E\mapsto E(1)=E \otimes \cO_X(1)$,
by spherical twists $\ST_\cO$ with respect to the spherical object $\cO$, 
and by automorphisms of $X$ 
itself (or rather, its formal completion at $\IP^2$). For any spherical object $S$,  the spherical twist
$\ST_S$ acts on $E\in \cC$ via \cite{seidel2001braid}
\be
\label{STS}
\ST_S\colon E \mapsto \Cone\left(\Hom^{\bullet}_{K_{\IP^2}}(S,E) \otimes S \stackrel{\rm ev}{\to} E \right)
\ee 
where $\Cone(f)$ 
is defined by the exact triangle $A \stackrel{f}{\to} B \to \Cone(f) \to A[1]$. As a result
$\ST_S$ maps the Chern vector  
\be
\ch E\mapsto \ch E - \langle \ch S, \ch E\rangle \ch S
\ee
As shown in \cite{Bayer:2009brq}, the translation $E\mapsto E(1)$ and the spherical twist $\ST_\cO$ 
generate $\Gamma_1(3)$, the subgroup of $SL(2,\IZ)$ matrices defined below \eqref{defZPi}.

\subsection{Stability conditions and Donaldson-Thomas invariants\label{sec_Brid}}

A stability condition on a triangulated category 
$\cC$ with Grothendiek group $\Gamma$ consists of a pair $\sigma=(Z,\cA)$ such that  \cite{MR2373143}
\begin{itemize}
\item[i)] $\cA$ is the heart of a bounded $t$-structure, in particular it is an Abelian subcategory of~$\cC$;
\item[ii)] $Z\colon \Gamma \to \IC$ is a linear map, called the central charge;
\item[iii)] For any $0\neq E\in \cA$, $Z(E)=\rho(E) e^{\I\pi \phi(E)}$ where $\rho(E)>0$
and $0<\phi(E)\leq 1$; in other words, $Z(E)$ is contained in $\IH_B=\IH \cup (-\infty,0)$;
\item[iv)] {\it (Harder-Narasimhan property)} Every $0\neq E\in \cA$ admits a finite filtration $0 \subset E_0\subset E_1 \cdots \subset E_n =E$ by objects $E_i$ in $\cA$, such that each factor $F_i\coloneqq E_i/E_{i-1}$ is $\sigma$-semistable
and $\phi(F_1)>\phi(F_2)>\ldots > \phi(F_n)$;
\item[v)] {\it (Support property)} There exists a quadratic form $Q$ on $\Gamma\otimes \IR$ such that the kernel of $Z$ in 
$\Gamma\otimes \IR$ is negative definite with respect to $Q$, and moreover, for any 
$\sigma$-semistable object $E\in \cA$, $Q(\gamma(E))\geq 0$.  Equivalently \cite{ks}, 
there exists a non-negative constant $C$ such that, for all  $\sigma$-semistable object $E\in \cA$,
\be\label{supp-prop}
\| \gamma(E) \| \leq C \, |Z(E)|
\ee
where $\| \cdot \|$ is a fixed Euclidean norm on $\Gamma\otimes \IR$.
\end{itemize}
In the last two items above, we define $\sigma$-semistability of an object $F\in \cA$
 by requiring that $\phi(F')\leq \phi(F)$
for every non-zero subobject of $F$. Unlike common practice, we do not declare
that homological shifts $F[k]$ of a $\sigma$-semistable object $F$ are also stable, 
but we compensate for this in the definition of the DT invariants below.

According to  \cite{MR2373143}, the space  of stability conditions  $\Stab(\cC)$
is a complex manifold of dimension $\rk\Gamma$, such that the  map $\Stab(\cC)\to 
\Hom(\Gamma,\IC)$ which sends $\sigma=(Z,\cA)\mapsto Z$ is a local homeomorphism of
complex manifolds. Moreover, it admits an action of $\GLt \times \Aut(\cC)$  
\cite[Lemma 8.2]{MR2373143},
where $\GLt$ is the universal cover of  the group of $2\times 2$ real matrices  with 
positive determinant. The group $GL(2,\IR)^+$ acts on the central charge $Z$ via 
\be
\begin{pmatrix} \Re Z \\ \Im Z \end{pmatrix} \mapsto 
\begin{pmatrix} a & b \\ c & d \end{pmatrix} \begin{pmatrix} \Re Z \\ \Im Z \end{pmatrix}\ ,\quad ad-bc>0
\ee
preserving the orientation on $\IR^2$. Its universal cover acts on the stability condition $(Z,\cA)$
by suitably tilting the heart $\cA$. The subgroup $\IC^\times$ of matrices of the form 
$\lambda(\begin{smallmatrix} \cos\theta & \sin\theta \\  -\sin\theta  & \cos\theta  \end{smallmatrix})$
with $\lambda>0$ preserves the complex structure and has trivial stabilizers, so the quotient 
 $\Stab(\cC)/\IC^*$ is still a complex manifold. When $\cC$ is a category of coherent sheaves, there is a  notion of derived duality\footnote{Physically, $\gamma\mapsto\gamma^\vee$ corresponds
 to time reversal $T$, which reverses the sign of the Dirac pairing, while $\gamma\mapsto-\gamma$
 corresponds to CPT symmetry.}
 $E\mapsto E^\vee$, 
 and therefore an 
involution sending $(Z,\cA)\mapsto (Z^\vee, \cA^\vee)$ where
$\cA^\vee$ is the derived dual of $\cA$ (up to tilt) and $Z^\vee(E):=-\overline{Z(E^\vee)}$.
This involution allows to extend  $\GLt$ to the full group $\widetilde{GL}(2,\IR)$.

For any class $\gamma\in\Gamma$  and  stability condition $\sigma\in \Stab(\cC)$, 
we denote by $\cM_\sigma(\gamma)$ the moduli stack of  $\sigma$-semistable objects of charge 
$\pm\gamma$ in $\cA$. We further denote by 
$\Omega_\sigma(\gamma)$ the motivic Donaldson-Thomas invariant of
 $\cM_\sigma(\gamma)$.
Informally, one can think of $\Omega_\sigma(\gamma)$ as the Poincar\'e-Laurent polynomial of the cohomology with compact support 
of $\cM_\sigma(\gamma)$, 
with a cohomological shift ensuring that 
\be
\label{defOm}
\Omega_\sigma(\gamma) = \sum_{p\geq 0} b_{p}(\cM_\sigma(\gamma)) 
(-y)^{p-\virdim_{\IC}\cM_\sigma(\gamma)}
\ee
where $\virdim_\IC$ is the virtual dimension and $y^2=\IL$ is the motive of the affine line  (see Appendix \ref{sec_defDT} for a more precise mathematical definition). 
We further define the rational DT invariant $\bOm_\sigma(\gamma)$ via \eqref{defOmb}.
By construction, both $\Omega_\sigma(\gamma)$ and $\bOm_\sigma(\gamma)$ are 
invariant under the action of $\GLt \times \Aut(\cC)$, in particular 
under $\gamma\mapsto-\gamma$, as well as derived duality acting on $(\sigma,\gamma)$. 

These invariants are locally constant on $\Stab\cC$ but may  jump when some object $E\in \cA$
of charge $\gamma$ goes from being stable to unstable. This may happen when the central charge $Z(\gamma')$ of a subobject $E'\subset E$ of charge $\gamma'$ becomes aligned with  
$Z(\gamma)$, therefore along the
real-codimension one {\it wall of marginal stability}
\be
\cW(\gamma,\gamma') \coloneqq \{ \sigma=(Z,\cA) \in\Stab\cC : \Im(Z(\gamma') \overline{Z(\gamma)}) =0\}
\ee
The discontinuity across $\cW(\gamma,\gamma')$ is  determined from the invariants on either side of the wall by the wall-crossing formulae of \cite{ks,Joyce:2008pc}.

\subsection{Stability conditions on \texorpdfstring{$K_{\IP^2}$}{KP2} and \texorpdfstring{$\Pi$}{Pi}-stability\label{sec_Pistab}}

The space of Bridgeland stability conditions  $\Stab(\cC)$ on $\cC=D^b(\Coh_c K_{\IP^2})$
was studied in \cite{bridgeland2006stability,Bayer:2009brq}. After fixing the $\IC^\times$ action
such that skyscraper sheaves have central charge $Z_{[0,0,1]}=-1$, the central charge
can be parametrized by the two complex coefficients $(T,T_D)$ in \eqref{defZ},
\be
\label{defZ2}
Z(\gamma) = - r T_D + d T -\ch_2
\ee
We define $s=\frac{\Im T_D}{\Im T}$, in such a way that $\Im Z(\gamma)=\Im T(d - r s)$. 
We then denote by 
\begin{itemize}
\item $\Coh_c^{\leq s}$ the subcategory of $\Coh_c(X)$ generated (under extensions) by slope-semistable torsion-free
sheaves of slope $\frac{d}{r}\leq s$
\item $\Coh_c^{> s}$ the subcategory of $\Coh_c(X)$ generated by slope-semistable torsion-free
sheaves of slope $\frac{d}{r}> s$ and by torsion sheaves
\item $\Coh_c^{\sharp s}$ the subcategory of $\cC$ of objects $E$ such that $\cH^i(E)=0$
for $i\neq -1,0$, $\cH^{-1}(E)\in \Coh_c^{\leq s}$, $\cH^0(E)\in \Coh_c^{>s}$
\end{itemize}
The Abelian  category $\cA(s)=\Coh_c^{\sharp s}$ is obtained by the standard tilt 
procedure from the torsion pair 
$(\Coh_c^{\leq s},\Coh_c^{> s})$ and is the heart of a bounded $t$-structure on 
$\cC$. 
For $\Im T>0$, the construction ensures that $\Im Z(\gamma)\geq 0$
for any $E\in\cA(s)$.  In addition, if $\Im Z(\gamma)=0$, namely
$s=\mu$, we get
\be
\label{ReZ0}
\Re  Z(\gamma)  = r \left( w - \frac{\ch_2}{r} \right) = r \left( w  -\frac12 s^2 + \Delta \right)
\ee
where we defined  $w= - \Re T_D + s \Re T$ as in \eqref{defsw}. 
The coordinates $(s,w)$ in fact parametrize the orbits of the action
of $\GLt$ on $\Stab(\cC)$ in the region $\Im T>0$, 
such that the central charge \eqref{defZ2} 
is in the same orbit as the function used in \cite{li2019birational}
\be
\label{ZLZ}
Z^{LJ}_{(s,w)}(\gamma) = (w - \I s) r + \I d  -\ch_2  = (r w- \ch_2) + \I (d-s r)
\ee
The virtue of these coordinates is that walls of marginal stability become straight lines, given by the
vanishing of 
\be
\label{WallExact}
\Im\left[ Z^{LJ}_{(s,w)}(\gamma')\overline{Z^{LJ}_{(s,w)}(\gamma)} \right]
=   (rd'-r'd) w + (r'\ch_2 - r\ch_2') s +  (d\ch_2'-d'\ch_2)
\ee

Returning to \eqref{ReZ0}, the Dr\'ezet--Le Potier bound $\Delta(E)\geq \deltaLP(\mu(E))$ in \eqref{MGiesekerNotEmpty} (which also holds for slope-semi-stable objects)
 implies that $\Re  Z(\gamma)>0$ (and therefore $\Re Z(-\gamma)<0$ for the 
homologically shifted object $E[-1]$) in the region 
\be
\label{geostabsq}
\Im T>0, \qquad w > \frac12 s^2 -\deltaLPh(s) 
 \ee 
 where we define $\deltaLPh(\mu)=\deltaLP(\mu)$ when $\mu\notin\cE$, 
$\deltaLPh(\mu)=\Delta_\mu$ otherwise.
The region $w > \frac12 s^2 -\deltaLPh(s)$ is the region above the blue jagged curve and red points
in  Figure \ref{figsq} (the black and purple lines will be discussed momentarily). One can show that the remaining axioms for Bridgeland stability conditions (HN filtration and support property) are indeed satisfied by the pair $\sigma=(Z,\cA(s))$ in this region. Moreover, the resulting stability condition turns out to exhaust the subset $U\subset \Stab(\cC)$ of geometric stability conditions, defined as those stability conditions for which all skyscraper sheaves $\cO_x$ with $x\in \IP^2$ are $\sigma$-stable with the same phase. The connected component of $\Stab(\cC)$ containing $U$,
denoted by  $\Stab^\dagger(\cC)$,  is the  union of the images of (the closure of)~$U$ under the group $\Gamma_1(3)$ of auto-equivalences 
of $\cC$ generated by spherical twists, and is simply connected \cite{Bayer:2009brq}. It is unknown whether $\Stab(\cC)$ might have other connected components.

\begin{figure}[ht]
\begin{center}
\includegraphics[height=8cm]{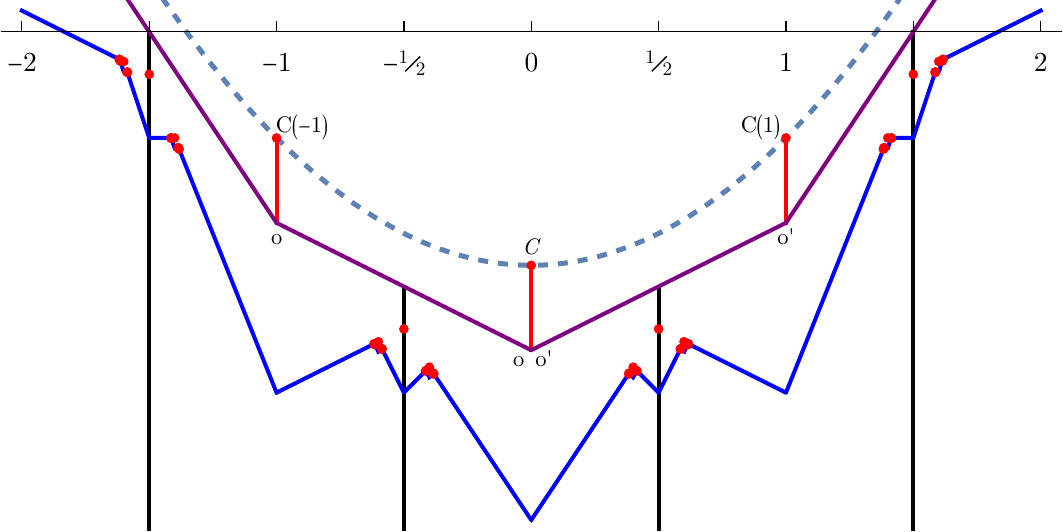} 
\end{center}
\caption{Image in the $(s,w)$ plane of the fundamental domain $\cF_C$ and some of its translates.
The dotted line correspond to the parabola $w\geq \frac12 s^2$. The jagged blue curve corresponds
to $w=\frac12s^2-\deltaLP(s)$, and 
the red dots indicate the points $(\mu_\alpha, \frac12 \mu_\alpha^2 - \Delta_\alpha)$ corresponding to exceptional slopes. The vertical segments in red interpolate between the conifold point at $\tau=n$ and the orbifold point at $\tau=\tau_o+n$. The segments in purple are the images of infinitesimal arcs around the orbifold point $\tau_o$ (from $(-1,\frac16)$ to $(0,-\frac13)$), around $\tau_{o'}$ (from $ (0,-\frac13)$ to $(1,\frac16)$), etc. These segments are connected at points with $s$ integer along the parabola $w=\frac12s^2-\frac13$. 
The vertical lines in black correspond to straight $\tau_1=n+\frac12$ lines connecting the orbifold point $\tau_o+n$ and conifold point $n-\frac12$ in the $\tau$ plane.\label{figsq}}
\end{figure}
 
Having identified the space of stability conditions $\Stab^\dagger\cC$ inside the two-dimensional space of central charges parametrized by $(T,T_D)$, it remains to determine which part of it is covered by the physical slice of $\Pi$-stability conditions. As mentioned in the introduction and further detailed in Appendix \ref{sec_Eichler}, the slice $\Pi$ is isomorphic to the universal cover of the modular curve $X_1(3)$, and conveniently parametrized by $\tau$ in the Poincar\'e upper half-plane $\IH$. 
As shown in 
 \cite{Bayer:2009brq}, there is an embedding $\IH \hookrightarrow \Stab^\dagger\cC$ sending
 $\tau\mapsto (Z_\tau,\cA_\tau)$ which is equivariant with respect to the action of $\Gamma_1(3)$ on both sides.  The central charge function $Z_\tau$ \eqref{defZPi} is determined by the 
representation \eqref{Eichler0} for the periods, 
\be
\label{Eichler1}
 T(\tau)= -\frac12 + \int_{\tau_o}^{\tau} C(u) \de u\ ,\quad
 T_D(\tau) = \frac13  + \int_{\tau_o}^{\tau} C(u) \,u \de u
\ee
where $\tau_o=\frac{1}{\sqrt3} e^{5\pi\I/6}$ is a preimage of the orbifold point, and $C (\tau) = \frac{\eta(\tau)^9}{\eta(3\tau)^3}$ is a weight 3 modular form for $\Gamma_1(3)$.
Indeed, one checks that the period vector  $\Pi=(1,T,T_D)$ transforms under $\tau\mapsto\tau+1$
and $\tau\mapsto -\frac{\tau}{3\tau-1}$ with the same monodromy matrices $M_{\rm LV}, M_{C}$ (see
\eqref{monDG}) as predicted by the translation $E\mapsto E(1)$ and spherical twist
$E\mapsto \ST_\cO E$ on the derived category side. 
It follows from the reality properties
\eqref{ReTTD} of the periods $T$ and $T_D$ that 
the action of derived duality preserves the slice of $\Pi$-stability conditions, 
\be
\label{eqreZ}
Z_{-\bar \tau}(\gamma^\vee) = -\overline{Z_\tau(\gamma)}
\ee
where $[r,d,\ch_2]^\vee=[-r,d,-\ch_2]=[-r,d,3d-\chi)$. Moreover, one can show 
that the periods
$T(\tau),T_D(\tau)$ satisfy the geometric stability conditions \eqref{geostabsq}
in the fundamental domain $\cF_C$ and its translates  \cite[\S C]{Bayer:2009brq}.

In the limit $\tau\to\I\infty$, $C(\tau)= 1+\cO(q)$ with $q=e^{2\pi\I\tau}$, the central charge becomes a quadratic polynomial in $\tau$, 
\be
\label{ZLVq}
Z_\tau(\gamma) = - \frac{r}{2} \left(\tau^2+\tfrac18 \right) + d \tau -\ch_2 + \cO\bigl(|\tau|q\bigr)
\ee
In fact, in the domain  $\IH^{\rm LV}$ defined below \eqref{defsw} as the
preimage of the region above the the curve $w =\frac12 s^2$ inside
the fundamental domain $\cF_C$ and its translates
(see Figure \ref{figfundO}),
the $\cO(|\tau|q)$ corrections in \eqref{ZLVq} as well as the constant $\frac18$
can be absorbed by 
an action of $(\begin{smallmatrix}1 & (s-\Re T)/\Im T \\
0 & t/\Im T \end{smallmatrix}) \in GL(2,\IR)^+$ with $t=\sqrt{2w-s^2}$ 
on $Z_\tau(\gamma)$ (note that $\Im T>0$ throughout  $\IH^{\rm LV}$). 
The new central charge is then $Z_{(s,t)}^{\rm LV}(\gamma)$ in
\eqref{defZLV}.

The heart $\cA_\tau$ is more subtle \cite{Bayer:2009brq}. 
In the  fundamendal domain $\cF_C$, we define the heart as $\cA_\tau=\cA(s(\tau))$ where $s(\tau)$ is defined as usual by $s=\frac{\Im T_D}{\Im T}$ and  $\cA(s)=\Coh_c^{\sharp s}$, and extend this definition to the full upper half-plane by
$\Gamma_1(3)$ equivariance. In particular, in the translates $\cF_C(m)$ under $\tau\mapsto \tau+m$
(which includes the region $\IH^{\rm LV}$ and  the  wedge of angle $2\pi/3$ above the orbifold point), the definition $\cA_\tau=\cA(s(\tau))$
continues to hold.  
However it is no longer valid below the arcs which connect the conifold points $\tau\in \IZ$ 
and orbifold points $\tau\in \tau_o+\IZ$. In particular, near the real axis
the objects in $\cA_\tau$ may include complexes of arbitrary length. As a result, the slice
of $\Pi$-stability conditions only covers the part of the space $U$ of geometric stability conditions
which lies above the purple line in Figure \ref{figsq}. 
While the image of the region $\Im T>0$ in the upper half plane 
 does extend below this line, the relevant heart no longer coincides with the one appropriate for 
geometric stability conditions.

\section{Scattering diagrams and attractor flows\label{sec_diag}}

In this section, we recall the construction of stability scattering diagram
for quivers with potential, generalize this notion to arbitrary 
triangulated categories, and explain its relation to the Split Attractor Flow
Conjecture in the 
case of non-compact CY threefolds with one-dimensional K\"ahler moduli space.

\subsection{Scattering diagram for quivers with potentials\label{sec_scatquiv}}

Let $(Q,W)$ be a quiver with potential. A representation $R$ of $(Q,W)$ is a set of finite dimensional
vector spaces $M_i$ for every node $i\in Q_0$, and linear maps $M_a\colon M_i\to M_j$ for every
arrow $(a\colon i\rightarrow j)\in Q_1$ such that for any arrow $a\in Q_1$, the element $\partial W/\partial a$ in the path algebra of $Q$ evaluates to zero on $R$.  The set of representations $R$ forms an Abelian category $\cA$  graded by the lattice $\IZ^{Q_0}$ of dimension vectors $\gamma=\dim R=(\dim M_i)_{i\in Q_0}$. For any set $(z_i)_{i\in Q_0} $ of points 
$z_i=-\theta_i+\I\rho_i$ in the upper half-plane~$\IH_B$, 
the pair $\sigma=(Z,\cA)$ with central charge $Z(\gamma)=\sum_{i\in Q_0} n_i z_i$
with $n_i=\dim M_i$ 
defines a stability condition in the sense of \S\ref{sec_Brid}. 
We denote by $\cM_Z(\gamma)$
the moduli space of $\sigma$-semi-stable representations of dimension vector $\gamma$, by 
$\Omega_Z(\gamma)$ its motivic Donaldson-Thomas invariant \eqref{defOm} and by
$\bOm_Z(\gamma)$ its rational counterpart \eqref{defOmb}. 
These invariants vanish unless $\gamma$ belongs to the positive quadrant 
$\Gamma_+=\IN^{Q_0}$. 

As shown in  \cite{bridgeland2016scattering}, the set of DT invariants is conveniently encoded
in a stability scattering diagram $\cD_Q$ in the space $\IR^{Q_0}$ spanned by King stability parameters $\theta=(\theta_i)_{i\in Q_0}$. To define $\cD$, we first introduce the hyperplane orthogonal to $\gamma$
\be
\Rgeo_Q(\gamma)=\{\theta : (\theta,\gamma) = 0 \}
\ee
which we call {\it geometric ray}\footnote{The rays are sometimes called walls of second kind or BPS walls. 
The word 'ray' avoids possible confusion with walls of marginal stability (or first kind), but admittedly is most adequate for two-dimensional scattering diagrams, where rays are in fact one-dimensional.}. 
At any point along $\Rgeo_Q(\gamma)$, the
notion of $Z$-semi-stability coincides with the notion of $\theta$-semi-stability (\ie $R$ is
$\theta$-semi-stable if and only if  
$(\theta, \gamma')\leq 0$ for any subrepresentation $R'\subset R$ of dimension vector $\gamma'$). 
In particular, the index $\Omega_Z(\gamma)\eqqcolon\Omega_\theta(\gamma)$  is independent of $\rho$. 
We then define the active ray $\Ract_Q(\gamma)$ as the subset  of $\Rgeo_Q(\gamma)$
\be
\Ract_Q(\gamma) = \{ \theta :  (\theta,\gamma) = 0,\quad \bOm_\theta(\gamma)\neq 0\}
\ee
Since $\bOm_\theta(\gamma)$  is invariant under 
rescaling $\theta\to \lambda\theta$ with $\lambda\in\IR_{>0}$, and since it can only jump on a finite set of hyperplanes $\Rgeo_Q(\gamma')$ corresponding to the destabilization by 
a subobject of dimension vector $\gamma'<\gamma$, the set of rays 
$\{ \Ract_Q(\gamma) : \gamma\in\Gamma_+\}$ decomposes into 
a complex of convex rational polyhedral cones in the space $\IR^{Q_0}$ of King stability conditions. 

Furthermore, to each point  $\theta\in \Ract_Q(\gamma)$ 
we associate an automorphism $\Uact_\theta(\gamma)$ of the quantum torus algebra 
$\hat\cT$ defined as follows. 
Let $\cT$ be the algebra $\IC(y)[[\cX_{\Gamma_+}]]$ generated by 
formal variables $\cX_\gamma$ for any $\gamma\in \Gamma_+$, with coefficients
in $\IC(y)$, subject to the relations 
\be
\label{qalg0}
\cX_\gamma\, \cX_{\gamma'} =(-y)^{\langle \gamma,\gamma'\rangle} \cX_{\gamma+\gamma'}
\ee
where 
\be
\langle \gamma,\gamma'\rangle=\sum_{a:(i\to j)\in Q_1} (n'_i n_j - n_i n'_j)
\ee
is the antisymmetrization of the 
Euler form $\chi_Q(\gamma,\gamma')=\sum_{i\in Q_0} n_i n'_i - \sum_{a:(i\to j)\in Q_1} n_i n_j'$. 
For any positive integer $M$, we denote by $\cT_{M}$ the ideal spanned by generators 
$\cX_\gamma$ with total dimension $\sum_{i\in Q_0} n_i >M$, and define the  
pro-nilpotent algebra $\hat\cT$ as the inverse limit of $\cT/\cT_M$ as $M\to\infty$.
We denote by $\hat\cG=\exp(\hat\cT)$ the corresponding pro-unipotent group. 
The automorphism $\Uact_\theta(\gamma)$ is the element of $\hat\cG$ defined by
\be
\label{defbUQ}
\Uact_\theta(\gamma) =\exp\biggl( \frac{\bOm_\theta(\gamma) \cX_\gamma}
{y^{-1}-y}\biggr)
\ee
The scattering diagram $\cD_Q$ can be defined as the set of decorated rays 
$\{\Ract_Q(\gamma) : \gamma\in\Gamma_+\}$ 
equipped with the automorphism $\Uact_\theta(\gamma)$
at each point. The wall-crossing formula for DT invariants ensures
that $\cD_Q$ is consistent in the following sense  \cite{bridgeland2016scattering}: for any 
generic closed path $\cP\colon t\in[1,0]\to\IR^{Q_0}$ (where generic means that
the intersection of $\cP$ with a ray $\Ract(\gamma_i)$ at $t=t_i$ is transverse
and does not meet any cone of codimension larger than one), 
 the  ordered product of automorphisms associated to each intersection is trivial, 
  \be
\label{constD0}
\prod_i  \cU_{\theta(t_i)} (\gamma_i)^{\epsilon_i} = 1\ ,\quad 
\epsilon_i=\sign \left( \frac{\de \theta}{\de t}, \gamma_i\right)
\ee
This consistency property  ensures that 
all rays can be deduced from the knowledge of the initial 
rays, defined as those rays $\Ract(\gamma)$ which contain the
self-stability condition $\theta_\star(\gamma):=\langle -,\gamma\rangle$ 
(such that $(\theta_\star(\gamma),\gamma')=\langle \gamma',\gamma\rangle$ for
all $\gamma'$). The attractor
tree formula of  \cite{Mozgovoy:2021iwz} and the flow tree formula of \cite{Arguz:2021zpx}
provide an algorithm to compute $\Omega_\theta(\gamma)$ on any ray in terms
of the attractor invariants  $\Omstar(\gamma)$ for the initial rays.

For a general quiver with potential $(Q,W)$, the determination of the initial rays is a difficult problem.
For an acyclic quiver however, it is easy to prove that the only initial rays are those associated to the
simple representations at each node \cite[Theorem 1.5]{bridgeland2016scattering}, with 
\be
\Omstar(\gamma_i)=1, \quad \Omstar(k\gamma_i)=0 \ \mbox{for}\ k>1
\ee 
More generally, one can show that $\Omstar(\gamma)=0$ unless the support of $\gamma$ (defined as the set of vertices $i\in Q_0$ such that $n_i\neq 0$) is strongly connected 
\cite[Theorem 3.8]{Mozgovoy:2020has}. For quivers associated to non-compact Calabi-Yau threefolds, it is conjectured in \cite{Beaujard:2020sgs,Mozgovoy:2020has,Descombes:2021snc} that 
$\Omstar(\gamma)=0$ unless $\gamma$ belongs to the kernel of the antisymmetrized Euler form.

\subsection{Scattering diagrams for Kronecker quivers\label{sec_kron}}

As an example which will play a central role later in this paper, let us consider the Kronecker quiver with 2 nodes and $\kappa$ arrows $a_i\colon 1\to 2$. Since the quiver is acyclic, 
the only initial rays are those associated to $\gamma_1=(1,0)$ and $\gamma_2=(0,1)$, 
with $\Omega_\theta(\gamma_1)$ and $\Omega_\theta(\gamma_2)=1$ along the axes
$\theta_1=0$ and $\theta_2=0$, respectively. For 
other dimension vector $(n_1,n_2)$, the moduli space $\cM_\theta(\gamma)$ 
has virtual dimension
\be
\label{expectdimQ}
\virdim_\IC \cM_\theta(\gamma)=\kappa n_1 n_2 - n_1^2 - n_2^2 + 1 
\ee
It is empty unless $\theta_1>0,\theta_2<0$ and \eqref{expectdimQ} is non-negative, in which case it
coincides with the actual dimension. 
Note that the formula \eqref{expectdimQ}  is  symmetric  under the exchange
$(n_1,n_2)\mapsto (n_2,n_1)$, and under 
$(n_1,n_2)\mapsto (n_1,\kappa n_2-n_1)$, which corresponds to a mutation of the quiver.

For the $A_2$ quiver ($\kappa=1$, left panel on Figure \ref{figKronScatt}), the only active ray in the quadrant $\theta_1>0,\theta_2<0$ 
is associated to $\gamma_1+\gamma_2$, with $\Omega_\theta(\gamma_1+\gamma_2)=1$
along the diagonal $\theta_1+\theta_2=0$, and the consistency condition \eqref{constD0}
reproduces the usual five-term relation
\be
\cU_{\gamma_1} \cU_{\gamma_2} = \cU_{\gamma_2} \cU_{\gamma_1+\gamma_2} \cU_{\gamma_1}   
\ee
For the affine $A_1$ quiver ($\kappa=2$, middel panel on Figure \ref{figKronScatt}), there is an infinity of active rays of the form $(k,k-1),(k-1,k)$ and $(k,k)$ with
$k\geq 1$, with DT invariant $1,1$ and $-y-1/y$, respectively. 
For  the Kronecker quiver with $\kappa\geq 3$, arrows there is a discrete set  of active
rays $(n_1,n_2)$ with unit DT invariant, obtained
by successive mutations of $(1,0)$ and $(0,1)$, 
and a dense set of rays in the region $\kappa-\sqrt{\kappa^2-4}< \frac{2n_1}{n_2}, \frac{2n_2}{n_1} < \kappa+\sqrt{\kappa^2-4}$. For $\kappa=3$ (right panel on Figure \ref{figKronScatt}), the discrete rays correspond to $(n_1,n_2)=(F_{2k},F_{2k\pm 2})$ where $F_{2k}$ are the even Fibonacci numbers $1,3,8,21,\dots$.
In Table \ref{Kroneckertab}, we tabulate some of the DT indices $K_\kappa(n_1,n_2)\coloneqq\Omega_\theta(n_1,n_2)$
for low values of $(\kappa, n_1,n_2)$,
restricting to the case $\kappa=0 \mod 3$ which is most relevant for the present work.
It is worth noting that for $(n_1,n_2)=(1,1)$, the moduli space of stable representations
 is $\IP^{\kappa-1}$, hence 
\be
\label{K11}
K_{\kappa}(1,1)=(-1)^{\kappa+1} \frac{(y^\kappa-y^{-\kappa})}{y-1/y} 
\stackrel{y\to 1}{\rightarrow} (-1)^{\kappa+1} \kappa
\ee
More generally, for $(n_1,n_2)=(1,n)$ or $(n,1)$, the moduli space  $\cM_\theta(\gamma)$ 
is the Grassmannian
of $n$-dimensional planes in $\IC^\kappa$, hence $K_{\kappa}(1,n)=0$ if $n>\kappa$. 
From the point of view of the  flow tree or attractor tree formulae, this vanishing occurs
as a result of cancellations between numerous different trees.

\begin{figure}[ht]
\begin{center}
\includegraphics[height=5.5cm]{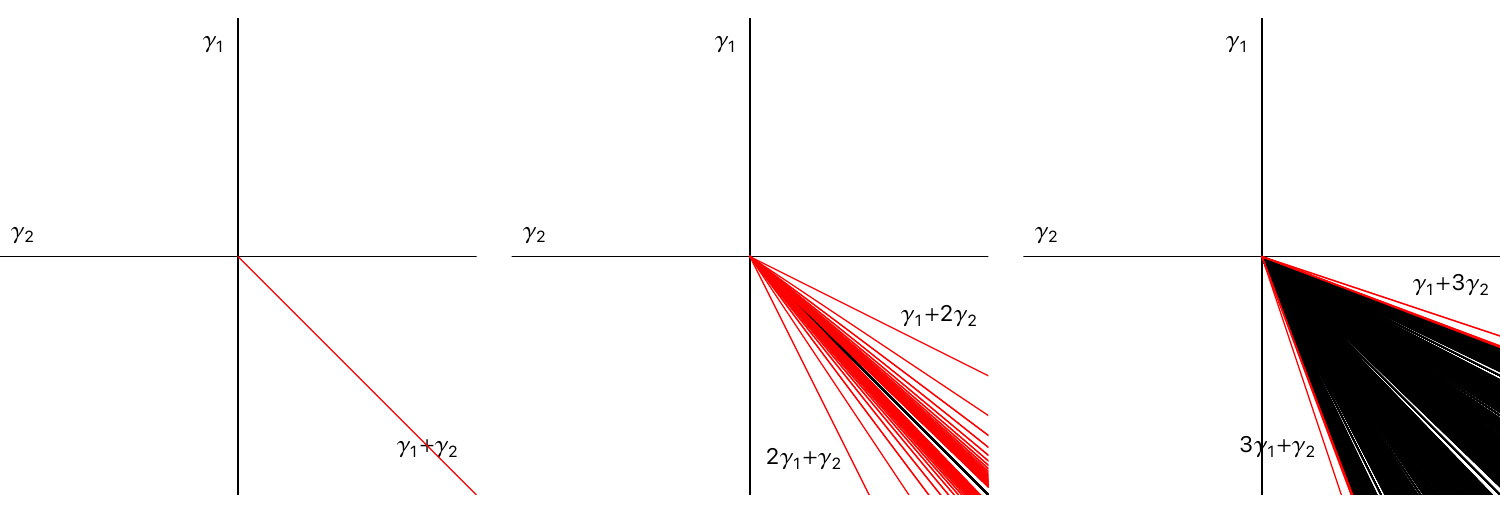}
\end{center}
\caption{Scattering diagram for the Kronecker quiver with $\kappa=1,2,3$ arrows \label{figKronScatt}}
\end{figure}

\begin{table}[ht]
\caption{
Table of indices $K_{\kappa}(n_1,n_2)$ for the Kronecker quiver with $\kappa$ arrows, dimension vector $(n_1,n_2)$ in the chamber $\theta_1>0,\theta_2<0$ with $n_1\theta_1+n_2\theta_2=0$. 
Negative powers of $y$ omitted in the dots are determined by invariance under Poincar\'e duality $y\mapsto 1/y$.
\label{Kroneckertab}}
\vspace{-.5\baselineskip}
$\renewcommand{\arraystretch}{1.05}
\begin{array}{|l|l|} \hline
 K_3(1,1) & y^2+1+\dots\\
 K_3(1,2) & y^2+1+\dots\\
 K_3(1,3) & 1 \\
 K_3(2,2) & -y^5-y^3-y-\dots\\
 K_3(2,3) & y^6+y^4+3 y^2+3+\dots\\
 K_3(2,4) & -y^5-y^3-y-\dots\\
 K_3(2,5) & y^2+1+\dots\\
 K_3(3,3) & y^{10}+y^8+2 y^6+2 y^4+2 y^2+2+\dots\\
 K_3(3,4) & y^{12}+y^{10}+3 y^8+5 y^6+8 y^4+10
   y^2+12+\dots\\
 K_3(3,5) & y^{12}+y^{10}+3 y^8+5 y^6+8 y^4+10
   y^2+12+\dots\\
 K_3(4,4) & -y^{17}-y^{15}-3 y^{13}-4 y^{11}-6 y^9-6 y^7-7 y^5-7 y^3-7
   y-\dots\\
 K_3(4,5) & y^{20}+y^{18}+3 y^{16}+5 y^{14}+10 y^{12}+14 y^{10}+23 y^8+30 y^6+41 y^4+46
   y^2+51+\dots\\
 K_3(5,5) & y^{26}+y^{24}+3 y^{22}+5 y^{20}+9 y^{18}+13 y^{16}+18 y^{14}+22 y^{12}+26
   y^{10}+28 y^8+30 y^6 \\
   &+30 y^4+31 y^2+31+\dots\\
 K_6(1,1) & -y^5-y^3-y-\dots\\
 K_6(1,2) & y^8+y^6+2 y^4+2 y^2+3+\dots\\
 K_9(1,1) & y^8+y^6+y^4+y^2+1+\dots\\
 K_9(1,2) & y^{14}+y^{12}+2 y^{10}+2 y^8+3 y^6+3 y^4+4
   y^2+4+\dots\\ \hline
\end{array}
$
\end{table}

\subsection{Stability scattering diagrams\label{sec_Scattdiag}}

The construction of the stability scattering diagrams in the space of King stability conditions
for quivers with potential can be generalized to the space of Bridgeland stability conditions on any triangulated category $\cC$ of CY3 type, at the cost of several complications.

The main complication is that charges $\gamma$ such that $\Omega_Z(\gamma)\neq 0$ 
are no longer restricted to a fixed cone $\Gamma_+$, although they are still restricted by 
the condition $\Im Z(\gamma)\geq 0$. Secondly, there is no analog of King stability conditions, so no reason to restrict to rays  where $\Re Z(\gamma)=0$. In this subsection, we define a family of 
scattering diagrams $\cD_\psi \subset \Stab\cC$ labeled by a phase $\psi\in\IR/2\pi\IZ$, supported on active rays where semi-stable objects with $\arg Z(\gamma)=\psi+\frac{\pi}{2}$ exist. 

We first need to properly define an analogue of the quantum torus algebra $\hat\cT$.
Let $\cT_\Gamma$ be the algebra $\IC(y)[\cX_{\Gamma}]$ generated by 
formal variables $\cX_\gamma$ for any $\gamma\in \Gamma=K(\cC)$, with coefficients
in $\IC(y)$, subject to the relations 
\be
\label{qalg}
\cX_\gamma\, \cX_{\gamma'} =(-y)^{\langle \gamma,\gamma'\rangle} \cX_{\gamma+\gamma'}
\ee
For any stability condition $\sigma\in\Stab\cC$,  phase $\psi\in \IR/(2\pi \IZ)$  
and mass cut-off $M>0$ we denote by  
$\cT_{\sigma,\psi,M}$ the ideal spanned by generators $\cX_\gamma$ with $\Im(e^{-\I\psi} Z_\sigma(\gamma))>M$. We define the pro-nilpotent algebra $\hat\cT_{\sigma,\psi}$ 
as the inverse limit
of $\cT/\cT_{\sigma,\psi,M}$ as $M\to\infty$, and the pro-unipotent group $\hat\cG_{\sigma,\psi}=\exp(\hat\cT_{\sigma,\psi})$.

For any phase $\psi\in \IR/(2\pi \IZ)$ and charge vector $\gamma\in\Gamma\backslash\{0\}$, we define the  geometric ray as the codimension-one locus inside $\Stab(\cC)$
\be
\label{defgeoray}
\Rgeo_\psi(\gamma) \coloneqq \{ \sigma=(Z,\sigma) \in \Stab(\cC) : 
\Re(e^{-\I\psi}Z(\gamma))=0, \Im(e^{-\I\psi}Z(\gamma))>0 \}
\ee
along which the argument of the central charge $Z(\gamma)$ is equal to $\psi+\frac{\pi}{2}$ 
modulo $2\pi$. We further define the active 
 ray $\Ract_\psi(\gamma)$ as a
subset of $\Rgeo_\psi(\gamma)$ along which 
the rational DT invariant $\bOm_Z(\gamma)$ defined in \eqref{defOmb} is not zero,
\be
\label{defelemray}
\Ract_\psi(\gamma) \coloneqq \{ \sigma :\  \Re(e^{-\I\psi}Z(\gamma))=0, \Im(e^{-\I\psi}Z(\gamma))>0,  \ 
\bOm_\sigma(\gamma)\neq 0\} 
\ee
At any point $\sigma$ along an active ray, we associate an automorphism
\be
\label{defbU}
\Uact_\sigma(\gamma) =\exp\biggl( \frac{\bOm_\sigma(\gamma) \cX_\gamma}
{y^{-1}-y}\biggr)
\ee
valued in $\cG_{\sigma,\psi}$. 
Since the geometric ray \eqref{defgeoray} is invariant under rescaling $\gamma$ by any positive
real number, and since the automorphisms $\Uact_\sigma(\gamma)$ associated to collinear charges
commute among each other, it is convenient to combine all active rays $\Ract_\psi(k\gamma)$
for a given primitive charge $\gamma\in \Gamma$
into a single `effective' ray 
\be
\label{defeffectiveray}
\Reff_\psi(\gamma) \coloneqq \{ \sigma :\  \Re(e^{-\I\psi}Z(\gamma))=0, \Im(e^{-\I\psi}Z(\gamma))>0,  \exists k\geq 1 \ ,
\Omega_\sigma(k\gamma)\neq 0\} 
\ee
equipped with the automorphism 
\be
\label{defUeff}
\Ueff_\sigma(\gamma) 
= \exp\biggl( \sum_{k=1}^{\infty} 
\frac{\bOm_\sigma(k\gamma) \cX_{k\gamma}}{y^{-1}-y} \biggr)
= \Exp\biggl( \sum_{k=1}^{\infty} 
\frac{\Omega_\sigma(k\gamma) \cX_{k\gamma}}{y^{-1}-y} \biggr)
\ee
where  
$\Exp$ is the plethystic exponential. When  $\gamma$ is non-primitive, we set 
$\Ueff_\sigma(\gamma)=\Ueff_\sigma(\gamma/\ell)$ where $\ell$ is the largest integer which divides $\gamma$, so that both $\Reff_\psi(\gamma)$ and $\Ueff_\sigma(\gamma)$ are invariant under rescaling $\gamma$
by a non-negative rational number. 

The stability scattering diagram $\cD_\psi(\cC)$ is defined as the union of
all active rays $\Ract_\psi(\gamma)$ with $\gamma\in\Gamma$, equipped with
their respective automorphism $\Uact_\sigma(\gamma)$ (equivalently, the union 
of all effective rays $\Reff_\psi(\gamma)$ with $\gamma$ primitive, equipped with  $\Ueff_\sigma(\gamma)$). We note that $\cD_\psi(\cC)$ is invariant under $\psi\mapsto\psi+\pi$
and $\psi\mapsto -\psi$,  upon relabeling the charges $\gamma$ into $-\gamma$ or 
$\gamma^\vee$ (defined below \eqref{eqreZ}),
respectively. In the following, we shall restrict to $\psi\in(-\frac{\pi}{2},\frac{\pi}{2}]$, such that 
the rays are supported on loci where semi-stable objects $\cA$ exist in the heart.

As in the case of quivers, the wall-crossing formula for DT invariants again ensures
that the scattering diagram $\cD_\psi(\cC)$ is consistent, but in a more restricted sense than in the quiver context: 
rather than considering an arbitrary closed path $\cP$, we pick any 
 point $\sigma$ on a codimension two intersection, 
 and consider an infinitesimal path $t\in[1,0]\to\Stab\cC$ circling counterclockwise\footnote{A change of orientation maps the product \eqref{constD} to its
inverse, so does not affect the consistency property, but it does exchange the notions
of incoming and outgoing rays. In \S\ref{sec_flow} we restrict to two-dimensional scattering diagrams
where the rays inherit a global orientation from the complex structure.}
 around $\sigma$ in a  two-dimensional plane containing $\sigma$ and intersecting the
rays $\cR_\psi(\gamma_i)$ transversally (and away from cones of codimension greater than one).
The consistency condition is then that the ordered product of automorphisms associated to each 
active ray $\Ract_\psi(\gamma_i)$
intersecting at $\sigma$ is trivial: 
 \be
\label{constD}
\prod_i  \cU_{\sigma(t_i)} (\gamma_i)^{\epsilon_i} = 1
\ee 
Here, $t_i$ are the intersection points of the path with all rays passing through~$\sigma$  and $\epsilon_i=\pm$ is given by 
\be
\label{defepsi}
\epsilon_i=\sign\Re\left( e^{-\I\psi} \frac{\de}{\de t} Z_{\sigma(t_i)}(\gamma_i)\right)
\ee
Rays such that  $\epsilon_i=1$ (respectively  $\epsilon_i=-1$) are called outgoing (respectively incoming) at the point $\sigma$.  
Expanding out the exponential in each factor and using the algebra \eqref{qalg}, 
the consistency property \eqref{constD} allows to determine
DT invariants on outgoing rays from the knowledge of DT invariants on incoming rays.
The result takes the form
\be
\label{AFTlocal}
\bOm^+_\sigma(\gamma)=\sum_{n\geq 1} \sum_{\gamma=\sum_{i=1}^n \gamma_i} 
\frac{g(\{\gamma_i\}) }{{\rm Aut}(\{\gamma_i\})} 
\prod_i \bOm^-_\sigma(\gamma_i)
\ee
where only a finite number of decompositions contribute (this follows from  
$\Im(e^{-\I\psi}Z_\sigma(\gamma_i))=|Z(\gamma_i)|$ and the support condition).
The coefficients $g(\{\gamma_i\})$ can be computed either using 
the attractor tree formula of~\cite{Mozgovoy:2021iwz} or the flow tree formula of~\cite{Arguz:2021zpx}\footnote{The Coulomb branch formula of \cite{Manschot:2010qz,Manschot:2014fua} gives yet another prescription, which gives
the same coefficient $g(\{\gamma_i\})$ whenever 
all charges $\gamma_i$ whose rays intersect $\sigma$ lie in a common two-dimensional sublattice, such that scaling solutions do not occur. }. 
Both involve a sum over rooted trees decorated with 
charges $\gamma_e$ along the edges and stability parameters $\theta_v$ at the vertices,
valued in an  auxiliary space $\IR^n$. The former involves trees of arbitrary valency
and stability parameters at the root vertex 
given by  $\theta_i=\Re(e^{-\I\psi} Z_\sigma(\gamma_i))$, while the latter involves binary
trees only, at the cost of perturbing the stability parameters at the root vertex. 

Since the consistency condition \eqref{constD} determines the outgoing rays from the incoming rays at each codimension-two intersection point of the scattering diagram, the full diagram is determined 
if one can identify a set of  initial rays, from which all other rays originate by repeated scattering. In the
context of quivers with potential, the initial rays  in the space
of King stability parameters are those which contain the self-stability condition 
$\theta_\star(\gamma):=\langle -, \gamma \rangle$. We do not know a simple characterization of
initial rays for a general scattering diagram, but in the context of scattering diagrams restricted
to the slice of $\Pi$-stability condition, it is natural to conjecture that they either 
 emanate from a boundary point of $\Pi$ where the central charge $Z_\sigma(E)$ tends to zero,
or from a regular attractor point $\sigma_\star(E)$ such that central charge 
$|Z_{\sigma_\star(E)}(E)|$ attains a local minimum with $\arg Z_\sigma(E)=\psi+\frac{\pi}{2}$.
Once such initial data has been fixed, by similar arguments as for standard scattering 
diagrams \cite{gross2011tropical}, 
the scattering diagram $\cD_\psi$ is
expected to be uniquely determined, and the DT invariant $\Omega_\sigma(\gamma)$ at any point $\sigma\in\Stab\cC$
can be read off from the automorphism $\cU_\sigma(\gamma)$ along the ray $\Ract_\psi(\gamma)$ 
with  $\psi=\arg Z_\sigma(\gamma)- \frac{\pi}{2}$
passing through the desired point $\sigma$.

\subsection{Attractor flows and Split Attractor Flow Conjecture\label{sec_SAF}}

We define the attractor flow $\AF(\gamma)$ as the flow on 
the physical slice of $\Pi$-stability conditions\footnote{More generally, one could 
consider the attractor flow along any complex subspace in the space of Bridgeland
stability conditions, such as the large volume slice, equipped with a hermitean metric.} induced by 
 the gradient of the modulus square of the central charge $Z(\gamma)$ for a fixed charge 
 $\gamma\in \Gamma\backslash\{0\}$. In local complex coordinates $z^a$ on $\Pi$, 
\be
\label{attflow}
\frac{\de z^a}{\de \mu} = - g^{a\bar b} \partial_{\bar b} | Z(\gamma) |^2
\ee
where $\mu$ is a coordinate parametrizing
the flow, and $g^{a\bar b}$ is the inverse of 
the K\"ahler metric $g_{a\bar b} \de z^a \de \bar z^{\bar b}$ on $\Pi$ 
specified by local mirror symmetry. Note that the flow depends only on the conformal
class of the metric, up to reparametrization of $\mu$.

A key property of \eqref{attflow} is that the modulus of the central charge necessarily decreases along the flow,
\be\label{attflow-mass}
\frac{\de}{\de \mu}  |Z(\gamma)|^2 = -2 \partial_a | Z(\gamma)|^2 g^{a\bar b} \partial_{\bar b} 
| Z(\gamma)|^2 \leq 0
\ee
and $z^a(\mu)$ must therefore reach a local minimum of $|Z(\gamma)|$ as $\mu\to\infty$, unless it encounters a singularity (or reaches the boundary) along the way. We denote by $z_\star(\gamma)$
the endpoint of the maximally extended flow; it is independent of the initial value of $z^a$ at $\mu=0$  within a given basin of attraction. If 
$Z_{z_\star(\gamma)}(\gamma)\neq 0$, the endpoint of the flow is said to be a
regular attractor point, otherwise it is a singular  attractor point. 
Following common lore, we define the attractor index 
$\Omstar(\gamma)$ as the limit of 
$\Omega_z(\gamma)$ along the flow $z\to z_\star(\gamma)$, 
and  similarly for the rational attractor
index $\bOmstar(\gamma)$. This definition overlooks the possibility that a given charge may admit different attractor points (depending on the basin of attraction $b$), in which case one should attach an
collection of attractor indices $\bOmstar^{(b)}(\gamma)$ to the same charge $\gamma$,
but we shall gloss over this complication in this work,  since it does not arise for local CY threefolds.

An attractor flow tree is an oriented rooted tree $T$, decorated with charges $\gamma_e\in \Gamma\backslash\{0\}$ along each edge $e\in E_T$ and equipped with a continuous map $\pi\colon T\to \Pi$ such that
\begin{enumerate}
\item  {\it (charge conservation)}  
\ At each vertex $v\in V_T$,
$\gamma_{v}=\sum_{e\in{\rm ch}(v)} \gamma_e$
where $\gamma_v$ is the charge along the edge ending at $v$ and ${\rm ch}(v)$ the set of
children edges leaving from $v$;
\item {\it (attractor flow along edges)}
\ For each edge $e\in E_T$, the map $\pi|_e\colon e\to\Pi$ is an embedding and its image $\pi(e)$ follows the  flow lines of $\AF(\gamma_e)$;
\item {\it  (marginal stability at vertices)} 
\ At each vertex $v\in V_T$ and  for each 
child edge  $e\in {\rm ch}(v)$, the phase of the central charge
$Z_{\pi(v)}(\gamma_{e})$ is equal to that of 
$Z_{\pi(v)}(\gamma_{v})$, and 
$\Im\bigl( Z_{\pi(v^+)}(\gamma_e) \overline{   Z_{\pi(v^+)}(\gamma_{v}) }\bigr) 
\langle \gamma_e, \gamma_v \rangle>0$, with $v^+$ 
a point 
infinitesimally close to 
$v$  along the edge ending at $v$;
\item The leaves of the tree $v_i$ are mapped to the attractor points $z_\star(\gamma_i)$,
where $\gamma_i$ is the charge along the edge ending at $v_i$.
%, such that $\bOmstar(\gamma_i)\neq 0$. 
\end{enumerate}
The tree $T$ is called active if $\bOm_Z(\gamma_e)\neq 0$ along each  edge.
We denote by $\cT_z(\{\gamma_i\})$ the set of attractor flow trees whose
root vertex $v_0$ is mapped to $\pi(v_0)=z\in \Pi$, and whose leaves carry charge $\gamma_i$.
For fixed charges $\{\gamma_i\}$, this set is obviously finite (most of the time empty).

\medskip

The Split Attractor Flow Conjecture, originally proposed in \cite{Denef:2000nb,Denef:2001xn}
and sharpened in \cite{Denef:2007vg}, amounts to the statement that for any $z\in\Pi$ and $\gamma\in\Gamma$, there exists only a finite number 
of decompositions $\gamma=\sum_i \gamma_i$ such that $\cT_z(\{\gamma_i\})$ is non-empty; the rational index $\bOm_z(\gamma)$ is then obtained by summing over all
attractor flow trees,   
\be
\label{SAF}
\bOm_z(\gamma) = \sum_{\gamma=\sum_i \gamma_i} 
\frac{1}{\abs{\Aut(\{\gamma_i\})}} 
\sum_{T\in \cT_z(\{\gamma_i\})} g_z(T)\, 
\prod_{i} \bOm_\star(\gamma_i)
\ee
Here, $\abs{\Aut(\{\gamma_i\})}$ is a symmetry factor, given by the order of the subgroup of permutations
of $n$ elements which preserves the ordered list of charges $\{\gamma_i\}$, and the prefactor 
$g_z(T)$ is obtained recursively by applying the wall-crossing formula at each of the vertices of the tree. In the case of a binary tree, such that each vertex $v$
has two descendants edges ${\rm ch}(v)=\{L(v), R(v)\}$ (defined up to exchange), it is simply given
by a product over all vertices,
\be
\label{gT}
g_z(T) = \prod_{v\in V_T} (-1)^{\langle\gamma_{L(v)},\gamma_{R(v)}\rangle+1}
| \langle\gamma_{L(v)},\gamma_{R(v)}\rangle |
\ee
or in the case of refined invariants,
\be
\label{gTref}
g_z(T,y) = \prod_{v\in V_T} (-1)^{\langle\gamma_{L(v)},\gamma_{R(v)}\rangle+1}
\sign(\langle\gamma_{L(v)},\gamma_{R(v)}\rangle )\,
\frac{y^{\langle\gamma_{L(v)},\gamma_{R(v)}\rangle}-y^{-\langle\gamma_{L(v)},\gamma_{R(v)}\rangle}}{y-y^{-1}}
\ee
More generally, if the tree has vertices with higher valency, the prefactor $g_z(T)$ is given by 
the product of the local factors appearing in the formula \eqref{AFTlocal}, evaluated on the charges which descend from the vertex $v$:
\be
g_z(T) =  \prod_{v\in V_T} g(\ch(v))
\ee

Alternatively, in cases where the Dirac pairing $\langle -,-\rangle$ is degenerate (as is the case for local CY threefolds) such that there exists a non-zero $\delta\in\Gamma$ with $\langle \delta,-\rangle=0$, one may perturb the charges $\gamma_i \to \gamma_i + \epsilon_i \delta$ where
$\epsilon_i$ are small parameters, such that the attractor flow tree is resolved into a union of binary trees, for which \eqref{gT} or \eqref{gTref} can be applied. Such a perturbation is in fact a
necessary step when applying the flow tree formula \cite{Alexandrov:2018iao,Arguz:2021zpx} at each vertex, although it need not extend to a global perturbation of the full set of trees in 
$\cT_z(\{\gamma_i\})$. 

It is worth noting that there can be cancellations between different trees with  the same embedding. This occurs for example for the Kronecker quiver with $\kappa$ arrows and dimension
vector $(1,n)$ with $n>\kappa$, as noted below \eqref{K11}, and can be viewed as a consequence of the Pauli exclusion principle for multi-centered black holes \cite{Denef:2002ru}. A {\it weak} version of the 
Split Attractor Flow Conjecture is then that the number of {\it active} attractor flow trees
is finite. 
 
While the formula \eqref{SAF} is largely a consequence of the wall-crossing formula, 
the main problem is to ensure that only a finite number of 
decompositions $\gamma=\sum \gamma_i$ can occur as leaves of attractor flow trees, and to
provide an algorithm for finding them in practice.
The fact that the modulus of the central charge $|Z(\gamma_e)|$ decreases along each edge, 
and is additive at each vertex, implies that the central charges of the constituents are bounded 
by 
\be
\sum_i |Z_{z_\star(\gamma_i)}(\gamma_i) | \leq |Z_z(\gamma)|
\ee
This constraint was used in \cite[\S C]{Denef:2007vg} to show that, in the context
of compact CY threefolds at large radius, the number of attractor
flow trees terminating in a fixed compact region of K\"ahler moduli space is finite. 
However, this constraint becomes less and less stringent in regions where $|Z_z(\gamma)|$
becomes large (\eg in the large volume limit), and  moot for constituents which are massless 
at their respective attractor points. 
As we demonstrate in \S\ref{sec_inidata},  the Split Attractor Flow Conjecture (in its strong form)
holds for $K_{\IP^2}$ along the large volume slice defined by the quadratic central charge 
\eqref{defZLV}. In \S\ref{sec_exactD} we describe how the SAFC along the slice 
of $\Pi$-stability conditions is proven for~$K_{\IP^2}$ (Theorem~\ref{thm:SAFC}).

\subsection{From scattering sequences to attractor flow trees\label{sec_flow}}

The connection between the scattering diagram $\cD_\psi(\cC)$ and the attractor flow is based
on the observation that when $Z(\gamma)$ is a holomorphic function on $\Pi$, its argument is constant along $\AF(\gamma)$,
\be
\frac{\de }{\de \mu}\arg Z(\gamma)
= \Im \Bigl( \frac{\de}{\de \mu} \log Z(\gamma) \Bigr)
= \Im \Bigl( -\partial_a Z(\gamma) g^{a\bar b} \partial_{\bar b}\bar Z(\gamma)  \Bigr)
= 0
\ee
This property holds for any non-compact CY threefold \cite{Denef:1998sv}, and is tied to the fact that the central charge of the D0-brane is independent of K\"ahler moduli\footnote{On a compact CY threefold, the physical central charge involves an extra factor $Z_{D0}=e^{-K} X_0$ equal to the 
D0-brane central charge, 
which is not holomorphic and cannot be trivialized by a K\"ahler transformation. }.
A second important consequence of the holomorphy of $Z$ is that minima of 
$|Z(\gamma)|$ can only occur when $Z(\gamma)=0$ or at the boundary, hence
regular attractor points never occur. Since $\bOm_z(\gamma)=0$
 at a point $z$ where $Z(\gamma)=0$ by the support property~\eqref{supp-prop},
 the endpoint of the flow must be a singular point in K\"ahler moduli space, hence belongs to the boundary of the space of stability conditions.

Since the argument of $Z(\gamma)$ is constant along the flow, the flow lines of $\AF(\gamma)$ must lie inside the geometric ray $\Rgeo_\psi(\gamma)$, where $\psi$ is
fixed in terms of the argument of the central charge at $\mu=0$. Since the flow is only meaningful when $\bOm(\gamma)\neq 0$, and can only split when the
central charge of the descendants $Z(\gamma_e), e\in {\rm ch}(v)$ become aligned
with that of the incoming charge $\gamma_v$ at every vertex $v$, 
the whole attractor flow tree is in fact embedded inside the scattering diagram $\cD_\psi(\cC)$, with each vertex $v$ lying along the intersection of active rays along the wall of marginal stability 
$\cW(\gamma_v,\gamma_e)$.

Specializing further to the case where the slice of $\Pi$-stability conditions has complex dimension~$1$, which is of main interest in this paper,  both the flow lines and the rays have real dimension one, so the intersection $\cD_\psi(\cC) \cap \Pi$
must in fact coincide with the set of all possible attractor flows carrying central charges of fixed argument
$\psi+\frac\pi{2}$. Moreover, the complex structure on $\Pi$ induces an orientation 
of the plane around each intersection point, such that the orientation of the rays defined locally 
by  \eqref{defepsi} extends to a global orientation, in which incoming and outgoing rays have decreasing and increasing values of
$\Im(e^{-\I\psi} Z(\gamma_i))=| Z(\gamma_i)| $, respectively.
For a given charge $\gamma\in\Gamma$ and point $z\in \Pi$, the set of split 
attractor flows determines all possible scattering sequences, starting from initial rays 
$\Ract_\psi(\gamma_i)$ specified by the leaves of the tree and producing an active ray 
$\Ract_\psi(\gamma)$ going through the desired point $z\in \Pi$. In the (admittedly restrictive) context of local CY threefolds with a single K\"ahler modulus, the Split Attractor Flow Conjecture then amounts to the finiteness of the number of such scattering sequences.

%%%%%%%%%%%%%%%%%%%%%%%%

\section{The large volume scattering diagram\label{sec_LV}}

In this section, we determine the scattering diagram $\cD^{\rm LV}_\psi$ for the category $\cC=D^b(\Coh K_{\IP^2})$  along the large volume slice $(Z^{\rm LV}_{(s,t)},\cA(s))$ with
$(s,t)\in\IR \times\IR^+$, defined by the 
quadratic central charge \eqref{defZLV}
\be
\label{defZLVst}
Z^{\rm LV}_{(s,t)} (\gamma)= - \frac{r}2  (s+\I t)^2+ d (s+\I t) -\ch_2
\ee  
and heart $\cA(s)= \Coh^{\sharp s}$ defined in \S\ref{sec_Pistab}.
For brevity we write \eqref{defZLVst} as $Z(\gamma)$ in this section.
The scattering diagram  $\cD^{\rm LV}_\psi$ for $\psi=0$ was analyzed in
\cite{Bousseau:2019ift},
using adapted coordinates $(x=s,y=-\frac12(s^2-t^2))$ such that the
rays $\Re Z(\gamma)=0$ are straight lines $ry+dx-\ch_2=0$ lying above the parabola $y=-\frac12 x^2$. We shall recast the construction in coordinates $(s,t)$, which have a more transparent relation to the coordinate $\tau$ on the physical slice, and generalize it to any phase $\psi\in(-\frac{\pi}{2},\frac{\pi}{2})$. We shall also describe an algorithm for determining which scattering sequences can contribute to the index $\Omega_{(s,t)}(\gamma)$ for given charge $\gamma$
and moduli $(s,t)$, and apply it to compute the Gieseker index $\Omega_\infty(\gamma)$
for small Chern vectors. As a consequence of this algorithm, the Split Attractor Flow Conjecture
holds along the large volume slice.

\subsection{Scattering rays and walls of marginal stability\label{sec_raywall}}
For simplicity we first consider the case $\psi=0$.  Evaluating the real and imaginary parts of \eqref{defZLVst},
\be
\Re Z(\gamma) = -\frac12 r (s^2 - t^2 ) + d s - \ch_2\ ,\quad 
\Im Z(\gamma) = t (d-s r)
\ee
we readily see that the geometric rays $\Rgeo_0(\gamma)=\{ (s,t) : \Re(Z(\gamma))=0, \Im(Z(\gamma))>0\}$ are given as follows:\footnote{Note that there are no rays associated to skyscraper sheaves with $r=d=0$.}
\begin{itemize}
\item $r\neq 0$ and $\Delta>0$:  a branch of  hyperbola intersecting the real axis at 
$s=\mu -\sign(r) \sqrt{2\Delta}$ (where $Z(\gamma)$ vanishes) and asymptoting  to $s=\frac{d}{r}-t \sign r$ from below;
\item $r\neq 0$ and $\Delta=0$:  the half-line $s=\frac{d}{r}-t \sign r$;
\item  $r\neq 0$ and $\Delta<0$:  a branch of hyperbola starting at $(s,t)=(\mu,\sqrt{-2\Delta})$ (where $Z(\gamma)$ vanishes) and  asymptoting to 
$s=\frac{d}{r}-t \sign r$   from above;
\item $r=0$ and $d>0$: a vertical line at $s=\ch_2/d$;
\item $r=0$ and $d\leq 0$: the geometric ray is empty.
\end{itemize}
Recall that the slope $\mu$ and discriminant $\Delta$ were defined in \eqref{defmuDelta}. 
Moreover, the rays are oriented in the direction of increasing $t$ (or equivalently decreasing $r s$ when $r\neq 0$), such that the modulus of the central charge $|Z(\gamma)|=\Im Z(\gamma)$ increases along the ray.

On the other hand, the walls of marginal stability $\cW(\gamma,\gamma')$ where the phases of $Z(\gamma)$ and $Z(\gamma')$ align (or anti-align) are given by the vanishing of 
\be
\label{WallLV}
W(\gamma,\gamma')\coloneqq \frac{\Im( Z(\gamma') \overline{Z}(\gamma))}{t}
\coloneqq \frac12(s^2+t^2) (rd'-r'd) -s (r\ch_2'-r'\ch_2) +  (d\ch_2'-d'\ch_2)
\ee
When $rd'-r'd=0$, $\cW(\gamma,\gamma')$ is a vertical line at $s_{\gamma,\gamma'}=\mu=\mu'$, otherwise, it is 
a semi-circle centered at $(s_{\gamma,\gamma'},0)$ of radius  
$R_{\gamma,\gamma'}$,
with
\begin{align}
\label{sRform}
s_{\gamma,\gamma'}&\coloneqq\frac{r\ch_2'-r'\ch_2}{rd'-r'd} =\frac12 (\mu+\mu') - \frac{\Delta-\Delta'}{\mu-\mu'} \\
R_{\gamma,\gamma'} &\coloneqq \sqrt{s_{\gamma,\gamma'}^2 - 2 \frac{d\ch_2'-d'\ch_2}{rd'-r'd}} 
= \sqrt{ (s_{\gamma,\gamma'} - \mu)^2-2\Delta} =  \sqrt{ (s_{\gamma,\gamma'} - \mu')^2-2\Delta'}
\end{align}
We shall denote this half-circle by $\cC(s_{\gamma,\gamma'}, R_{\gamma,\gamma'})$. 
Whenever distinct geometric rays $\Rgeo_0(\gamma)$ and $\Rgeo_0(\gamma')$ intersect, they do so  
 at the highest point $(s_{\gamma,\gamma'},R_{\gamma,\gamma'})$ 
along the half-circle $\cW(\gamma,\gamma')$, and bound states exist on the side of the wall where $t$ is large, \ie
$(r d'-r' d)W(\gamma,\gamma')>0$. Assuming that $\Delta\geq 0$ and $\Delta'\geq 0$, the intersection is not empty provided 
$\mu\neq\mu'$ and, depending on the signs of $r,r'$,   
\be
\label{condint}
\left\{
\begin{array}{@{}rll}
i) & r>0, r'<0 :& \mu'+\sqrt{2\Delta'}<  \mu-\sqrt{2\Delta} \\
ii) & r>0, r'>0 \text{ and } \mu<\mu' :& \mu'-\sqrt{2\Delta'}< \mu-\sqrt{2\Delta} \\
iii) & r<0, r'<0 \text{ and } \mu<\mu' :& \mu'+\sqrt{2\Delta'}< \mu+\sqrt{2\Delta} \\
iv) & r>0, r'=0 :& \frac{\ch'_2}{d'} <   \mu-\sqrt{2\Delta} \\
v) & r=0, r'<0 :& \mu'+\sqrt{2\Delta'} < \frac{\ch_2}{d}  
\end{array}
\right.
\ee
The remaining cases ($r<0,r'>0$, or $r<0,r'=0$, or $r=0,r'>0$, or $r,r'\gtrless 0$ with $\mu>\mu'$) are given by exchanging 
$\gamma$ and $\gamma'$, while the intersection is evidently empty if $r=r'=0$. 
The case $i)$ is depicted in Figure \ref{figlocalscatt}, other cases can be understood similarly. It is interesting to note that whenever the intersection is not empty, one has 
\be
\label{condpos}
d d' - r \ch'_2 - r' \ch_2 \geq 0 
\ee
In cases $i)-iii)$, this follows by writing $\ch_2=\frac{r}{2}s(2\mu-s)$, where
$s=\mu-\sign r\, \sqrt{2\Delta}$ is the point where the ray $\Rgeo_0(\gamma)$ crosses the real axis, and similarly for  $\Rgeo_0(\gamma')$, such that 
\be
d d' - r \ch'_2 - r' \ch_2  = 
\frac{rr'}{2} \left[ (\mu-s) ^2 + (\mu-s' )^2 - (\mu-\mu)^2 \right]
\ee
which is manifestly positive. Cases $iv)$ and $v)$ are obvious. 
If $\Delta$ or $\Delta'$ is negative, the conditions \eqref{condint} for non-empty intersection
continue to hold upon setting
$\sqrt{2\Delta}\to 0$ or $\sqrt{2\Delta'}\to 0$, but the condition 
\eqref{condpos} no longer needs to hold.

\begin{figure}[ht]
\begin{center}
\includegraphics[height=7cm]{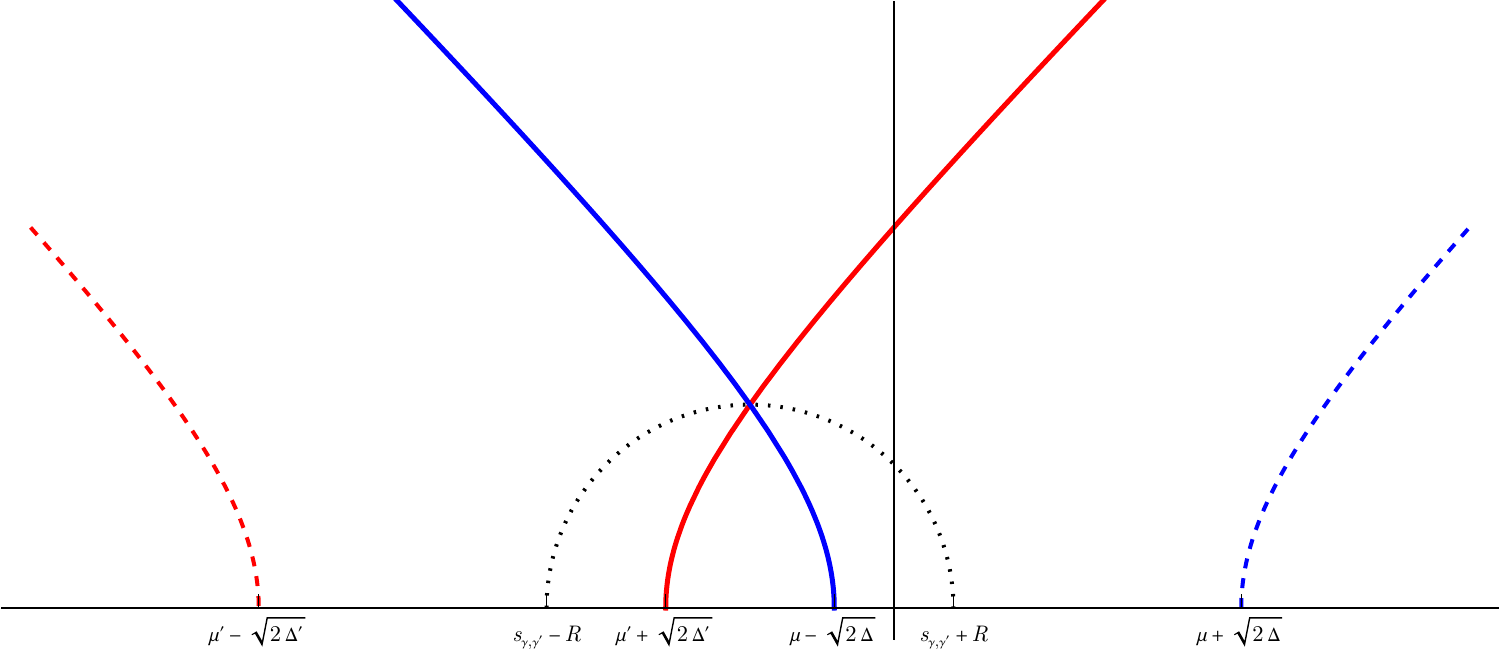}
\end{center}
\caption{Assuming $r>0, r'<0, \mu'+\sqrt{2\Delta'}<\mu-\sqrt{2\Delta}$ with $\Delta,\Delta'>0$, 
the rays $\Rgeo_0(\gamma)$ (in blue) and $\Rgeo_0(\gamma')$ (in red) 
intersect at the apex
of the wall of marginal stability  $\cW(\gamma,\gamma')$ (in dotted black). The dashed rays corresponding to $\Rgeo_0(-\gamma)$ 
and $\Rgeo_0(-\gamma')$ do not intersect. Outgoing rays, not shown here, are emitted in the forward region between the blue and red lines.  
For $\psi\neq 0$ with $|\psi|<\frac{\pi}{2}$, the intersection point is displaced by an angle $\psi$ along the dotted circle.
\label{figlocalscatt}}
\end{figure}

The structure of the walls of marginal stability for fixed charge $\gamma$ was analyzed in \cite{arcara2013minimal,bertram2014birational,
maciocia2014computing}. The main result is that
for $r>0$, there is a finite number
of walls, forming two sequences\footnote{This duplication is due to our definition
of DT invariants, which count semi-stable objects of charge $\epsilon\gamma$
where $\epsilon Z(\gamma)\in\IH_B$.}
 of nested half-circles
on either side of a vertical wall  at $s=\frac{d}{r}$.
For $(s,t)$ outside the largest walls (sometimes called Gieseker walls) the index  
$\Omega_{(s,t)}(\gamma)$ agrees with the Gieseker index  $\Omega_{\infty}(\gamma)$, justifying the notation.
For $r=0$ and $d\neq 0$, there is a single sequence of nested walls, and 
 $\Omega_{(s,t)}(\gamma)$ similarly agrees with the Gieseker index  $\Omega_{\infty}(\gamma)$ for  outside the  largest wall. 
When  $\gamma$ is a multiple of the Chern vector
associated to $\cO(m)$ for some $m\in \IZ$, the chamber structure is actually trivial, 
with $\Omega_{(s,t)}(\gamma)=\delta_{r,1}$ for any $s<m$ and $t>0$.
As each wall in the nested sequence is crossed, 
the moduli space $\cM_\sigma(\gamma)$ of $\sigma$-semistable objects undergoes a birational transformation contracting
a particular nef divisor associated to the wall (but keeping the dimension $\dim_\IC\cM_\sigma(\gamma)$
unchanged), until the innermost wall (sometimes known as collapsing wall) is crossed, after which $\cM_\sigma(\gamma)$ becomes empty. In particular, the index vanishes at the point $(\mu\pm\sqrt{2\Delta},0)$ or $(\ch_2/d,0)$, respectively, where $Z(\gamma)=0$,
which always lies inside the innermost wall. 
In the rest of this section, we shall confirm this structure using the scattering diagram. Before
doing so, we make a couple of remarks.

\begin{itemize}
\item Viewing 
$(s,t)$ as `space' and `time' coordinates in a two-dimensional Minkovski space, one may think of  
the active  ray $\Ract_0(\gamma)\subset\Rgeo_0(\gamma)$ as the worldline of a fictitious relativistic particle
of mass-squared  $m(\gamma)^2=r^2\Delta=\frac12 d ^2 - r \ch_2$, global charge $\gamma$ and electric charge $q=r$ and immersed in a constant electric field.  Indeed, when $r\neq 0$ the trajectory for such a particle is a branch of 
hyperbola  asymptoting to the light cone $s\simeq t\sign r$, while it is of course a straight line inside the light-cone when $r=0$. When two such particles collide, their charges add up and their mass increases by 
\be
m^2(\gamma+\gamma')-m^2(\gamma)-m^2(\gamma')=dd'-r\ch'_2-r'\ch_2
\ee
which is positive when $\Delta,\Delta'\geq 0$ by virtue of \eqref{condpos}. 
This analogy is not perfect, since the initial position and velocity are
fixed by the charge (in particular, the only allowed trajectory for $r=0$ is a vertical line).

\medskip

\item Two key properties following from this analogy are that the rays remain within the forward lightcone 
$|\delta s|\leq \delta t$ (the {\it causality} property), and that the `electric potential'
\be
\label{defvarphi}
\varphi_s(\gamma) \coloneqq 2(d - s r)
\ee
can only increase as $t$ increases (for $r=0$, it remains constant since the ray is vertical). The fact that  $\Im Z(\gamma)=t \varphi_s(\gamma) $ increases with $t$ is an obvious consequence of the monotonicity of $|Z(\gamma)|$ along the attractor flow, but the point is that $\Im Z(\gamma)$ increases faster than $t$. These properties will be instrumental in the next subsection for obtaining bounds on allowed constituents for a given total charge $\gamma$.

\medskip

\item In the simplest case of a collision of two incoming rays 
$\Reff_0(\gamma_1)$ and $\Reff_0(\gamma_2)$ with primitive vectors $\gamma_i$, each of them having $\Omega(k\gamma_i)=\delta_{k,1}$,  there is an infinite fan of outgoing effective rays of the form 
$\Reff_0(n_1\gamma_1+n_2\gamma_2)$,
carrying an index $\Omega(n_1\gamma_1+n_2\gamma_2)=K_{\kappa}(n_1,n_2)$
given  by the DT invariant  of a Kronecker quiver with $\kappa=|\langle \gamma_1,\gamma_2\rangle|$ arrows (a multiple of 3 due to \eqref{DSZ}) and dimension vector $(n_1,n_2)$
 (see Table \ref{Kroneckertab} for a table of some relevant values). 
More generally,  the outgoing rays at each collision are determined 
from the incoming rays by requiring that the product \eqref{constD} of automorphisms 
$\cU({\gamma_i})^{\epsilon_i}$ around the 
collision point is equal to one, or equivalently by using 
 the flow tree formula of \cite{Alexandrov:2018iao,Arguz:2021zpx},
 with the indices associated to incoming rays playing the role of attractor indices. 

\end{itemize}

\subsection{Initial rays and scattering sequences\label{sec_inidata}}
The main result of \cite{Bousseau:2019ift} was to identify the initial rays of the scattering diagram 
$\cD^{\rm LV}_0$. Namely, it was shown that the initial rays form two infinite families
associated to $\cO(m)$ and $\cO(m)[1]$, 
with charge $\gamma_m=[1,m,\frac12m^2]$ and $-\gamma_m$ respectively, 
emitted from the points $(m,0)$ along the boundary, for any $m\in \IZ$.  
Moreover, the automorphism $\Ueff(\gamma)$ on each such ray is specified by 
\be
 \Omstar(\pm k\gamma_m)=\delta_{k,1}\ ,\quad \bOmstar(\pm k\gamma_m)=\frac{y-y^{-1}}{k(y^k-y^{-k})}
\ee
for every $k\geq 1$. 
The fact 
that initial rays can only emanate from integer points along the boundary follows by noting that 
in any triangle bounded by $t<\max(s-m,m+1-s)$ in the $(s,t)$ plane, $\sigma$-semistable objects must be in the heart 
of the quiver associated to the tilting sequence $\langle E_1(m-1),E_2(m-1),E_3(m-1) \rangle$ (see \eqref{excepcoll}), but such objects necessarily have $\Re Z(\gamma)\neq 0$. The invariance
of the initial rays under translations $(s,t)\mapsto(s+1,t)$ is a straightforward consequence of the action of auto-equivalences $E\mapsto E(1)$ on $\cC$, while the invariance under $(s,t)\mapsto (-s,t)$ 
is consistent with the derived duality \eqref{eqreZ} in the 
large volume limit.

Given this simple structure for the initial rays, the consistency requirement  \eqref{constD} at any intersection determines in principle the automorphism $\Ueff_\sigma(\gamma)$ at any point 
$(s_0,t_0)$ along the geometric ray $\Rgeo_0(\gamma)$ for any primitive charge vector~$\gamma$, and therefore the indices $\Omega_{(s,t)}(k\gamma)$ along the same ray. In particular, the invariance of initial rays under  translations 
$(s,t)\mapsto(s+k,t)$ makes it obvious that the resulting indices will be invariant under the action of 
auto-equivalences $E\mapsto E(k)$,
\be
\label{Omflow}
\Omega_{(s,t)}\left([r,d,\chi)\right) = 
\Omega_{(s+k,t)} \left( [r, d+ k r , \chi+ k( d+ \tfrac32 r) +  \tfrac{k^2}{2} r) \right) \ ,\quad k\in \IZ
\ee
while the invariance of initial rays under  $(s,t)\mapsto (-s,t)$ 
implies the symmetry under derived duality 
\be
\label{Omdual}
\Omega_{(s,t)}( [r,d,\chi) ) = \Omega_{(-s,t)}( [-r, d, 3d-\chi) )
\ee
For $r\neq 0$ and $2d$ divisible by $r$, this can be combined with \eqref{Omflow}
to conclude that $\Omega_{(s,t)}( [r,d,\chi) )$ is invariant under $s\mapsto \frac{2d}{r}-s$.
Similarly, for $r=0$ and $2\chi$ divisible by $d$ one concludes
 that $\Omega_{(s,t)}( [0,d,\chi) )$  is invariant under
$s\mapsto \frac{2\chi}{d}-3-s$.

In  order to compute the indices in practice, the difficulty is to determine which rays, among the infinite set of initial rays associated to $\cO(m)$ and $\cO(m)[1]$, can produce an outgoing ray with the desired charge 
$\gamma$ passing through  the desired point $(s_0,t_0)$. The electromagnetic analogy
mentioned in the previous subsection gives a way to tackle this apparently formidable problem.
Indeed, the problem amounts to determining the set of all possible  particles of charge
$\gamma_i=k_i \ch  \cO(m_i)$ ($i=1,\dots n$) and anti-particles of charge $\gamma'_j= - k'_j \ch \cO(m'_j)$ ($j=1,\dots n'$),
emitted from the boundary $t=0$ at spatial positions $s=m_i$ and $s=m'_j$, respectively, such that 
their scattering products contains a particle of charge $\gamma$ going through  $(s_0,t_0)$.

A necessary condition is of course that  all initial particles lie in the past light-cone of $(s_0,t_0)$
and their global charges add up to the desired charge,
\be
\label{pastLC}
s_0-t_0\leq m_i,m'_j\leq s_0+t_0
\ee
\be
\label{totcharge}
[r,d,\chi) = 
\sum_{i=0}^{n-1} k_i \left[ 1, m_i, 1+\frac{m_i(m_i+3)}{2} \right) 
- \sum_{j=0}^{n'-1} k'_j \left[ 1, m'_j, 1+\frac{m'_j(m'_j+3)}{2} \right) 
\ee
but there can be cancellations between the charges of $k_i\cO(m_i)$ and $k'_j\cO(m'_j)[1]$ so these requirements alone do not yet give a finite set of 
decompositions. One can further reduce  the set of possible decompositions to a finite list  by exploiting the monotonicity of the `electric potential'  defined in \eqref{defvarphi}.
 Indeed, since each of the initial particles and anti-particles is emitted along the `light-cone' 
 $t=|s-m_i|$ or $t=|s-m'_j|$ at the boundary, the first scattering can only take place after a time $t\geq \frac12$, by which time the electric potential $\varphi_s(\gamma_i)$ has increased from 0 to $k_i$, and similarly for $\gamma'_j$. This immediately gives a bound
 \be
\label{bound1}
 \sum_{i=0}^{n-1} k_i + \sum_{j=0}^{n'-1} k'_j \leq \varphi_{s_0}(\gamma)
 \ee
 In particular, since the multiplicities $k_i$ and $k'_j$ are $\geq 1$, the  number $n+n'$ of possible constituents is bounded by  $\varphi_{s_0}(\gamma)$.
 Moreover, ordering the slopes $m_i, m'_j$ such that  
\be
m'_0 \leq m'_1 \leq \dots \leq m'_{n'-1} \ , \quad 
m_{n-1} \leq \cdots \leq m_1 \leq m_0 
\ee
with $k_i\leq k_{i-1}$ whenever $m_i=m_{i-1}$, and similarly $k'_j\leq k'_{j-1}$ whenever $m'_j=m'_{j-1}$, it is clear that the scattering between the left-moving particles $k_i\cO(m_i)$ 
and right-moving anti-particles $k'_j\cO(m'_j)[1]$ can only take place if $m'_{n'-1} < m_0$ and 
$m_0'<m_{n-1}$. Denoting by $(s_1,t_1)$ the point where the last scattering takes place,
such that $m_0'<s_1<m_0$,  the left-most anti-particle must accumulate at least $2s_1-2m'_0-1$ electric potential in going
from $m'_0+\frac12$ to $s_1$, while the right-most particle must accumulate at least $2m_0-2s_1-1$ electric potential in going
from $m_0-\frac12$ to $s_1$. Thus, the bound  \eqref{bound1}
can be strengthened to 
\be
\label{bounddec}
\sum_{i=0}^{n-1} k_i + \sum_{j=0}^{n'-1} k'_j + 2 (m_0 - m'_0 -1) \leq \varphi_{s_0}(\gamma)
\ee
This restricts the initial positions $m_i, m'_j$  to a finite interval around~$s_1$
that is typically tighter than the causality bound~\eqref{pastLC} for large $t$.  At any rate, the set of allowed initial integer
 positions $m_i, m'_j$  and multiplicities $k_i,k'_j$ is finite. 
 Thus, the list of possible decompositions having  a particle of charge 
 $\gamma$ at position $(s_0,t_0)$ among all their scattering products is finite, and the 
 SAFC holds along the large volume slice, for trees rooted at a point where $\Re Z(\gamma)=0$ (we relax this restriction in \S\ref{sec_LVpsi}).
 
 Unfortunately, many of decompositions  $\gamma=\sum_{i=1}^n \gamma_i + \sum_{=1}^{n'} \gamma'_j$ which satisfy \eqref{totcharge} and \eqref{bounddec} turn out to not 
include $\gamma$ among their scattering products, as the worldlines of the constituents may fail to intersect. To determine which of them do,  one way is to construct all possible binary trees with constituents of charge $\gamma_i$ and $\gamma'_j$, and retain those that satisfy the conditions  \eqref{condint}  at each vertex. Identifying vertices connected by edges of vanishing length (\ie mapped to the same point in $(s,t)$ plane), one obtains a set of attractor flow trees associated to the decomposition. The contribution of each attractor flow tree to the total index 
$\Omega_{(s_0,t_0)}(\gamma)$ can then be computed by applying the 
(attractor or flow) tree formulae  locally at each vertex, as explained in \S\ref{sec_SAF}, or more globally 
 by perturbing the charges of the initial constituents $\gamma_i\to \gamma_i+\epsilon_i \delta, \gamma'_j\to \gamma'_j+\epsilon'_j \delta$ where $\delta$ is the D0-brane charge vector
 and $\epsilon_i,\epsilon'_j$ are small enough and generic.
  
 While the resulting (unperturbed) attractor flow trees typically involve vertices of higher valency
 and tend to proliferate, it is
 sufficient to keep track of {\it scattering sequences},
  where descendent edges whose charges are multiples of the same primitive charge 
 $\gamma_e$ are aggregated together as a single edge of charge $k\gamma_e$.
 Conversely, the original family 
 of attractor flow trees can be recovered from the aggregated flow trees by replacing 
 each edge of charge $k\gamma_e$ by edges of charge $\{ k_1 \gamma_e, k_2 \gamma_e, \dots\}$
 where $\{k_i\}$ runs over all integer partitions of $k$. We shall make use of this bookkeeping
 device in the next subsection, in order to keep the list of attractor flow trees (or rather, scattering
 sequences) within reasonable length.

\subsection{Examples: Hilbert scheme of \protect\texorpdfstring{$n$}{n} points on \protect\texorpdfstring{$\IP^2$}{P2}\label{sec_Hilb}}
 We now apply the procedure outlined above to the case  of rank 1 sheaves with Chern vector
 $\gamma=[1,0,1-n)$. At large volume $t\to\infty$, the 
 moduli space $\cM_{(s,t)}(\gamma)$ reduces to the Hilbert scheme of $n$ points on $\IP^2$, 
for which the Poincar\'e polynomials are well-known (see \eqref{genHilb}). From the analysis in 
\cite{arcara2013minimal,coskun2014birational}, it follows that the chamber structure is 
given by a set of nested walls, the largest of which being  the Gieseker wall corresponding
to the destabilizing object $\cO(-1)$. From the discussion below \eqref{WallLV}, this is the
half-circle $\cC(-n-\frac12,n-\frac12)$. Thus it suffices to consider all scattering sequences 
which go through
the points $(s,t)=(-n-\frac12,n-\frac12)$, with constituents with slopes in the range
$-2n\leq m\leq -1$, and electric potential bounded by $\varphi_{-n-\frac12}(\gamma)=2n+1$.

We find that for any $n$, there is indeed a scattering sequence
whose first splitting lies on the Gieseker wall $\cC(-n-\frac12,n-\frac12)$,
namely
\be
T_{\rm Gieseker} = \begin{cases} \{ -\cO(-2), 2\cO(-1) \} & n=1 \\ 
\{\{ -\cO(-n-1), \cO(-n)\}, \cO(-1) \} & n\geq 2
\end{cases} 
\ee
For $n=1$ the scattering sequence is a shorthand for two distinct attractor flow trees
with a single splitting into either two or three descendants, namely $ \{ \gamma_1, 2\gamma_2 \}$ and $\{  \gamma_1, \gamma_2, \gamma_2 \}$ with $\gamma_1=-\cO(-2), \gamma_2=\cO(-1)$. The contributions of the two trees sum up to 
\be
\frac12 \kappa(\langle\gamma_1,\gamma_2\rangle)^2 \bOm(\gamma_1) \bOm(\gamma_2)^2 + 
 \kappa(\langle\gamma_1,2\gamma_2\rangle) \bOm(\gamma_1) \bOm(2\gamma_2)
 =
\frac12 (y^2+1+y^{-2})^2 - \frac{y^6-y^{-6}}{y-y^{-1}} \frac{y-y^{-1}}{2(y^2-y^{-2})} 
\ee
reproducing the index $K_3(1,2)= 
y^2+1+1/y^2$ of a Kronecker quiver with 3 arrows and dimension vector $(1,2)$
(see Table \ref{Kroneckertab}). For $n=2$, there is a single attractor flow tree with two vertices
$\{\{\gamma_1,\gamma_2\},\gamma_3\}$ 
 contributing 
 \be
  \kappa(\langle\gamma_1,\gamma_2\rangle)  \kappa(\langle\gamma_1+\gamma_2,\gamma_3\rangle)
  = K_3(1,1)^2=(y^2+1+1/y^2)^2
 \ee
For $n=1$ or $n=2$, there is a single wall of marginal stability associated to either of the two 
sequences, and the total index agrees with the prediction \eqref{genHilb} for
the Gieseker index outside the wall, and  vanishes inside. This is consistent with the fact that the moduli space $\cM_\infty(\gamma)$ for $n\leq 2$ has a single stratum.

\begin{figure}[ht]\centering
\includegraphics[width=\textwidth]{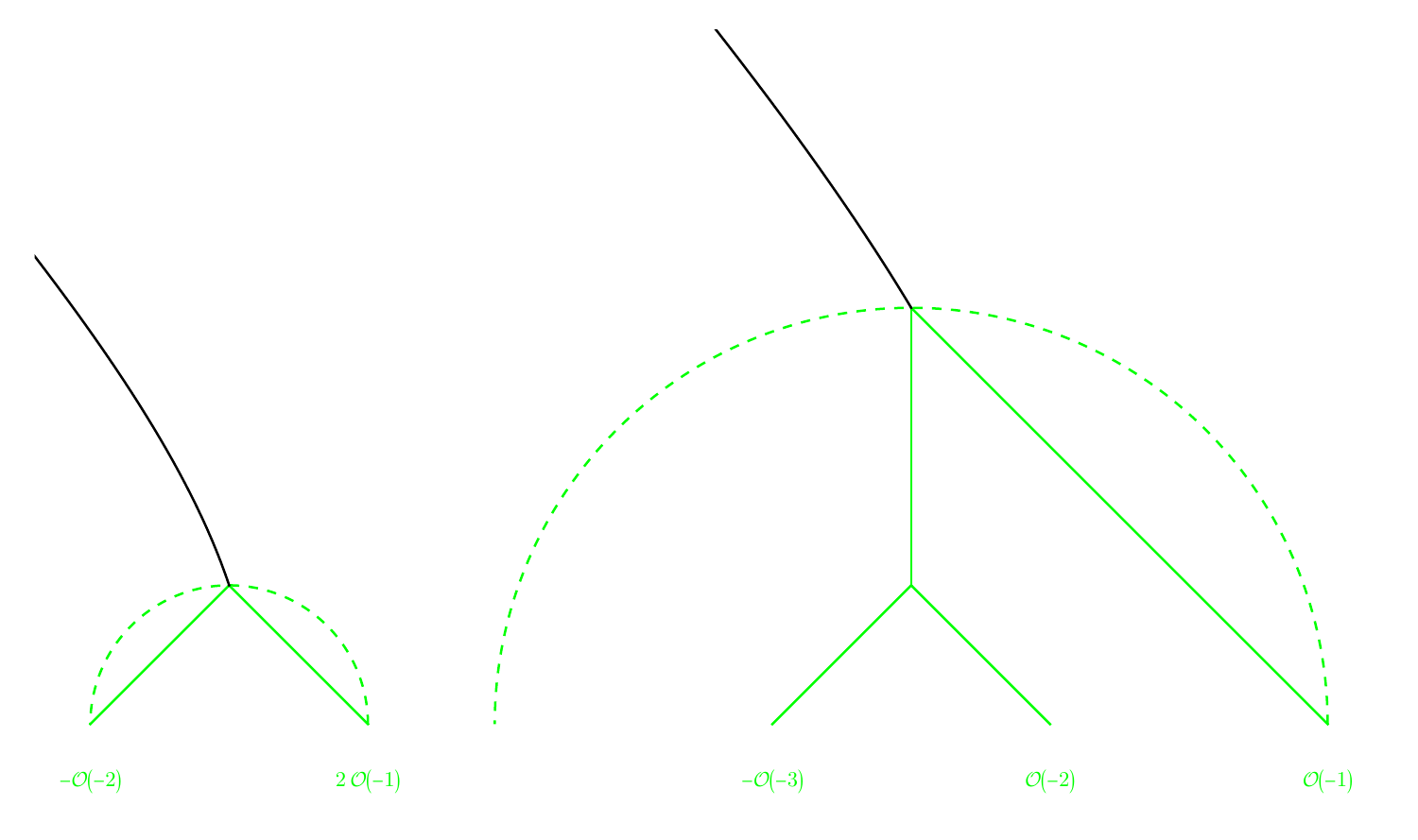}
\caption{Scattering sequences for Hilbert scheme of 1 and 2 points on $\IP^2$.\label{figH12}}
\end{figure}

For $n\geq 3$, there are additional walls of the form $\cC(s_{\gamma,\gamma'},R_{\gamma,\gamma'})$, associated to scattering sequences with final vertex of the form $\gamma'+(\gamma-\gamma')\to \gamma$. Below we list the corresponding scattering sequences 
for $n\leq 7$, along with their respective contribution to the index in the region 
above the respective wall (for brevity, we only indicate the unrefined limit $y\to 1$, but the full
refined index can be obtained from the Kronecker indices in Table \ref{Kroneckertab}). 
In all cases, we find agreement with \eqref{genHilb}:

\begin{itemize}

\item For $n=3$, there are two nested walls, each of which is associated to one sequence,
\begin{equation*}
\begin{array}{l@{\hskip 5mm}l@{\hskip 5mm}l}
\cC(-\frac72,\frac52) &  \{\{-\cO(-4),\cO(-3)\},\cO(-1)\} & K_3(1,1)^2
\\
\cC(-\frac52,\frac12) & \{ -2\cO(-3),3\cO(-2)\} & K_3(2,3)
\end{array}
\end{equation*}
Again we stress that the scattering sequence $\{ -2\cO(-3),3\cO(-2)\}$ stands for 6 different
attractor flow trees, corresponding to the $2\times 3$ integer partitions of $2$ and $3$, which all split
at the point and add up to the index of a Kronecker quiver with 3 arrows and dimension vector $(2,3)$. 
In the Gieseker chamber, the contributions of the two sequences add up to $9+13=22$ as $y\to 1$.

\item For $n=4$ there are 3 nested walls, and correspondingly 3 sequences,
\begin{equation*}
\begin{array}{l@{\hskip 5mm}l@{\hskip 5mm}l}
\cC(-\frac92,\frac72) & \{\{-\cO(-5), \cO(-4)\}, \cO(-1)\} & K_3(1,1)^2 \\
\cC(-\frac72,\frac{\sqrt{17}}{2}) &  \{\{-\cO(-4), \cO(-3)\}, \{-\cO(-3), 2 \cO(-2)\}\} & K_3(1,1)^2 K_3(1,2) \\
\cC(-3,1) &  \{-\cO(-4), 2 \cO(-2)\} & K_6(1,2)
\end{array}
\end{equation*}
In the Gieseker chamber, their contributions add up to  $9+27+15=51$  as $y\to 1$.

\begin{figure}[ht]
\begin{center}
\includegraphics[height=7cm]{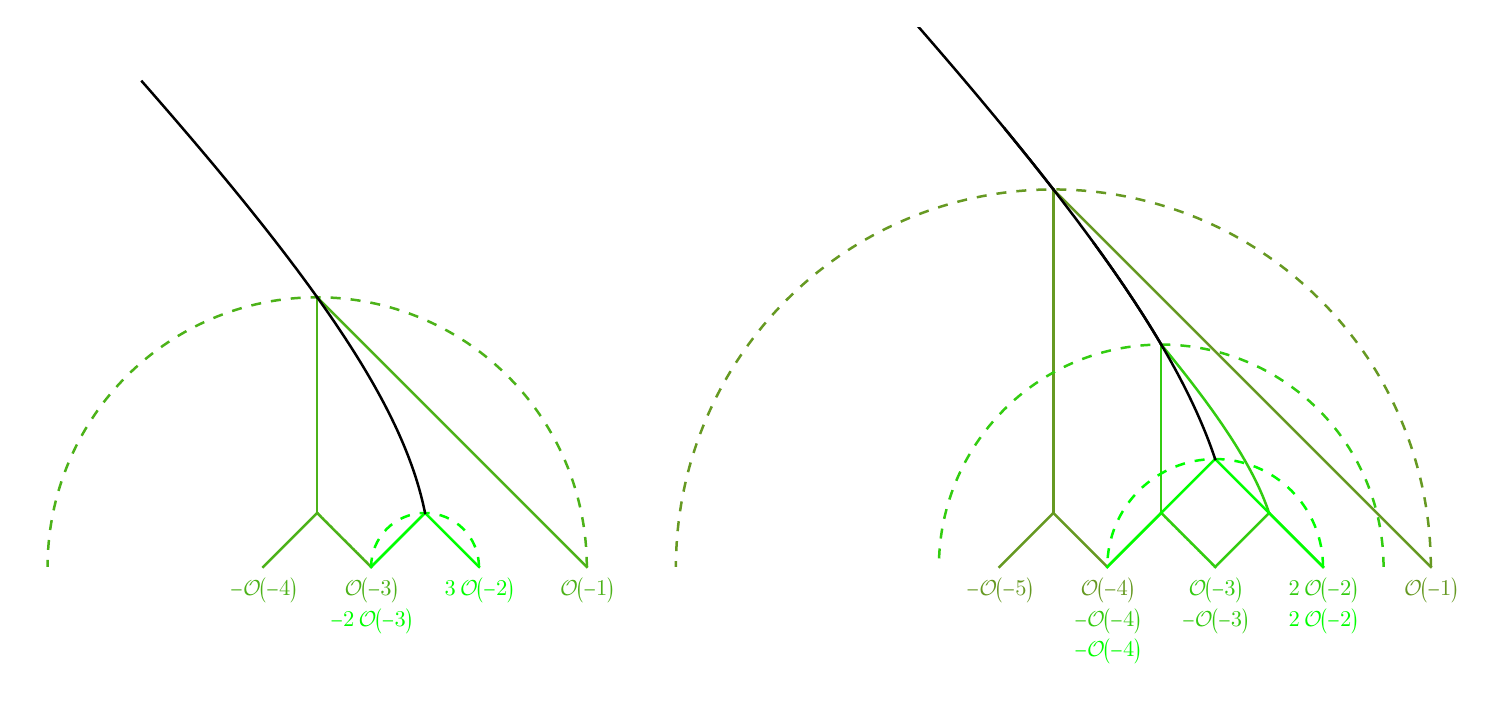}
\end{center}
\caption{Scattering sequences for Hilbert scheme of 3 and 4 points on $\IP^2$.\label{figH34}}
\end{figure}

\item For $n=5$ there are 3 nested walls associated to 3 sequences,
\begin{equation*}
\begin{array}{l@{\hskip 5mm}l@{\hskip 5mm}l}
\cC(-\frac{11}{2},\frac92) & \{\{-\cO(-6), \cO(-5)\}, \cO(-1)\} & K_3(1,1)^2  \\
\cC(-\frac92.\frac{\sqrt{41}}{2}) & \{\{-\cO(-5), \cO(-4)\}, \{-\cO(-3), 2 \cO(-2)\}\} &  K_3(1,1)^2 K_3(1,2) \\
\cC(-\frac72,\frac32) & \{\{-2 \cO(-4), 2 \cO(-3)\}, \cO(-2)\} &  K_{-3,3,6}(2,2,1)
\end{array}
\end{equation*}
It is worth explaining in some detail the computation of the index associated to the last sequence
$\{\{2\gamma_1,2\gamma_2\},\gamma_3\}$ with $\gamma_1=-\cO(-4),\gamma_2=\cO(-3),
\gamma_3=\cO(-2)$ (which we have denoted by $K_{a,b,c}(2,2,1)$
with $a=\langle\gamma_1,\gamma_2\rangle$, $b=\langle\gamma_2,\gamma_3\rangle$,
$c=\langle\gamma_3,\gamma_1\rangle$). This sequence actually stands for 
5 different attractor flow trees, namely
\be
\begin{aligned}
& \{\{2\gamma_1,2\gamma_2\},\gamma_3\} , \quad
&& \{\{\gamma_1,\gamma_1,2\gamma_2\},\gamma_3\}, \\
& \{\{2\gamma_1,\gamma_2,\gamma_2\},\gamma_3\}, \quad
&& \{\{\gamma_1,\gamma_1,\gamma_2,\gamma_2\},\gamma_3\}, \\
&&& \{\{\gamma_1,\gamma_2\},\{ \{\gamma_1,\gamma_2\}, \gamma_3\}\}
\end{aligned}
\ee
The first four of those combine to produce the rational index ${\overline K}_3(2,2)\to -\frac{21}{4}$ of a Kronecker quiver with $3$ arrows and dimension vector $(2,2)$. The fifth involves two copies of 
the Kronecker quiver with $3$ arrows and dimension vector $(1,1)$. The total index is therefore
\be
\label{Om221}
\begin{aligned}
K_{-3,3,6}(2,2,1) & = \kappa( \langle 2\gamma_1+2\gamma_2,\gamma_3\rangle) {\overline K}_3(2,2)
+\frac12 \kappa( \langle \gamma_1+\gamma_2,\gamma_3\rangle)^2 K_3(1,1)^2 \\
& \to 6 \frac{21}{4}  + \frac12 3^2 3^2 = 72
\end{aligned}
\ee
Equivalently, this factor arises by applying
the flow tree formula
to a local scattering diagram with two incoming rays of charge $\alpha=\gamma_1+\gamma_2$
and $\beta=\gamma_3$ with $\Omega^-(\alpha)=K_3(1,2)=y^2+1+1/y^2, \Omega^-(2\alpha)
=K_3(2,2)=-y^5-y^3-y-1/y-1/y^3-1/y^5$ and $\Omega^-(\beta)=1$,
selecting the outgoing ray of charge $2\alpha+\beta$. 
In the Gieseker chamber, the contributions of the three sequences
 add up to $9+27+72=108$  as $y\to 1$.

\item For $n=6$ there are 5 nested walls associated to 5 sequences,
\begin{equation*}
\begin{array}{c@{\hskip 5mm}l@{\hskip 5mm}l}
\cC(-\frac{12}{2},\frac{11}{2}) & \{\{-\cO(-7), \cO(-6)\}, \cO(-1)\} & K_3(1,1)^2 \\
\cC(-\frac{11}{2},\frac{\sqrt{73}}{2}) &  \{\{-\cO(-6), \cO(-5)\}, \{-\cO(-3), 2 \cO(-2)\}\} &  K_3(1,1)^2 K_3(1,2) \\
\cC(-\frac92,\frac{\sqrt{33}}{2}) &\{\{-\cO(-5), \cO(-4)\}, \{\{-\cO(-4), \cO(-3)\}, \cO(-2)\}\}&  K_3(1,1)^4 
\\
\cC(-4,2) & \{\{-\cO(-5), \cO(-3)\}, \cO(-2)\}&  K_6(1,1)^2\\
\cC(-\frac72,\frac12) & \{-3 \cO(-4), 4 \cO(-3)\} & K_3(3,4) 
\end{array}
\end{equation*}
In the Gieseker chamber, their contributions add up to $9+27+81+36+68=221$  as $y\to 1$.

\begin{figure}[ht]
\begin{center}
\includegraphics[height=7cm]{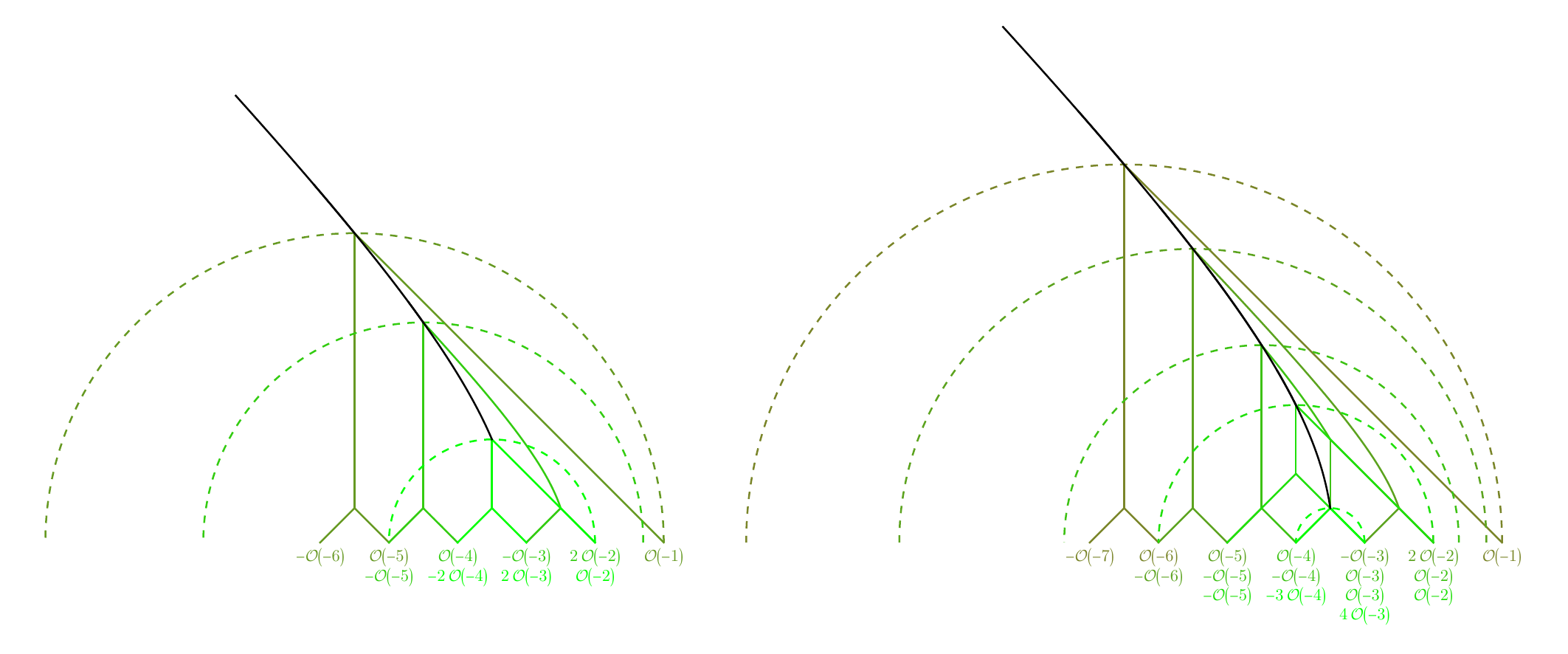}
\end{center}
\caption{Scattering sequences for Hilbert scheme of 5 and 6 points on $\IP^2$.\label{figH56}}
\end{figure}

\item For $n=7$ there are 6 nested walls, but one of them is associated to 2 sequences:
\begin{equation*}
\begin{array}{cll}
\cC(-\frac{15}{2},\frac{13}{2}) & \{\{-\cO(-8), \cO(-7)\}, \cO(-1)\}& K_3(1,1)^2  \\
\cC(-\frac{13}{2},\frac{\sqrt{113}}{2})  & \{\{-\cO(-7), \cO(-6)\}, \{-\cO(-3), 2 \cO(-2)\}\}  & K_3(1,1)^2 K_3(1,2) \\
\cC(-\frac{11}{2},\frac{\sqrt{65}}{2}) & \{\{-\cO(-6), \cO(-5)\}, \{\{-\cO(-4), \cO(-3)\}, \cO(-2)\}\}& K_3(1,1)^4 \\
\cC(-\frac92,\frac52) & \{\{-2 \cO(-5), 2 \cO(-4)\}, \cO(-2)\}& K_{-3,6,9}(2,2,1) \\
 " & \{\{-\cO(-5), \cO(-4)\}, \{-2 \cO(-4), 3 \cO(-3)\}\}& K_3(1,1)^2 K_3(2,3)
 \\
\cC(-4,\sqrt{2}) & \{\{-\cO(-5), \cO(-3)\}, \{-\cO(-4), 2 \cO(-3)\}\}& K_3(1,2) K_6(1,1)^2 \\
\cC(-\frac{39}{10},\frac{11}{10}) &\{-\cO(-5), \{-\cO(-4), 3 \cO(-3)\}\}&  K_3(1,3) K_{15}(1,1) 
\end{array}
\end{equation*}
The index for the first tree associated to the wall $\cC(-\frac92,\frac52)$
is computed in the same way as explained around \eqref{Om221}.  Note also that 
the scattering sequence associated to $\cC(-4,\sqrt{2})$ is non-planar, as the ray 
$\cO(-3)$ has to cross through $-\cO(-4)$ before colliding with $-\cO(-5)$. 
The contributions of the 7 sequences to the Gieseker index 
add up to $9+27+81+72+117+108+15=429$  as $y\to 1$, in agreement with  \eqref{genHilb}.

\begin{figure}[ht]
\begin{center}
\includegraphics[height=7cm]{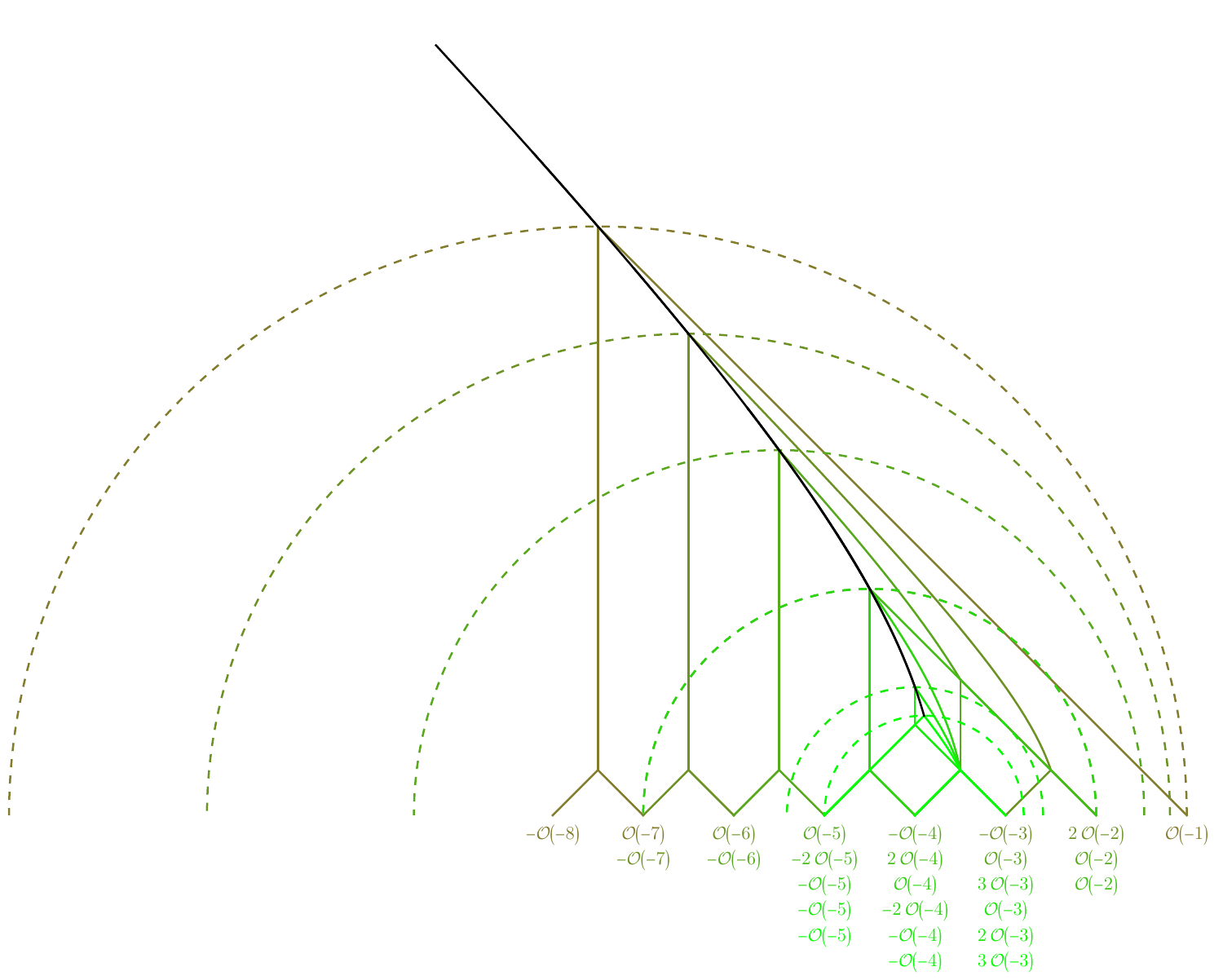}
\end{center}
\caption{Scattering sequences for Hilbert scheme of 7 points on $\IP^2$.\label{figH7}}
\end{figure}

\end{itemize}

\subsection{Examples: D2-D0 indices\label{sec_D2D0}}
We now turn to the case of rank 0 sheaves with  Chern vector $\gamma=[0,d,\chi)$ with $d>0$.
At large volume $t\to\infty$, the  index $\Omega_{(s,t)}(\gamma)$ is determined in terms of 
refined Gopakumar-Vafa invariants via \eqref{OmD2GV}, in particular it is  independent of 
$\chi$. This is in general no longer true for finite $t$, due to the existence of a sequence of nested walls of marginal stability. The outermost wall (or Gieseker wall) is determined as follows
\cite[Proposition 7.5]{woolf2013nef}. Let $\eta$ be an integer in the range 
$-\frac{d}{2} \leq \eta\leq \frac{d}{2}$ such that $\chi+\frac12d (d-3) = k d+ \eta$ with $k\in \IZ$. 
Then the Gieseker wall is the half-circle $\cC(s_0,R_0)$ of radius 
$R_0=\frac{d}{2} - \frac{|\eta|}{d}$ centered
at $s_0=\frac{\chi}{d}-\frac32= k+\frac{\eta}{d}-\frac{d}{2}$. 
When $\eta\geq 0$, the destabilizing sub-object along the wall is $\cO(k)$,
whereas for $\eta\leq 0$ it is the ideal sheaf $\cI_{Z(-\eta)}(k)$ where $Z(w)$ is a scheme of dimension 0 and length $w$. For $\eta=0$, these two notions agree, while for $|\eta|=\frac{d}{2}$ they do not, but destabilizing subobjects of both types may occur.  It follows that all constituents of the trees contributing to $\Omega_\infty(\gamma)$ must be emitted  in an interval of width $\frac{d}{2} - \frac{|\eta|}{d}$ centered around $s_0$, and electric potential bounded by $\varphi_{s_0}(\gamma)=2d$.

We shall now discuss examples with low degree $d$, making use of  the symmetry properties \eqref{Omflow} and 
\eqref{Omdual} to restrict the values of $\chi$.  In all cases, the wall structure matches 
with \cite{bertram2014birational} and the index outside the Gieseker wall agrees with  the GV invariants in \eqref{GVtab}. The trees arising at each wall also agree with the stratification
in \cite{drezet2011geometry,maican2011moduli,maican2013classification}, 
with one exception for $(d,\chi)=(5,3)$ mentioned below.

\begin{itemize}
\item For $(d,\chi)=(1,1)$, such that $(k,\eta)=(0,0)$, there is only one wall,
\begin{equation*}
\begin{array}{c@{\hskip 10mm}l@{\hskip 5mm}l@{\hskip 5mm}l}
\chi=1: & \cC(-\tfrac12,\tfrac12) \qquad \{ -\cO(-1), \cO \} \qquad K_3(1,1) 
\end{array}
\end{equation*}
The tree corresponds to the fact that the structure sheaf of a curve 
$\cO_C$ is given by the short exact sequence $0\to \cO(-1)\to \cO_C\to \cO\to 0$, and its index 
$K_3(1,1)=y^2+1+1/y^2$ reproduces the Poincar\'e polynomial of the linear system $[C]=\IP^2$.

\item For $(d,\chi)=(2,0),(2,1)$, such that $(k,\eta)=(0,-1),(0,0)$, respectively, 
there is a single wall 
\begin{equation*}
\begin{array}{c@{\hskip 10mm}l@{\hskip 5mm}l@{\hskip 5mm}l}
\chi=0: & \cC(-\frac32,\frac12) & \{ -2\cO(-2), 2\cO(-1) \} & K_3(2,2)  \\[2mm]
\chi=1: & \cC(-1,1) & \{ -\cO(-2), \cO\} & K_3(1,2)
\end{array}
\end{equation*}
giving the same index $-y^5-y^3-y-\dots$ in the Gieseker chamber in either case. 

\item For $(d,\chi)=(3,0),(3,1)$, such that $(k,\eta)=(0,0),(0,1)$, respectively,
there are two walls in the first case, and a single wall in the second case,
\begin{equation*}
\begin{array}{c@{\hskip 10mm}l@{\hskip 5mm}l@{\hskip 5mm}l}
\chi=0: & \cC(-\frac32,\frac32)&\{-\cO(-3),\cO(0)\} & K_9(1,1) \\[1pt]
"  &  \cC(-\frac32,\frac12) &\{ -3\cO(-2), 3\cO(-1) \} \ & K_3(3,3) 
 \\[2mm]
\chi=1:& \cC(-\frac76,\frac76) & \{ \{ -2\cO(-2), \cO(-1)\}, \cO\} & 
K_3(1,2) K_9(1,1)
\end{array}
\end{equation*}
contributing $27=18+9$ as $y\to 1$  in the Gieseker chamber in either case. 

\begin{figure}[ht]
\begin{center}
\includegraphics[height=6cm]{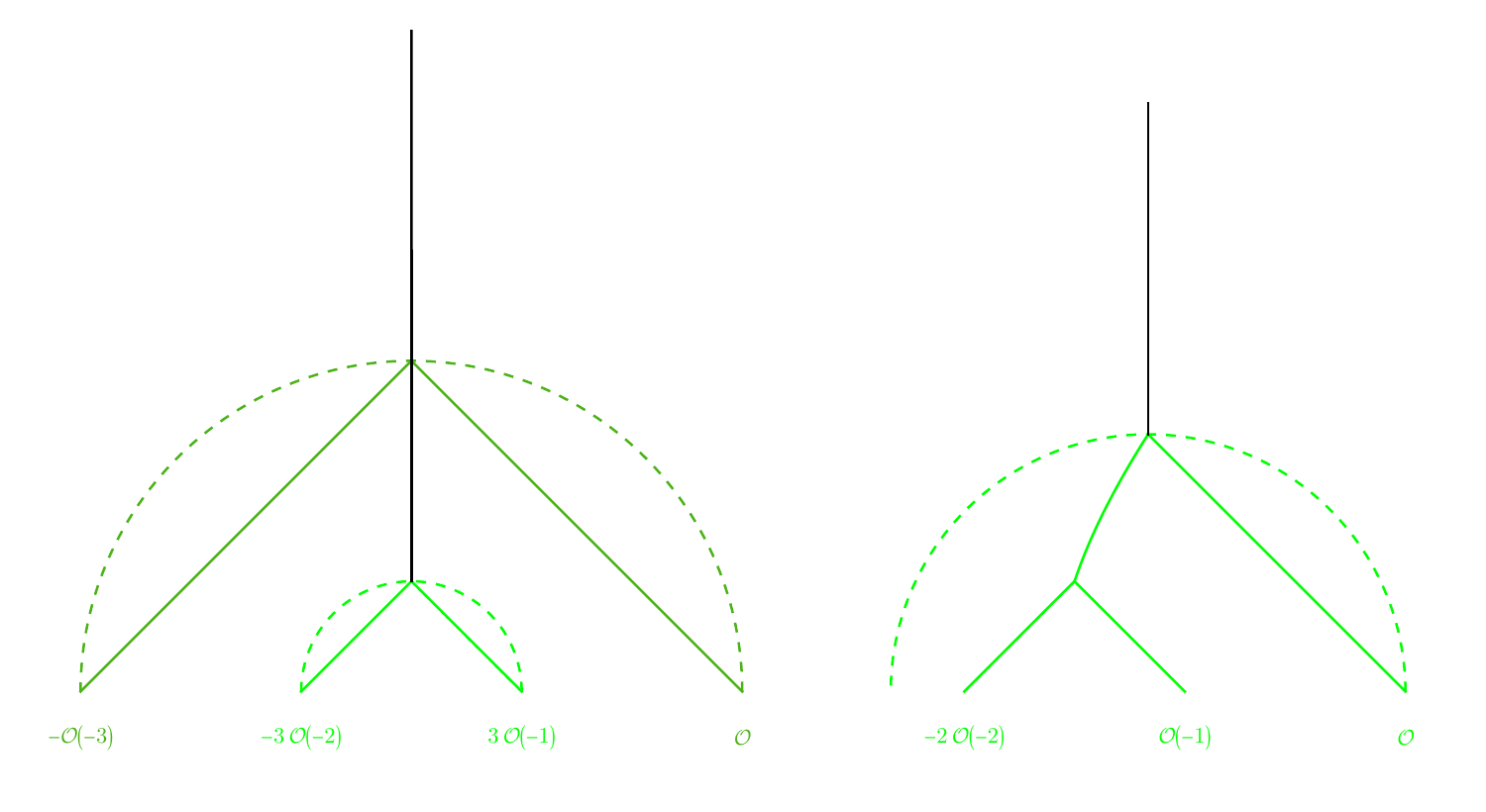}
\end{center}
\caption{Scattering sequences for $(d,\chi)=(3,0),(3,1)$.
\label{fig3D2}}
\end{figure}

\item  For $(d,\chi)=(4,0),(4,1),(4,2)$, such that $(k,\eta)=(0,2),(1,-1),(1,0)$, respectively,
there are two walls in the first case, and three in the third case,
\begin{equation*}
\begin{array}{c@{\hskip 10mm}l@{\hskip 5mm}l@{\hskip 5mm}l}
\chi=0: 
 &\cC(-\frac32,\frac32) & \{-\cO(-3),\{\{-\cO(-2),\cO(-1)\},\cO\}\}   &  K_3(1,1)^2 K_{12}(1,1)\\[1pt]
"  & \cC(-\frac32,\frac12) & \{-4\cO(-2),4\cO(-1)\}   &  K_3(4,4)  \\ [2mm]
\chi=1: 
& \cC(-\frac54,\frac74)&\{-\cO(-3),\{-\cO(-1),2\cO\}\}   &  K_3(1,2) K_{12}(1,1) \\[1pt]
"  &\cC(-\frac54,\frac54)&\{\{-3\cO(-2),2\cO(-1)\},\cO\}   &  K_3(2,3) K_{12}(1,1)  \\[2mm]
\chi=2: & \cC(-1,2)&\{-\cO(-3),\cO(1)\} &  K_{12}(1,1)  \\[1pt]
" & \cC(-1,\sqrt2) &\{\{-2\cO(-2),\cO(-1)\},\{-\cO(-1),2\cO\}\} & K_3(1,2)^2 K_{12}(1,1)  \\[1pt]
" & \cC(-1,1) &\{-2\cO(-2),2\cO\}  & K_6(2,2) 
\end{array}
\end{equation*}
 contributing $-108-84=-36-156=-108-72-12=-192$  in the Gieseker chamber in each case. 
 Note that for $\chi=0$, the first wall involves three charges $-\cO(-3), -\cO(-2)+\cO(-1)$
 and $\cO$ colliding at the same point, which can be resolved into two successive collisions by perturbing the incoming rays.
 For each value of $\chi$, the sequences match with the stratification in \cite{drezet2011geometry}.

\begin{figure}[ht]
\begin{center}
\includegraphics[height=5cm]{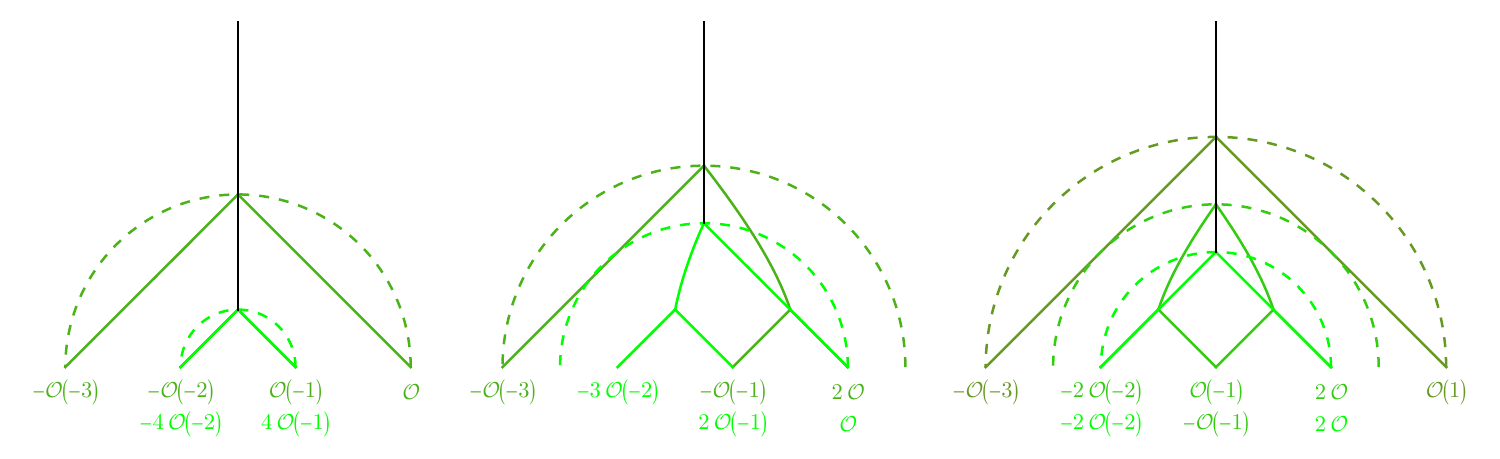}
\end{center}
\caption{Scattering sequences for $(d,\chi)=(4,0),(4,1),(4,2)$.
\label{fig4D2}}
\end{figure}

\item  For $(d,\chi)=(5,0),(5,1),(5,3)$, such that $(k,\eta)=(1,0),(1,1),(2,-2)$, we find
\begin{equation*}
\hspace*{-7mm}
\begin{array}{c@{\hskip 3mm}l@{\hskip 3mm}l@{\hskip 3mm}l}
\chi=0 & \cC(-\frac32,\frac52) & \{-\cO(-4), \cO(1)\}  &  K_{15}(1,1) \\[1pt]
" & \cC(-\frac32,\frac{\sqrt{17}}{2}) & \{\{-2 \cO(-3), \cO(-2)\}, \{-\cO(-1), 2 \cO\}\}   &  K_3(1,2)^2 K_{15}(1,1) \\[1pt]
" & \cC(-\frac32,\frac32) & \{-\cO(-3), \{\{-2 \cO(-2), 2 \cO(-1)\}, \cO\}\}   &  K_{15}(1,1) K_{-3,3,6}(2,2,1) \\[1pt]
  & \cC(-\frac32,\frac12) & \{-5 \cO(-2), 5 \cO(-1)\}   &  K_3(5,5) 
 \\[2mm]
\chi=1  & \cC(-\frac{13}{10},\frac{23}{10})  & \{\{-2 \cO(-3), \cO(-2)\}, \cO(1)\}  &  K_3(1,2)  K_{15}(1,1) \\[1pt]
" & \cC(-\frac{13}{10},\frac{\sqrt{329}}{10}) &
\begin{array}[t]{@{}r@{}} \{\{-\cO(-3), \{-\cO(-2), \cO(-1)\}\},\quad\qquad \\ \{-\cO(-1), 2 \cO\}\} \end{array}  &  
K_3(1,1)^2 K_3(1,2) K_{15}(1,1) \\[1pt]
" & \cC(-\frac{13}{10},\frac{17}{10})  & \{-\cO(-3), \{-\cO(-2), 2 \cO\}\}   &  K_6(1,2) K_{15}(1,1) \\[1pt]
& \cC(-\frac{13}{10},\frac{13}{10}) & \{\{-4 \cO(-2), 3 \cO(-1)\}, \cO\}   &  K_3(3,4) K_{15}(1,1) 
\\[2mm]
\chi=3 & \cC(-\frac{9}{10},\frac{21}{10})  & \{-\cO(-3), \{\{-\cO(-1), \cO\}, \cO(1)\}\}  &  K_3(1,1)^2 K_{15}(1,1)\\[1pt]
" & \cC(-\frac{9}{10},\frac{19}{10}) & \{\{-3 \cO(-2), 2 \cO(-1)\}, \cO(1)\}  &  K_3(2,3) K_{15}(1,1) \\[1pt]
" & \cC(-\frac{9}{10},\frac{\sqrt{241}}{10})  & \{ \{-2 \cO(-2), \cO(-1)\}, \{-2 \cO(-1), 3 \cO\}\}  &  K_3(1,2) K_3(2,3) K_{15}(1,1) \\[1pt]
" & \cC(-\frac{9}{10},\frac{\sqrt{161}}{10}) & \{ \{ -2\cO(-2),\cO\}, \{-\cO(-1),2\cO\}\}   &  K_3(1,2) K_6(1,2) K_{15}(1,1) \\[1pt]
" &  \cC(-\frac{9}{10},\frac{11}{10}) & \{-2 \cO(-2), \{-\cO(-1), 3 \cO\}\}  &  K_3(1,3) K_{15}(1,2)  
\end{array}
\end{equation*}
contributing $15+135+1080+465=
45+405+225+1020=
135+195+585+675+105=
1695$ in the Gieseker chamber. For $\chi=0$, the third wall involves three charges $-\cO(-3), -2\cO(-2)+2\cO(-1)$ and $\cO$ colliding at the same point; it can be treated by  
perturbing in the same way as for $(d,\chi)=(4,0)$ and using the same reasoning as explained
below \eqref{Om221}. 
Note that the fourth sequence for $\chi=3$  is non-planar, and corresponds to a codimension one stratum, which appears to be missing in \cite{maican2011moduli}.

\begin{figure}[ht]
\centerline{\includegraphics[width=18cm]{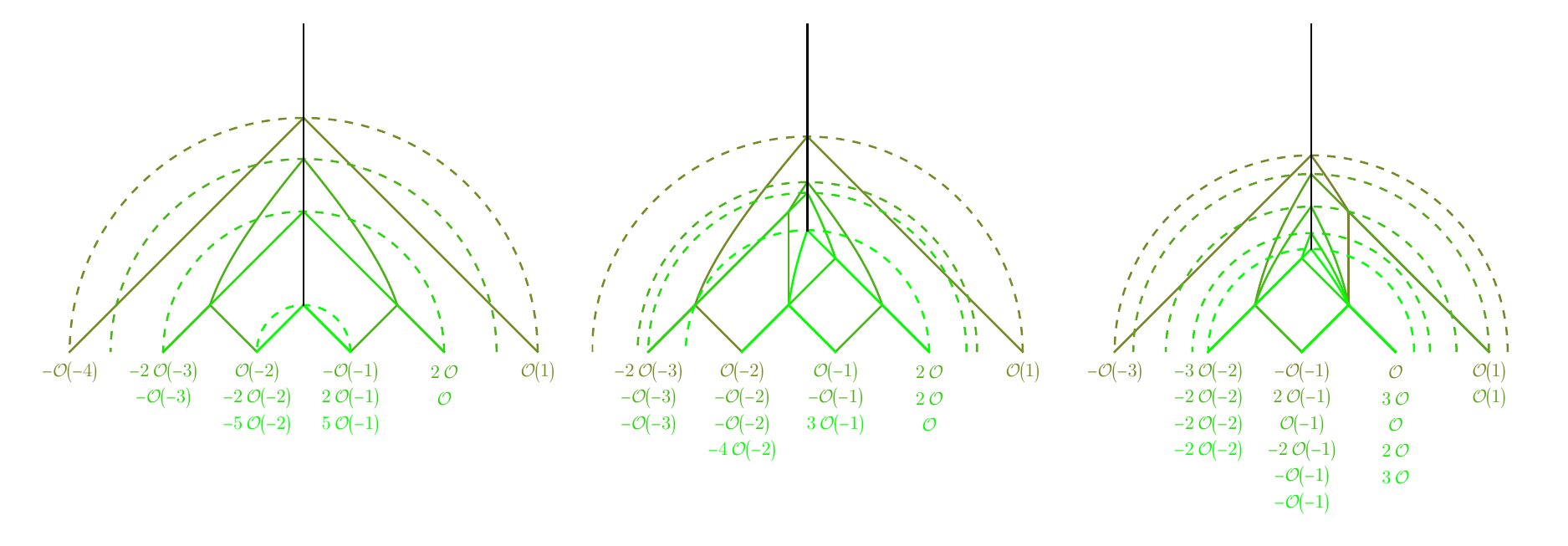}}
\vspace{-.8cm}
\caption{Scattering sequences for $(d,\chi)=(5,0),(5,1),(5,3)$\label{fig5D2}}
\end{figure}

\item For  $(d,\chi)=(6,1)$, corresponding to $(k,\eta)=(2,-2)$,
\begin{equation*}
\hspace*{-5mm}
\begin{array}{lll}
\cC(-\frac43,\frac83)& \{-\cO(-4), \{\{-\cO(-1), \cO\}, \cO(1)\}\}   &   K_3(1,1)^2 K_{18}(1,1)\\[1pt]
\cC(-\frac43,\frac73) & \{\{\{-2 \cO(-3), \cO(-2)\}, \{-\cO(-2), \cO(-1)\}\}, 
   \cO(1)\}  &  K_3(1,1)^2 K_3(1,2)  K_{18}(1,1) \\[1pt]
& \{\{-2 \cO(-3), \cO(-1)\}, \cO(1)\}   &  K_6(1,2)  K_{18}(1,1)
\\[1pt]
\cC(-\frac43\frac{\sqrt{46}}{3}) & \{\{-2 \cO(-3), \cO(-2)\}, \{-2 \cO(-1), 3 \cO\}\}  &   
 K_3(1,2) K_3(2,3)  K_{18}(1,1)  \\[1pt]
\cC(-\frac43,\frac{\sqrt{31}}{3}) & \{\{-\cO(-3), \{-2 \cO(-2), 2 \cO(-1)\}\}, \{-\cO(-1), 
    2 \cO\}\}  &  K_{-3,6,3}(2,2,1) K_3(1,2) K_{18}(1,1)  \\[1pt]
\cC(-\frac43,\frac{2\sqrt7}{3}) & \{\{-\cO(-3), \{-\cO(-2), \cO(-1)\}\}, \{-\cO(-2), 
    2 \cO\}\}  &  K_3(1,1)^2 K_6(1,2)  K_{18}(1,1) \\[1pt]
\cC(-\frac43,\frac53) & \{-\cO(-3), \{\{-2 \cO(-2), \cO(-1)\}, 2 \cO\}\}  &  K_3(1,2) K_9(1,2)  K_{18}(1,1) \\[1pt]
\cC(-\frac43,\frac43) & \{\{-5 \cO(-2), 4 \cO(-1)\}, \cO\}  &  K_{3}(4,5) K_{18}(1,1)  
\end{array}
\end{equation*}
contributing
$-162-270-486-702-2430-2430-1944-7182=-17064$ in the Gieseker chamber
(here the factor $K_{-3,6,3}(2,2,1)$ in the fourth sequence can be computed as
explained below \eqref{Om221}, and still evaluates to 72). 
Note that the sequence associated to $\cC(-\frac43,\frac{2\sqrt7}{3})$ is non-planar.
We refrain from considering other values of $\chi$ in this case.

\begin{figure}[ht]
\begin{center}
\includegraphics[width=9cm]{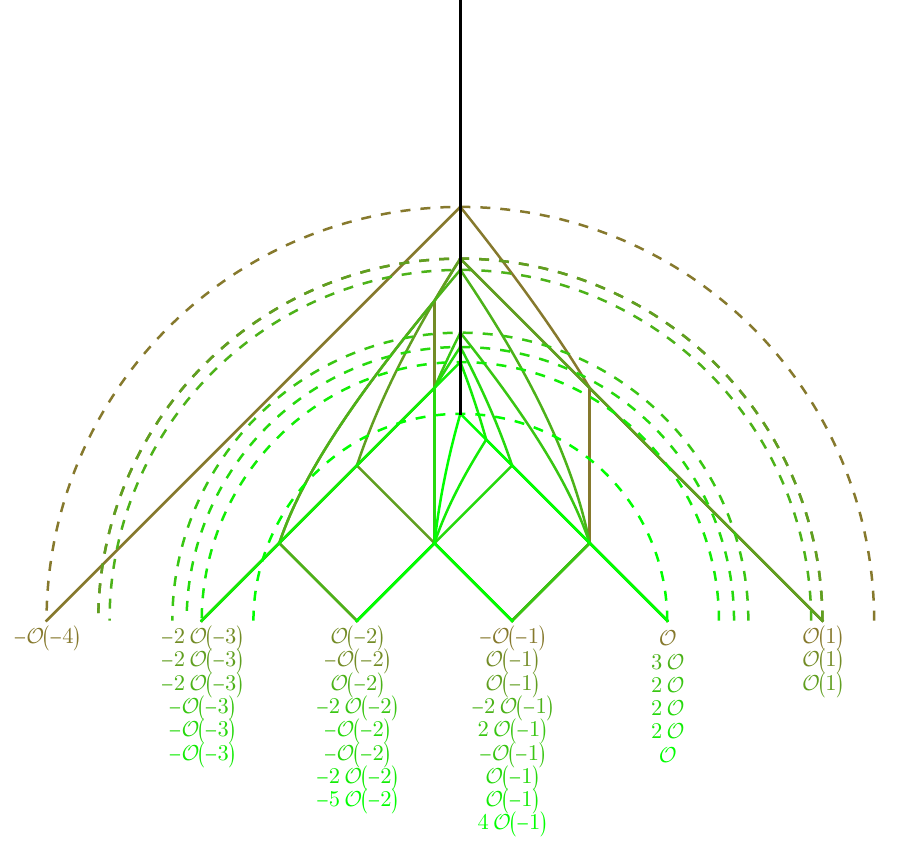} 
\end{center}
\caption{Scattering sequences for $(d,\chi)=(6,1)$\label{fig61}}
\end{figure}

\end{itemize}

\subsection{Generalization to  \texorpdfstring{$\psi\neq 0$}{non-zero psi}\label{sec_LVpsi}}

We now discuss how this picture generalizes to an arbitrary phase $\psi\in(-\frac{\pi}{2},\frac{\pi}{2})$. 
For $r\neq 0$, the geometric ray  $\Rgeo_\psi(\gamma)$
is now a branch of hyperbola (degenerating to a straight line for vanishing discriminant $\Delta=0$)
\be
s = \tfrac{d}{r} - t \tan\psi - \frac{\sign r}{\cos\psi} \sqrt{t^2+2\Delta \cos^2\psi}  
\ee
with $t>0$, 
while for $r=0$, it is the line $s=\frac{\ch_2}{d}-t \tan\psi$ when $d>0$ (or the empty set 
when $d\leq 0$). In particular, the electromagnetic analogy of \S\ref{sec_raywall} breaks down for $\psi\neq 0$,
since one of the two rays emanating from $s=m$ is no longer inside the light-cone (indeed, the angle
between the two rays remains equal to $\frac{\pi}{2}$ independently of $\psi$).

Moreover, under the same conditions \eqref{condint} 
two rays $\Rgeo_\psi(\gamma)$ and $\Rgeo_\psi(\gamma')$ still
intersect on the same wall $\cW(\gamma,\gamma')$, but at the point $(s_{\gamma,\gamma'} - R_{\gamma,\gamma'}\sin\psi, R_{\gamma,\gamma'}\cos\psi)$ at an angle $\psi$ from the top of the half-circle $\cC(s_{\gamma,\gamma'},R_{\gamma,\gamma'})$. In fact, observing that 
\be
\Re\left( e^{-\I\psi} Z^{\rm LV}_{(s,t)}(\gamma)\right) = \cos\psi\, \Re \left( 
Z^{\rm LV}_{(s_\psi,t_\psi) }(\gamma)\right)
\ee
where $(s_\psi,t_\psi)$ are related to $(s,t)$ through
\be
\label{eqstpsi}
s_\psi= s+ t \tan\psi, \quad t_\psi=\frac{t}{\cos\psi}
\ee
we see that the loci $\Re\left( e^{-\I\psi} Z(\gamma)\right)=0$ in the $(s,t)$ 
upper half-plane are mapped to the loci $\Re( Z(\gamma))=0$ in the $(s_\psi,t_\psi)$ upper 
half-plane. The relation for the imaginary parts is however more complicated,
\be
\label{ImZLVpsi}
\Im\bigl( e^{-\I\psi} Z^{\rm LV}_{(s,t)}(\gamma) \bigr)
= \Im \bigl( Z^{\rm LV}_{(s_\psi,t_\psi)}(\gamma) \bigr)
- \sin\psi \, \Re \bigl( Z^{\rm LV}_{(s_\psi,t_\psi)}(\gamma) \bigr)
+ r (\sin\psi) t_\psi^2
\ee
In particular, the transformation $e^{-\I\psi} Z^{\rm LV}_{(s,t)}(\gamma) \rightarrow 
Z^{\rm LV}_{(s_\psi,t_\psi)}(\gamma)$ is {\it not} a $\GLt$ transformation, due to 
the last term in \eqref{ImZLVpsi}. Rather, one should view $(s_\psi,t_\psi)$
as variants of the affine coordinates $(x,y)$ defined in \eqref{defxysgen0}, 
indeed $(x,y)=(s_\psi, \frac12(t_\psi^2-s_\psi^2))$.

Thus, we conclude that whenever $\psi\in(-\frac{\pi}{2},\frac{\pi}{2})$,
the scattering diagram $\cD^{\rm LV}_\psi$ in
the $(s,t)$ upper half-plane coincides with the scattering diagram $\cD^{\rm LV}_0$ in
the $(s_\psi,t_\psi)$ upper half-plane, 
and both have the same image in the $(x,y)$ plane.
Thus, the SAFC holds on the large volume slice, as long as $\cos\psi\neq 0$.
 For $\psi=\pm \frac{\pi}{2}$, 
the scattering diagram $\cD^{\rm LV}_\psi$ 
becomes degenerate, as all loci $\Im Z(\gamma)=0$ become
either vertical lines $s=d/r$ for $r\neq 0$ or the horizontal line $t=0$ on the boundary. 
As we shall see in \S\ref{sec_global}, the scattering diagram $\cD^{\Pi}_\psi$ along the slice of
$\Pi$-stability conditions is better behaved in this limit.

%%%%%%%%%%%%%%%%%%%%%%%%

\section{The orbifold scattering diagram\label{sec_orbifold}}

In this section, we construct the scattering diagram $\cD_Q$ for the 
orbifold quiver shown in Figure~\ref{figQuiver}, following the general construction 
of \S\ref{sec_scatquiv}. In~\S\ref{sec_P2proof}, we identify the initial rays, 
giving a rigorous proof of the Attractor Conjecture for this case (Theorem~\ref{thm:attractor-conj} in the introduction). In \S\ref{sec_scattorb},
we restrict  $\cD_Q$ to a two-dimensional slice $\cD_o$, which we shall identify later in \S\ref{sec_orbregion}
with a subset of the exact scattering diagram $\cD^{\Pi}_{\psi}$. In \S\ref{sec_rank01}
we illustrate the use of the diagram $\cD_o$ for computing the index for small
dimension vectors.

\subsection{Quiver descriptions\label{sec_quiv}}

As explained for example in \cite[\S 4.2]{Mozgovoy:2020has}, any tilting sequence $\cS=(E_1,E_2,E_3)$ in the derived category $\cC=D^b\Coh(K_{\IP^2})$  (in the sense that $\cS$ generates $\cC$ and $\Ext^k(E_i,E_j)=0$ for $k\neq 0$ and all $i,j$)  provides an isomorphism $\cC \sim D^b(\Rep J(Q,W))$ 
with the  derived category of representations of the Jacobian algebra $J(Q,W)$ of a certain 
quiver with potential $(Q,W)$ associated to $\cS$. A tilting sequence can be obtained 
from any strong exceptional collection $(F_1,F_2,F_3)$ on $\IP^2$ by setting 
$\cS=(i_*(F_1[1]), i_*(F_2), i_*(F_3)[-1])$.
Starting from the strong exceptional collection $(\cO(-1),\Omega(1),\cO)$, one
arrives at the tilting sequence in \eqref{excepcoll}, associated to the 
quiver with potential depicted in Figure \ref{figQuiver}. All other strong exceptional collections are obtained by successive mutations of $(\cO(-1),\Omega(1),\Omega)$,
 and similarly lead to 3-node quivers $(Q_{a,b,c},W_{a,b,c})$ with  $a$ arrows 
 from node 1 
to node 2, $b$ arrows from 2 to 3 and $c$ arrows from 3 to 1, with 
$(a,b,c)$ any set of positive integers satisfying the Markov condition $a^2+b^2+c^2=abc$.

While the isomorphism $\cC \simeq D^b(\Rep J(Q,W))$ holds for any choice of tilting sequence $\cS=(E_1,E_2,E_3)$, the heart $\cA_\sigma$ only coincides (up to the action of $\GLt$) with the Abelian category $\Rep J(Q,W)$ %of representations of $J(Q,W)$ 
in a region $\IH_\cS$ of $\Stab\cC$ where the
objects $E_i$ are stable and  the phases of their central charges $Z(E_i)$ lie in a common
half-space (see for example Lemma 3.16 in \cite{macri2007stability}). In that case,  
the moduli space of $\sigma$-semistable objects $E$ of Chern vector $\gamma$ in 
$\cC$ coincides with the moduli space of semi-stable representations of $J(Q,W)$, with
dimension vector $(n_1,n_2,n_3)$ and King parameters $(\theta_1,\theta_2,\theta_3)$
 determined by 
  \be
  \label{thetafromZ}
  \gamma=\sum_{i=1}^3 n_i \ch E_i\ ,\quad 
 \theta_i=\lambda \, \Re[e^{-\I\psi}Z(E_i)]
 \ee
with $\psi=\arg Z(E)-\frac{\pi}{2}$ such that $n_1\theta_1+n_2\theta_2+n_3\theta_3=0$, and $\lambda$ is an arbitrary non-negative factor. At the boundary of the region $\IH_\cS$, one of the objects $E_i$ exits from the common half-space and one should apply a left or right mutation so as to obtain a new quiver description.

Let us consider the exceptional collection in \eqref{excepcoll}, corresponding to the fractional
branes on the orbifold $\IC^3/\IZ_3$:
\be
\label{excepcollo}
E_1=i_*(\cO)[-1], \quad E_2=i_*(\Omega(1)), \quad E_3 = i_*(\cO(-1))[1]
\ee
with charges  $\gamma_i=\ch E_i$ given by
\be
\gamma_1=[-1,0,0], \quad \gamma_2=[2,-1,-\frac12], \quad \gamma_3=[-1,1,-\frac12]
\ee
The corresponding quiver, shown in  Figure \ref{figQuiver}, has $3$ arrows $E_i\to E_{i-1}$
for each $i$ modulo 3, consistently with $\langle \gamma_i,\gamma_{i-1}\rangle =-3$. 
Using  the Eichler integral representation \eqref{Eichler0}, it is easy to check that the central
charges of the simple objects  coincide at the orbifold point $\tau_o$, namely $Z_{\tau_o}(\gamma_i)=1/3$. Thus, this quiver description will be appropriate\footnote{
Note that
the objects $E_i$ are not stable at $\tau=\tau_o$ strictly, since their central charge is not in
the half-plane $\IH_B$; the objects $E_i[1]$ are stable, and lead to the same quiver
but  opposite dimension vector.}
 in a region $\IH^o$ around $\tau_o$
in the physical slice $\Pi$, which we identify in \S\ref{sec_orbregion}. 
The dimension vector  associated to an object of charge $\gamma=[r,d,\ch_2]$ is given by
\be
\label{n123}
(n_1,n_2,n_3) = \Bigl(-\frac32 d-\ch_2-r,  -\frac12 d-\ch_2, \frac12 d-\ch_2\Bigr) = (-\chi, r+d-\chi,r+2d-\chi)
\ee
or conversely
\be
r=2n_2-n_1-n_3, \quad d=n_3-n_2, \quad \ch_2 =-\frac12(n_2+n_3)
\ee
Abusing notation once again, we shall write $\gamma=(n_1,n_2,n_3)$, using round brackets on both sides to distinguish it from the other notations $\gamma=[r,d,\ch_2]=[r,d,\chi)$.  

A similar quiver description holds true around any image of $\tau_o$ under an element $g\in \Gamma_1(3)$, with the tilting sequence being replaced by its image $(g(E_1),g(E_2),g(E_3))$. In particular, around $\tau=\tau_o+k$, the  tilting sequence \eqref{excepcoll} is shifted to 
\be
\label{excepcollk}
E_1(k)=i_*(\cO(k))[-1], \quad E_2(k)=i_*(\Omega(k+1)), \quad E_3(k) = i_*(\cO(k-1))[1]
\ee
with Chern vectors $\gamma_1(k),\gamma_2(k),\gamma_3(k)$. 
The relevant dimension vector is 
then 
 \be
 \begin{cases}
n_1(k) = n_1 - \frac12 k (3n_1-4n_2+n_3) + \frac12 k^2 (n_1-2n_2+n_3) \\
 n_2(k) =n_2 - \frac12 k(n_1-n_3)+ \frac12 k^2 (n_1-2n_2+n_3) \\
 n_3(k) =n_3 - \frac12 k(4n_2-n_1-3n_3) +  \frac12 k^2 (n_1-2n_2+n_3) 
 \end{cases}
 \ee
where $(n_1,n_2,n_3)$ denotes the value for $k=0$, given in \eqref{n123}. We note that
for fixed $\gamma=[r,d,\ch_2]$ and large $k$, the entries $n_i(k)$ grow like $-\frac12 k^2 r$ for 
$r\neq 0$, or like $k d$ for $r=0$, in particular they all have the same sign for large $k$. 

\subsection{Initial rays for the orbifold scattering diagram\label{sec_P2proof}}

In this subsection, we prove Theorem~\ref{thm:attractor-conj}, which states that
for the quiver with potential $(Q,W)$ 
shown in Figure \ref{figQuiver}, 
the attractor invariant $\Omstar(\gamma)$ vanishes for all dimension vectors $\gamma=(n_1,n_2,n_3)$ except for
\be
\gamma\in\bigl\{ (1,0,0), (0,1,0), (0,0,1), (n,n,n) : n\geq 1\bigr\}
\ee
where it takes     the values stated in \eqref{attindexP2}. 
This result was conjectured in \cite{Beaujard:2020sgs}, but the proof outlined in that paper 
was not mathematically rigorous. Here we complete the proof, combining ideas from \cite[\S3.3]{Beaujard:2020sgs} and \cite[\S 3.2]{Descombes:2021egc}. We denote by $a_i\colon 2\to 1$, $b_i\colon 3\to 2$ and $c_i\colon 1\to 3$ the arrows of the quiver. 

\begin{proof}[Proof of Theorem~\ref{thm:attractor-conj}]
Consider a representation $\phi$ of $(Q,W)$ with dimension vector $(n_1,n_2,n_3)$, 
assumed to be stable for a King stability condition $\theta=(\theta_1,\theta_2,\theta_3)$. 
Consider a cycle $w$ of the quiver: because $(Q,W)$ gives a noncommutative crepant resolution of $\IC^3/\IZ_3$ \cite[Proposition 3.13]{Mozgovoy:2009fi}, the cycle $w$ is a central element of the Jacobian algebra of $(Q,W)$, hence it defines an endomorphism of $\phi$. Because $\phi$ is stable, its only automorphisms are rescalings, hence $w$ acts as a scalar on $\phi$. Suppose now that $\phi$ is not of dimension vector $(n,n,n)$ for $n\in\IN^\ast$: $w$ cannot be an automorphism (otherwise all its arrows would be isomorphisms, hence one would have $n_1=n_2=n_3$), hence $w$ vanishes on $\phi$.

We suppose now that the King stability parameters satisfy $\theta_1<0$, and $\theta_3>0$ 
(since $\theta_1+\theta_2+\theta_3=0$, other cases follow by circular permutation of the nodes). We now show by contradiction that  $\phi_{c_j}$ vanishes for every arrow $c_j:1\to 3$. Suppose that there exists $x\in V_1$ and $j\in\{1,2,3\}$ such that $\phi_{c_j}(x)\in V_3$ does not vanish. Because $\theta_1<0$ and $\phi$ is stable, one must have $x\in \bigoplus_{k=1}^3 {\rm Im}(\phi_{a_k})$. Because $\theta_3>0$ and $\phi$ is stable, $\langle\phi_{c_j}(x)\rangle$ cannot be a sub-representation of $\phi$, hence there must be an arrow $b_i\colon 3\to 2$ such that $\phi_{b_i}\phi_{c_j}(x)\in V_2$ does not vanish. Because all the cycles of $Q$ vanish on $\phi$, one must have $\phi_{b_i}\phi_{c_j}(\bigoplus_{k=1}^3 {\rm Im}(\phi_{a_k}))=0$, a contradiction. Hence the $\phi_{c_j}$ vanish on $\phi$ for all $j$. The set $I=\{c_1,c_2,c_3\}$ of vanishing arrows provides a cut of the potential $W$, in the sense that each cycle of $W$ contains exactly one arrow of $I$.

The stable representation $\phi$ is then a representation of the quiver with relation $(Q_I,\partial_I W)$,  where the arrows $c_j\in I$ have been traded for relations $\partial_{c_j} W=0$ for $c_j\in I$. From general arguments of geometric invariant theory, the moduli $\mathcal{M}^{\theta,s}_{Q_I}$ of stable representations of $Q_I$ is smooth of dimension $3n_3n_2+3n_2n_1-n_1^2-n_2^2-n_3^2+1$. The moduli space $\mathcal{M}^{\theta,s}_{Q_I,\partial_I W}$ of stable representations of $(Q_I,\partial_I W)$ is cut out  by $3n_3n_1$ bilinear relations $\partial_{c_j} W=0$
inside $\mathcal{M}^{\theta,s}_{Q_I}$. We denote the infinitesimal versions of these relations on the tangent space of $\phi$ by $\delta(\partial_{c_j}W)\colon V_3\to V_1$. A linear dependence between these infinitesimal relations would be given by maps $\tilde{\phi}_{c_j}\colon V_1\to V_3$ such that:
\begin{align}
    \Tr(\tilde{\phi}_{c_j}\delta(\partial_{c_j}W))=0
\end{align}
As shown in \cite[\S 3.2]{Descombes:2021egc}, this equation is equivalent to the fact that the representation $\bar{\phi}\coloneqq(\phi_a,\phi_b,\tilde{\phi}_c)$ of $Q$ (which, like~$\phi$, satisfies the relations $\delta_{c_j}W=0$) satisfies the relations $\delta_{a_i}W=0$ and $\delta_{b_i}W=0$. 
Hence $\tilde{\phi}$ is a representation of $(Q,W)$, which is $\theta$-stable (because it has fewer subobjects than 
the $\theta$-stable representation $\phi$) of dimension $\gamma\neq (n,n,n)$. By the above arguments $\tilde{\phi}_{c_j}=0$, therefore the relations $\partial_{c_j} W=0$ are transverse, and $\mathcal{M}^{\theta,s}_{Q_I,\partial_I W}$ is smooth of dimension $3n_3n_2+3n_2n_1-3n_3n_1-n_1^2-n_2^2-n_3^2+1$. In particular, because this dimension is positive, one has
\begin{align}
\label{dimineq1}
    n_1^2+n_2^2+n_3^2-3n_3n_2-3n_2n_1-3n_1n_3+6n_3n_1\leq 1
\end{align}
Now the proof proceeds exactly as in \cite[\S 3.3]{Beaujard:2020sgs}. Assume that $\theta$ is an attractor stability condition, hence a small deformation of the self-stability condition $\langle-,\gamma\rangle$
for some dimension vector $\gamma\neq (n,n,n)$. 
Then $\theta_1<0$ implies that $n_1\geq n_2$, and $\theta_3>0$ implies that $n_3\geq n_2$: hence $6n_3n_1\geq 2n_3n_2+2n_2n_1+2n_1n_3$, and then, using \eqref{dimineq1}:
\be\label{eq}
\begin{aligned}
    q(\gamma)\coloneqq {} &\frac{1}{2}((n_1-n_2)^2+(n_2-n_3)^2+(n_3-n_1)^2)= 
     n_1^2+n_2^2+n_3^2-n_1n_2-n_2n_3-n_3n_1\\
    \leq {} & n_1^2+n_2^2+n_3^2-3n_1n_2-3n_2n_3-3n_2n_3+6n_3n_1\leq 1
\end{aligned}
\ee
The kernel of the quadratic form $q(\gamma)$ is given by dimensions vectors $(n,n,n)$, but we have assumed that $\gamma$ was not of this form. One has $q(\gamma)=1$ only for $\gamma=(n+1,n,n)$ and circular permutations thereof, but the inequality is strict in the second line of \eqref{eq} unless $n=0$. Hence, $\gamma$ must be $(1,0,0)$, $(0,1,0)$, or $(0,0,1)$, with attractor invariants equal to~$1$. The DT invariants for $\gamma=(n,n,n)$ are independent of the stability conditions, and given by
the shifted Poincar\'e polynomial 
$(-y)^{-3} P(K_{\IP^2},y)=-y^3-y-1/y$ by \cite[Remark 5.2]{Mozgovoy:2020has}. 
\end{proof}

\begin{figure}[ht]
\begin{center}
\includegraphics[width=8cm]{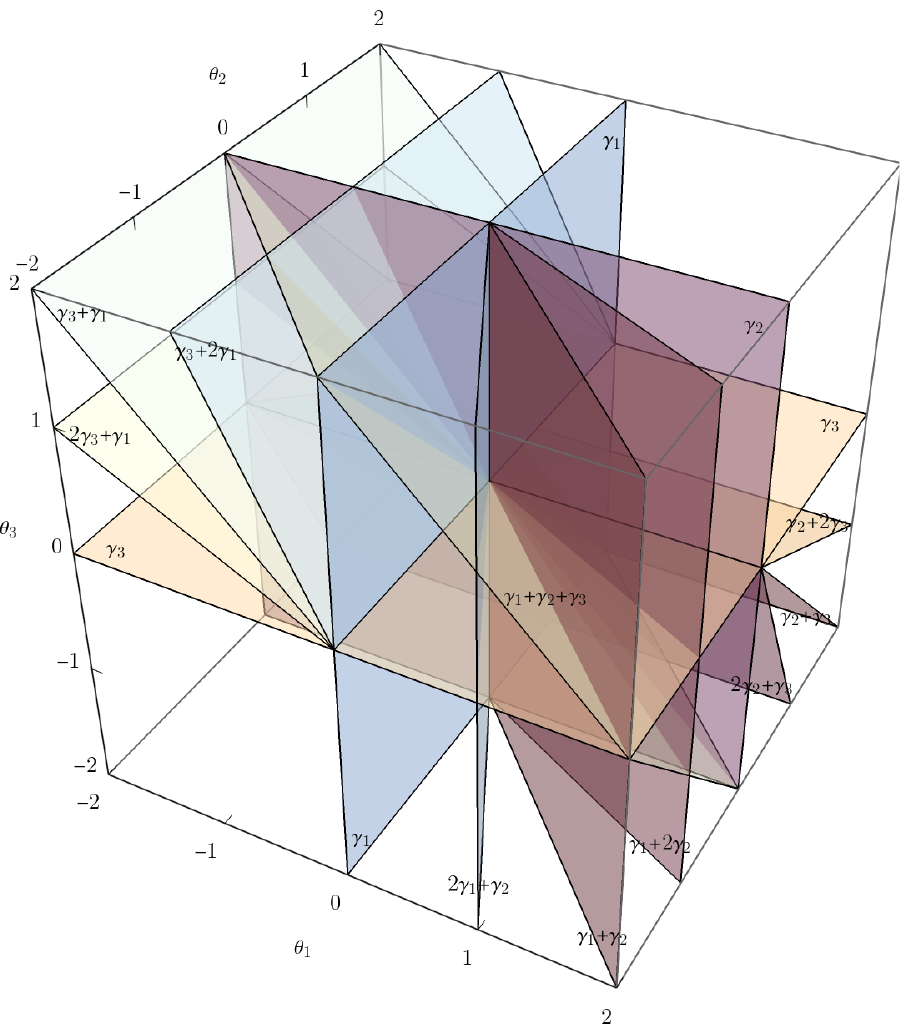}
\end{center}
\caption{Quiver scattering diagram $\cD_Q$, showing only  the initial rays 
 $\Ract_o(\gamma_i)$ and $\Ract_o(\delta)$ and secondary rays 
 $\Ract_o(\gamma_j+\gamma_k)$.
\label{figMcKayScat3D}}
\end{figure}

Having determined the initial rays for the orbifold quiver, it is now straightforward at least in principle to construct the stability scattering diagram $\cD_Q$. The result
is represented in Figure~\ref{figMcKayScat3D}, including only initial and some secondary rays. The full
diagram is dense except in a cone which includes the positive and negative octant, as will become clear shortly. 

\begin{figure}[ht]
\begin{center}
\includegraphics[width=10cm]{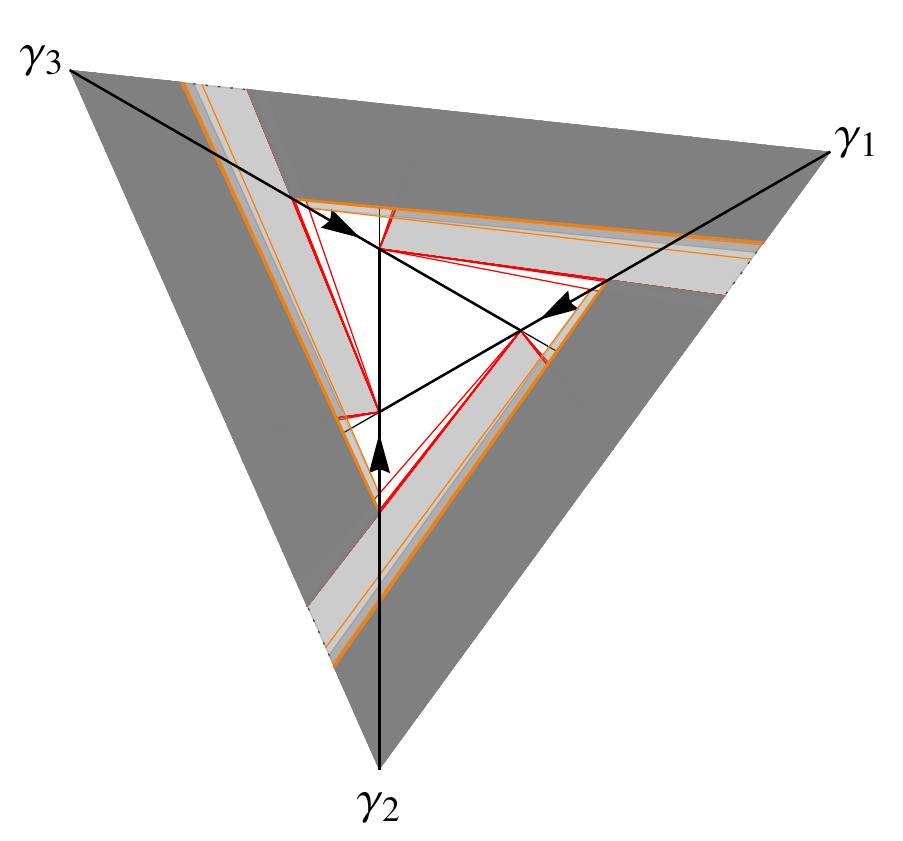}
\end{center}
\caption{Two-dimensional section $\cD_o$ of the orbifold scattering diagram $\cD_Q$ along the hyperplane 
$\theta_1+\theta_2+\theta_3=1$. We show only the restriction to the triangle $\Delta_\psi$ (see \S\ref{sec_tessel}) 
 bounded by the points $p_i(\cV_\psi)$ defined in \eqref{pilambda} 
(with $\psi=-1.4$ for illustration), but the scattering diagram $\cD_o$ in the space of King stability conditions is unbounded.
Regions with a dense set of rays are shown in gray. The rays in red correspond to the discrete set of real roots of the Kronecker quiver with 3 arrows, with dimension vector $(F_{2k},F_{2k+2},0)$ or permutations thereof, see \S\ref{sec_kron}.
\label{figMcKayScatdense}}
\end{figure}

\subsection{Restricted scattering diagram\label{sec_scattorb}}

Due to the rescaling symmetry in the space of King stability parameters, there is no significant loss 
of information\footnote{The price to pay is that the ray $\Reff_o(\delta)$ is no longer visible, but 
this causes no problem since the automorphism $\Ueff(\delta)$ is central; moreover, the 
self-stability condition $\langle -,\gamma\rangle$ is now pushed to infinity in the direction 
opposite to the vector $\nu(\gamma)$ introduced below. }
 in restricting the scattering diagram $\cD_Q$ to its intersection $\cD_o$ with 
  the hyperplane $\cH=\{\theta_1+\theta_2+\theta_3=1\}$. 
  It is furthermore convenient to parametrize this hyperplane by  $(u,v)\in\IR^2$ such that 
\be
\label{thetauv}
\theta_1 = -u+ v \sqrt3 +\tfrac13 ,\quad
\theta_2 =2 u+\tfrac13,\quad
\theta_3 = -u-v \sqrt3 +\tfrac13
\ee
In these coordinates, the $\IZ_3$ symmetry permuting the nodes of the quiver cyclically acts by  a rotation of angle $2\pi/3$ around the origin
(see  Figure \ref{figMcKayScat}). The geometric rays 
$\Rgeo_o(\gamma)=\{ (u,v) \in\IR^2 :  \sum_i n_i \theta_i(u,v)=0\}$
are given by straight lines 
\be
\label{mckayray}
 (n_1+n_3-2n_2) u +  (n_3-n_1)  v \sqrt3 - \frac{1}{3}(n_1+n_2+n_3) = 0
\ee
oriented along the vector  $ \nu(\gamma)=(\sqrt3 (n_3-n_1), 2n_2-n_1-n_3)$
pointing from the attractor stability condition $\theta_i=\lambda(n_{i-1}-n_{i+1})$ towards
the anti-attractor  stability condition $\theta_i=\lambda(n_{i+1}-n_{i-1})$, with indices $i=1,2,3$
taken modulo 3. In particular,
the initial rays intersect at $\Rgeo_o(\gamma_i)\cap \Rgeo_o(\gamma_{i+1})=\{p_i\}$ with 
\be
\label{p123}
p_1=\left(-\frac16,-\frac{1}{2\sqrt3}\right), \quad 
p_2=\left(-\frac16,\frac{1}{2\sqrt3}\right), \quad 
p_3=\left(\frac13,0\right)
\ee

Interpreting the ray as the worldline of a fiducial particle 
 traveling at point $ r=(u,v)$ with velocity $ \nu(\gamma)$, the condition
 \eqref{mckayray} implies that the particle has  
angular momentum $ r \wedge  \nu(\gamma)=-\frac13(n_1+n_2+n_3)$ with respect
to the origin $(0,0)$ in the $(u,v)$ plane, therefore rotates clockwise around the origin 
(assuming that  $n_i\geq 0$).  In fact, \eqref{mckayray} implies that  the particle has  
angular momentum 
\be
( r-p_i) \wedge \nu(\gamma)=-n_{i-1}
\ee
with respect to any of the points $p_i$, 
so the particle rotates clockwise around each $p_i$, and can only pass through it
if $n_{i-1}=0$ (this generalizes the previous statement since $p_1+p_2+p_3=0$).

Another consequence of \eqref{mckayray} is that
the linear function on $\Gamma$
\be
\label{defphip}
\varphi_o(\gamma) = (n_3-n_1) u \sqrt3 + (2n_2-n_1-n_3) v = ( r, \nu(\gamma))
\ee
increases monotonically along the ray
and is additive at each vertex, similar to the electric potential \eqref{defvarphi} 
in the large volume scattering diagram. Unlike the latter however,
the function \eqref{defphip} is not positive, 
in particular the first scattering between $\Ract_o(\gamma_i)$ and 
$\Ract_o(\gamma_{i+1})$ may take place at arbitrary negative values of $\varphi_o(\gamma_i),
\varphi_o(\gamma_{i+1})$. Fortunately, we do not need to rely on such a 
cost function to define the quiver scattering diagram and enumerate all possible
scattering sequences, instead we can use the fact that only positive dimension vectors support
non-zero DT invariants. 
In the next subsection, we  use the orbifold scattering diagram to compute DT invariants for some simple dimension vectors.

\subsection{Examples\label{sec_rank01}}

We first consider the case $\gamma=[1,0,1-n)$, corresponding to the Hilbert scheme of $n$ points on $\IP^2$. The corresponding dimension vector is $\gamma=(n-1,n,n)$. As in \S\ref{sec_Hilb} 
we find that the index in the anti-attractor chamber, which we denote by 
$\Omega_c(\gamma)\coloneqq\Omega_{\langle\gamma,-\rangle}(\gamma)$, 
agrees with the Gieseker index \eqref{genHilb}
for all cases considered, although the set of scattering sequences for the  large volume and orbifold 
scattering diagrams may differ. As explained in \S\ref{sec_scattB}, the absence of walls between
the anti-attractor chamber and the large volume chamber is a general property of objects
with slope $\mu\in [-1,0]$.

\begin{figure}[ht]
\begin{center}
\includegraphics[height=12cm]{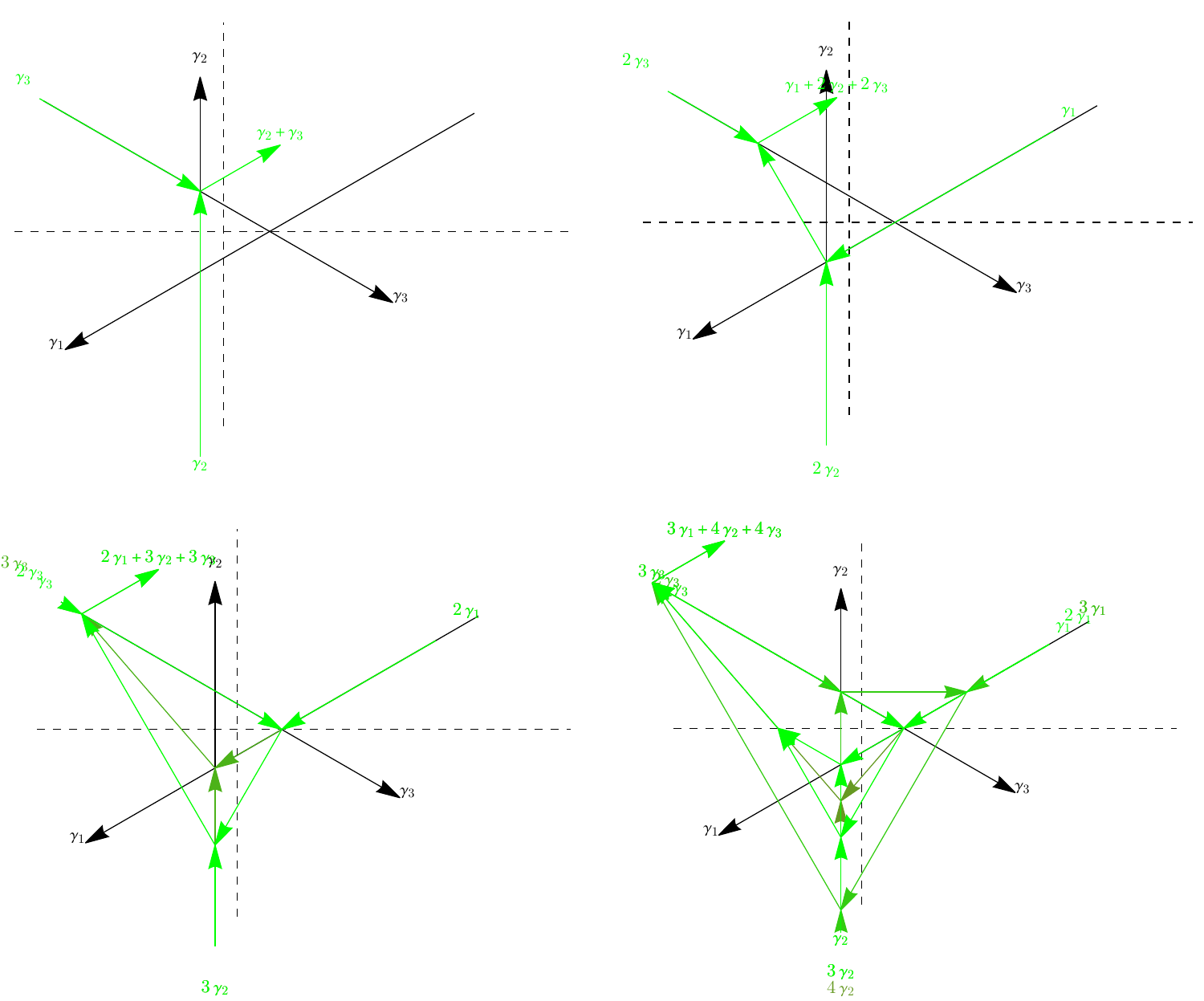}
\end{center}
\caption{Scattering sequences for dimension vector $(n-1,n,n)$ corresponding the Hilbert scheme of $n$ points on $\IP^2$,  with $n=1,2,3,4$\label{figH1234}}
\end{figure}

\begin{itemize}

\item For $n=1$ we find a single scattering sequence
$\{\gamma_2,\gamma_3\}$ contributing $K_3(1,1)=y^2+1+1/y^2\to 3$ in the anti-attractor chamber.

\item For $n=2$, we find a single scattering sequence
$\{\{\gamma_1,2\gamma_2\},2\gamma_3\}\}$ contributing $K_3(1,2)^2=(y^2+1+1/y^2)^2\to 9$.

\item For $n=3$ we find two scattering sequences,
\be
\begin{array}{l@{\hskip 5mm}l}
\left\{\left\{2 \gamma _1,3 \gamma _2\right\},3 \gamma _3\right\} &
 K_3(1,3) K_3(2,3) \to 13
\\
\left\{\left\{\left\{2 \gamma _1,\gamma _3\right\},3 \gamma _2\right\},2 \gamma _3\right\} & 
K_3(1,2){}^2 K_3(1,3) \to 9
\end{array}
\ee
for a total index of $y^6+2 y^4+5 y^2+6+\dots\to 22$.

\item For $n=4$ we find 3  scattering sequences,
\be
\begin{array}{l@{\hskip 5mm}l}
\left\{\left\{\left\{3 \gamma _1,\gamma _3\right\},4 \gamma _2\right\},3 \gamma _3\right\} & K_3(1,3){}^2 K_6(1,4) \to 15
\\
\left\{\left\{\left\{3 \gamma _1,\left\{\gamma _2,2 \gamma _3\right\}\right\},3 \gamma _2\right\},2 \gamma _3\right\} & K_3(1,2){}^2 K_3(1,3){}^2 \to 9
\\
\{\{\{ \{  2 \gamma _1,\gamma _3\}, 3 \gamma _2 \}, \{\gamma _1, \gamma _2  \}\} ,3 \gamma _3\} & 
 K_3(1,1){}^2 K_3(1,2) K_3(1,3){}^2 \to 27
 \end{array}
\ee
for a total index of $y^8+2 y^6+6 y^4+10 y^2+13+\dots\to 51$.

\end{itemize}

Next we consider $\gamma=[0,1,1-n)$,  corresponding to a D2-brane bound to $n$ anti-D0-branes.
The corresponding dimension vector is $\gamma =(n,n+1,n+2)$. Unlike the Gieseker index at large volume, the index $\Omega_c(\gamma)$ in the anti-attractor chamber depends on $n$, though it is
still invariant under $n\mapsto -n-2$. After a circular permutation, the dimension vector becomes
$\gamma'=(n+1,n+2,n)=[3,-2,-1-n)$, and $\Omega_c(\gamma')$ agrees
with the Gieseker index for rank 3 sheaves with $d=-2$ and $\chi=-1-n$. 
\begin{itemize}

\item  For $n=-1$, the dimension vector $(-1,0,1)$ has mixed signs so the index vanishes.

\item For $n=0$,  we find a single  scattering sequence
$\{\gamma_2,2\gamma_3\}$ contributing $K_3(1,2)=y^2+1+1/y^2\to 3$.
\item For $n=1$, we find $2$ scattering sequences,
\be
\begin{array}{ll}
\{\gamma_1, \{2\gamma_2,3\gamma_3\}\} & K_3(2,3)K_3(1,1)\to 39 \\
\{\{\gamma_1,2\gamma_2\},3\gamma_3\} & K_3(1,2) K_3(1,3)\to 3
\end{array} 
\ee
giving a total index of  $y^8+2y^6+5y^4+8y^2+10+\dots \to 42$. 

\item For $n=2$, we find  $3$ scattering sequences,
\be
\begin{array}{ll}
\{ 2 \gamma _1, \{3\gamma_2,4\gamma_3\}\} &  K_3(1,2) K_3(3,4) \to 204 \\
\{\{\{2\gamma_1,\gamma_3\},3\gamma_2\},3\gamma_3\} & K_3(1,2) K_3(1,3)^2 \to 3\\
\{\{\{\gamma_1,3\gamma_2\},4\gamma_3\},\gamma_1\} &
K_3(1,1) K_3(1,3) K_6(1,4) \to 45  \\
\{\{ \{\gamma_1,2\gamma_2\},3\gamma_3\},\{\gamma_2,\gamma_3\}\},\gamma_1\} & 
K_3(1,1)^3 K_3(1,2) K_3(1,3) \to 81
\end{array} 
\ee
giving a total index of $333$.

\item For $n=3$, we find $7$ scattering sequences,
\begin{equation*}
\begin{array}{ll}
\left\{3 \gamma _1,\left\{4 \gamma _2,5 \gamma _3\right\}\right\} & K_3(1,3) K_3(4,5) \to 399
\\
\left\{2 \gamma _1,\left\{\left\{\left\{\gamma _1,3 \gamma _2\right\},5 \gamma _3\right\},\gamma _2\right\}\right\} & K_3(1,2) K_3(1,3) K_6(1,5) K_{12}(1,1) \to 216
\\
\left\{2 \gamma _1,
\{\{ \left\{\gamma _1,3 \gamma _2\right\} , 4 \gamma _3\},\{ 
 \gamma _2 , \gamma _3 \}\}
\right\} & K_3(1,1){}^2 K_3(1,2) K_3(1,3) K_6(1,4) \to 405
\\
\left\{2 \gamma _1,\{\{
 \left\{\gamma _1,2 \gamma _2\right\} , 3 \gamma _3 \},\{
 2 \gamma _2 , 2 \gamma _3 \}\}
 \right\} & K_3(2,2,1) K_3(1,2){}^2 K_3(1,3)\to  648
\\
\left\{\left\{\left\{3 \gamma _1,\left\{\gamma _2,2 \gamma _3\right\}\right\},3 \gamma _2\right\},3 \gamma _3\right\} & K_3(1,2) K_3(1,3){}^3 \to  3
\\
\left\{\left\{\left\{2 \gamma _1,4 \gamma _2\right\},5 \gamma _3\right\},\gamma _1\right\} & K_{6,6,3}(2,2,5) K_3(1,1) \to  216
\\
\left\{\{\{
 \left\{\left\{2 \gamma _1,\gamma _3\right\},3 \gamma _2\right\} , 3 \gamma _3 \},\{
 \gamma _2 , \gamma _3 \}\}
,\gamma _1\right\} & K_3(1,1){}^3 K_3(1,2) K_3(1,3){}^2 \to 81
\end{array}
\end{equation*}
giving a total index of $1968$. Similar to the discussion below \eqref{Om221}, the factor $K_{6,6,3}(2,2,5)$
for the penultimate sequence is computed by applying the flow tree formula
to a local scattering diagram with two incoming rays of charge $\alpha=\gamma_1+2\gamma_2$
and $\beta=\gamma_3$ with $\bOm^-(\alpha)=K_3(1,2)=y^2+1+1/y^2, \bOm^-(2\alpha)
=K_3(2,4)=-y^5-y^3-y-1/y-1/y^3-1/y^5$ and $\bOm^-(k\beta)=\delta_{k,1}$,
selecting the outgoing ray of charge $2\alpha+5\beta$.

\end{itemize}

\section{The exact scattering diagram\label{sec_global}}

In this section, we determine the scattering diagram $\cD_{\psi}^{\Pi}$ along the slice of 
$\Pi$-stability conditions, by combining results on the large volume 
and orbifold scattering diagrams constructed in \S\ref{sec_LV} and \S\ref{sec_orbifold} with invariance under $\Gamma_1(3)$. We start by analyzing the attractor flow for the exact central
charge \eqref{defZPi}, first in the Poincar\'e upper half-plane (\S\ref{sec_piflow}) and then in affine coordinates (\S\ref{sec_affine}). 
In \S\ref{sec_orbregion} we then identify the orbifold scattering diagram $\cD_o$ as a particular subset of  $\cD_{\psi}^{\Pi}$  
in a region around the orbifold point. 
In \S\ref{sec_exactD} we describe the full diagram~$\cD_\psi^\Pi$ and prove  the SAFC (Theorem~\ref{thm:SAFC}).

\subsection{Exact attractor flow\label{sec_piflow}}
In this section, we study the attractor flow \eqref{attflow} for the central charge \eqref{defZPi}
along the slice of $\Pi$-stability conditions, equipped with an general hermitean metric 
$\de s^2=g_{\tau\bar\tau} \de\tau\de\bar\tau$. 
Using the Eichler integral representation 
\eqref{Eichler0} for the coefficients $T,T_D$ in the central charge, 
the attractor flow $\AF(\gamma)$ \eqref{attflow} reduces to 
\be
\label{attflowC}
\frac{\de\tau}{\de\mu}= - g^{\tau\bar\tau} \partial_{\bar\tau} | Z_\tau(\gamma)|^2 =
- g^{\tau\bar\tau} ( d- r \bar\tau)\, \overline{C(\tau)} Z_\tau(\gamma)
\ee
The flow is only meaningful in the region where the charge $\gamma$
is populated, $\bOm_\tau(\gamma)\neq 0$. 
As already noted in \S\ref{sec_flow}, the modulus $|Z_\tau(\gamma)|= \Im[e^{-\I\psi} Z_\tau(\gamma)]$ decreases along the flow while the argument of $Z_\tau(\gamma)$ is
preserved, so the trajectories of \eqref{attflowC} are included in 
the active ray $\Ract_\psi(\gamma)$ defined in \eqref{defelemray}
as the locus $\{\Re[e^{-\I\psi} Z_\tau(\gamma)]=0, 
\Im(e^{-\I\psi}Z_\tau(\gamma))>0, 
\bOm_\tau(\gamma)\neq 0\}$ 
inside $\IH$, where the phase $\psi$ is determined by the argument of $Z_{\tau_0}(\gamma)$
at the starting point $\tau_0$.  The attractor flow stops when
either a) the flow crosses a wall of marginal stability, after which
$\bOm_\tau(\gamma)$ jumps, b)
 $|Z_\tau(\gamma)|$ reaches a local minimum in the interior of the upper half-plane,
 or c)  $\tau$ reaches the boundary $\Im\tau=0$. 
 In \S\ref{sec_must_end}
we rule out the possibility of an infinitely long flow.
 
 Case a) arises at a point $\mu_1$ where  the quantity $\tilde W_\tau(\gamma,\gamma')\coloneqq \Im[ Z_\tau(\gamma') \overline{Z_\tau(\gamma)}]$ vanishes for some charge $\gamma'$. Since  $\tilde W_\tau(\gamma,\gamma')$ varies along the flow according to
 \be
\label{dWdmu}
\frac{\de \tilde W_\tau(\gamma,\gamma')}{\de\mu} = -g^{\tau\bar\tau}  |C|^2 
\left[ 2 \tau_2 (rd'-r'd) |Z_\tau(\gamma)|^2 +  \tilde W_\tau(\gamma,\gamma') |d-r\tau|^2 \right]
\ee
the sign of $(rd'-r'd) \tilde W_\tau(\gamma,\gamma')$ is positive for $\mu<\mu_1$ and negative for $\mu>\mu_1$. Thus, the flow crosses from the side of the wall where a two-particle bound state of charges $\gamma$ and $\gamma-\gamma'$ is stable, into the side where the bound state decays.
In particular, it follows from \eqref{dWdmu} that the flow cannot cross the same wall more than once,
and that BPS states do not decay as one follows the attractor flow in reverse \cite{Denef:2001xn}. 

As for case b), since $C(\tau)\neq 0$ for 
$\Im\tau>0$, $|Z_\tau(\gamma)|$ can only reach a local 
 minimum at a point $\tau_i$  in the interior of $\IH$ if $Z_{\tau_i}(\gamma)=0$, but this 
is ruled  out by the assumption that $\bOm_\tau(\gamma)\neq 0$ along the flow (and the support property). This leaves
case c) with $\Im\tau_i=0$. Consider the image of the attractor flow on the quotient $\IH/\Gamma_1(3)$. Since the quotient can be compactified by adding the large volume and conifold points, $\tau_i$ is either in the orbit of a large volume point or a conifold point, \ie $\tau_i=p/q$ with $(p,q)$ coprime, $q=0\mod 3$ in the first case and $q\neq 0\mod 3$ in the second case. Using suitable bounds on the central charge of semi-stable objects in the large volume region, we show in \S\ref{sec_growth} that $\tau_i$ can only be a conifold point, where some spherical object $E$ of charge $\gamma_C$ becomes massless, $Z_{\tau_i}(\gamma_C)=0$.

To determine $E$, it suffices to find an element $g\in\Gamma_1(3)$
which maps $\tau=0$ to $\tau_i=p/q$, and act  with the corresponding
auto-equivalence on the object $\cO$ which is massless at $\tau=0$. 
Since $(T,T_D)=(\I \cV,0)$ at $\tau=0$,  where $\cV$ is the 
quantum volume defined in \eqref{defV0}, 
the periods at $\tau_i$ are then 
\be
\label{TTDcon}
T(\tau_i) = m + \I q\cV, \quad T_D(\tau_i) = m_D + \I p \cV
\ee
where $(m,m_D)$ are the off-diagonal elements of the monodromy matrix $M(g)$ in 
\eqref{monodromyact}. The charge of the massless object $E$ at $\tau=\tau_i$
is then $\gamma_C=[q,p,p m-q m_D]$ up to overall sign, see Table \ref{Conifoldtab}
for examples with $0\leq p\leq q\leq 5$.

Returning to the active ray $\Ract_\psi(\gamma)$, it follows from \eqref{TTDcon} that it can only reach the conifold point $\tau_i= p/q$ if 
\be
\label{Znearcon}
0 = \Re[e^{-\I\psi} Z_{p/q}(\gamma)] =
(d m - r m_D -\ch_2) \cos\psi + (dq-p r)\cV \sin\psi 
\ee
For generic $\psi$, this is only possible if $\gamma$ is a multiple of $\gamma_C$,
in which case $\cR_\psi(\gamma)$ coincides with the ray $\cR_\psi(\gamma_C)$ originating
 from $\tau_i=p/q$. For special
values of $\psi$ such that $\cV_\psi\coloneqq\cV \tan\psi$ is rational, namely 
\be
\label{psicrit}
\cV_\psi= \frac{\ch_2 + r m_D - d m}{d q-pr}
\ee
it is possible that a ray  $\Ract_\psi(\gamma)$
with $\langle \gamma_C,\gamma\rangle  = d q-pr \neq 0$ originates from the point 
$\tau_i=\frac{p}{q}$,
even though $Z_{\tau_i}(\gamma)\neq 0$. Conversely, for the same critical values of $\psi$ there can also be rays which terminate at a conifold point, rather than escaping at a large 
volume point, as proven in Proposition~\ref{prop:critical} (this phenomenon occurs for instance in the $\psi\to\psicr{-1/2}$ limit of Figure \ref{figscattPicr}).
In the next subsection, we introduce adapted coordinates where these critical rays become
much easier to detect.

\subsection{Affine coordinates\label{sec_affine}}
While the attractor flow $\AF(\gamma)$ is complicated on the $\tau$-plane, it becomes a straight line
 in affine  coordinates $(x,y)$ defined globally by\footnote{The affine coordinate $y$ should not be confused with the refinement variable $y$ in the definition of motivic DT invariants.}
\be
\label{defxysgen}
x= \frac{\Re\left( e^{-\I \psi} T\right)}{\cos\psi}\ ,\quad 
y=  -\frac{\Re\left( e^{-\I \psi} T_D\right)}{\cos\psi}
\ee
It is useful to introduce the dual coordinates 
 \be
\label{defxydual}
\tilde x=\frac{\Im(e^{-\I\psi}T)}{\cos\psi} , \quad 
\tilde y=-\frac{\Im(e^{-\I\psi} T_D)}{\cos\psi} 
\ee
such that the geometric ray $\Rgeo_\psi(\gamma)$, oriented along the direction of
decreasing $\Im[ e^{-\I\psi} Z(\gamma)]=r\tilde y+d\tilde  x$, are given by oriented straight lines 
\be
\label{rayxy}
ry+dx-\ch_2=0,\quad \frac{\de}{\de\mu} ( r \tilde y + d \tilde x ) <0
\ee  
Notice that for the large volume central charge $Z^{\rm LV}_{(s,t)}$, the 
coordinates $(x,y)$ coincide with $(s_\psi,\frac12(t_\psi^2-s_\psi^2))$
defined in \eqref{eqstpsi} and generalize the coordinates
$(x,y)$ of \cite{Bousseau:2019ift} to any $\psi\in(-\frac{\pi}{2},\frac{\pi}{2})$.
Notably, the  image of the upper half-plane
$t>0$ lies above the parabola $y= -\frac12 x^2$.

In contrast, for the 
exact central charge \eqref{defZPi}, the image of 
 the fundamental domain 
$\cF_C$ and its translates in the $(x,y)$-plane
extend below the parabola $y= -\frac12 x^2$, 
while remaining above the parabola $y=-\frac12 x^2 - \frac{5}{24}$ passing through the images 
of the orbifold points at $\tau=\tau_o+n$  (see Figure \ref{figGrTTD})
\be
(x_{o(n)},y_{o(n)}) = \left(n -\tfrac12,  -\tfrac13 - \tfrac12 n(n-1)\right)
\ee
For $\psi=0$, the preimage of the region $y> -\frac12 x^2$ 
includes the domain $\IH^{\rm LV}$ defined by the condition $w > \frac12 s^2$ (or equivalently $t>0$).
In fact, the curves  $t=0$  and $y=-\frac12 x^2$ are tangent to the same point 
$\tau\simeq\frac12+0.559926 \I$ corresponding to $(s,w)=(\frac12,\frac18)$ or $(x,y)=(\frac12,-\frac18)$, and are nearly indistinguishable for any $\tau_1$,
with the former lying below the latter. 
For general $\psi$, using \eqref{shalfintegertau1} along the line $\tau_1=\frac12$
we get 
\be
y+\frac12x=w-\frac12 s^2 + \frac12 \Im T \tan\psi ( 1+ \Im T \tan\psi )
\ee
so the relative position of the two curves depends on the
sign of $\Im T \tan\psi$.

\begin{figure}[ht]
\begin{center}
\includegraphics[width=7.5cm]{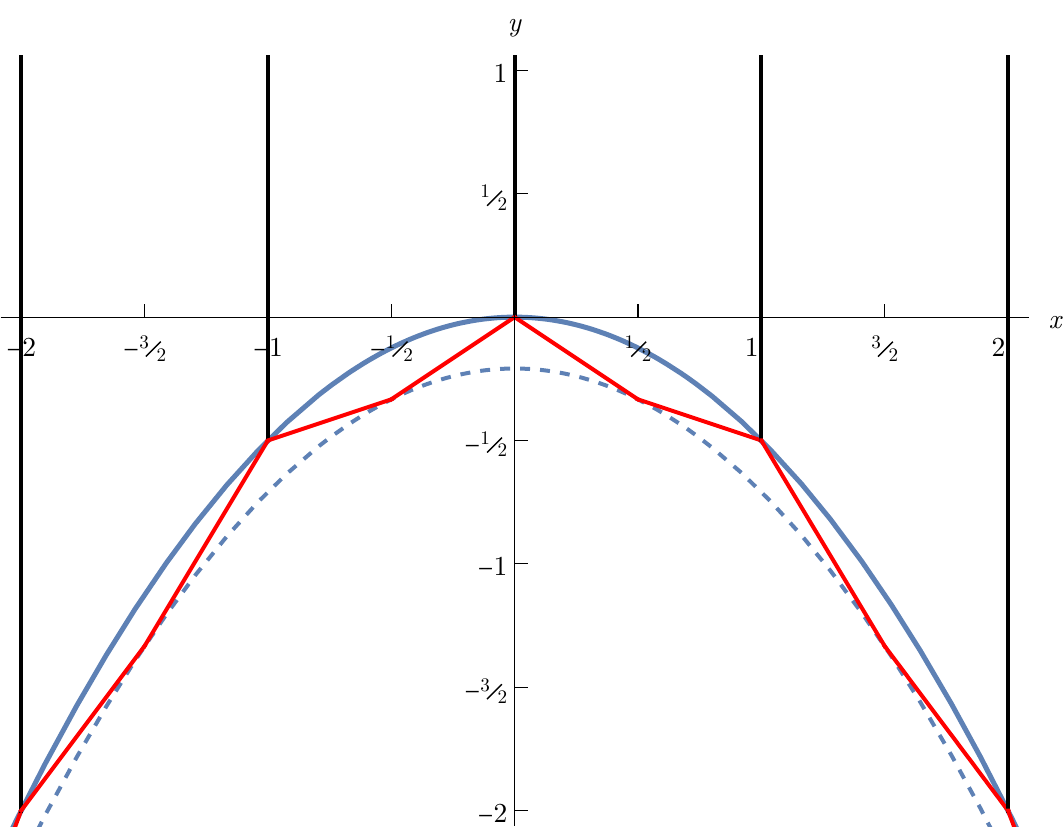}\hspace*{1cm}
\includegraphics[width=7.5cm]{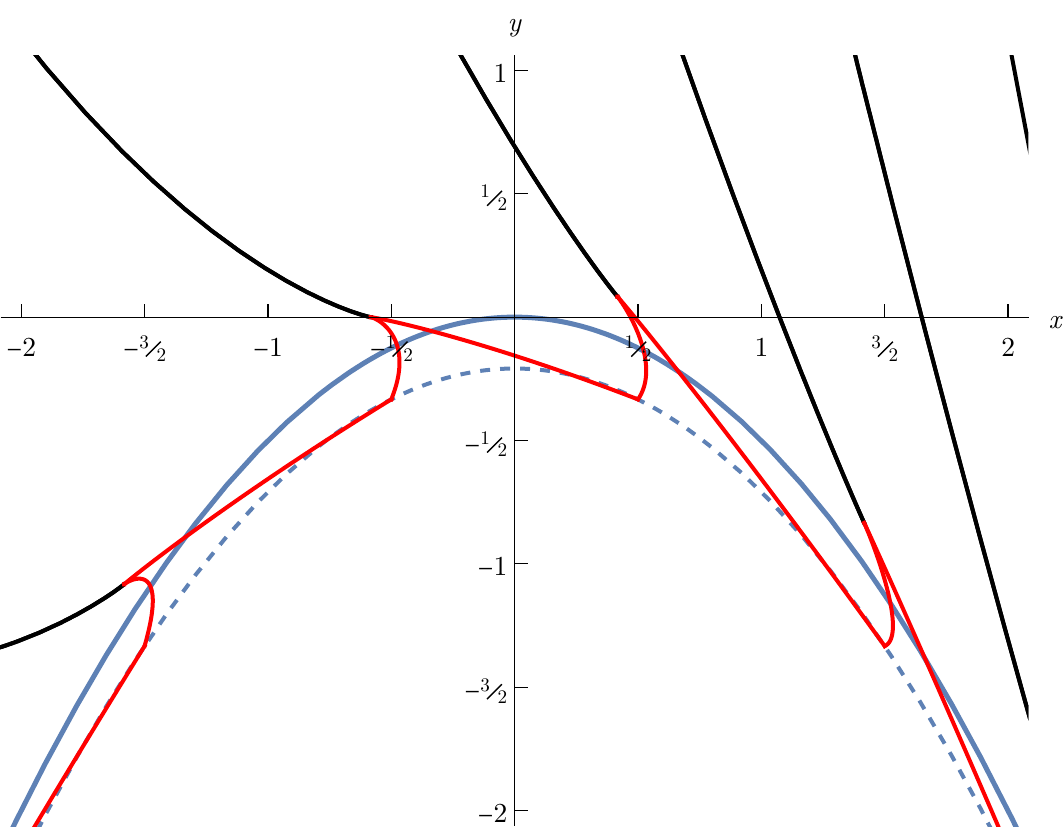}
\end{center}
\caption{Image of the fundamental domain $\cF_C$ and its translates $\cF_C(m)$
in the $(x,y)$-plane for $\psi=0$ (left) and $\psi=-0.9<-\psicr{1/2}$ (right). The parabolas $y=-\frac12 x^2$
and $y=-\frac12 x^2-\frac{5}{24}$ are shown for reference. For any $\psi\in(-\frac{\pi}{2},\frac{\pi}{2})$,
the map $\tau\mapsto(x,y)$ is  injective in $\bigcup_{m\in\IZ} \cF_C(m)$. 
\label{figGrTTD}}
\end{figure}

Geometrically, the coordinates $(x,y)$ and $(\tilde x,\tilde y)$ provide two systems of local Darboux
coordinates for the K\"ahler form
\be
\de x \wedge \de y = \de \tilde x \wedge \de \tilde y
=-\frac14 (\de T \wedge \de \bar T_D+\de \bar T \wedge \de T_D)
\ee
Hence, locally there exists a function $H_\psi(x,y)$ such that $\tilde x=\partial_y H_\psi, \tilde y= -\partial_x H_\psi$, satisfying  the Monge-Amp\`ere type equation
\be
\label{eqMA}
(\partial^2_x H_\psi) (\partial^2_y H_\psi) - (\partial_x \partial_y H_\psi)^2=1
\ee
One easily checks that the Legendre transform to the imaginary part of the tree-level 
prepotential \eqref{eqF0},
sometimes known as the Hesse potential \cite{Cardoso:2010gc}, provides such a function,
\be
\label{HLeg}
H_\psi(x,y) = \langle \Im\left[ -\frac{2\pi\I}{3} \frac{e^{-2\I\psi}}{\cos^2\psi} 
F_0\left( \cos\psi e^{\I\psi} (x+\I \tilde x) \right)\right] + \tilde x y \rangle_{\tilde x}
\ee
where $F_0$ is normalized such that $T_D=-\frac{3}{2\pi\I} \partial_T F_0$ (see \eqref{eqF0}). In
the large volume limit,  \eqref{HLeg} evaluates to
\be\begin{aligned}
 H^{\rm LV}_\psi(x,y) & = \frac{1}{3\cos\psi} (x^2+2y)^{\frac32} - \frac13 x(x^2+3y) \tan\psi \\
 & = \frac{t^3}{3\cos\psi} + \frac{s}{6}(s^2-3t^2) \tan\psi -\frac{t}{6} (t^2-3s^2) \tan^2\psi
\end{aligned}
\ee
An explicit computation shows that the Hessian matrix of $H_\psi$ with respect to $(x,y)$ coincides with that of $H^{\rm LV}_\psi$ (\ie the Hessian is insensitive to the factor $C(u)$ appearing in the Eichler integral \eqref{Eichler0}), namely
 \be
 \begin{pmatrix} 
 \partial_{xx} H_\psi & \partial_{xy} H_\psi \\ 
  \partial_{xy} H_\psi  &  \partial_{yy} H_\psi
   \end{pmatrix} 
   =   \frac{1}{\tau_2}  \begin{pmatrix}  |\tau|^2 & \tau_1 \\
 \tau_1 & 1 \end{pmatrix} 
 \ee
In particular it is positive definite. This observation allows to write the orientation condition 
 in \eqref{rayxy} purely in terms of the coordinates $(x,y)$: indeed, it becomes
\be
\label{phidec}
- \frac{1}{r}
 \left[ r^2 \partial_x^2 H_\psi -2 d r \partial_x\partial_y H_\psi + d^2 \partial_y^2 H_\psi  \right] 
 \frac{\de x}{\de\mu} <0
 \ee
 Thus we conclude that for $r\neq 0$ the ray is oriented such that $r \frac{\de x}{\de \mu}>0$. In particular,
 the electric potential $\varphi_x(\gamma)=2(d - x r)$ 
 (generalizing \eqref{defvarphi} for the exact central charge)
 decreases along the ray. When $r=0$, we get instead $\frac{\de (d \tilde x)}{\de\mu}=\frac{d}{\tau_2}  \frac{\de y}{\de\mu}<0$, so the vertical ray is oriented towards decreasing $y$ for $d>0$ (or increasing
 $y$ for $d<0$).

We now determine the images of the conifold points in the $(x,y)$ plane. Using \eqref{TTDcon},
we find that  the point $\tau=p/q$ where the object of charge $\gamma_C=[q,p,p m-q m_D]$ becomes massless maps to 
\be
(x_C,y_C) = ( m + q \cV_\psi, -m_D - p \cV_\psi)
\ee
where $\cV_\psi=\cV\tan\psi$.
In particular,  the conifold point $\tau=m$ where 
the object  $\cO(m)$ becomes massless is mapped to the point
\be\label{xyOm}
(x_{\cO(m)},y_{\cO(m)}) = \Bigl(m + \cV_\psi,  -\frac12{m^2} - m \cV_\psi\Bigr)
\ee
on the parabola $y = -\frac12 x^2 + \frac12 \cV_\psi^2$. 
Moreover, the associated rays $\Rgeo_\psi(\pm\cO(m))$
are contained in the tangent to the parabola $y=-\frac12 x^2$ at $x=m$, just as
for the large volume scattering diagram. 
Similarly, the point $\tau=n-\frac12$ where  $ \Omega(n+1)$ becomes massless is mapped to 
\be
(x_{\Omega(n+1)},y_{\Omega(n+1)}) = \Bigl( n -\frac12 -2 \cV_\psi,  -\frac12(n^2-n+1) +(2n-1) \cV_\psi\Bigr)
\quad 
\ee
on the parabola $y= -\frac12 x^2 + 2\cV_\psi^2 -\frac38$,
and the corresponding rays $\Rgeo_\psi(\pm\Omega(n+1))$ are contained in the tangent to the parabola $y=-\frac12x^2-\frac38$ at $x=n-\frac12$.
The corresponding points are marked in Figure \ref{figGrxy}. 

\begin{figure}[ht]
\begin{center}
\includegraphics[width=10cm]{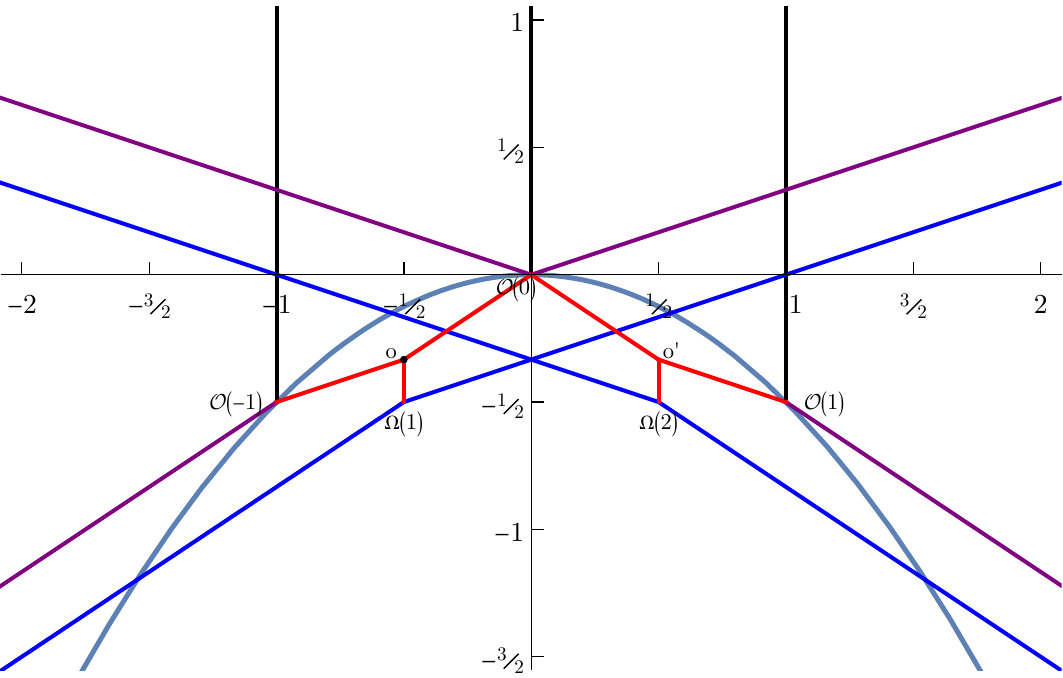}
\end{center}
\caption{Images of the  fundamental domains $\cF_o$, $\cF_{o'}=\cF_o(1)$ and their $\IZ_3$ images in the $(x,y)$ plane for $\psi=0$. The vertical black lines are the images of the semi-infinite lines $n+\frac{\I}{2\sqrt3} + \I \IR^+$ for $n=-1,0,1$, while the blue and purple lines correspond to their images under $\IZ_3$. 
\label{figGrxy}}
\end{figure}

\subsection{Orbifold region\label{sec_orbregion}}

We now discuss the domain of validity of the orbifold quiver description inside the
space of  $\Pi$-stability conditions, and how the two-dimensional scattering diagram constructed
in \S\ref{sec_scattorb} fits in the scattering diagram for the exact central charge. 

 As explained in \S\ref{sec_quiv}, 
the quiver description holds  in the region $\IH^o\subset \Pi$ where the central 
charges $Z_\tau(\gamma_i)$ of the simple objects belong to a common half-plane. 
In that region, the heart $\cA(\tau)$ is equivalent under $\GLt$ to the Abelian category
of representations of $J(Q,W)$, and the moduli space of semi-stable objects of phase 
$\arg Z_\tau(\gamma) = \psi+\frac{\pi}{2}$ in $\cA(\tau)$ coincides with the 
moduli space of semi-stable representations for  King stability parameters 
$\theta_i\propto\Re[e^{-\I\psi}Z_\tau(\gamma_i)]$.  Choosing the scale
such that $\theta_i(\gamma)=\Im [Z(\gamma_i) \bar Z(\gamma)]=\tilde W(\gamma,\gamma_i)$,
we can apply the relation \eqref{dWdmu} to obtain the variation of $\theta_i$ along
the attractor flow $\AF(\gamma)$, 
\be
\frac{\de \theta_i}{\de \mu} = -g^{\tau\bar\tau}  |C|^2 
\left[ 2 \tau_2 \langle \gamma,\gamma_i \rangle  |Z(\gamma)|^2 +   |d-r\tau|^2  \theta_i \right]
\ee
The second term simply amounts to a rescaling of $\theta_i$, while the first term shows that 
$\theta_i$ flows towards the self-stability condition $\langle \gamma_i,\gamma \rangle$.
On the other hand, 
since $Z_\tau(\gamma_1)+Z_\tau(\gamma_2)+Z_\tau(\gamma_3)=Z_\tau(\delta)=1$ 
for any $\tau\in\IH$,
 it makes more sense to fix the scale by requiring
 \be
 \label{thetaZ}
 \theta_i =\frac{ \Re[e^{-\I\psi}Z_\tau(\gamma_i)] }{\cos\psi}
 \ee
 such that $\theta_1+\theta_2+\theta_3=1$. The King stability parameters \eqref{thetaZ} are
 then linear combinations
 of the affine coordinates defined in \eqref{defxysgen}, such that the lines of gradient flow \eqref{rayxy}
 corresponds to straight lines $n_1\theta_1+n_2\theta_2+n_3\theta_1=0$. 
 
Since the central charges $Z_{\tau}(\gamma_i)$ are all equal at the orbifold point $\tau_o=\frac{1}{\sqrt3} e^{5\pi\I/6}$,  the region $\IH^o\subset \Pi$ includes an open set around $\tau=\tau_o$. 
Using the Eichler integral representation \eqref{Eichler0}, it is immediate to compute
the central charge 
to first order in $\tau-\tau_o$, 
\be
Z_\tau(\gamma) \simeq Z_{\tau}^o(\gamma)\coloneqq -\frac{r}{3}-\frac{d}{2}-\ch_2 +  C(\tau_0) (d-r\tau_o)\, (\tau-\tau_o)
\ee
leading to the central charges for the simple objects 
\be
\label{Ziexp}
\begin{cases}
Z_\tau(\gamma_1) \simeq
 \tfrac13 + \tau_o \, C(\tau_o) \, (\tau-\tau_o) \\ 
Z_\tau(\gamma_2) 
 \simeq   \tfrac13 -(2 \tau_o+1) \, C(\tau_o) \, (\tau-\tau_o) \\
Z_\tau(\gamma_3) 
\simeq   \tfrac13 + (\tau_o+1) \, C(\tau_o) \, (\tau-\tau_o) \\
\end{cases}
\ee
Computing the stability parameters via \eqref{thetaZ}, one finds that the coordinates $(u,v)$
on the two-dimensional section \eqref{thetauv} are related to the distance away from the
orbifold point via
\be
\label{tauouv}
\tau- \tau_o \simeq  \frac{2\I \sqrt3}{C(\tau_o)}  e^{\I\psi} (u+\I v) \cos\psi 
\ee 
In particular, as $\psi\to\pm \frac{\pi}{2}$, the scattering diagram $\cD_o$ is mapped to an infinitesimal
neighbourhood of $\tau_o$ (see Figure \ref{figslopecontour}). 

Going away from  the orbifold point, the region $\IH^o$ is bounded  in the 
fundamental domain~$\cF_o$
by the condition 
\be
\label{ImZ13}
\Im( \overline{Z_\tau(E_1)}Z_\tau(E_3))>0
\ee
 and similarly by conditions 
$ \Im( \overline{Z_\tau(E_2)}Z_\tau(E_1))>0$, $\Im( \overline{Z_\tau(E_3)}Z_\tau(E_2))>0$ in the images
$M_o \cF_o$ and $M_o^2 \cF_o$ under the $\IZ_3$ symmetry.  
In terms of the variables $(s,w)$ defined in \eqref{defsw}, where walls of (anti)marginal stability
are straight lines (see  \eqref{WallExact}), 
 the inequality \eqref{ImZ13} in $\cF_o$ amounts to 
 \be
 \label{qs20}
w+\frac12 s<0
 \ee
In that region, the coordinates $(u,v)$ on the two-dimensional section of the quiver scattering diagram
can be expressed in terms of the affine coordinates \eqref{defxysgen},
\be
\label{defuvxy}
u = - \frac12 x + y + \frac1{12}, \quad v = - \frac{1}{2\sqrt3} \left(x+\frac12 \right)
\ee
such that the rays in the exact scattering diagram map to segments contained in the geometric
rays \eqref{mckayray}.

However, while the initial rays $\cR_o(\gamma_i)$ in the quiver scattering diagram extend to infinity
inside the space of King stability conditions, they actually originate
from conifold points at $\tau=0,-\frac12,-1$ in the slice  of $\Pi$-stability conditions.
Using the values of the periods in Table \ref{tabvalues}, one finds that  
these conifold points are mapped  
to points at finite distance in the coordinates defined by \eqref{defuvxy}. 
To make this precise, let us parametrize the initial rays  $\cR_o(\gamma_i)$ in the quiver scattering diagram as
\begin{equation}\label{pilambda}
  p_1(\lambda) = \Bigl(\frac{1}{12}-\frac{\lambda}{2} , \frac{-2\lambda-1}{4\sqrt{3}} \Bigr) , \quad
  p_2(\lambda) = \Bigl(-\frac{1}{6} , \frac{\lambda}{\sqrt{3}} \Bigr) , \quad
  p_3(\lambda) = \Bigl(\frac{1}{12}+\frac{\lambda}{2} , \frac{-2\lambda+1}{4\sqrt{3}}\Bigr) 
\end{equation}
 with $\lambda\in\IR$.
In this parametrization, the ray $\cR_o(\gamma_i)$ starts from $\lambda=-\infty$ 
and intersects $\cR_o(\gamma_{i-1})$ and $\cR_o(\gamma_{i+1})$ at $p_i(-1/2)$
and $p_i(1/2)$, respectively, reproducing the intersection points in \eqref{p123}.  

In contrast to the rays $\cR_o(\gamma_i)$ in the orbifold scattering diagram, the rays 
$\cR_\psi(\gamma_i)$  in the slice  of $\Pi$-stability conditions start from the images of the conifold points at $\tau=0,-\frac12,-1$ in the $(u,v)$-plane, given by the points $p_i(\cV_\psi)$
with $\lambda=\cV_\psi$. 
For $\psi=-\frac{\pi}{2}+\epsilon$ with $\epsilon\to 0^+$, these initial points recess to infinity, and one recovers the
infinite two-dimensional scattering diagram constructed above. 
Similarly, for $\psi=\frac{\pi}{2}-\epsilon$ with $\epsilon\to 0^+$, the initial points of the rays $\cR_\psi(-\gamma_i)$
associated to the exceptional collection $E_i[1]$ recess to infinity. 
On the contrary, for $\psi$ small enough such that $|\cV_\psi|<\frac12$, the initial points are such that the initial rays $\Ract_\psi(\gamma_i)$
can no longer interact among themselves (see Figure~\ref{figMcKayTriangle}). 

\begin{figure}
\begin{center}
\includegraphics[height=5cm]{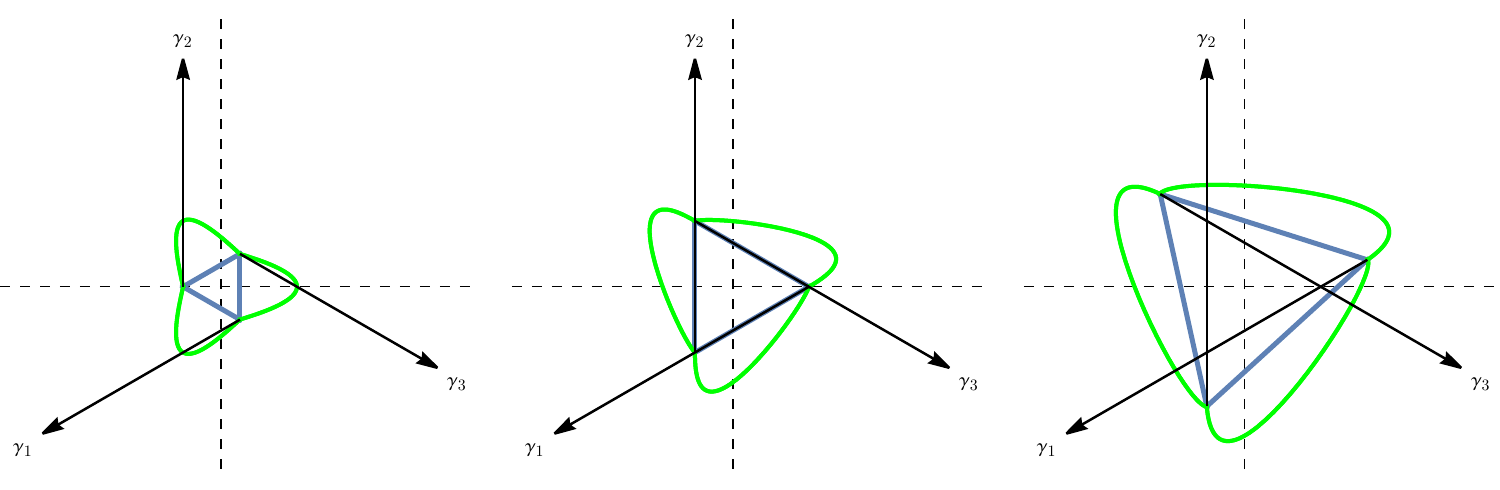}
\end{center}
\caption{Initial rays around the orbifold point in $(u,v)$ plane for $\psi=0$ (left), $\psi=-\psicr{1/2}$ (middle) and $\psi<-\psicr{1/2}$ (right). The dashed lines represent the $u=0$ and $v=0$ axis. 
The green curves correspond to the boundary of the region $\IH^o$ defined around \eqref{ImZ13}.
The blue triangle is the face $\Delta_\psi$ in a tesselation of~$\IH$ defined in \S\ref{sec_tessel}.
As $\psi\to-\frac{\pi}{2}$, the initial points recess to infinity,
and the scattering diagram reduces to Figure \ref{figMcKayScat}.\label{figMcKayTriangle}}
\end{figure}

In between these two regimes, there are critical phases $\psi$ such that some ray in the
scattering diagram $\cD_o$ passes through one of the initial points 
$p_i(\cV_\psi)$ (see Figure \ref{figMcKayScatdense}). 
The first of these critical values arise from the discrete series of rays
$\cR_\psi(F_{2k} \gamma_{i}+F_{2k+2} \gamma_{i+1})$
emitted in the scattering of $\gamma_{i}$ and $\gamma_{i+1}$, where $F_k$ 
are the Fibonacci numbers. As  mentioned above \eqref{defpsicr} in the introduction,
this leads to jumps in the topology of the scattering sequence contributing to the 
index $\Omega_\infty(\gamma)$ along the ray $\cR_\psi(\gamma)$ as a function of the phase $\psi$.

\subsection{Exact scattering diagram\label{sec_exactD}}

Having identified the image of the large volume and orbifold region in affine coordinates,  
we are ready to describe the exact scattering diagram $\cD_\psi^\Pi$ in the same coordinates.
We prove the Split Attractor Flow Conjecture on the slice of $\Pi$-stability conditions, and give an algorithm to list the finite set of trees that contribute to a given (rational) DT 
invariant.
Depending on the phase~$\psi$, we find that $\cD_\psi^\Pi$ has the following structure,
modulo the action of $\Gamma_1(3)$:
\begin{itemize}
\item For small phases $|\psi|<\psicr{1/2}$ (where $\psicr{\alpha}$ was defined in \eqref{defpsicr}), 
$\cD_\psi^\Pi$ coincides with the large volume diagram $\cD^{\rm LV}_0$
in the $(x,y)$ plane, except for a shift in the origin of the initial rays.
\item For large phases $\psicr{1/2}<|\psi|<\pi/2$ that are not critical (as defined at the end of~\S\ref{sec_orbregion}),  $\cD^{\rm \Pi}_\psi$ includes a triangular portion of the orbifold diagram, whose outbound rays seed scattering diagrams in the large volume region.
\item For $\psi=\pi/2$,  $\cD^{\rm \Pi}_\psi$ coincides with the $\theta_1+\theta_2+\theta_3=0$ slice
of the quiver scattering diagram $\cD_Q$  concentrated at the orbifold point
$\tau_o$ in the $\tau$-plane, supplemented with a vertical ray $\Ract(\cO[1])$.
\end{itemize}
We refrain from discussing critical phases, since additional rays emanating from the conifold points complicate the analysis.\footnote{Our proof of the SAFC should go through for critical~$\psi$, except that one must prove that the additional initial rays have positive values of the potential $\tilde\varphi_\tau(\gamma)$ defined in~\eqref{tildevarphixgamma}, a point which we have not investigated.}

\subsubsection{Initial rays and a tessellation\label{sec_tessel}}

A basic constraint on the scattering diagram $\cD_{\psi}^{\Pi}$ is that it is invariant under the group $\Gamma_1(3)$  of auto-equivalences of $D^b\Coh(K_{\IP^2})$, generated by tensor product by $\cO_X(1)$ (acting by $T\colon\tau\mapsto \tau+1$) and by the spherical twist $\ST_\cO$ (acting by $V\colon\tau\mapsto \frac{\tau}{1-3\tau}$).
The quotient $\Pi/\Gamma_1(3)$ has two boundary points from which initial rays could emanate: the large volume point, which is ruled out because central charges grow without bound in this limit, and the conifold point. Without loss of generality, we can thus restrict to initial rays emanating from the point $\tau=0$ in $\Pi$.  The analysis in~\eqref{Znearcon} (with $\tau=p/q=0/1$ hence $m=m_D=0$)
shows a ray with charge $\gamma=[r,d,\ch_2]$ can only emanate from $\tau=0$  if $d\cV_\psi=\ch_2$. 
If $\cV_\psi$ is irrational, this is only possible when $d=0=\ch_2$. In order to determine
the associated attractor invariant, we use the fact that the orbifold and large volume descriptions are valid in regions $\IH^o$ and~$\IH^{\rm LV}$ that cover a neighborhood of $\tau=0$. 
The outcome of this analysis in \S\ref{sec_init} is that the initial rays of the exact scattering diagram are\footnote{To lighten the notation, we omit the $\psi$-dependence of the rays.}
\begin{equation}\label{initial-rays-data}
  \begin{aligned}
    \Reff(\cO) \colon \quad && \Omega([k,0,0]) & = \delta_{k,1} , \quad && \text{and its } \Gamma_1(3) \text{ images,}
    \\
    \Reff(\cO[1]) \colon \quad && \Omega([-k,0,0]) & = \delta_{k,1} , \quad && \text{and its } \Gamma_1(3) \text{ images,}
  \end{aligned}
\end{equation}
provided $\psi$ is \emph{not critical} in a sense of Definition \ref{def:critical}.  
In particular this holds for all small phases $|\psi|<\psicr{1/2}$.
Notably, since $V\in\Gamma_1(3)$ maps $\cO(0)[n]$ to $\cO(0)[n+2]$, the initial rays in~$\cD^\Pi_\psi$ include infinitely many effective rays starting from $\tau=0$, which lie in images of $\cF_C$ that border this point.  By a mild abuse of notation we will denote by $\Reff(\cO(0)[n])$ these distinct rays, even though their charge vector only depends on the parity of $n$.
More generally, there are an infinite set of rays emanating from every
rational number $\tau=\frac{p}{q}$ with $q\neq 0 \mod 3$, corresponding to the images of 
$\cO(0)[n]$ under $\Gamma_1(3)$.  We denote by $\Reff(\cO(p/q)[n])$ the corresponding effective ray.
For non-critical~$\psi$, these images exhaust all initial rays. In particular, there are no rays emanating from the large volume points $\tau=\frac{p}{q}$ with $q=0\mod 3$, nor from points with $\tau\notin\IQ$.

For any $\psi\in(-\pi/2,\pi/2)$ we introduce a $\Gamma_1(3)$-invariant tessellation of~$\Pi$.
Edges of the tiles consist of all $\Gamma_1(3)$ images of the straight line joining the conifold points $(x_{\cO(-1)},y_{\cO(-1)})$ to $(x_{\cO(0)},y_{\cO(0)})$ in the $(x,y)$ plane.
There are two types of tiles: an orbifold region around each $\Gamma_1(3)$ image of~$\tau_o$ and a large volume region around each image of $\tau=\I\infty$.
The orbifold region around~$\tau_o$ is the aforementioned triangle $\Delta_\psi$, whose corners are images of $\tau=0,-1/2,-1$.
The large volume region~$\diamondsuit_\psi$ around $\tau=\I\infty$ maps in $(x,y)$ coordinates to the region above the piecewise linear function (or `jagged parabola', depicted in Figure~\ref{fig:Delta-diamondsuit} on page \pageref{fig:Delta-diamondsuit})
\begin{equation}\label{eqjagged}
y \geq \mathsf{p}(x) \coloneqq \Bigl(m + \cV_\psi + \frac{1}{2}\Bigr) (\cV_\psi - x) + \frac{m(m+1)}{2} \quad \text{with } 
m=\lfloor x-\cV_\psi\rfloor
\end{equation}
zigzaging between\footnote{To see this,
rewrite $\mathsf{p}(x)=\frac{(\cV_\psi)^2-x^2}{2} + \frac{(2 \{x-\cV_\psi\} - 1)^2 - 1}{8}$
where $\{x\}=x-\lfloor x \rfloor$ denotes the fractional part.}
the parabolas $y+ \frac12x^2 - \frac12(\cV_\psi)^2\in\{0,\frac18\}$.
The boundary of~$\diamondsuit_\psi$ joins consecutive conifold points $(x_{\cO(m)},y_{\cO(m)})$, $m\in\IZ$, so that $\diamondsuit_\psi$ borders all translates of~$\Delta_\psi$, while $\Delta_\psi$ borders three large volume regions.
The initial rays~\eqref{initial-rays-data} emanate from vertices of the tessellation by construction, and they enter different regions depending on~$\psi$.

\subsubsection{Small phases $|\psi|<\psicr{1/2}$}

Within $\Delta_\psi$, the exact diagram 
$\cD_\psi^\Pi$ coincides with a triangular region of the orbifold diagram~$\cD_o$.
For our first case of interest, $|\cV_\psi|<1/2$, this region lies inside the empty central triangle of~$\cD_o$, thus $\Delta_\psi$ does not contain any ray, as depicted in Figure~\ref{figMcKayTriangle}.
The large volume regions are thus disconnected from each other in~$\cD^\Pi$, so that we can focus on~$\diamondsuit_\psi$.

The initial rays in that region are $\Reff(\cO(m))$ and $\Reff(\cO(m)[1])$, emanating from integer conifold points $\tau=m\in\IZ$.
In $(x,y)$ coordinates, the geometric ray $\Rgeo(\pm[1,m,m^2/2])$  is contained in the line $y+mx=m^2/2$ tangent to the parabola $y+x^2/2=0$ at $x=m$.
The active rays $\Reff(\cO(m))$ and $\Reff(\cO(m)[1])$ emanate leftwards/rightwards from the point $(x_{\cO(m)},y_{\cO(m)})=(m+\cV_\psi,-m^2/2-m\cV_\psi)$ on this line: for $\psi=0$ this coincides with the point of tangency with the parabola, while for $\psi\lessgtr 0$ it is moved left/right along the tangent line.

For $\psi=0$, the initial rays $\Reff(\cO(m))$ and $\Reff(\cO(m)[1])$ for all $m\in\IZ$ coincide with the initial rays for~$\cD^{\rm LV}_0$.  Under conditions that are fulfilled here, consistent scattering diagrams are uniquely determined by the initial rays, so that $\cD^\Pi_{0}\cap\diamondsuit_\psi$ coincides with~$\cD^{\rm LV}_0$ up to the map $\tau\mapsto(x,y)$.
The rest of the exact diagram is completed by $\Gamma_1(3)$ invariance, namely $\cD^\Pi_{\psi=0}$ consists of all images of $\cD^{\rm LV}$ under $\Gamma_1(3)$, in which translations simply map $\cD^{\rm LV}$ to itself.

In the large volume diagram $\cD^{\rm LV}$, the first intersection along each initial ray is the intersection with the neighboring initial ray, with intersection points
\begin{equation}
  \Rgeo(\cO(m)) \cap \Rgeo(\cO(m-1))
  = \Bigl\{ (x,y) = \Bigl(m-\tfrac{1}{2},-\tfrac12m(m-1)\Bigr) \Bigr\} 
\end{equation}
In other words, along each geometric ray there is a neighborhood of $(x,y)=(m,-m^2/2)$ in which there is no intersection.
When $\psi$ is varied away from zero, the initial point  $(m+\cV_\psi,-m^2/2-m\cV_\psi)$ remains in this intersection-free portion as long as $|\cV_\psi|<1/2$.
We conclude that for small phases $|\psi|<\psicr{1/2}$, the exact diagram $\cD^\Pi_\psi$ consists of disjoint copies (images under $\Gamma_1(3)$) of $\cD^{\rm LV}$ with initial points shifted horizontally by~$\cV_\psi$  along the geometric rays.

To prove the (strong) SAFC for~$\cD^\Pi_\psi$ with $|\psi|<\psicr{1/2}$, we cannot rely on the absence of active rays in~$\Delta_\psi$ and its $\Gamma_1(3)$ images, as only the leaves of attractor flow trees are required to be active.
We must instead prove that edges of attractor flow trees cannot enter~$\Delta_\psi$.
Recall that tree edges are straight line segment in $(x,y)$ coordinates.
We call such an edge `outbound' if it lies in a large volume region $g\cdot\diamondsuit_\psi$ and if 
the line segment, prolonged in the attractor flow direction, intersects the boundary $\partial(g\cdot\diamondsuit_\psi)=g\cdot\mathfrak{p}$.
In other words, an edge is outbound if it moves away from that boundary when following the scattering diagram direction.
For instance, the initial rays $\Reff(\cO)$ and $\Reff(\cO[1])$, supported on the line $y=0$, are outbound: indeed, apart from the initial point $(x_{\cO(0)},y_{\cO(0)})=(\cV_\psi,0)$ of these rays, all other vertices $(x_{\cO(m)},y_{\cO(m)})$ of the polygonal curve~$\mathfrak{p}$ are below $y=0$ since
\begin{equation}
  y_{\cO(m)} = - \tfrac{1}{2} m (m+2\cV_\psi) \leq 0
\end{equation}
thanks to $m$ and $m+2\cV_\psi$ having the same sign for $m\in\IZ$ and $|\cV_\psi|<1/2$.
The leaves of an attractor flow tree are among the initial rays, which are $\Gamma_1(3)$ images of $\Reff(\cO)$ and $\Reff(\cO[1])$, hence are outbound.
Two outbound rays $\cR(\gamma')$, $\cR(\gamma'')$ can only intersect if they lie in the same large volume region, and any bound state $\cR(n'\gamma'+n''\gamma'')$ lies in the angular sector between $\cR(\gamma')$ and $\cR(\gamma'')$ hence is outbound.  We deduce that in an attractor flow tree all rays are outbound and must lie in the same large volume region.
The SAFC for~$\cD^\Pi_\psi$ is then equivalent to the SAFC for  $\cD^{\rm LV}$ since trees contributing to a given $\bOm_\tau(\gamma)$ are simply trees in one of the copies of~$\cD^{\rm LV}$.
Assuming without loss of generality (up to $\Gamma_1(3)$ transformation of $\tau$ and~$\gamma$) that $\tau\in\diamondsuit_\psi$, the electric potential $\varphi_\tau(\gamma)=2(d-rx)$ bounds the number of leaves of the tree as in \S\ref{sec_inidata}, while a causality argument~\eqref{pastLC} bounds the slope 
$m\in\IZ$ of the initial rays
which can contribute.
This establishes Theorem~\ref{thm:SAFC} for $|\psi|<\psicr{1/2}$.

\subsubsection{Large phases $\psicr{1/2}<|\psi|<\pi/2$}

More generally, the exact diagram combines aspects of the orbifold and large volume diagrams.
For definiteness, we consider $-\pi/2<\psi<-\psicr{1/2}$, namely $\cV_\psi<-1/2$.
Then, the $\Gamma_1(3)$ images of $\Reff(\cO)$ enter large volume regions while those of $\Reff(\cO[1])$ enter orbifold regions.  Specifically, $\Reff(\cO)$ enters~$\diamondsuit_\psi$ while its  homological shift $\Reff(\cO[-1])$ enters~$\Delta_\psi$.\footnote{For $\cV_\psi>1/2$, the situation is analogous: $\Reff(\cO[1])$ enters~$\diamondsuit_\psi$ while $\Reff(\cO)$ enters~$\Delta_\psi$.}
In the region $\Delta_\psi$ the exact scattering diagram coincides with the corresponding portion of~$\cD_o$, as argued in \S\ref{sec_orbregion} on general grounds.\footnote{An alternative proof of the matching is that the three $\IZ_3$ images of~$\Reff(\cO[-1])$ carry the same initial data as in the orbifold diagram and consistency allows for a unique diagram with that initial data.}
The initial ray $\Reff(\cO)$, together with rays exiting from $\Delta_\psi$ to~$\diamondsuit_\psi$, and all of their translates, then serve as incoming rays for the scattering diagram in the large volume region $\diamondsuit_\psi$.

By the same argument as for small~$\psi$, all active rays in~$\diamondsuit_\psi$ are outbound (point away from $\mathfrak{p}$ in the $(x,y)$ plane), and more generally, the $\diamondsuit_\psi$ portion of any attractor flow tree consists only of outbound rays.
In particular, none of these rays can ever re-enter~$\Delta_\psi$.
The resulting picture for~$\cD^\Pi$ is that three initial rays scatter within~$\Delta_\psi$, producing some rays that leave~$\Delta_\psi$ towards~$\diamondsuit_\psi$;~these rays  further scatter with~$\Reff(\cO)$ and with translates of all of these rays so as to produce an ever denser set of outbound rays, similarly to the $|\psi|<\psicr{1/2}$ case but with additional incoming rays from the orbifold regions.
Throughout this process, only a limited set of initial rays participate in the construction of the scattering diagram in~$\diamondsuit_\psi$:
\begin{equation}\label{rays-halfint}
  \Reff(\cO(m)[-1]) , \ \Reff(\Omega(m+1)) , \  \Reff(\cO(m-1)[1])
  \ \ \text{and} \ \ \Reff(\cO(m))
\end{equation}
for each $m\in\IZ$, which amount to one initial ray from each conifold point $\tau=m-1/2$ and three from each integer conifold point.

\medskip

Let us now prove the Split Attractor Flow Conjecture.
As we have shown, edges of an attractor flow tree cannot exit the orbifold region~$\Delta_\psi$ (in the attractor flow direction).
Thus, for $\tau\in\Delta_\psi$ the problem reduces to the orbifold diagram, where the dimension vector 
$\gamma=(n_1,n_2,n_3)$ has a finite number of decompositions into elementary charges of the form $k\gamma_i$, and each decomposition leads to a finite number of trees with those leaves.
This proves the conjecture for $\tau\in\Delta_\psi$.

For $\tau\in\diamondsuit_\psi$, infinitely many leaf rays~\eqref{rays-halfint} could in principle contribute.
Since edges cannot exit orbifold regions, attractor flow trees ending at~$\tau$ can be sawed into  a {\it trunk} within~$\diamondsuit_\psi$  rooted at $\tau$, with  leaves that are initial rays $\Ract(\cO(m))$, 
and various {\it shrubs} in 
orbifold regions~$\Delta_\psi(m)$.
This is exemplified in Figure~\ref{fig:composite-tree} on page~\pageref{fig:composite-tree}.
The trunk must lie in $H\cap\diamondsuit_\psi$, where $H$~is the convex hull of $\mathfrak{p}=\partial\diamondsuit_\psi$ and of~$\tau$: this is easily seen by induction starting from the root, using the fact that all rays of the tree within~$\diamondsuit_\psi$ are outbound.
Since $H$ is contained in a vertical strip of finite width in the $(x,y)$ plane, we learn that the set of orbifold regions that can be reached by the split attractor flow starting from~$\tau$ is finite, therefore attractor flow trees have finitely many possible constituents.

This does not suffice for establishing finiteness however: as in the large volume diagram, some charges of possible constituents are opposite to each other, hence there are still infinitely many decompositions of~$\gamma$.
To deal with this, we introduce a variant of the electric potential of the large volume diagram, defined piecewise by
\begin{equation}\label{tildevarphixgamma}
  \tilde\varphi_\tau(\gamma) = \begin{cases}
    2 (d - r \lfloor x - \cV_\psi \rfloor) & \text{in } \diamondsuit_\psi \\
    c_\psi (n_1 + n_2 + n_3) & \text{in } \Delta_\psi 
  \end{cases}
  \qquad
  \text{with } c_\psi \coloneqq \frac{4(-2\cV_\psi-1)}{12\cV_\psi^2+1} 
\end{equation}
extended to all of~$\IH$ by the action of~$\Gamma_1(3)$.\footnote{The definition in~$\diamondsuit_\psi$ is consistent with $\Gamma_1(3)$ covariance: $\tau\to\tau+1$ maps $x\to x+1$ and $d\to d+r$.}
We shall now prove that the potential is monotonically decreasing from the root to the (sum of) leaves of any attractor flow tree, as in the large volume diagram.
It is manifestly additive in the charge~$\gamma$, so that (the sum over all branches of) $\tilde\varphi_\tau(\gamma)$ does not jump at nodes of the tree.
For a given edge in~$\Delta_\psi$, the potential $\tilde\varphi_\tau(\gamma)$ is constant.
Next, we note that $\tilde\varphi_\tau(\gamma)$ is non-increasing along edges in~$\diamondsuit_\psi$: if $r=0$ this is immediate, while if $r\gtrless 0$ 
it follows from the monotonicity of $x$ hence of $\lfloor x-\cV_\psi\rfloor$ from~\eqref{phidec}.

It remains to check that $\tilde\varphi_\tau(\gamma)$ decreases when an edge $\{ry+dx=\ch_2\}$ of the tree crosses from $\diamondsuit_\psi$ to~$\Delta_\psi$ in the attractor flow direction.
The boundary between these regions is the line segment between conifold points $(\cV_\psi-1,\cV_\psi-1/2)$ and $(\cV_\psi,0)$, namely $\{x-\cV_\psi = y/(\frac{1}{2}-\cV_\psi) \in (-1,0)\}$.
The intersection point of the two lines is easy to find, and must satisfy $x-\cV_\psi\in(-1,0)$.
In addition, the tangents $(r,-d)$ and $(1,1/2-\cV_\psi)$ must have the correct relative orientation for the attractor flow to enter~$\Delta_\psi$:
\be\label{xmcVpsi}
x - \cV_\psi = \frac{\ch_2 - \cV_\psi d}{d + r(\frac{1}{2}-\cV_\psi)} \in (-1,0) ,
\qquad
\left|\begin{matrix} r & -d \\ 1 & \frac{1}{2} - \cV_\psi \end{matrix}\right|
= d + r\Bigl(\frac{1}{2}-\cV_\psi\Bigr) > 0 
\ee
Combining these inequalities yields $\cV_\psi d - \ch_2 < d + r(\frac{1}{2}-\cV_\psi)$, or
equivalently  $- \frac{d}{2} - \ch_2 < (d+r)(\frac{1}{2}-\cV_\psi)$.
We deduce
\be
\begin{aligned}
  n_1 + n_2 + n_3
  = - r - \frac{3}{2} d - 3 \ch_2
  & < - r + 3 (d+r) \Bigl(\frac{1}{2}-\cV_\psi\Bigr) \\
  & < \Bigl[\frac{1}{-\cV_\psi-1/2} + 3 \Bigl(\frac{1}{2}-\cV_\psi\Bigr)\Bigr] (d+r) 
\end{aligned}
\ee
In terms of the constant~$c_\psi$ spelled out in~\eqref{tildevarphixgamma}, this inequality reads $c_\psi(n_1+n_2+n_3)<2(d+r)$, namely the potential $\tilde\varphi_\tau(\gamma)$ evaluated on the $\Delta_\psi$ side of the boundary is less than the potential evaluated on the $\diamondsuit_\psi$ side (to evaluate the latter we have used~\eqref{xmcVpsi} to see that $\lfloor x-\cV_\psi\rfloor=-1$).
This establishes that the total electric potential is monotonically decreasing in any tree
in the  attractor flow direction.

The final step is to prove that the potential $\tilde\varphi_\tau(\gamma)$ is positive for all initial rays~\eqref{rays-halfint} at their starting point.
Up to translations we only need to consider $\cR(\cO)$ and the three initial rays of the orbifold diagram.
The ray $\cR(\cO)$ starts at $(\cV_\psi-\epsilon,0)$ with $\epsilon\to 0^+$, thus has a positive potential
\be
\tilde\varphi(\cR(\cO)) = 2 (0 - \lfloor -\epsilon\rfloor) = 2 > 0 
\ee
The three initial rays of the quiver diagram have $(n_1,n_2,n_3)=(1,0,0)$ and permutations thereof, hence $\tilde\varphi=c_\psi>0$.
We conclude that the number of constituents participating in an attractor flow tree rooted at $\tau$ with total charge~$\gamma$ is at most $\tilde\varphi_\tau(\gamma) / \min(2,c_\psi)$.

The SAFC now follows by noting that trees contributing to $\bOm_\tau(\gamma)$ are made of a bounded number of constituents taken among a finite set, hence there are finitely many possible lists of constituents.
For each such list there are finitely many ways that these constituents can be arranged into a topological tree.
For each topology there is at most one attractor flow tree in~$\IH$, constructed by following the attractor flow starting from the root.
There is no ambiguity about where the flow should split: if some edge with charge $\gamma'$ splits (in the topological tree) into $\gamma'_1+\gamma'_2$, then the attractor flow must split when hitting the wall of marginal stability of $\gamma'$ and~$\gamma'_1$, which can only be crossed once according to~\eqref{dWdmu}.
The flow may fail to hit that wall, in which case the topological tree does not correspond to any attractor flow tree.
This establishes Theorem~\ref{thm:SAFC} for $\psicr{1/2}<|\psi|<\pi/2$.

\subsubsection{Behavior in the limits $\psi\to \pm \pi/2$\label{sec_scattB}}

For $\psi=\pi/2$, the scattering diagram $\cD^\Pi_{\psi}$ drastically simplifies.
 In that case, geometric rays $\Rgeo(\gamma)_\psi$ satisfy 
\begin{equation}
  \Im Z_\tau(\gamma)
  = - r \Im T_D + d \Im T 
\end{equation}
thus are included
in contours of constant slope $s=\frac{\Im T_D}{\Im T}=\frac{d}{r}$, independent of
the value of $\ch_2$ (see Figure~\ref{figslopecontour}). The contours only
intersect at points where $s$ is ill-defined,  \ie  when $\Im T=\Im T_D=0$. 
These two curves intersect at  the orbifold point~$\tau_o$ and  $\Gamma_1(3)$ images thereof. 
At $\tau_o$, there are three incoming rays
\begin{equation}
\label{raysorbc}
  \Reff(\cO) , \ \Reff(\Omega(1)[1]) , \  \Reff(\cO(-1)[2])
\end{equation}
The first and last emanate from $\tau=0$ and $\tau=-1$ and lie along the boundary of the fundamental domain~$\cF_o$, where $s=0$ and $s=-1$, respectively while $\Reff(\Omega(1)[1])$ is a vertical line emanating from $\tau=-1/2$, with slope $s=-1/2$. 
In addition, there are vertical rays $\Reff(\cO(m)[1])$
for any $m\in \IZ$, which do not interact with the rays \eqref{raysorbc}. 
Consistency at~$\tau_o$ fixes uniquely the outgoing rays, which then continue on towards a large volume limit, either at $\tau\to -\frac12 + \I\infty$, $\tau\to -\frac23+\I 0$ or $\tau\to -\frac13+\I 0$.
The rays \eqref{raysorbc} are associated to the homological shifts $E_i[1]$ of the objects $E_i$ in the exceptional collection \eqref{excepcoll}. Their scattering diagram around $\tau_o$, which we denote by 
$\cD^\vee_o$, is simply obtained from $\cD_o$ 
by sending $(u,v)\mapsto (-u,-v)$ and reversing the direction of the arrows. Equivalently, we can send 
$(u,v)\mapsto (-u,v)$ and exchange $\gamma_1$ and $\gamma_3$, which amounts to applying the
derived duality \eqref{eqreZ} and a suitable translation. 
Thus, the scattering diagram $\cD^\Pi_{\pi/2}$ consists of the scattering diagram $\cD^\vee_o$
concentrated at the point $\tau_o$, the vertical ray $\Reff(\cO[1])$ and all $\Gamma_1(3)$ images 
thereof. 

Similarly, for $\psi=-\pi/2+\epsilon$ with $\epsilon\to 0^+$, the rays associated to the 
objects $E_i$ in the exceptional collection \eqref{excepcoll}.
\begin{equation}
\label{raysorb}
  \Reff(\cO[-1]) , \ \Reff(\Omega(1)) , \  \Reff(\cO(-1)[1])
\end{equation}
intersect in a region of size $\epsilon$ around  $\tau_o$. As $\epsilon\to 0$, the scattering diagram $\cD^\Pi_{-\pi/2+\epsilon}$  reduces  to the orbifold scattering diagram ${\cD}_o$ concentrated in a region of size $\epsilon$ at the point $\tau_o$, the vertical ray $\Reff(\cO)$  and all their $\Gamma_1(3)$ images. 

In either case,  the SAFC on $\cD^{\Pi}_{\pi/2}$ reduces to the flow tree formula for quivers, which is proven in~\cite{Arguz:2021zpx}. Moreover, the fact that rays do not intersect away from orbifold points
makes it clear that the  index $\Omega_c(\gamma)$ computed using the quiver
associated to the exceptional collection 
$E_i(k)$ will coincide with the Gieseker index 
$\Omega_\infty(\gamma)$ along the ray $\cR_{\pi/2}(\gamma)$ provided 
$k-1\leq \frac{d}{r}\leq k$ for $r\neq 0$, as observed in \cite{Douglas:2000qw,Beaujard:2020sgs}.
This concludes the proof of Theorem~\ref{thm:SAFC} for all non-critical phases $\psi\in(-\pi/2,\pi/2]$.

\subsection{Case studies}

In this final section, we study the scattering sequences which contribute to the Gieseker index 
$\Omega_\infty(\gamma)$ for two simple choices of charges, as function of the phase $\psi$.
The initial rays in these sequences turn out to jump at a subset of the critical 
phases of Definition~\ref{def:critical}, while the index remains unaffected. This is reminiscent of the `fake walls' encountered in earlier studies \cite{Andriyash:2010yf,Chowdhury:2012jq,Alexandrov:2018iao},
where the structure of the bound state changes while the index remains unchanged.

\medskip

%\subsubsection{$\cO_C$}
We first consider $\gamma=[0,1,1)$, corresponding to the Chern vector of 
 the structure sheaf $\cO_C$ of a rational curve in the  hyperplane class. As noted below \eqref{dimMGieseker0}, the moduli space $\cM_\infty(\gamma)$ is  $\IP^2$ itself, 
 so $\Omega_\tau(\gamma)$ should equal 
 $y^2+1+1/y^2$ for  $\tau_2\gg 1$ and arbitrary values of $\tau_1$. 
Thus, for any $\psi\in(-\pi/2,\pi/2)$ the scattering diagram should contain a ray $\Reff_\psi(\gamma)$ reaching
$\tau_2=\I\infty$ with the above index. The scattering sequences which contribute to this index however depend sensitively on the value of $\psi$  (which we can assume to be negative, 
using the reflection symmetry of the scattering diagram):
 \begin{itemize}
\item $-\psicr{\frac12}< \psi < \psicr{\frac12}  \simeq  0.82406$:  
$T_0=\{\cO(-1)[1], \cO(0)\}$ 
contributing $K_3(1,1)$;
\item $-\psicr{\frac32} \simeq -1.27155 < \psi < -\psicr{\frac12}$:
$T_{1}=\{ 2\cO(0)[1],\Omega(2)\}$
contributing $K_3(1,2)$;
\item $-\psicr{\frac52}\simeq -1.38766 < \psi < -\psicr{\frac32}$: 
$T_{2}=\{3 \cO(1)[1], \{2 \Omega(3), \cO(2)[-1]\}\}$
contributing $K_3(1,2) K_3(1,3)$;
\item $-\psicr{\frac72}\simeq -1.43934 < \psi < -\psicr{\frac52}$: 
$T_{3}=\{3 \cO(2)[1], \{ 3\Omega(4), \{2 \cO(3)[-1], \cO(2)[1]\}\}\}$
contributing $K_3(1,2) K_3(1,3)$;
\end{itemize}
In all cases, the tree produces the desired index $y^2+1+1/y^2$. More generally, when 
\be
-\psicr{n+1/2} < \psi < - \psicr{n-1/2}\ ,
\ee
we find a single tree $T_n$ with constituents  in the positive cone spanned by the  exceptional
collection $\langle \cO(n)[-1], \Omega(n+1), \cO(n-1)[1]\rangle$ with 
Chern vectors $\gamma_1(n), \gamma_2(n),  \gamma_3(n)$, obtained 
from  \eqref{excepcoll} by $n$ units
of spectral flow  (and homology shift $[-1]$).
The tree $T_n$ has $n+1$ constituents and is of the iterated form
\be
T_n = \{ k_1 \gamma_3(n), \{  k_2 \gamma_2(n),  \{ k_3 \gamma_1(n),  \{ k_4 \gamma_3(n), \{ k_5 \gamma_2 (n),  \dots \}\}\}\}, 
\ee
where $k_i=3$ for $i=1,\dots n-1$ and $k_n=2, k_{n+1}=1$. Equivalently, $T_n$ is defined inductively by 
\be
T_{n} = \{ 3 \gamma_3(n), \sigma_n^{-1} \cdot T_{n-1}(1) \} \ ,\quad T_1=\{ 2\gamma_3(1), \gamma_2(1) \}
\ee
where $T(1)$ is the tree $T$ shifted by one unit of spectral flow, and $\sigma_n$ acts by cyclic permutation
\be
\sigma_n\colon \gamma_1(n)\mapsto \gamma_2(n)\ ,\quad 
\gamma_2(n)\mapsto \gamma_3(n)\ ,\quad  \gamma_3(n)\mapsto \gamma_1(n)
\ee
The charges indeed
add up to $(n-1) \gamma_1(n) + n \gamma_2(n) + (n+1)\gamma_3(n) = [0,1,1)$, and the index evaluates to $K_3(1,2) K_3(1,3)^{n-1}=y^2+1+1/y^2$. In Figure \ref{fig011psi} we plot the scattering sequences for the 
 first few values of $n$.
 
\begin{figure}[ht]
\begin{center}
\includegraphics[width=17cm]{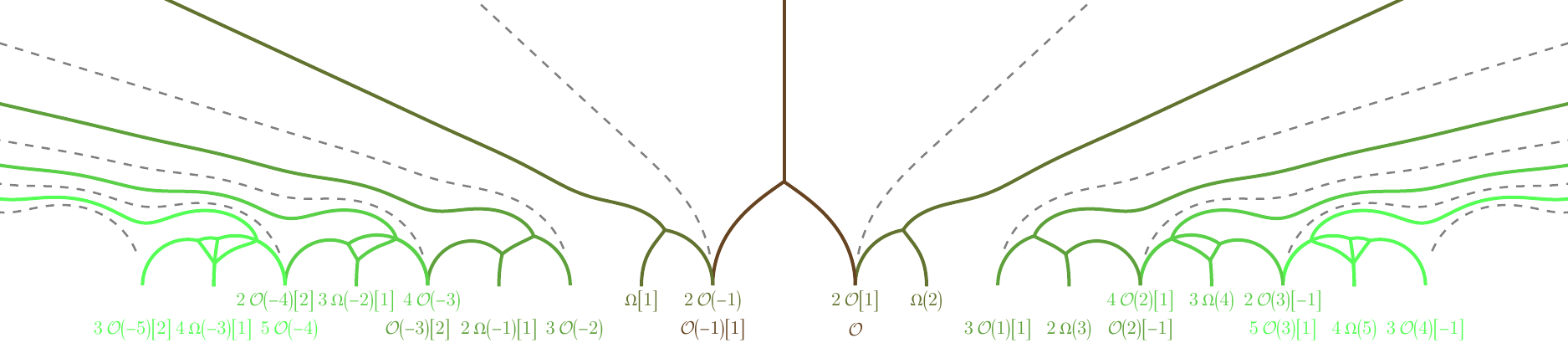}
\end{center}
\caption{Attractor flow trees contributing to $\gamma=[0,1,-\frac12]=[0,1,1)$ for $\psi$ varying from $-\frac{\pi}{2}$  (right) to $\frac{\pi}{2}$ (left). The dashed lines corresponding to the incoming rays for $\psi=\psicr{\alpha}$ with $\alpha\in\{\pm\frac12,\pm\frac32,\pm\frac52,\dots\}$.
\label{fig011psi}}
\end{figure}

\medskip

%\subsubsection{$\cO$}

We now consider $\gamma=[1,0,1)$,  the Chern vector of the structure sheaf
$\cO$. This object is spherical and stable throughout the large volume region.
In particular, for large $\tau_2$ and any $\tau_1<0$, the index
$\Omega_\tau(\gamma)$ should equal $1$. 
\begin{itemize}
\item In the range $-\frac{\pi}{2} < \psi < \psicr{1/2}$,
there is a single ray originating from the conifold point $\cO(0) = \gamma_1(0)[1]$;
\item 
In the range $\psicr{\frac12}<\psi<\psicr{1}$, we find a single sequence
$\{ 3\cO(-1), \Omega(0)[1]\}$
contributing $K_3(1,3)=1$;
\item 
In the range $\psicr{1}<\psi<\psicr{\frac32}$, we find a single sequence 
$\{6 \cO(-2), \{ 3\Omega(-1)[1],\cO(-3)[2]\}\}$ 
contributing $K_3(1,3)K_6(1,6)=1$;
\end{itemize}
More generally, for $\psicr{\frac{n}{2}}<\psi<\psicr{\frac{n+1}{2}}$ 
 there is a sequence $T_n$  obtained inductively through
\be
T_{n} = \{ (3n) \gamma_1(-n)[1], \sigma_{-n} \cdot T_{n-1}(-1) \} \ ,\quad T_0=\{ \gamma_1(0)[1] \}
\ee
with total charge $[1,0,1)$ contributing 1 to the index. 
In Figure \ref{fig101psi} we plot the scattering sequences for the first few values of $n$.

\begin{figure}[ht]
\begin{center}
\includegraphics[width=17cm]{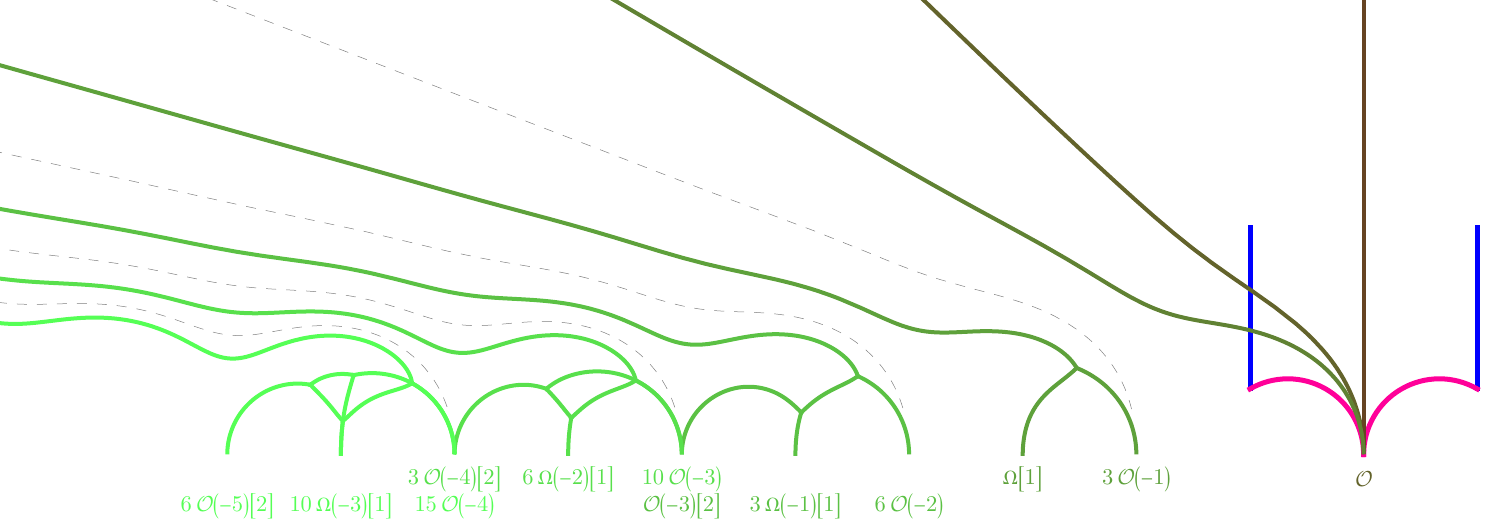}
\end{center}
\caption{Attractor flow trees  contributing to $\gamma=[1,0,0]=[1,0,1)$ for $\psi$ varying from $-\frac{\pi}{2}$ (right) to $\frac{\pi}{2}$ (left).  The dashed lines corresponding to the incoming rays for $\psi=\psicr{\alpha}$ 
with $\alpha\in\{\frac12,1,\frac32,\dots\}$.
\label{fig101psi}}
\end{figure}

\appendix

\section{Periods as Eichler integrals\label{sec_Eichler}}

In this section, we review the modular description of the K\"ahler moduli space of $K_{\IP^2}$, 
derive the Eichler integral representation \eqref{Eichler0} of the periods $(T,T_D)$, and use it to obtain asymptotic expansions around the large volume, conifold and orbifold points and the  behavior under monodromies. We refer to \cite{Chiang:1999tz,Diaconescu:1999dt,Aspinwall:2004jr,Aganagic:2006wq,Haghighat:2008gw,Bayer:2009brq,Alim:2013eja} for earlier studies in the literature.\footnote{We are grateful to Thorsten Schimannek for his help 
about the material in this section.}

\subsection{K\"ahler moduli space as the modular curve \texorpdfstring{$X_1(3)$}{X1(3)}}

Recall that the mirror of $K_{\IP^2}$ is  a family of genus-one curves $\Sigma(z)=x_1^3+x_2^3+x_3^3-z^{-1/3}x_1x_2x_3=0$ in $\IP^2$  parametrized by the complex structure modulus $z\in 
\cM_K=\IC\backslash \{0,-\frac{1}{27}\}$, which is identified as the complexified K\"ahler moduli space of $K_{\IP^2}$. The points $z=0,-\frac{1}{27}, \infty$ then correspond 
to the large radius, conifold and  orbifold points, respectively.
 Alternatively, we can consider the family of genus-one curves $\Sigma'(z')\colon x+y+1-z' x^3/y=0$
 in $\IC^\times_x \times \IC^\times_y$, 
 related to $\Sigma(z)$ with $z'=-z-\frac{1}{27}$ by a 3-isogeny.
 The points $z'=0,-\frac{1}{27}, \infty$ then correspond to 
the conifold, large radius and orbifold points, respectively.

\begin{figure}[ht]
\begin{center}
\includegraphics[width=8cm]{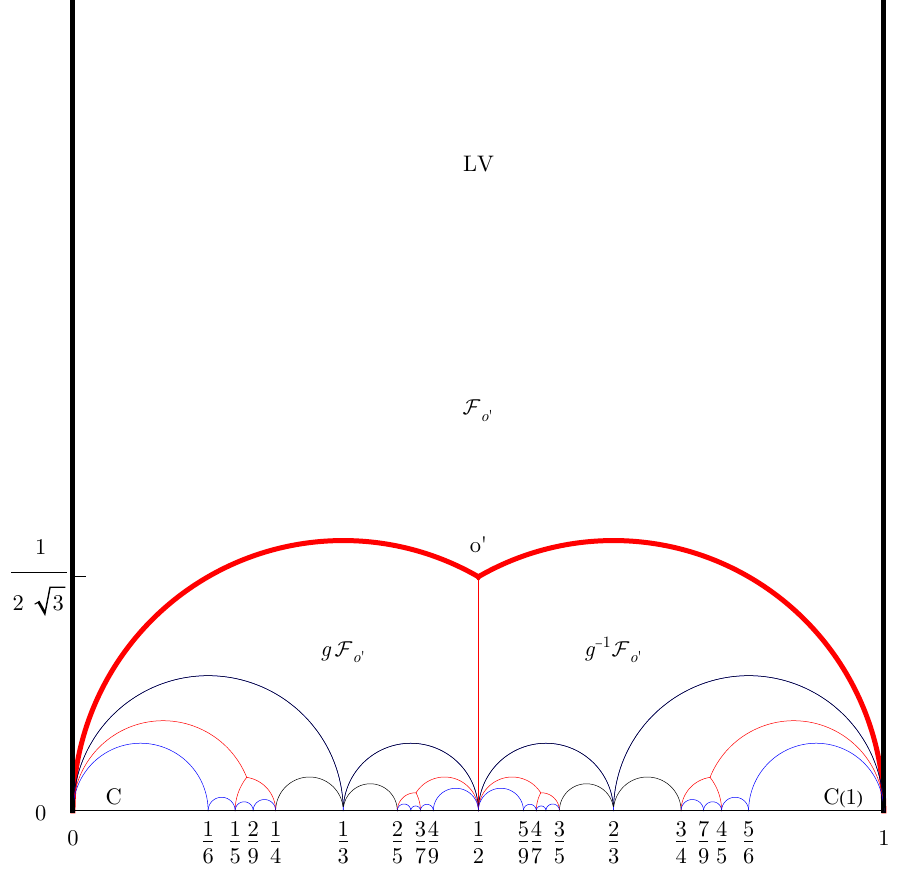}
\end{center}
\caption{Fundamental domain $\cF_{o'}$ 
centered around the orbifold point $\tau_{o'}=e^{\I\pi/6}/\sqrt3$, and some of its images. Here $g=ST^{-2}STS\colon \tau\mapsto \frac{\tau-1}{3\tau-2}$
and $g^{-1}=TST^3S\colon \tau\mapsto \frac{2\tau-1}{3\tau-1}$. The fundamental domain $\cF_{o}$
is the same domain translated to the interval $[-1,0]$.
\label{figfundOp}}
\end{figure}

The space $\cM_K$ is isomorphic to the modular curve $X_1(3)=\IH/\Gamma_1(3)$, where
$\Gamma_1(3)$ is the group of integer matrices $(\begin{smallmatrix} a & b \\ c & d \end{smallmatrix})$ 
such that  $a,d=1 \mod 3$ and $c=0\mod 3$. This is an index 4 subgroup of $PSL(2,\IZ)$, with two cusps of width 1 and 3, corresponding to the large volume and conifold points, respectively, and one elliptic point of order 3 corresponding to the orbifold point.
A convenient choice of fundamental domain is
the domain centered
around the orbifold point at $\tau_o=\frac{1}{\sqrt3} e^{5\pi\I/6}=(-\frac12,\frac{1}{2\sqrt3})$,
shifted horizontally by $-1/2$ compared to \eqref{defFC}, 
\be
\label{defFo}
\cF_o \coloneqq \Bigl\{\tau: \, -1\leq\tau_1<0, (\tau_1+\frac23)^2+\tau_2^2\geq \frac19, 
 (\tau_1+\frac13)^2+\tau_2^2> \frac19\Bigr\}
\ee
This domain (or rather its translate $\cF_{o'}\coloneqq \cF_o(1)$ under $\tau\mapsto \tau+1$) is depicted in Figure \ref{figfundOp},
along with several of its images under $\Gamma_1(3)$. 
Alternatively, one may choose a fundamental  domain centered around $\tau=0$,
\be
\label{defFC}
\cF_C \coloneqq \Bigl\{\tau: \, -\frac12\leq\tau_1<\frac12, (\tau_1+\frac13)^2+\tau_2^2\geq \frac19, 
 (\tau_1-\frac13)^2+\tau_2^2> \frac19\Bigr\}
\ee
shown in Figure \ref{fighomshifted}. It has the nice property of being invariant under the 
Fricke involution $\tau\to -1/(3\tau)$.

The isomorphism $\cM_K \simeq X_1(3)$ is given explicitly by
\be
\label{J3z}
J_3(\tau) =  -15 - \frac{1}{z}  = \frac{12-405 z'}{1+27 z'},
\ee
where $J_3(\tau)$ is the normalized Hauptmodul of $\Gamma_1(3)$, 
\be
\label{J3}
J_3 \coloneqq\left( \frac{\eta(\tau)}{\eta(3\tau)}\right)^{12} + 12 = \frac{1}{q}+54 q-76 q^2 -243 q^3 + 1188 q^4 - 1384 q^5 - 2916 q^6+ \dots
\ee
where $q=e^{2\pi\I\tau}$. 
The latter maps the points $\tau=\I\infty, 0$ and $\tau_o$ to $J_3=\infty,12$ and $-15$, corresponding to the large volume, conifold and orbifold points, respectively. One easily checks that the Fricke involution maps  $J_3\mapsto \frac{729}{J_3-12}+12$ and $z\mapsto z'$. Plugging the $q$-expansion 
\eqref{J3} into \eqref{J3z}, we get
\be\label{ztoq}
\begin{aligned}
q & =-z + 15 z^2 - 279 z^3 + 5729 z^4 - 124554 z^5+2810718 z^6 +  \dots\\
z & = -q + 15 q^2 - 171 q^3 + 1679 q^4 - 15054 q^5 + 126981 q^6 +\dots
\end{aligned}
\ee
The Klein invariant is expressed in terms
of $J_3$ via 
\be
J\coloneqq \frac{E_4^3(\tau)}{\eta^{24}(\tau)}  
= J_3 + 744+  \frac{196830}{J_3-12} + \frac{19131876}{(J_3-12)^2} + \frac{387420489}{(J_3-12)^3}
\ee
leading to 
\be
J = \frac{(216 z-1)^3}{z(1+27 z)^3} = -\frac{(1+24 z')^3}{z'^3(1+27 z')} 
\ee
Using the isomorphism $\cM_K\simeq X_1(3)$, the universal cover of $\cM_K$ is therefore identified with the Poincar\'e upper half plane, tesselated by an infinite number of copies of the fundamental domain $\cF_C$.

\subsection{Periods as Eichler integrals}
Under mirror symmetry, the central charges $(T(\tau),T_D(\tau))$ of the D2-brane and D4-brane are identified with periods $(\varpi,\varpi_D)=(\int_{\ell} \lambda,\int_{\ell_D} \lambda)$ of a suitable meromorphic differential $\lambda$ over a basis of one cycles $(\ell,\ell_D)$ on the mirror curve. Each of these 
periods satisfies the degree~$3$ 
Picard-Fuchs equation 
\be
 \bigl[ \Theta^2+3z (3\Theta+2)(3\Theta+1) \bigr] \Theta \cdot  
 \begin{pmatrix} \varpi \\ \varpi_D \end{pmatrix}=0
\ee
where $\Theta=z\partial_z$. On the other hands, the periods $(\varpi',\varpi'_D)=(\int_{\ell} \omega,\int_{\ell_D} \omega)$ of the unique (up to scale) holomorphic differential $\omega$ on the mirror curve
satisfy the degree~$2$
Picard-Fuchs equation 
\be
\left[ \Theta^2+3z (3\Theta+2)(3\Theta+1) \right] \cdot  \begin{pmatrix} \varpi' \\ 
\varpi'_D \end{pmatrix}=0
\ee
It follows that 
\be\label{Thvarpi}
 \Theta \cdot  \begin{pmatrix} \varpi \\ \varpi_D \end{pmatrix}= \begin{pmatrix} \varpi' \\ 
\varpi'_D \end{pmatrix}
 \ee
Identifying the modular parameter as the ratio $\tau=\frac{\varpi'_D}{\varpi'}$, we find that $\varpi'_D$ transforms as  a modular form of 
weight 1 under $\Gamma_1(3)$, hence it is given (up to normalization) by the 
weight 1 Eisenstein series (denoted by $A$ in \cite{Bousseau:2020ckw}),
\be
\varpi'(\tau) = 1+6\sum_{m\geq 1} \biggl( \sum_{n|m}  \chi(n)  \biggr) q^m   = 1 + 6 q+ 6 q^3+ 6 q^4 +12 q^7+ \dots
\ee
where $\chi(n)=\binom{n}{-3}$ is the Dirichlet character equal to $+1$ for $n=1\mod 3$, $-1$ for $n=2\mod 3$ and 0 otherwise.  Let $C$ be the weight~$3$ Eisenstein series with the same character,
\be
C(\tau) =1 -9 \sum_{m\geq 1} \biggl( \sum_{n|m} n^2 \chi(n)  \biggr) q^m = 1-9 \sum_{n\geq 1}
\frac{n^2 \chi(n) q^n }{1-q^n}
\ee
This modular form can also be written as an eta product
\be
C (\tau) \coloneqq \frac{\eta(\tau)^9}{\eta(3\tau)^3}=  
\sum_{n=1}^{\infty} c_n q^n = 1 - 9 q + 27 q^2 -9 q^3 - 117 q^4+ \dots
\ee
which makes it clear that it does not vanish anywhere in $\IH$. The ratio $\frac{C-\varpi'^3}{27 
\varpi'^3}$ is a meromorphic function on $X_1(3)$, which can be shown coincides with $z$. Using the differential identities in  $\frac{C-\varpi'^3}{27 \varpi'^3}$, we obtain immediately that $z\partial_z=\frac{\varpi'}{C} \partial_\tau$. Substituting into \eqref{Thvarpi}, we get 
\be  
\frac{d}{d\tau}  \begin{pmatrix} T \\ T_D  \end{pmatrix} = 
C  \begin{pmatrix} 1 \\ \tau  \end{pmatrix}
\ee
Using the values $(T,T_D)=(-\frac12,\frac13)$ at the orbifold point $\tau_o$ 
\cite{Diaconescu:1999dt} we 
can write the periods as a holomorphic Eichler integral
\be
\label{Eichler2}
 \begin{pmatrix} T \\ T_D  \end{pmatrix}
= \begin{pmatrix} -\frac12 \\ \frac13  \end{pmatrix} 
+  \int_{\tau_o}^{\tau} \begin{pmatrix} 1 \\u \end{pmatrix} 
\, C(u) \de u
\ee
Equivalently, using the identity
\be
\label{eq224}
 \int_{\tau_o}^{\tau_{o'}} \begin{pmatrix} 1 \\u \end{pmatrix} \, 
\, C(u) \de u =  \begin{pmatrix} 1 \\ 0 \end{pmatrix} 
\ee
proven at the end of \S\ref{sec_mono}, 
one can also take $\tau_{o'}=\tau_o+1$ as a base point and write
\be
 \begin{pmatrix} T \\ T_D  \end{pmatrix}
= \begin{pmatrix} \frac12 \\ \frac13  \end{pmatrix} 
+  \int_{\tau_{o'}}^{\tau} \begin{pmatrix} 1 \\u \end{pmatrix} \, 
\, C(u) \de u
\ee
This representation (or equivalently \eqref{Eichler2}) provides a global formula for the analytic continuation of $T$ and $T_D$ throughout
the upper half-plane, which  gives immediate access to the asymptotic expansions near all singular points and monodromies around them. It also proves to be very efficient for numerical evaluations. 
In Table \ref{tabvalues} we record the values of the periods at some special points. 
Using
$\overline{C(\tau)}=C(-\bar\tau)$ and \eqref{eq224}, it is easy to establish the reality properties
\be
\label{ReTTD}
\overline{T(\tau)} = - T(-\bar\tau)\ ,\quad \overline{T_D(\tau)} = T_D(-\bar\tau)\ ,
\ee
At the end of \S\ref{sec_mono}, we use \eqref{ReTTD} to conclude that 
$\Im T_D=0$ on the semi-circle $\cC(-\frac13,\frac13)$ passing through the orbifold point $\tau_o$.
As a result, the slope $s=\Im T_D/\Im T$ vanishes on that semi-circle.

\begin{table}[ht]
\caption{Periods at some special points.
\label{tabvalues}}
\vspace{-.5\baselineskip}
$
\renewcommand{\arraystretch}{1.2}
\begin{array}{|c|c|c|c|c|c|}
\hline
\tau & z & j & J_3 & T & T_D \\ 
\hline
n \in \IZ &  -\frac{1}{27}  & \infty & 12 & n+\I\,\cV &\frac{n^2}{2} +  n\I\,\cV \\
\hline
\frac12  &  -\frac{1}{27}  & \infty & 12 & \frac12 -2\I\,\cV & \frac12-\I\,\cV \\
\hline
-\frac12  &  -\frac{1}{27}  & \infty & 12 & -\frac12 -2\I\,\cV & \frac12+\I\,\cV \\
\hline
\frac14 &  -\frac{1}{27}  & \infty & 12 &  \frac52+4\I\,\cV & \frac12+\I\,\cV \\
\hline
\frac34 &  -\frac{1}{27}  & \infty & 12 & - \frac32+4\I\,\cV & -\frac32+3\I\,\cV \\
\hline
\frac15 &  -\frac{1}{27}  & \infty & 12 &  2-5\I\,\cV & \frac12-\I\,\cV \\
\hline
\frac25 &  -\frac{1}{27}  & \infty & 12 &  4-5\I\,\cV & 2-2\I\,\cV \\
\hline
\frac35 &  -\frac{1}{27}  & \infty & 12 &  -3-5\I\,\cV & -\frac32-3\I\,\cV \\
\hline
\frac37&  -\frac{1}{27}  & \infty & 12 &  3+7\I\,\cV & \frac32+3\I\,\cV \\
\hline
\tau_o+n= \frac{e^{\I\pi/6}}{\sqrt3} & \infty & 0 & -15 &n- \frac12 & \frac13 +\frac{n(n-1)}{2} \\
\hline
\frac{9+\I\sqrt3}{42} &\infty &0  & -15 &  \frac32 & \frac13  \\
\hline 
\frac{33+\I\sqrt3}{42} &\infty &0  & -15 &  -\frac12 & -\frac23  \\
\hline 
\frac{33\I\sqrt3}{78} &\infty &0  & -15 &  \frac52 & \frac43  \\
\hline 
\frac{45+\I\sqrt3}{78} &\infty &0  & -15 &  -\frac32 & -\frac23  \\
\hline
e^{2\pi\I/3} & \frac{1}{216} & 0 & -231  &  -\frac12 + 0.859778\, \I & -0.118598 - 0.429889\,\I\\ 
\hline
e^{\pi\I/3} & \frac{1}{216} & 0 & -231  &  \frac12 + 0.859778 \,\I & -0.118598 + 0.429889\,\I \\ 
\hline
\I &  \frac{5-3\sqrt{3}}{108} &   1728 & 255+162\sqrt3 & 
1.00267 \I & -0.378093 
\\
\hline
\end{array}
$
\end{table}

\subsection{Expansion around large radius\label{sec_LVexp}}

In the large radius limit, integrating term-by-term the $q$-expansion of $C$, we get 
\be
 \int_{\tau_o}^{\tau}  \begin{pmatrix} 1 \\u \end{pmatrix} \, C(u) \de u = 
 \begin{pmatrix} \tau-\tau_o + \frac{1}{2\pi\I} \left( \bar f_1(\tau) -  \bar f_1(\tau_o) \right) \\
\frac12 \tau^2-\frac12 \tau_o^2 + \frac{1}{2\pi\I} \left( \tau\, \bar f_1(\tau) -  \tau_o \bar f_1(\tau_o)\right)
+ \frac{1}{(2\pi\I)^2} \left( \bar f_2(\tau) -  \bar f_2(\tau_o) \right) 
\end{pmatrix}
\ee
where
\be
 \bar f_1 (\tau) \coloneqq \sum_{n=1}^{\infty} \frac{c_n}{n} q^n\ ,\quad 
 \bar f_2(\tau) \coloneqq - \sum_{n=1}^{\infty} \frac{c_n}{n^2} q^n
\ee
We observe numerically that for   $\tau=\tau_o$,
\be
\tau_o +  \frac{1}{2\pi\I} \bar f_1(\tau_o) = - \frac12\ ,\quad
- \frac12 \tau_o^2  -  \frac{\tau_o}{2\pi\I} \bar f_1(\tau_o)  + \frac{1}{(2\pi\I)^2} \bar f_2(\tau_o) 
= -\frac{1}{8} 
\ee
Thus we find that 
\be
\label{defTTD}
T=\frac{f_1}{2\pi\I}-\frac12\ ,\quad T_D = \frac{f_2}{(2\pi\I)^2} -\frac{f_1}{4\pi\I} +\frac14
\ee
 with 
\be
f_1 \coloneqq 
\log(- q)  + \bar f_1 \ ,\quad  
f_2 
\coloneqq \frac12[\log(-q) ]^2 + \bar f_1  \log(-q)  + \bar f_2\ ,\quad 
\ee
Consequently, we have
\be
\label{TTDbarf}
T = \tau + \frac{\bar f_1}{2\pi\I}\ ,\quad T_D = \frac12\tau^2+\frac18 + \frac{\tau\,  \bar f_1}
{2\pi\I} + \frac{\bar f_2}{(2\pi\I)^2}
\ee
In particular, on any vertical line with $2\tau_1\in\IZ$, we have
\be
\label{shalfintegertau1}
\Re T = \frac{\Im T_D}{\Im T}= \tau_1
\ee
since $\bar f_1, \bar f_2$ are real.
In the limit $\tau\to\I\infty$, $\bar f_1(\tau), \bar f_2(\tau)$ are exponentially suppressed 
hence $(T,T_D)\sim (\tau,\frac12\tau^2)$. Substituting $q$ in terms of $z$ via \eqref{ztoq},
we find agreement with the usual representations in terms of Meijer G-functions \cite{Diaconescu:1999dt}
\begin{align}
f_1(z)&= -\frac{1}{\Gamma(\frac13)\Gamma(\frac23)} G_{3,3}^{2,2} \left( 
\begin{smallmatrix} \frac13 & \frac23 & 1 \\ 0& 0 & 0 \end{smallmatrix}
\middle| 27 z \right)  \\
f_2(z)&=  \frac12 \left[  G_{3,3}^{3,1}  
  \left( \begin{smallmatrix} \frac13 & \frac23 & 1 \\ 0& 0 & 0 \end{smallmatrix} \middle| 27 z \right) + 
G_{3,3}^{3,1} 
\left( \begin{smallmatrix} \frac23 & \frac13 & 1 \\ 0& 0 & 0 \end{smallmatrix} \middle| 27 z \right) \right] 
-  \frac{\pi^2}{3}
\end{align}
Expressing $\tau$ and $T_D$ in terms of the flat coordinate $T$ by inverting the $q$-expansions, we 
recover the usual expansion in terms of Gromov-Witten invariants, 
\be
\label{TDexp}
T_D =\frac{T^2}{2} + \frac18 +
\frac{1}{(2\pi\I)^2} \left(  -9 Q + \frac{135}{4} Q^2 
+244 Q^3 + \frac{36999}{16} Q^4 + \frac{635634}{25} Q^5 + 
 307095 Q^6 + \dots\right)
\ee
where
\be
Q \coloneqq e^{2\pi\I T} =  q + 9 q^2 + 54 q^3 + 246 q^4 + 909 q^5 + 2808 q^6 +\dots 
\ee
The expansion \eqref{TDexp} can be integrated term-by-term to obtain the tree-level 
prepotential
\be
\label{eqF0}
F_0= -\frac{T^3}{18}  - \frac{T}{24} -\frac{1}{36} 
+ \frac{1}{(2\pi\I)^3} \left( 3 Q - \frac{45}{8} Q^2 + \frac{244}{9} Q^3  
-\frac{12333
   Q^4}{64}+\frac{211878 Q^5}{125}-\frac{102365 Q^6}{6}+ \dots  \right)
\ee
such that $T_D=-3 \partial_T F_0$.

\subsection{Expansion around conifold point\label{sec_conexp}}
The expansion near the conifold can be obtained by applying the Fricke involution $\tau\mapsto \tau'=-1/(3\tau)$ which maps $\tau_o$ to $\tau_{o'}=\tau_o+1$ and $C$ to 
\be
C'(\tau') \coloneqq
\frac{\eta(3\tau')^9}{\eta(\tau')^3} = \sum_{n=1}^{\infty} c'_n q'^n = 
q' + 3 q'^2 + 9 q'^3 + 13 q'^4 + 24 q'^5 + 27 q'^6 + \dots
\ee
This is again recognized as an Eisenstein series,
\be
C'(\tau') = \sum_{m\geq 1} \big( \sum_{n|m} n^2 \chi(m/n)  \big) q'^m = \sum_{n\geq 1} \chi(n)
\frac{ q'^n(1+q'^n) }{(1-q'^n)^2}
\ee
Changing variables $u\to u'=-1/(3u)$ in \eqref{Eichler2} and using $\eta(-1/\tau)=\sqrt{-\I\tau}\eta(\tau)$,  we get 
\be
\begin{pmatrix} T \\ T_D  \end{pmatrix}
= \begin{pmatrix} -\frac12 \\ \frac13  \end{pmatrix}  + \I\, 3^{\frac52} \int_{\tau_{o'}}^{\tau'} C'(u')
\begin{pmatrix} 3 u' \\  -1 \end{pmatrix}   \de u '
\ee
Integrating terms by terms  we get 
\be
 \int_{\tau_{o'}}^{\tau} C'(u)
\begin{pmatrix} 3 u' \\  -1 \end{pmatrix}   \de u 
= \begin{pmatrix}
 \frac{3}{2\pi\I} \left( \tau' f^c_1(\tau') - \tau_{o'} f^c_1(\tau_{o'})\right)
+ \frac{3}{(2\pi\I)^2} \left( \bar f^c_2(\tau')-  \bar f^c_2(\tau_{o'})  \right) 
\\  - \frac{1}{2\pi\I} \left(  f^c_1(\tau') -  f^c_1(\tau_{o'})  \right)
\end{pmatrix} 
\ee
where we defined
\be
\begin{aligned}
f^c_1(\tau') &\coloneqq  \sum_{n=1}^{\infty} \frac{c'_n}{n} q'^n = 
q' + \frac32 q'^2 + 3 q'^3 + \frac{13}{4} q'^4 + \dots
\\
f^c_2(\tau') &\coloneqq f_1^c  \log q'  + \bar f_2^c(\tau') 
\\
\bar f_2^c(\tau') &\coloneqq - \sum_{n=1}^{\infty} \frac{c'_n}{n^2} q'^n = 
-q'- \frac34 q'^2 - q'^3 - \frac{13}{16} q'^4 - \dots
\end{aligned}
\ee
We observe numerically that for   $\tau=\tau_o$,
\be
\kappa\, f_1^c(\tau_o) = -1, \quad 
\frac{\kappa}{2\pi\I}  f_2^c(\tau_o) = \tfrac12 -\I\,\cV
\ee
where
\be
\label{defV}
\kappa\coloneqq\frac{27\sqrt3}{2\pi}\,, \quad 
\cV \coloneqq \frac{27}{4\pi^2} \Im \Li_2\left(e^{2\pi\I/3}\right)
\simeq 
0.462758
\ee
Thus we find that $T, T_D$ can be expressed as 
\be\label{TTD-near-coni}
T=  \frac{\kappa}{2\pi\I} f_2^c +\I\,\cV\ ,\quad T_D = -\frac{\kappa}{3} f_1^c 
\ee
In particular, $(T,T_D)=(\I \cV,0)$ at the conifold point. Substituting $q'$ in terms of 
$z'$ using \eqref{ztoq} (with $q\to q', z\to z')$, we recover the usual expression in 
terms of Meijer G-functions, and arrive at an alternative representation for 
the `quantum volume' \cite{Chiang:1999tz},
\be
\cV =  \frac{
G_{3,3}^{2,2} \left( 
\begin{smallmatrix} \frac13 & \frac23 & 1 \\ 0& 0 & 0 \end{smallmatrix} 
\middle| -1 \right)}
{2\pi\, \Gamma(\frac13)\Gamma(\frac23)} +\frac{ \I}{2}
\ee
The value for $\cV$ observed numerically in \eqref{defV}
can be determined exactly by evaluating  \eqref{TTDbarf} at $\tau=0$ 
using zeta function regularization.
 Indeed, the $L$-series of $\bar f_1, \bar f_2$
easily evaluate to 
\be
L_{\bar f_1}(s) \coloneqq \sum_{m\geq 1} \frac{c_m}{m^{1+s}} = -9 \zeta(s+1)\, L(s-1)\ ,\quad 
L_{\bar f_2}(s) \coloneqq -\sum_{m\geq 1} \frac{c_m}{m^{2+s}} = 9 \zeta(s+2)\, L(s)\
\ee
where $L(s)\coloneqq\sum_{m\geq 1} \chi(m) m^{-s}=3^{-s}\left( \zeta(s,\frac13)-\zeta(s,\frac23) \right)
$ is the Dirichlet $L$-series. Its completion 
\be
L^\star(s) = \left( \tfrac{3}{\pi} \right)^{\frac{s+1}{2}} \Gamma\left( \tfrac{s+1}{2}\right) L_s(s)
\ee
is analytic for $\Re(s)>1$ and invariant under $s\mapsto 1-s$, and so is the completed Riemann zeta function $\zeta^\star(s)=\pi^{-s/2}\Gamma(s/2) \zeta(s)$. Using this, one can evaluate the limit as $s\to 0$,
\begin{align}
\label{f12L0}
\bar f_1(0) &= -9 \lim_{s\to 0}  \zeta(s+1)\, L(s-1) = -\frac{27 \sqrt3 L(2)}{8\pi^2} \\
\bar f_2(0) &= 9 \lim_{s\to 0} \zeta(s+2)\, L(s) = \frac{\pi^2}{2}
\end{align}
where we used $L(0)=1/3$. This reproduces the expected value $(T,T_D)=(\I\,\cV,0)$ at $\tau=0$.

\subsection{Orbifold point\label{sec_oexp}}
Near the orbifold point, integrating the Taylor expansion of $C$ term by term we get 
\begin{align}
\label{TTDorb}
T &= -\frac12 +C(\tau_o) (\tau-\tau_o)  + \frac12 C'(\tau_o)(\tau-\tau_o)^2 + \dots \\
T_D&= \frac13 + \frac12 C(\tau_o) (\tau^2-\tau_o^2)+\frac16 C'(\tau_o) (\tau-\tau_o)^2(2\tau+\tau_o)  + \dots
\end{align}
We observe numerically that
\begin{align}
C(\tau_o) &= \frac{\Gamma(\frac13)^3}{\Gamma(\frac23)^6}\simeq 3.1185\ ,\quad 
C'(\tau_o) = -\I \frac{729}{4\pi^{\frac92}} \Gamma(\tfrac13)^3 \Gamma(\tfrac76)^3\simeq 16.2043\, \I \ ,\quad \\
J_3'''(\tau_o) &= \frac{18\sqrt3 \Gamma(\frac13)^9}{\I \Gamma(\frac23)^9}\simeq -14474\, \I 
\end{align}
with the same values for $\tau=\tau_{o'}$.
The flat coordinate $w=1/z$ is obtained by expanding \eqref{J3z},
\be
w= - J_3 - 15 = \frac{3\sqrt3 \Gamma(\frac13)^9}{\I\, \Gamma(\frac23)^9} (\tau-\tau_o)^3 +\cO( (\tau-\tau_o)^6)  
\ee

\subsection{Monodromies\label{sec_mono}}

The monodromies around the three singular points can be computed by using the transformation property of the Eichler integral,
\be
\label{EichlerTrans}
 \begin{pmatrix} T+\frac12 \\ T_D-\frac13  \end{pmatrix}\left( \frac{a\tau+b}{c\tau+d} \right)
 =
 \begin{pmatrix} d & c \\ b & a \end{pmatrix}  \begin{pmatrix} T+\frac12 \\ T_D-\frac13  \end{pmatrix}\left( \tau \right)
+ \int_{\frac{d\tau_o-b}{a-c\tau_o}}^{\tau_o} 
 \begin{pmatrix} c u+d \\a u+b \end{pmatrix} \, 
\, C(u) \de u
 \ee
 whenever $ad-bc=1$, $c=0 \mod 3$. The last term is independent of $\tau$, and is a degree 1 period polynomial for the weight 3 modular form $C$. 
It follows from \eqref{Eichler2} and \eqref{EichlerTrans}  that  the period vector $\Pi=(1,T,T_D)$
transforms as  
\be
\label{monodromyact}
\Pi^t \mapsto M \Pi^t \, \qquad M=\begin{pmatrix} 1 & 0 & 0 \\
m & d & c \\
m_D & b & a
\end{pmatrix}\ ,\quad
\ee
where 
\be
\label{mmDperiod}
\begin{pmatrix}
m \\
m_D
\end{pmatrix}
= \begin{pmatrix} \frac12(d-1)-\frac{c}{3} \\ \frac13(1-a)+\frac{b}{2} 
\end{pmatrix}
+ \int_{\frac{d\tau_o-b}{a-c\tau_o}}^{\tau_o} 
 \begin{pmatrix} c u+d \\a u+b \end{pmatrix} \, 
\, C(u) \de u
\ee
Consequently, the coordinates $(s,w)$ defined in \eqref{defsw} transform
as
\be
\label{sqtrans}
s\mapsto \frac{as+b}{cs+d}, \quad w\mapsto \frac{w}{cs+d} + \frac{as+b}{cs+d} m - m_D
\ee
Under the $\Gamma_1(3)$ transformations 
\be
 \tau\mapsto\tau+1\ ,\quad \tau\mapsto -\frac{\tau}{3\tau-1}\ ,\quad \tau\mapsto -\frac{\tau+1}{3\tau+2}
\ee
corresponding to monodromies around $\I\infty$, $0$ and $\tau_o$, we 
 find, in agreement with \cite{Douglas:2000qw}\footnote{This also agrees with
\cite{Diaconescu:1999dt} upon  conjugating the matrices in  
by $\Bigl(\begin{smallmatrix} 1 &0 &0 \\
\!\! 1/2 & 1 &0  \\ 0 & 0 & 1 \end{smallmatrix}\Bigr)$, due to a shift $T\to T+\frac12$.} 
\be
\label{monDG}
 M_{\rm LV}=\begin{pmatrix}
1 & 0 & 0 \\
1 & 1 & 0 \\
\frac12 & 1 & 1 
\end{pmatrix}
 \, \quad
M_{C}=\begin{pmatrix}
1 & 0 & 0 \\
0 & 1 & -3 \\
0 & 0 & 1 
\end{pmatrix}
 \, \quad
 M_{o} =\begin{pmatrix}
1 & 0 & 0 \\
-\frac12 & -2 & -3 \\
\frac12 & 1 & 1 
\end{pmatrix}
\ee
satisfying $M_o=M_C M_{\rm LV}, M_o^3=\mathds{1}$.

\medskip

In the remainder of this section, we prove the identity \eqref{eq224} and the statement
below \eqref{eqreZ} by studying the action of the monodromy $M_C$. First, we observe that 
$M_C$ maps $\tau_o$ to $\tau_{o'}=\tau_o+1$, while preserving $T_D$. Thus, \eqref{Eichler2}
implies the second equation in \eqref{eq224}, namely
\be
 \int_{\tau_{o}}^{\tau_{o'}} u \, C(u) \de u=0
 \ee
As a result, $T_D(\tau_{o'})=T_D(\tau_{o})=1/3$.
Similarly, $M_C$ maps $T$ to $T-3T_D$, therefore $T(\tau_{o})=T(\tau_{o'})-3T_D(\tau_{o'})$,
which implies the first equation in \eqref{eq224},
\be
 \int_{\tau_{o}}^{\tau_{o'}} C(u) \de u= T(\tau_{o'}) - T(\tau_{o}) = 3T_D(\tau_{o'}) = 1
 \ee
Secondly, since $M_C$ sends $\tau=\tau_1+\I\tau_2$ to 
\be - \frac{\tau}{3 \tau-1} 
=-\frac{3\tau_1^2+3\tau_2^2-\tau_1}{3(3\tau_1^2+3\tau_2^2-2\tau_1)+1}
+\I \frac{\tau_2}{3(3\tau_1^2+3\tau_2^2-2\tau_1)+1} 
\ee
we see that on the half circle $\cC(-1/3,1/3)$ defined by $3\tau_1^2+3\tau_2^2-2\tau_1=0$,
the action of $M_C$ restricts to $\tau\mapsto -\bar\tau$. 
Since $T_D$ is invariant under $M_C$, it follows that 
for any $\tau\in \cC(-1/3,1/3)$
\be
T_D(\tau)=T_D(-\bar{\tau})
\ee
Since $\overline{T_D(\tau)} = T_D(-\bar\tau)$ by  \eqref{eqreZ}, it follows that
$\Im T_D=0$ on the half circle $\cC(-1/3,1/3)$. $\square$

\section{Massless objects at conifold points \label{sec_massless}}

The structure sheaf $\mathcal{O}$ of $\IP^2$ is a spherical object in the derived category  $D^b(\Coh_c K_{\IP^2})$, whose central charge vanishes at the conifold point
$\tau=0$. The action of the group 
$\Gamma_1(3)$ on the $\tau$ upper half-plane lifts to an action by auto-equivalences on the derived category of $K_{\IP^2}$. Thus, for every 
$g \in \Gamma_1(3)$, $E=g(\mathcal{O})$ is a spherical object whose central charge vanishes at the conifold point $\tau=g(0)=p/q$ with $q\neq 0\mod 3$. In this section we compute the object $E$
for low values of $p,q$. The results are summarized in Table \ref{Conifoldtab} on page \pageref{Conifoldtab}.

For example, the element 
$V=(\begin{smallmatrix} 1 & 0 \\ -3 & 1 \end{smallmatrix})$
(equal to the monodromy around the conifold point $\tau=0$) acts on 
$\tau$ via $V \colon \tau \mapsto -\frac{\tau}{3\tau-1}$, and 
 on the derived category via the spherical twist $\ST_\mathcal{O}$ around the spherical object 
 $\mathcal{O}$. The latter is given by the exact triangle (see \eqref{STS} and \cite[\S 9.1]{Bayer:2009brq})
\be
\mathrm{Hom}^\bullet_{K_{\IP^2}}(\mathcal{O},E)\otimes \mathcal{O} \xrightarrow{\rm ev} E\to \ST_\mathcal{O}( E )  \xrightarrow{+1}
\ee
Similarly for $U=V^{-1}$ the action on the derived category is the inverse of the spherical twist $\ST_\mathcal{O}$ around the spherical object $\mathcal{O}$, which is given, following
 Proposition 2.10 of \cite{seidel2001braid}, by the exact triangle
\be
\ST^{-1}_\mathcal{O}(E) \to E\xrightarrow{\rm coev} \Hom(\Hom^\bullet_{K_{\IP^2}}(E,\mathcal{O}),\mathcal{O})  \xrightarrow{+1}
\ee

At the point $\tau=-1/2$, obtained by acting with $VT$ on $\tau=0$, the spherical object becoming massless is $\ST_{\mathcal{O}}(\mathcal{O}(1))$. 
We have $\Hom^0_{\IP^2}(\mathcal{O},\mathcal{O}(1))=\IC^3$, 
$\Hom^k_{\IP^2}(\mathcal{O},\mathcal{O}(1))=0$ for $k>0$, and 
$\Hom^k_{\IP^2}(\mathcal{O}(1), \mathcal{O})=0$ for $k \geq 0$. 
Hence, using \eqref{lemma46}, 
\be \Hom^0_{K_{\IP^2}}(\mathcal{O},\mathcal{O}(1))=\Hom^0_{\IP^2}(\mathcal{O}, \mathcal{O}(1))=\IC^3\ ,\quad 
\Hom^{k\neq 0}_{K_{\IP^2}}(\mathcal{O},\mathcal{O}(1))=0
\ee
Finally, from the Euler exact sequence 
\be
0 \rightarrow \mathcal{O} \rightarrow \mathcal{O}(1)^{\oplus 3}
\rightarrow T \rightarrow 0
\ee
where $T$ is the tangent bundle of $\IP^2$,
we obtain the exact sequence 
\begin{equation}\label{eq_seq}
0 \rightarrow \Omega(1) \rightarrow \mathcal{O}^{\oplus 3} \rightarrow \mathcal{O}(1) \rightarrow 0
\end{equation}
and so an exact triangle 
$\mathcal{O}^{\oplus 3} \rightarrow \mathcal{O}(1) \rightarrow \Omega(1)[1] \xrightarrow{+1}$. 
Hence the massless object at $\tau=-\frac12$ is 
\be\label{STO-O1}
\ST_{\mathcal{O}}(\mathcal{O}(1))=\Omega(1)[1]
\ee

For the point $\tau=4/5$, obtained by acting  by $TV$ on $\tau=-1/2$, the spherical object becoming massless is $\ST_{\mathcal{O}}(\Omega(1)[1])(1)$. 
Using \eqref{eq_seq} and the Bott vanishing theorem, we obtain
$\Hom^0_{\IP^2}(\Omega(1), \mathcal{O})=\IC^3$. 
On the other hand, using the Bott vanishing theorem and Riemann--Roch formula, we find that 
$\Hom^k_{\IP^2}(\Omega(1),\mathcal{O})=0$ for all 
$k \neq 0$, and 
$\Hom^k_{\IP^2}(\mathcal{O}, \Omega(1))=0$ for all $k$. 
Hence, by \eqref{lemma46}, 
\begin{gather}
\Hom^{\bullet}_{K_{\IP^2}}(\mathcal{O},\Omega(1)) 
=\IC^3[-3]
\\
\Hom^{3}_{K_{\IP^2}}(\mathcal{O},\Omega(1))=\IC^3\ , \quad
 \Hom^{k\neq 3 }_{K_{\IP^2}}(\mathcal{O},\Omega(1)) 
\end{gather}
It follows that the exact triangle defining 
$E=\ST_\mathcal{O}(\Omega(1)[1])(1)$ is of the form 
\be
\label{massless45}
 \mathcal{O}^{\oplus 3}(1)[-2] \rightarrow \Omega(2)[1] 
\rightarrow E  \xrightarrow{+1}
\ee
Note in particular that $E$ has rank $-5$ and degree $-4$, as expected. 

It is  important to remark that  the map $\mathcal{O}^{\oplus 3}(1)[-2] \rightarrow \Omega(2)[1]$ in the derived
category of $K_{\IP^2}$ does not come from maps in the derived category of $\IP^2$: indeed all extension groups $\Hom^3_{\IP^2}$ between sheaves are zero since 
$\IP^2$ is of dimension $2$. As mentioned at the beginning of \S\ref{sec_dkp2}, 
in general, an object $E$ in the derived category of coherent sheaves 
on $K_{\IP^2}$ supported set-theoretically on the zero section  can be viewed as a pair $(F,\phi)$ with $F$ an object in the  derived category of $\IP^2$ and  $\phi\colon F \rightarrow F \otimes K_{\IP^2}$ 
a nilpotent Higgs field. For the object $E$ defined by \eqref{massless45},  
since every map $\mathcal{O}^{\oplus 3}(1)[-2] \rightarrow \Omega(2)[1]$ in the derived category 
of $\IP^2$ is zero, the underlying object $F$ 
is simply the direct sum 
\be
F=\Omega(2)[1] \oplus \mathcal{O}^{\oplus 3}(1)[-1]
 \ee
but there is a non-trivial nilpotent Higgs field on $F$ coming from a map 
$ \mathcal{O}^{\oplus 3}(1)[-1] \rightarrow \Omega(-1)[1]$
in $D^b(\Coh \IP^2)$.
Note that by Serre duality, we have 
\be
 \Hom^2_{\IP^2}(\mathcal{O}(1), \Omega(-1))
 =\Hom_{\IP^2}(\Omega(2),\mathcal{O}(1))^\vee
 \ee
which is indeed non-zero.

The same type of analysis provides the exact triangles giving the following
massless objects at other conifold points.
\begin{itemize}
    \item $\tau=1/5$, $g=U^2T^{-1}$. 
    Because of the relations $(VT)^3=1$ and $U=V^{-1}$,
    \be\label{UTinvO}
    UT^{-1}(\mathcal{O})=TVTV(\mathcal{O})=TVT(\mathcal{O}[-2])=\Omega(2)[-1]
    \ee
    In the quiver associated to the exceptional collection $(\mathcal{O},\Omega(2)[-1],\mathcal{O}(1)[-2])$ there are 3 relations and no arrows from the second to the first node, so we deduce successively
    \begin{align}
      \Hom^\bullet_{K_{\IP^2}}(\Omega(2)[-1],\mathcal{O}) & = \IC^3[-2] \\
      \Hom( \Hom^\bullet_{K_{\IP^2}}(\Omega(2)[-1],\mathcal{O}),\mathcal{O}) & = \mathcal{O}^{\oplus 3}[2]
    \end{align}
    The object $E=U(\Omega(2)[-1])$ is then given by the exact triangle
    \begin{align}
        E\to\Omega(2)[-1]\to\mathcal{O}^{\oplus 3}[2]  \xrightarrow{+1}
    \end{align}

    \item $\tau=1/4$, $g=UT$. One has $\Hom(\Hom^\bullet_{K_{\IP^2}}(\mathcal{O}(1),\mathcal{O}),\mathcal{O}) = \Hom(\IC^3[-3],\mathcal{O}) = \cO^{\oplus 3}[3]$,
      thus the object $E=U(\mathcal{O}(1))$ is given by the exact triangle
      \be
      E\to \mathcal{O}(1)\to\mathcal{O}^{\oplus 3}[3] \xrightarrow{+1}
      \ee

    \item $\tau=2/5$, $g=UT^{-2}$. One has $\Hom(\Hom^\bullet_{K_{\IP^2}}(\mathcal{O}(-2),\mathcal{O}),\mathcal{O}) = \Hom(\IC^6,\mathcal{O}) = \mathcal{O}^{\oplus 6}$,
      thus $E=U(\mathcal{O}(-2))$ is given by the exact triangle
      \be
      E\to\mathcal{O}(-2)\to\mathcal{O}^{\oplus 6} \xrightarrow{+1}
      \ee

%    \item $\tau=3/7$, $g=UT^{-1}VT$. We know $T^{-1}VT(\cO)=\Omega[1]$ from our study of $\tau=-1/2$. One has $\Hom^\bullet_{K_{\IP^2}}(\Omega,\mathcal{O}(-1))=\IC^3$ and $\Hom^\bullet_{K_{\IP^2}}(\mathcal{O}(-1),\mathcal{O})=\IC^3$, hence using the spectral sequence given by the composition yields
%      \be
%      \Hom^\bullet_{K_{\IP^2}}(\Omega,\mathcal{O})=\IC^9
%      \ee
%      Therefore $\Hom(\Hom^\bullet_{K_{\IP^2}}(\Omega[1],\mathcal{O}),\mathcal{O})=\cO^{\oplus 9}[1]$ and
%      $E=U(\Omega[1])$ is given by the exact triangle $E\to\Omega[1]\to\mathcal{O}^{\oplus 9}[1] \xrightarrow{+1}$, which we then rotate to the shorter form
%      \be
%      \Omega\to\mathcal{O}^{\oplus 9} \to E \xrightarrow{+1}
%      \ee

    \item $\tau=1/2$, $g=TVT$. This is a translate of~\eqref{STO-O1}, namely $E=\Omega(2)[1]$.

%    \item $\tau=4/7$, $g=TVTUT^{-1}$. One has $TUT^{-1}(\mathcal{O})=\Omega(3)[-1]$ thanks to~\eqref{UTinvO}. Similarly, one has $\Hom^\bullet_{K_{\IP^2}}(\mathcal{O},\mathcal{O}(1))=\IC^3$ and $\Hom^\bullet_{K_{\IP^2}}(\mathcal{O}(1),\Omega(3))=\IC^3$,  and the spectral sequence of the composition yields $\Hom^\bullet_{K_{\IP^2}}(\mathcal{O},\Omega(3))=\IC^9$. Hence $E=TV(\Omega(3)[-1])$ is given by the exact triangle $\mathcal{O}(1)^{\oplus 9}[-1]\to \Omega(4)[-1]\to E \xrightarrow{+1}$, rotated to
%    \be
%    E \to \mathcal{O}(1)^{\oplus 9} \to \Omega(4) \xrightarrow{+1}
%    \ee

    \item $\tau=3/5$, $g=TVT^2$. One has $\Hom^\bullet_{K_{\IP^2}}(\mathcal{O},\mathcal{O}(2))=\IC^6$, hence $E=TV(\mathcal{O}(2))$ is given by the exact triangle
    \be
    \mathcal{O}(1)^{\oplus 6}\to\mathcal{O}(3)\to E \xrightarrow{+1}
    \ee

    \item $\tau=3/4$, $g=TVT^{-1}$. One has $\Hom^\bullet_{K_{\IP^2}}(\mathcal{O},\mathcal{O}(-1))=\IC^3[-3]$, hence $E=TV(\mathcal{O}(-1))$ is given by the exact triangle
    \begin{align}
        \mathcal{O}(1)^{\oplus 3}[-3]\to\mathcal{O}\to E \xrightarrow{+1}
    \end{align}

%    \item $\tau=4/5$, $g=TV^2T$. As shown above, $E=TV(\Omega(1)[1])$ is given by the exact triangle:
%    \begin{align}
%        \mathcal{O}(1)^{\oplus 3}[-2]\to\Omega(2)[1]\to E \xrightarrow{+1}
%    \end{align}

    \item $\tau=1$, $g=T$.  Trivially, $E=\mathcal{O}(1)$.
    
\end{itemize}

These results are summarized in Table~\ref{Conifoldtab} on page~\pageref{Conifoldtab}.

\section{Endpoints of attractor flows for local \protect\texorpdfstring{$\IP^2$}{P2}\label{app:bounds}}

In this section we derive  several bounds on the behaviour of the attractor flow on the slice of $\Pi$-stability conditions in the case $\CY=K_{\IP^2}$, which we used in \S\ref{sec_piflow}.
As explained there, the central charge $Z_\tau(\gamma)$ is a holomorphic function of $\tau\in\IH$ with no critical point, so that the attractor flow can either end at a marginal stability wall, end at a conifold point, end at a large volume point, or continue indefinitely.
In \S\ref{sec_growth} and \S\ref{sec_must_end} we rule out the last two cases in turn.
In  \S\ref{sec_init}, \S\ref{sec_crit} and \S\ref{sec_start}  we determine conditions for an attractor flow to end or start at a conifold point, which leads to the equivalent definitions of the  critical phase~$\psi$  in Definition~\ref{def:critical}.

\subsection{Attractor flows avoid large volume points\label{sec_growth}}
% {Growth of central charges of semistable objects at large volume}

Let us assume that an attractor flow associated to some $\gamma\in\Gamma$ ends at a large volume point, which we take without loss of generality to be $\tau=\I\infty$.
Such a flow $\mu\mapsto \tau(\mu)\in\IH$ would at late times lie in the fundamental domain $\cF_C$ centered around the conifold point $\tau=0$, or its translates centered around $\tau=n$.
By applying a transformation $g(\mu)\in\Gamma_1(3)$ (which depends discretely on the flow parameter~$\mu$), we can map $\tau(\mu)$ into $\tilde\tau = g\cdot\tau\in\cF_C$, at the cost of also mapping $\gamma$ to $\tilde\gamma=g\cdot\gamma$.
The central charge is unchanged, and in particular the mass $|Z_{\tilde\tau(\mu)}(\tilde\gamma(\mu))| = |Z_{\tau(\mu)}(\gamma)|$ is monotonically decreasing, hence lower than its initial value.
We now exclude such attractor flows by proving that this upper bound on $|Z_{\tilde\tau}(\tilde\gamma)|$ would translate into an upper bound on $\Im\tau$, regardless of the charge vector $\tilde\gamma$ (provided the corresponding DT invariant is nonzero).
Passing to the contrapositive statement and dropping tildes,\footnote{Importantly, images of $\gamma\in\Gamma\setminus\IZ\delta$ under $\Gamma_1(3)$ are in $\Gamma\setminus\IZ\delta$ because $\delta$ is invariant.} we shall prove the following statement (recall that 
$\delta$ denotes the D0-brane charge vector).

\begin{proposition}\label{prop_bound_2}
For every $M>0$, there exists $D>0$ such that for every $\tau \in \mathcal{F}_C$ with $\Im(\tau)>D$, and every $\gamma \in \Gamma \setminus \IZ \delta$ with $\Omega_\tau(\gamma) \neq 0$, one has $|Z_\tau(\gamma)|>M$.
\end{proposition}

The proof relies on approximating the exact central charge $Z$ by the large volume central charge~\eqref{defZLV} $Z_{(s,t)}^{\rm LV}(\gamma) = - \frac{r}2  (s+\I t)^2+ d (s+\I t) -\ch_2$, and we begin by proving bounds on it.  We denote
\be
Y = \frac12(t^2-s^2) 
\ee

\begin{lemma}\label{lem_bound_1}
For every $(s,t)$ with $Y \geq 0$, and every $\gamma \in \Gamma$ such that $r\geq 0$ and $d^2-2r \ch_2\geq 0$, one has
\begin{equation}\label{eq_bound_1}
  |Z_{(s,t)}^{\rm LV}(\gamma)|^2 \geq r^2 Y^2+\frac{d^2}{2}Y 
\end{equation}
\end{lemma}

\begin{proof}
We have
\be
\begin{aligned}
|Z_{(s,t)}^{\rm LV}(\gamma)|^2
& = (r Y + d s -\ch_2)^2 + (2Y+s^2) (d - rs)^2 \\
& \geq r^2 Y^2+2rY(ds-\ch_2) +2Y(d-rs)^2 
\end{aligned}
\ee
If $r=0$, this gives $|Z_{(s,t)}^{\rm LV}(\gamma)|^2 \geq 2d^2Y$ and so in particular \eqref{eq_bound_1}. If $r>0$, we rephrase the inequality in terms of the slope $\mu=d/r$ and use the assumption $-\ch_2/r \geq -\mu^2/2$ to obtain
\be
|Z_{(s,t)}^{\rm LV}(\gamma)|^2 \geq r^2 \bigl( Y^2 + Y (2\mu s - \mu^2) + 2Y(\mu-s)^2\bigr) 
\ee
The coefficient of~$Y$ is $\mu^2-2s\mu+2s^2 \geq \mu^2/2$, which yields \eqref{eq_bound_1}.
\end{proof}

\begin{proof}[Proof of Proposition~\ref{prop_bound_2}]
We write $\tau=\tau_1+i\tau_2$ with $\tau_1=\Re \tau$ and $\tau_2=\Im \tau$. 
We denote by $\cO(1)$ any function on $\mathcal{F}_C$ which is bounded for $\tau_2$ large enough, uniformly in $r$ and $d$.
For instance, $1/\tau_2=\cO(1)$ and $\tau_1=\cO(1)$, as $-\frac{1}{2} \leq \tau_1\leq \frac{1}{2}$ on $\cF_C$.
According to the large volume expansion of the periods in \S\ref{sec_LVexp}, 
$T=\tau+\cO(1)$ and $T_D=\frac{\tau^2}{2}+\cO(1)$, hence
\be
\begin{aligned}
\Re T & = \cO(1) , & \quad
\Re T_D & = - \tau_2^2/2 + \cO(1)  \\
\Im T & = \tau_2 + \cO(1) , & \quad
\Im T_D & = \tau_1 \tau_2 + \cO(1) = \cO(1) \tau_2 
\end{aligned}
\ee
As a result
\be
\begin{aligned}
s & = \frac{\Im T_D}{\Im T} = \tau_1 + \cO(1)/\tau_2 = \cO(1)  \\
w & = - \Re T_D + s \Re T = \tau_2^2 / 2 + \cO(1)  \\
t & = \sqrt{2w - s^2} = \sqrt{\tau_2^2 + \cO(1)} = \tau_2 + \cO(1)/\tau_2 \\
Y & = w - s^2 = \tau_2^2 / 2 + \cO(1) 
\end{aligned}
\ee
In particular, large enough $\tau_2$, $t$ or~$Y$ are synonymous within $\cF_C$.
Another consequence is
\be
\tau = s + \I t + \cO(1)/\tau_2 , \qquad \tau^2 = (s+\I t)^2 + \cO(1) 
\ee
which implies that the large volume central charge is a good approximation of the central charge in the sense that
\be
\begin{aligned}
Z_\tau(\gamma)
& = -rT_D + dT - \ch_2 \\
& = -\frac{r}{2} (s+\I t)^2 + d (s+\I t) - \ch_2 + |r| \cO(1) + |d| \cO(1) \\
& = Z_{(s,t)}^{\rm LV}(\gamma) + |r| \cO(1) + |d| \cO(1) 
\end{aligned}
\ee

Next, we seek to apply Lemma~\ref{lem_bound_1}.
For large enough~$\tau_2$, the point $\tau \in \mathcal{F}_C$ does not belong to
the lower boundary of $\mathcal{F}_C$, thus $\tau$ defines a geometric stability condition.
Hence, by \cite{li2019birational} (Corollary 1.33 in published version, or Corollary 1.30 in arXiv version), if 
$\Omega_\tau(\gamma) \neq 0$, then 
$\gamma=n\gamma'$ is a multiple $n\in\IZ\setminus\{0\}$ of the class $\gamma'$ of a Gieseker semistable sheaf.
The latter obeys $r'\geq 0$ and $d'^2-2r' \ch'_2 \geq 0$ (see for example \cite[Lemma 3.4]{drezet1985fibres}).
Up to replacing $\gamma$ by $-\gamma$, which does not change the mass $|Z_\tau(\gamma)|$, one can assume that $n>0$, so that $r\geq 0$ and $d^2-2r \ch_2 \geq 0$.
Lemma~\ref{lem_bound_1} thus applies (for large enough $\tau_2$ to ensure $Y\geq 0$): for every
$\gamma \in \Gamma \setminus \IZ \delta$ such that 
$\Omega_\tau(\gamma)\neq 0$, we have
\be
|Z_{(s,t)}^{\rm LV}(\gamma)| \geq \sqrt{r^2 Y^2 + \frac12 d^2 Y} \geq \frac12  |r| Y +\frac12  |d| \sqrt{Y}
\ee
where we did not try to optimize the constants.
As a result,
\be
|Z_\tau(\gamma)|
\geq |Z_{(s,t)}^{\rm LV}(\gamma)| - |r \cO(1)| - |d \cO(1)| \geq\frac13  |r| Y/ +\frac13  |d| \sqrt{Y}
\ee
for large enough~$Y$.
Since we restrict to $\gamma\notin\IZ\delta$, one has $(r,d)\neq (0,0)$, so that we have proven $|Z_\tau(\gamma)|\geq\sqrt{Y}/3$ for large enough~$Y$, or equivalently for large enough~$\tau_2$.
This ends the proof of Proposition~\ref{prop_bound_2}, which confines any attractor flow away from all large volume points.
\end{proof}

\subsection{Attractor flows end at walls or conifold points\label{sec_must_end}}

We are now ready to prove that an attractor flow cannot continue indefinitely.  We denote by $\overline{\IH}=\IH\cup\IR$ the closed upper half plane.

\begin{proposition}\label{prop:init-conifold}
  For a charge vector $\gamma \in \Gamma \setminus \IZ \delta$ and a starting point $\tau(\mu_0)\in\IH$, consider the attractor flow $[\mu_0,\mu_\infty)\ni\mu\mapsto\tau(\mu)$ that is maximally extended subject to the condition $\Omega_{\tau(\mu)}(\gamma)\neq 0$.
  Then the limit $\tau(\mu_\infty) = \lim_{\mu\to\mu_\infty}\tau(\mu)\in\overline{\IH}$ exists and lies either at a conifold point or on a wall of marginal stability of~$\gamma$.
\end{proposition}

\begin{proof}
The modular curve $X_1(3)=\Gamma_1(3)\backslash \IH$ has a natural compactification
$\overline{X_1(3)} \simeq \IP^1$ obtained by adding the large volume point~$z_{\rm LV}$ and the conifold point~$z_C$.
Let $\pi \colon \IH \rightarrow X_1(3)$ be the quotient map.
Then $\pi\circ\tau\colon[\mu_0,\mu_\infty)\to X_1(3)$ takes values in the compact space $\overline{X_1(3)}$ hence admits at least one limit point $z_\infty\in\overline{X_1(3)}$. In other words there exists a sequence $(\mu_n)_{n=1,2,\dots}$ that tends to~$\mu_\infty$ and such that $\pi(\tau(\mu_n))\to z_\infty$.
For later purposes, it is useful to recall that the mass $|Z_{\tau(\mu)}(\gamma)|$ is monotonically decreasing hence has a limit as $\mu\to\mu_\infty$, which necessarily coincides with the limit of its subsequence $|Z_{\tau(\mu_n)}(\gamma)|$.

Consider the unique element $g_n\in\Gamma_1(3)$ that maps $\tau(\mu_n)$ to a point $\tau'_n=g_n\cdot\tau(\mu_n)$ in the fundamental domain~$\cF_C$ centered on the conifold point, and consider the corresponding charge $\gamma'_n=g_n\cdot\gamma$.
By $\Gamma_1(3)$-equivariance,
\be
\Omega_{\tau'_n}(\gamma'_n) = \Omega_{\tau(\mu_n)}(\gamma)\neq 0 , \qquad
|Z_{\tau'_n}(\gamma'_n)| = |Z_{\tau(\mu)}(\gamma)| \leq |Z_{\tau(\mu_0)}(\gamma)| 
\ee
Proposition~\ref{prop_bound_2} applied with $M=|Z_{\tau(\mu_0)}(\gamma)|$ implies
that the imaginary parts $\Im\tau'_n\leq D$ must be bounded above by some constant~$D$.
This bound excludes $\pi(\tau'_n)=\pi(\tau(\mu_n))$ from a neighborhood of the large volume point in $\overline{X_1(3)}$.
Therefore, the large volume point cannot be a limit point~$z_\infty$ of $\pi\circ\tau$.

Next, assume that the limit point~$z_\infty$ lies in $X_1(3)$.
Consider its lift $\tau_\infty\in\cF_C$, and note that\footnote{Strictly speaking, if $\tau_\infty$ lies on the boundary of~$\cF_C$, it can have multiple $\Gamma_1(3)$ images in the closure~$\overline{\cF_C}$.  Then $\tau'_n$ has a subsequence that converges to either of these images, which we then denote $\tau_\infty$.} $\tau'_n\to\tau_\infty$.
By the support property, $\tau_\infty$ admits an open neighbourhood~$U$ on which there are finitely many classes $\gamma' \in \Gamma \setminus \IZ \delta$ with non-zero DT invariants and with central charge less than the upper bound~$|Z_{\tau(\mu_0)}(\gamma)|$.
The point $\tau'_n$~lies in~$U$ for large enough~$n$, hence $\gamma'_n$ takes finitely many values. Up to passing to a subsequence we can assume that all $\gamma'_n=\gamma'$ are equal to the same charge vector.
Fix an arbitrary (Euclidean) norm $\|\ \|\colon\Gamma\to[0,+\infty)$.
The support property ensures that $\|\gamma'\|\lesssim |Z_\tau(\gamma')|$ whenever $\tau\in U$ and $\Omega_\tau(\gamma')\neq 0$, with an implied constant that is uniform in $\tau\in U$.  This gives a positive lower bound on $|Z_{\tau'_n}(\gamma')|$, hence on its limit
\be\label{minfty-Z}
m_\infty
\coloneqq \lim_{\mu\to\mu_\infty} |Z_{\tau(\mu)}(\gamma)|
= \lim_{n\to+\infty} |Z_{\tau(\mu_n)}(\gamma)|
= \lim_{n\to+\infty} |Z_{\tau'_n}(\gamma')|
= |Z_{\tau_\infty}(\gamma')|
> 0 
\ee
We learn that the point $\tau_\infty\in\IH$ is not a critical point of $Z_\tau(\gamma')$ and the gradient flow is smooth near~$\tau_\infty$.
Therefore, near~$\tau_\infty$ there exists local coordinates $m=|Z_\tau(\gamma')|$ along attractor flow lines and $\ell$ parametrizing the different flow lines: the attractor flow keeps $\ell$ constant and decreases~$m$.
Consider a neighborhood that is rectangular in these coordinates,
\be
V = (\ell_-,\ell_+)\times(m_-,m_+) \ni (\ell_\infty,m_\infty) = \tau_\infty 
\ee
For large enough~$n$ we have $\tau'_n=(\ell_n,m_n)\in V$.  The gradient flow of $|Z_\tau(\gamma')|$ starting from this point is $(\ell_n,m)$ with $m$ decreasing from $m_n$ all the way to $m_-$ at the boundary of~$V$.
Since $m_-<m_\infty$, the attractor flow must stop before, and specifically~\eqref{minfty-Z} requires the attractor flow to stop precisely at $m=m_\infty$.
Recall now that $\tau'_n=g_n\cdot\tau(\mu_n)$.
The image $\{g_n^{-1}\cdot(\ell_n,m), m_\infty<m\leq m_n\}$ of the gradient flow of $|Z_\tau(\gamma')|$ is the gradient flow of $|Z_\tau(\gamma)|$ starting from $\tau(\mu_n)$, which is precisely the attractor flow.
We have thus fully determined the end segment of the attractor flow: the end point $\tau(\mu_\infty)=g_n^{-1}\cdot(\ell_n,m_\infty)$ of the attractor flow exists.
The gradient flow of $|Z_\tau(\gamma)|$ could continue unimpeded beyond $m=m_\infty$, hence what stops the attractor flow must be that $\Omega_\tau(\gamma)=0$ for $\tau=g_n^{-1}\cdot(\ell_n,m)$ with $m<m_\infty$ (at least, close to~$m_\infty$).
This means that $\tau(\mu_\infty)$ is on a wall of marginal stability.

It remains to treat the case where none of the limit points of $\pi\circ\tau$ are of the above type, in which case the only remaining possibility for the limit point is the conifold point $z_\infty=z_C\in\overline{X_1(3)}$.
Since this is a unique limit point, we have $\lim_{\mu\to\mu_\infty}\pi(\tau(\mu))=z_C$.
For a constant $0<D<\Im\tau_o$, consider the (connected) set $W_D=\{\tau\in\cF_C,\Im\tau<D\}$, its projection $V_D=\pi(W_D)\subset X_1(3)$, and the union $U_D=\pi^{-1}(V_D)\subset\IH$ of all of its $\Gamma_1(3)$ images.
The sets $V_D$ and $U_D$ are manifestly open.
In fact, $V_D\cup\{z_C\}$ is a neighborhood of~$z_C$ in $\overline{X_1(3)}$, hence for large enough~$\mu$, one has $\pi(\tau(\mu))\in V_D$ thus $\tau(\mu)\in U_D$.
We learn that $\tau(\mu)$ remains in a fixed connected component of~$U_D$ for large enough~$\mu$.
For instance, the connected component of $U_D$ containing~$W_D$ consists of the union of images $g\cdot W_D$ for all elements $g\in\Gamma_1(3)$ that leave $\tau=0$ invariant, so if $\tau(\mu)$ lies in this connected component, it can only tend to the conifold point $\tau=0$.
Other connected components are $\Gamma_1(3)$ images of this one, which implies that $\tau(\mu)$ tends to a conifold point $\tau_C\in\IQ\subset\partial\overline{\IH}$.
\end{proof}

\subsection{Initial data of the exact diagram from the large volume diagram}\label{sec_init}

A key ingredient when building scattering diagrams is the initial data: DT invariants along the initial rays, which are rays that do not arise from the scattering of other rays.  Initial rays correspond to attractor flows that do not end on a wall of marginal stability, hence that end at a conifold point by Proposition~\ref{prop:init-conifold}.
In this subsection we identify the DT invariants of such a flow (near a conifold point) to some DT invariants in the large volume diagram.  This establishes the equivalence of the characterizations (\ref{def-critical-start}) and (\ref{def-critical-LV}) of critical phases in Definition~\ref{def:critical}.
Using $\Gamma_1(3)$ invariance we take the conifold point to be $\tau=0$.
We assume $\psi\in(-\pi/2,\pi/2)$ in this section as the scattering diagram for $\psi=\pi/2$ is highly degenerate and best treated separately in~\S\ref{sec_scattB}.
We recall $\cV_\psi=\cV\tan\psi$.

\begin{proposition}\label{prop:attr-LV}
  Consider an attractor flow that ends at the conifold point $\tau=0$ and has $Z_\tau(\gamma)\in\I e^{\I\psi}[0,+\infty)$ with $\psi\in(-\pi/2,\pi/2)$.
  Then the point $(s,t) = (\cV_\psi,|\cV_\psi|)$ lies on (the closure of) an active ray~$\Ract^{\rm LV}_0(\gamma)$ of charge~$\gamma$ in the large volume scattering diagram~$\cD^{\rm LV}_0$.
  Furthermore, each DT invariant $\Omega_\tau(k\gamma)$, $k\geq 1$, is eventually constant along the flow close to the conifold point, and coincides with the limit of $\Omega^{\rm LV}_{(s,t)}(k\gamma)$ as $(s,t)\to(\cV_\psi,|\cV_\psi|)$ along~$\Ract^{\rm LV}_0(\gamma)$ in~$\cD^{\rm LV}_0$.
  In particular, either $\gamma=[r,0,0]$ with DT invariants $\Omega_\tau([k,0,0])=\delta_{k,\sign r}$, or $\psi$~is a critical phase in the sense that $(\cV_\psi,|\cV_\psi|)$ is an intersection of active rays in~$\cD^{\rm LV}_0$.
\end{proposition}

\begin{proof}
While the statement is expressed in a uniform way, we shall distinguish $d=0$ from $d\neq 0$ momentarily as they require very different approaches.

There is a $\IZ$-worth of fundamental domains meeting at this conifold point, acted upon by the monodromy $V\colon\tau\mapsto\frac{\tau}{1-3\tau}$ around $\tau=0$.
We use the coordinate $\tau'=-1/(3\tau)$, in which the conifold point lies at $\tau'\to+\I\infty$. and the monodromy acts as $V\colon\tau'\mapsto \tau'+1$.
In terms of $q'=e^{2\pi\I\tau'}$ (which vanishes at the conifold point), the expansion~\eqref{TTD-near-coni} reads
\be
T = \I\cV + \kappa \tau' q' (1+o(1)) , \qquad T_D = -\frac{\kappa}{3} q' (1 + o(1)) 
\ee
where 
$o(1)$ denotes any function of~$\tau'$ (such as $1/\tau'$ or $q'$) that vanishes as $q'\to 0$.
Within an attractor flow, the phase of $Z_\tau(\gamma)$ is fixed and its modulus is monotonically decreasing, hence $Z_\tau(\gamma)=-rT_D+dT-\ch_2$ has a limit as $q'\to 0$.

If $d=0$ then $Z_\tau(\gamma)\to -\ch_2$, which does not belong to the half-line $\I e^{\I\psi}[0,+\infty)$ unless $\ch_2=0$, which corresponds to a charge $\gamma=[r,0,0]$, with $r\neq 0$ since $\gamma\neq 0$.
Such a charge is a multiple of the class of sheaves $\cO(0)[k]$ becoming massless at the conifold point.
It is easy to see that $Z_\tau(\gamma) = -rT_D = \frac{\kappa}{3} rq'(1+o(1))\in\I e^{\I\psi}(0,+\infty)$ requires
\be\label{appC-taup1-d0}
\tau'_1 = \Re\tau' = -n + \frac14 \sign r + \frac{\psi}{2\pi} + o(1)
\ee
for some integer $n\in\IZ$.  Close enough to the conifold point, the $o(1)$ term is less than $1/2$, so along a given attractor flow $n$~is eventually constant.
Applying a $\Gamma_1(3)$ translation $V^n\colon\tau'\mapsto\tau'+n$ reduces the problem to the case $n=0$, for which $|\tau'_1|<1/2$ close enough to the conifold point.
As the Fricke involution $\tau'=-1/(3\tau)$ maps the fundamental domain~$\cF_C$ to itself, we deduce that $\tau\in\cF_C$ at late enough times along the flow.
In addition,
\be\label{signRetau-signr}
\sign(\Re\tau) = \sign\Bigl(\frac{-\tau'_1}{3|\tau'|^2}\Bigr)
= - \sign(\tau'_1) = - \sign r 
\ee
Let us determine the DT invariants along this attractor flow that eventually lies in~$\cF_C$ and ends at $\tau=0$.
It follows from \cite[Corollary 1.24]{li2019birational} (Corollary 1.21 of the arXiv version) that $\Omega_\tau(\gamma)$ does not jump for $\gamma$ proportional to $\gamma(\cO)$ and $\tau$ geometric (apart from $\gamma \mapsto -\gamma$ across the vertical axis).
So it is enough to work at large volume, namely with Gieseker stability.
Let $E$ be a Gieseker-stable sheaf of class $\gamma = k \gamma(\mathcal{O})$ with $k\geq 1$. It is of slope $\mu=0$ and of discriminant $\Delta=0$. Hence, $E$ is exceptional in the sense of \cite[Section 4.2]{drezet1985fibres}. By \cite[Lemma 4.3]{drezet1985fibres}, there is a unique exceptional sheaf of given slope. As $\mathcal{O}$ is exceptional of slope $0$, we obtain $E=\mathcal{O}$. The moduli space of stable objects is a point for $k=1$ and empty for $k>1$.
So
\be\label{appC-standardDT}
\Omega_\tau(k \gamma(\mathcal{O})) = \delta_{k,1}
\ee
for $k\in\IZ$ and for every $\tau\in\cF_C$ with $\Re \tau <0$, and likewise $\Omega_\tau(k \gamma(\mathcal{O})) = \delta_{k,-1}$ for $\Re\tau>0$.
Thanks to~\eqref{signRetau-signr} this translates to the sign condition in the statement of the Proposition.

We henceforth assume that $d\neq 0$.

After proving that the phase must obey $\cV_\psi=\ch_2/d$, our strategy is to show that the attractor flow (near $\tau=0$) lies in the large volume region $\tau\in\IH^{\rm LV}$, map this point to large volume coordinates $(s,t)$, and finally use that the large volume scattering diagram for non-zero phase $\psi^{\rm LV}=\arg(-\I Z^{\rm LV}_{(s,t)}(\gamma))$  coincides with the diagram $\cD^{\rm LV}_0$ with zero phase up to a further change of coordinates~\eqref{eqstpsi},
\be\label{appC-tildest}
\Omega_\tau(\gamma) = \Omega^{\rm LV}_{(s,t)}(\gamma) = \Omega^{\rm LV}_{(\tilde{s},\tilde{t})}(\gamma) , \qquad
\tilde{s} = s + t \tan \psi^{\rm LV} , \qquad
\tilde{t}^2 = t^2 + \bigl(t \tan \psi^{\rm LV}\bigr)^2 
\ee
Our calculations show that $(\tilde{s},\tilde{t})\to(\cV_\psi,|\cV_\psi|)$ as $\tau\to 0$ along the flow, which suffices to conclude.

We start by determining how the conifold point is approached.
The central charge has a limit, hence $T$ has a limit, and equivalently $\tau' q'\to c$ for some $c\in\IC$.
Decomposing $\tau'=\tau'_1+\I\tau'_2$, we see that $\tau'_2 q' = \cO(\tau'_2 e^{-2\pi\tau'_2})$ vanishes at the conifold point so $\tau'_1q'\to c$.
If $c\neq 0$ then $|\tau'_1|\sim |c|/|q'| = |c|e^{2\pi\tau'_2}$, which means that the phase of $\tau'_1q'$ diverges, contradicting $\tau'_1q'\to c$.
Thus $c=0$ and altogether $T\to\I\cV$.
We learn that
\be\label{appC-Ztaugamma}
Z_\tau(\gamma) = \I\cV d - \ch_2 + \kappa \tau' q' (d + o(1))
\ee
with $\tau'q'\to 0$.
Along the flow, the central charge is fixed to lie in $Z_\tau(\gamma)\in\I e^{\I\psi}[0,+\infty)$ and to move towards~$0$ along this half-line, hence $Z_\tau(\gamma)-(\I\cV d - \ch_2)$ lies in the same half-line.
This implies (using $-\pi/2<\psi<\pi/2$ hence $\cos\psi>0$)
\be\label{appC-dcospsi}
d > 0 , \qquad
\frac{\ch_2}{d} = \cV_\psi , \qquad
2\pi\tau'_1 = - 2\pi n + \psi + o(1) 
\ee
where $\cV_\psi = \cV\tan\psi$ and $n\in\IZ$ is eventually constant (as in \eqref{appC-taup1-d0}, $\tau\in\cF_C$ if and only if $n=0$).
We then evaluate the large-volume coordinates $s,t$,
\be\label{appC-swt2}
\begin{aligned}
s & = \frac{\Im T_D}{\Im T} = -\frac{\kappa}{3\cV} e^{-2\pi\tau'_2} \bigl( \sin\psi + o(1)\bigr) 
\\
w & = - \Re T_D + s \Re T = \frac{\kappa}{3} e^{-2\pi\tau'_2} \bigl( \cos\psi + o(1) \bigr) 
\\
t^2 & = 2w - s^2 = \frac{2\kappa}{3} e^{-2\pi\tau'_2} \bigl( \cos\psi + o(1)\bigr) > 0 
\end{aligned}
\ee
Therefore, close enough to the conifold point, the ray lies in the large volume region $2w>s^2$, and one has $\Omega_\tau(\gamma) = \Omega^{\rm LV}_{(s,t)}(\gamma)$.

Next we consider the phase $\psi^{\rm LV}=\arg(-\I Z^{\rm LV}_{(s,t)}(\gamma))$ of the large volume central charge at $(s,t)$, and evaluate its tangent since this is what appears in the change of coordinates~\eqref{appC-tildest}:
\be
t \tan \psi^{\rm LV}
= - t \frac{\Re Z^{\rm LV}_{(s,t)}(\gamma)}{\Im Z^{\rm LV}_{(s,t)}(\gamma)}
= \frac{\ch_2 - d s + \frac{r}{2} (s^2-t^2)}{d-rs}
= \cV_\psi + \cO(e^{-\pi\tau'_2}) 
\ee
Thus
\be
\begin{aligned}
\tilde{s} & = s + t \tan \psi^{\rm LV} = \cV_\psi + \cO(e^{-\pi\tau'_2}) , \\
\tilde{t} & = \sqrt{t^2 + \bigl(t \tan \psi^{\rm LV}\bigr)^2} = |\cV_\psi| + \cO(e^{-\pi\tau'_2}) 
\end{aligned}
\ee
As announced, we learn that $(\tilde{s},\tilde{t})$ tends to the point $(\cV_\psi,|\cV_\psi|)$.
By construction, $(\tilde{s},\tilde{t})$ moves along the geometric ray $\cR^{\rm geo,LV}_0(\gamma)$ of the zero-phase large volume scattering diagram~$\cD^{\rm LV}_0$.
DT~invariants $\Omega_\tau(\gamma)$ near the end of the attractor flow are thus given by DT invariants along the ray $\cR^{\rm geo,LV}_0(\gamma)$ near its intersection $(\cV_\psi,|\cV_\psi|)$ with the initial ray $\Ract_0^{\rm LV}([\sign\cV_\psi,0,0])$.
The ray $\Rgeo_\psi(\gamma)$ is thus active close to $\tau=0$ precisely when the ray $\cR^{\rm geo,LV}_0(\gamma)$ is active close to $(\cV_\psi,|\cV_\psi|)$.
We conclude that there are non-trivial initial rays at $\tau=0$ if and only if $(\cV_\psi,|\cV_\psi|)$ is an intersection of active rays.
\end{proof}

\subsection{Initial data of the exact diagram from the orbifold diagram}\label{sec_crit}

In the previous section we have mapped DT invariants along an attractor flow ending at $\tau=0$ to DT invariants in the large volume scattering diagram.
We now map the DT invariants to the orbifold diagram, thus proving the equivalence of the criteria (\ref{def-critical-start}) and (\ref{def-critical-orbi}) in Definition~\ref{def:critical}, in terms of attractor flows and of the orbifold diagram.
The latter criterion involves the point $\theta=(0,\frac{1}{2}+|\cV_\psi|,\frac{1}{2}-|\cV_\psi|)$ whose $(u,v)$ coordinates are determined from~\eqref{thetauv} to be
\be\label{uv-of-psi}
u = \frac{1}{12} + \frac{1}{2} |\cV_\psi| , \qquad
v = \frac{1}{4\sqrt{3}} (- 1 + 2|\cV_\psi|)
\ee
Recall the functions $p_j\colon\IR\to\IR^2$ given in~\eqref{pilambda} for $j=1,2,3$ parametrizing the three initial rays of the orbifold diagram.
The point~\eqref{uv-of-psi} involved in Definition~\ref{def:critical} is $p_1(-|\cV_\psi|)$.  By $\IZ_3$-invariance (cyclic permutations of the~$\theta_j$) one can replace this point by $p_3(-|\cV_\psi|)$, which will appear more naturally in this section.
Rather than repeating what can already be learned about $\gamma=[k,0,0)$ from Proposition~\ref{prop:attr-LV}, we restrict our attention immediately to attractor flows with $\gamma\notin[1,0,0)$.

\begin{proposition}\label{prop:critical}
  Consider an attractor flow that ends at the conifold point $\tau=0$ and has $Z_\tau(\gamma)\in\I e^{\I\psi}[0,+\infty)$ with $-\pi/2<\psi<\pi/2$ and with $\gamma\notin[1,0,0)\IZ$.  Then the point $p_3(-|\cV_\psi|)$ is a ray intersection in~$\cD_o$.
\end{proposition}

\begin{proof}
For $|\cV_\psi|<1/2$ there are no such ray intersections in $\cD_o$, and we have shown as part of Proposition~\ref{prop:attr-LV} that there is no such attractor flow.
We thus concentrate on $\cV_\psi\leq -1/2$, fixing the sign to be negative by using the $\psi\mapsto-\psi$ symmetry.
In other words, $-\pi/2<\psi\leq -\psicr{1/2}$.

A translation $V^n\colon\tau'\mapsto\tau'+n$ maps the attractor flow~\eqref{appC-dcospsi} to that with $n=0$,
\be
\tau'_1 = \psi/(2\pi) + o(1) \in (-1/2,0) 
\ee
As discussed below~\eqref{appC-taup1-d0}, this condition on~$\tau'$ implies that $\tau\in\cF_C$.  Furthermore, the sign $\sign(\Re\tau)=-\sign(\tau'_1)=1$ implies that $\tau\in\cF_{o'}=\cF_o(1)$.
Within the orbifold fundamental domain~$\cF_o$, the region of validity~$\IH^o$ of the quiver description is described by the inequality \eqref{qs20} $2w+s<0$, namely $t^2<-s(1+s)$ for the point $\tau-1\in\cF_o$.
Since $s(\tau-1)=s(\tau)-1$ and $t^2$ is invariant under translations, this condition reduces to $t^2<(1-s)s$, namely $2w<s$, for the point~$\tau\in \cF_{o'}$.
Then, thanks to the asymptotics~\eqref{appC-swt2}, we evaluate
\be
s - 2w = \frac{\kappa\cos\psi}{3\cV^2} e^{-2\pi\tau'_2} \bigl( - \cV_\psi - 2 \cV^2 + o(1)\bigr) 
\ee
Since $-\cV_\psi \geq 1/2 > 2\cV^2 \simeq 0.4283$, this is positive, which ensures that the attractor flow lies in $\IH^{o'}=\IH^{o}(1)$ and its DT invariants are correctly given by the quiver scattering diagram.

The translated attractor flow $\mu\mapsto\tau(\mu)-1\in\cF_o$ tends to the $\tau=-1$ conifold point, hence
\be
(x,y) \longrightarrow (x_{\cO(-1)}, y_{\cO(-1)}) = (\cV_\psi - 1 , \cV_\psi - 1/2)
\ee
where coordinates of the conifold point were calculated in~\eqref{xyOm}.  The corresponding $(u,v)$ coordinates are
\be\label{appC-uv-long}
(u,v) = \biggl( \frac{1}{12} - \frac{1}{2} x + y , - \frac{2x + 1}{4\sqrt{3}}\biggr)
\longrightarrow
\biggl( \frac{1}{12} + \frac{1}{2} \cV_\psi , \frac{1 - 2\cV_\psi}{4\sqrt{3}}\biggr) = p_3(\cV_\psi) = p_3(-|\cV_\psi|)
\ee
where $p_3$ was defined in~\eqref{pilambda} and we used $\psi<0$.
Along the attractor flow, the quiver description is valid in a neighborhood of this point, so DT invariants of the exact ray $\cR_\psi(\gamma)$ starting at $\tau=0$ coincide with those of the ray $\cR_o(\gamma)$ starting at $p_3(-|\cV_\psi|)$ in the orbifold diagram.
In particular, there exists an active ray (with $\gamma\notin[1,0,0]\IZ$) ending at $\tau=0$ if and only if $p_3(-|\cV_\psi|)$ is a ray intersection.
\end{proof}

\subsection{Attractor flows starting at conifold points}
\label{sec_start}

To complete our description of the neighborhood of the conifold point, we now also study the attractor flows that emanate from $\tau=0$, namely $\mu\mapsto\tau(\mu)$ with $\mu\in(\mu_0,\mu_\infty)$ such that $\lim_{\mu\to\mu_0}\tau(\mu)=0$.
We establish the equivalence of the characterizations (\ref{def-critical-end}) and (\ref{def-critical-orbi}) in Definition~\ref{def:critical} by mapping such flows to rays passing through $p_1(-|\cV_\psi|)$ in the quiver.

As before, the central charge must have a limit as $\mu\to\mu_0$, but now this limit must be non-zero (as its modulus should decrease along the flow).
This rules out the case $d=0$ because $Z_{\tau=0}(\gamma)=-\ch_2$ cannot be in the open half-line $\I e^{\I\psi}(0,+\infty)$.
Thus, $d\neq 0$.  By symmetry under $\psi\mapsto-\psi$, we focus on $\psi\in(-\pi/2,0]$.

Returning to the expansion~\eqref{appC-Ztaugamma} of the central charge, we again find that $\tau' q'$ has a limit and that this limit must vanish to avoid a divergent phase.
Thus, $Z_\tau(\gamma)\to\I\cV d-\ch_2$ as $\mu\to\mu_0$.
Along the flow, the central charge is fixed to lie in $Z_\tau(\gamma)\in\I e^{\I\psi}[0,+\infty)$ and to move towards~$0$ along this half-line, hence $(\I\cV d - \ch_2) - Z_\tau(\gamma)$ lies in the same half-line.
This implies \eqref{appC-dcospsi} with a constant shift of $\tau'_1$,
\be
d > 0 , \qquad
\frac{\ch_2}{d} = \cV_\psi , \qquad
2\pi\tau'_1 = -2\pi n + \psi + \pi + o(1) 
\ee
The integer $n\in\IZ$ can be eliminated by a $\Gamma_1(3)$ transformation $V^n\colon\tau'\to\tau'+n$.
Then $\tau'_1\in(1/4,1/2]$, namely $\tau'$ is in the closure $\overline{\cF}_C$, which is stable under the Fricke involution, so $\tau\in\overline{\cF}_C$.
In addition, the sign of $\tau'_1$ yields $\Re\tau<0$, so that $\tau$ lies in the closure $\overline{\cF}_o$ of the orbifold fundamental domain.

The expansions of $s$ and~$w$ are then the opposites of~\eqref{appC-swt2}, so that for $\psi\in(-\pi/2,0]$ we have $w,s\leq 0$ close to the conifold point, hence the inequality $2w\leq -s$ defining the orbifold region~$\IH^o$ within~$\cF_o$ is satisfied.
DT invariants along the flow are thus correctly given by those of the orbifold diagram~$\cD_o$ at a suitable point~$(u,v)$.
The affine coordinates~\eqref{defxysgen} are
\be
x = \frac{\Re( e^{-\I \psi} T)}{\cos\psi} = \cV_\psi + o(1) = -|\cV_\psi| + o(1) ,\quad
y = -\frac{\Re( e^{-\I \psi} T_D)}{\cos\psi} = o(1) .
\ee
Thus, \eqref{appC-uv-long}~holds as well, and the attractor flow starts (in the conifold limit $\mu\to\mu_0$) at the point $p_1(-|\cV_\psi|)$ lying on the initial ray~$\Ract_o(\gamma_1)$ of the quiver scattering diagram.

As in the previous section, we find that flows with $\gamma\notin[1,0,0]\IZ$ that start (rather than end) at $\tau=0$ are given by rays of the quiver scattering diagram that end\footnote{We recall the opposite orientation of rays and of the attractor flow.} (rather than start) at $(u,v)=p_1(-|\cV_\psi|)$.
By consistency of the orbifold scattering diagram, the points $p_1(-|\cV_\psi|)$ of $\Ract_o$ with incoming rays or with outgoing rays are the same, and correspond by Proposition~\ref{prop:critical} to critical phases.
This situation, in which both incoming and outgoing rays at $\tau=0$ occur for a critical phase, is illustrated in Figure~\ref{figscattPicr} for $\cV_\psi\simeq -1/2$.

Since both incoming and outgoing rays at $\tau=0$ are seen in the orbifold diagram, it should be interesting to translate the consistency of the orbifold scattering diagram at $p_1(-|\cV_\psi|)$ into a notion of consistency of the exact diagram~$\cD^\Pi_\psi$ at the conifold point, which is a singular point in the moduli space.

\section{On the mathematical definition of DT invariants\label{sec_defDT}}

Here we provide mathematical details on the definition of the DT 
invariants $\Omega_\sigma(\gamma)$. These invariants are a direct generalisation of the integer BPS invariants of \cite{DavMein}. From \cite{BravDyck}, the objects of a smooth CY3 dg category $\cC$ (like the dg category of perfect complexes with compact support on a smooth CY3-fold) form a 
$-1$-shifted symplectic derived stack $\mathcal{M}$ in the sense of \cite{shifsymp}. We suppose that $\mathcal{M}$ admits an orientation, \ie a square root of the line bundle $\det(\mathbb{L}_\mathcal{M})$ given by the determinant of the cotangent complex of $\mathcal{M}$ (a canonical orientation was constructed in the case of sheaves with compact support on noncompact CY3-fold in \cite[Theorem 4.9]{JOYCEorient}). For $\sigma$ a stability condition on $\cC$,  $\sigma$-semistability is a Zariski open condition, hence from \cite[Proposition 2.1]{STV11} there is an open $-1$-shifted symplectic substack $\mathcal{M}_\sigma\hookrightarrow\mathcal{M}$ of $\sigma$-semistable objects.  From the definition of semistability, $\Ext^i(E,E)=0$ for any $E\in \mathcal{M}_\sigma$ and $i<0$, hence by \cite[Proposition 3.3]{BravDyck} $\mathbb{T}_{\mathcal{M}_\sigma}|_E=\Ext(E,E)[1]$, $\mathcal{M}$ is a $-1$-shifted symplectic Artin--1 stack. We then  define $\mathcal{M}_\sigma(\gamma)$ as the component of $\mathcal{M}_\sigma$ of objects of class $\gamma$ in the Grothendieck group of $\cC$. Suppose now that $\sigma$ is generic, \ie that two $\sigma$-semistable objects  $E,E'$ of the same phase have collinear charges. For $\gamma$ primitive, we define the DT invariants $\Omega_\sigma(k\gamma),k\geq 1$ by
\begin{align}
    \Exp\biggl(\sum_{k=1}^\infty\frac{\Omega_\sigma(k\gamma)}{y^{-1}-y}x^k\biggr)\coloneqq
    \sum_{k=0}^\infty H_c(\mathcal{M}_\sigma(k\gamma),P_{\mathcal{M}_\sigma(k\gamma)})x^k
\end{align}
where $\Exp$ denotes the plethystic exponential, $P_{\mathcal{M}_\sigma(k\gamma)}$ the monodromic mixed Hodge module on $\mathcal{M}_\sigma(k\gamma)$ constructed in \cite[Theorem 4.4]{darbstack} using the orientation data, and $H_c(M,P)$ the Hodge polynomial of the cohomology with compact support on $M$ with values in $P$. In the case of quiver with potentials and King stability conditions, the invariants $\Omega_\sigma(k\gamma)$ are integer  for any $\gamma$ and $k\geq 1$ by \cite{DavMein}, and we conjecture that this remains true in this more general framework. In particular, we conjecture that, as in \cite{DavMein},
\begin{align}
    \Omega_\sigma(k\gamma)=H_c\bigl(M_\sigma(k\gamma),\mathcal{H}^1(\JH_!P_{\mathcal{M}_\sigma(k\gamma)})\bigr)
\end{align}
where $\JH\colon\mathcal{M}_\sigma(k\gamma)\to M_\sigma(k\gamma)$ denotes the Jordan-H\"older map to the coarse moduli space, $\JH_!$ denotes the proper pushforward for the derived categories of monodromic mixed Hodge modules, and $\mathcal{H}^1$ denotes the first cohomology of a complex of monodromic mixed Hodge modules.

\section{Gieseker indices for higher rank sheaves \label{sec_higherk}}
In this section, we extend the list of examples in \S\ref{sec_LV} and determine the trees
contributing to the Gieseker index for some examples with higher rank. As explained in 
\cite{coskun2014ample,coskun2014birational}, for 
$(r,d)$ coprime and discriminant $\Delta(\gamma)\geq  \Delta_1(r,d)$ 
large enough, the Gieseker wall is 
$\cW(\gamma,\gamma')$ due to a subobject with Chern vector $\gamma'=[r',d',\chi')$ 
uniquely determined by the following conditions:
\begin{itemize}
\item $0 < r' \leq r$\ ,\quad $\mu(\gamma')<\mu(\gamma)$
\item Every rational number in the interval $(\mu(\gamma'),\mu(\gamma))$ has denominator greater than $r$
\item The discriminant of any stable bundle of slope $\mu(\gamma')$ and rank $\leq r$ is $\geq \Delta(\gamma')$
\item the rank of any stable bundle of slope $\mu(\gamma')$ and discriminant $\Delta(\gamma')$ 
is $\geq r'$
\end{itemize}
The minimal value $\Delta_1(r,d)$ for which conditions are applicable and the rightmost point $x_+=s_{\gamma,\gamma'}+R_{\gamma,\gamma'}$ of the Gieseker wall 
for the lowest discriminant $\Delta_0\geq \Delta_1$ are tabulated in \cite[Table 3]{coskun2014ample} for $r\leq 6$ and $0<\mu(\gamma)\leq 1$.

\subsection{Rank 2}

We consider rank 2 sheaves with $\gamma=[2,-1,1-n)$, discriminant $\Delta=\frac{n}{2}-\frac18$.
The condition \eqref{MGiesekerNotEmpty} gives $\Delta\geq \deltaLP(-\frac12)=\frac58$ for non-exceptional sheaves. 
The generating function of Gieseker indices is given by \cite{Yoshioka:1994}
 \cite[(A.38)]{Beaujard:2020sgs}
\be\begin{aligned}
\label{h2m1series}
h_{2,-1} &= q + (y^2+1+1/y^2)^2 q^2 + (y^8+2y^6+6y^4+9y^2+12+\dots) q^3 \\
& \quad + \left( y^{12}+2 y^{10}+6 y^8+13 y^6+24 y^4+35
   y^2+41+\dots\right) q^4+
\dots
\\
&\rightarrow q + 9 q^2 + 48 q^3 +203 q^4 + 729 q^5 + 2346 q^6 + \dots
\end{aligned}
\ee

\begin{itemize}

\item For $n=1$, corresponding to the exceptional sheaf $\Omega(1)$, there is a single wall $\cC(-\frac32,\frac12)$ associated to the scattering sequence  $\{-\cO(-2), 3\cO(-1)\}$ contributing $K_3(1,3)=1$. 

\item For $n=2$  there is a single wall  $\cC(-\frac52,\frac32)$ associated to $\{\{-\cO(-3), \cO(-2)\}, 2\cO(-1)\}$ contributing 
$K_3(1,1)K_3(1,2)=(y^2+1+1/y^2)^2$. The rightmost point on the wall
is at $s=-1$, consistent with the entries $\Delta_0=\frac78, x_+=0$ in \cite[Table 3]{coskun2014ample}.

\item For $n=3$, there is a single wall  associated to  two scattering sequences
\begin{equation*}
\begin{array}{lll}
\cC(-\frac72,\frac52) &\{\{-\cO(-4),\cO(-3)\},2\cO(-1)\} & K_3(1,1) K_3(1,2) \\
& \{\{ -2\cO(-3),3\cO(-2)\} ,\cO(-1)\} &  K_3(1,1) K_3(2,3)
\end{array}
\end{equation*}
 contributing $9+39=48$ in the unrefined limit.
 
 \item For $n=4$, there are 2 walls associated to four scattering sequences
\begin{equation*}
\begin{array}{lll}
\cC(-\frac92,\frac72)& \{\{-\cO(-5), \cO(-4)\}, 2 \cO(-1)\} &  K_3(1,1) K_3(1,2)  \\
& \{\{\{-\cO(-4), \cO(-3)\}, \{-\cO(-3), 2 \cO(-2)\}\}, \cO(-1)\} &  K_3(1,1)^2 K_3(1,2) \\
& \{\{-\cO(-4), 2 \cO(-2)\}, \cO(-1)\} &  K_3(1,1) K_6(1,2) \\
\cC(-\frac72,\frac12) & \{-3 \cO(-3), 5 \cO(-2)\} &  K_3(3,5)
\end{array}
\end{equation*}
 contributing $9+81+45+ 68=203$ in the unrefined limit.
\end{itemize}

For $\gamma=[2,0,2-n)$, with discriminant $\Delta=n/2$, the expected generating function is
\be\label{h20series}
\begin{aligned}
h_{2,0} & = -(y^5+y^3+y+\dots) q^2 - (y^9+2y^7+4y^5+6y^3+6y+\dots) q^3 \\
& \quad - \left( y^{13} + 2 y^{11} + 6 y^9 + 11 y^7 + 19 y^5 + 24 y^3 + 27 y+\dots\right) q^4 - \dots\\
& \rightarrow -6 q^2 -38 q^3 - 180 q^4 - 678 q^5- 2260q^6 - \dots
\end{aligned}
\ee
The  condition \eqref{MGiesekerNotEmpty}
gives $\Delta\geq \deltaLP(0)=1$ for non-exceptional sheaves. 
\begin{itemize}
\item For $n=2$ there is a single wall $\cC(-\frac32,\frac12)$ associated to $\{-2\cO(-2), 4\cO(-1)\}$ contributing $K_3(2,4)\rightarrow -6$.

\item For $n=3$  there are two walls 
\begin{equation*}
\begin{array}{lll}
\cC(-\frac52,\frac{\sqrt{13}}{2}) &\{-\cO(-3),3\cO(-1)\} &  K_6(1,3)
\\
\cC(-2,1) &\{\{ -\cO(-3),\cO(-2)\} ,\{-\cO(-2),3\cO(-1)\}\} &  K_3(1,1) K_3(1,3) K_6(1,1)
\end{array}
\end{equation*}
 contributing $-20-18=-38$ in the unrefined limit. The Gieseker wall 
 has  rightmost point  at $s=-\frac52+\frac{\sqrt{13}}{2}$,
 consistent with the values $\Delta_0=\frac32, x_+\simeq 0,30$ 
 in  \cite[Table 3]{coskun2014ample}. 
 
 \item For $n=4$  there are two walls
\begin{equation*}
\begin{array}{lll}
\cC(-\frac72,\frac{\sqrt{33}}{2}) & \{\{-\cO(-4), \cO(-3)\}, \{-\cO(-2), 3 \cO(-1)\}\}  & K_3(1,1) K_3(1,3) K_6(1,1) \\
\cC(-\frac52,\frac32) & \{\{-2 \cO(-3), 2 \cO(-2)\}, 2 \cO(-1)\} &  K_{-3,3,6}(2,2,2)
\end{array}
\end{equation*}
The index for the second scattering sequence is obtained 
in the same way as in \eqref{Om221}, \ie by 
applying the flow tree formula for a local scattering diagram with
 two incoming rays of charge $\alpha=\gamma_1+\gamma_2$
and $\beta=\gamma_3$ with $\Omega^-(\alpha)=K_3(1,2)=y^2+1+1/y^2, \Omega^-(2\alpha)
=K_3(2,2)=-y^5-y^3-y-1/y-1/y^3-1/y^5$ and $\Omega^-(\beta)=1$, and selecting the outgoing ray of charge $2\alpha+2\beta$. This leads to 
\be
K_{-3,3,6}(2,2,2)=-y^{13} - 2 y^{11} - 6 y^9 - 10 y^7 - 17 y^5 - 21 y^3 - 24 y- \dots \rightarrow -162
\ee
Adding up the contributions of the two scattering sequences, we get $-18-162=-180$
in the limit $y\to 1$, in agreement with \eqref{h20series}.

\end{itemize}

\subsection{Rank 3}
 
 We now turn to rank 3 sheaves, with $\gamma=[3,-1,2-n)$, discriminant $\Delta=\frac{n}{3}-\frac19$. 
 The  condition \eqref{MGiesekerNotEmpty}
 gives $\Delta\geq \deltaLP(-\frac13)=\frac59$ for non-exceptional sheaves. 
 The generating function is given by  \cite[Table 1]{Manschot:2010nc} \cite[(A.40)]{Beaujard:2020sgs} 
\be\label{h3m1series}
\begin{aligned}
h_{3,-1} &= (y+1+1/y^2) q^2 + (y^8+2y^6+5y^4+8y^2+10+\dots) q^3  + \dots
\\
&\rightarrow 3q^2 + 42 q^3 + 333 q^4 + 1968 q^5 + 9609 q^6+ \dots
\end{aligned}
\ee

\begin{itemize}

\item For $n=2$ there is a single wall $\cC(-\frac32,\frac12)$ associated to $\{-2\cO(-2), 5\cO(-1)\}$ contributing $K_3(2,5)=y^2+1+1/y^2$.

\item For $n=3$  there are three walls:
\begin{equation*}
\begin{array}{lll}
\cC(-\frac{7}{2},\frac{33}{2})&\{ \{\{-\cO(-3), \cO(-2)\}, \cO(-1)\},\{-\cO(-2), 3 \cO(-1)\}\} & K_3(1,1)^3 K_3(1,3)
\\
\cC(-\frac{5}{2},\frac{\sqrt{35}}{2\sqrt3})&\{\{ -\cO(-3),\cO(-2)\} ,\{-\cO(-2),4\cO(-1)\}\} & K_3(1,1) K_3(1,4) K_9(1,1)
\\
\cC(-2,1)&\{\cO(-3),4\cO(-1)\} & K_6(1,4)
\end{array}
\end{equation*}
 contributing $27+0+15=42$ in the unrefined limit. The Gieseker wall $\cC(-\frac{7}{2},\frac{33}{2})$
has rightmost point at  $-\frac52+\frac12 \sqrt{33}$, 
consistent with the values  
$\Delta_0=\frac89,  x_+\simeq 0,37$
 quoted in \cite[Table 3]{coskun2014ample}.
 \end{itemize}

For $\gamma=[3,1,5-n)$  with $\Delta=\frac{n}{3}-\frac19$, we find instead the following 
scattering sequences (not related to the previous ones by 
reflection)
\begin{itemize}
\item For $n=1$, there is a scattering sequence  $\{-\cO(-1),4\cO\}$ but its index $K_3(1,4)$ vanishes.
\item For $n=2$ there is a single wall
\begin{equation*}
\cC(-\tfrac32,\tfrac32)  \quad \{\{-\cO(-2), \cO(-1)\},3\cO(0)\} \quad K_3(1,3) K_3(1,1)
\end{equation*}
contributing $y^2+1+1/y^2$.

\item For $n=3$  there is a single wall but two scattering sequences,
\begin{equation*}
\begin{array}{lll}
\cC(-\frac52,\frac52) &\{ \{ -2\cO(-2), 3\cO(-1)\}, 2\cO \} & K_3(1,2) K_3(2,3)
\\
 &\{\{ -\cO(-3),\cO(-2)\} ,3\cO\} &  K_3(1,1)  K_3(1,3)
\end{array}
\end{equation*}
The wall $\cC(-\frac52,\frac52)$ has rightmost point at $x_+=0$,
consistent with the values  $\Delta_0=\frac89,  x_+=0$
quoted in \cite[Table 3]{coskun2014ample}.

\end{itemize}

For $\gamma=[3,0,3-n)$ with discriminant $\Delta=n/3$, the generating function is 
\cite[(A.40)]{Beaujard:2020sgs}
\be
\begin{aligned}
h_{3,0} &= (y^{10}+y^8+2y^6+2y^4+2y^2+2+\dots) q^3 \\
& \quad +
\left( y^{16}+2
   y^{14}+5 y^{12}+9 y^{10}+15 y^8+19 y^6+22 y^4+23
   y^2+24+\dots\right) q^4 + \dots
\\
&\rightarrow 18 q^3 +216 q^4+1512 q^5 + 8109 q^6 \dots
\end{aligned}
\ee

\begin{itemize}
\item For $n=3$, there is a single wall $\cC(-\frac32,\frac12)$ associated to $\{-3\cO(-2),6\cO(-1)\}$ contributing 
$K_3(3,6) \to 18$.

\item For $n=4$, there are two walls
\begin{equation*}
\begin{array}{lll}
\cC(-\frac52,\frac{\sqrt{43}}{2\sqrt3}) & \{\{-\cO(-3), \cO(-2)\}, \{-2 \cO(-2), 5 \cO(-1)\}\}  & 
 K_3(1,1) K_3(2,5) K_9(1,1) \\
\cC(-\frac{13}{6},\frac{\sqrt{73}}{6}) & \{\{-\cO(-3),  2 \cO(-1)\}, \{-\cO(-2), 3 \cO(-1)\}\}  & 
K_3(1,3) K_6(1,2)  K_9(1,1)
\end{array}
\end{equation*}
contributing $81+135=216$ in the unrefined limit. The Gieseker wall  has  rightmost point $-\frac52+\frac12\sqrt{\frac{43}{3}}= x_+-1$, consistent with the values $\Delta_0=\frac43, x_+\simeq 0,39$
quoted in  \cite[Table 3]{coskun2014ample}.

\end{itemize}

\section{Mathematica package  {\tt P2Scattering.m}\label{sec_mathematica}}

The  {\sc Mathematica} package {\tt P2Scattering.m}, available from
\begin{center}
  \url{https://github.com/bpioline/P2Scattering}
\end{center}
provides a suite of routines for analyzing the scattering diagrams considered in this work,  
both at large volume, around the orbifold and along the $\Pi$-stability slice. It was used
extensively in order to generate the figures and arrive at the global picture presented in this article. 
A list of routines is provided in the documentation {\tt P2Scattering.pdf} available in the 
GitHub repository, along with several demonstration worksheets. Here we simply
give a taste of the package capabilities.

After copying file {\tt P2Scattering.m} in the current directory, load the package via

\mathematica{1}{0.965}{SetDirectory[NotebookDirectory[]]; <\!\!\! < P2Scattering\`{}}{P2Scattering 1.4 - A package for evaluating DT invariants on $K_{\IP^2}$}

For a given charge $\gamma=[r,d,\chi)$ and point $(s,t)$ on the large volume slice, the scattering sequences contributing to the index $\Omega_{(s,t)}(\gamma)$ can be found by using the routine $\fun{ScanAllTrees}$, 
for example for $\gamma=[3,0,0)$ through the point $(s,t)=(-\frac32,2)$,

\mathematica{2}{0.965}{LiTrees = ScanAllTrees[\{0, 3, 0\}, \{-3/2, 2\}]
}{
\{\{-Ch[-3], Ch[0]\}, \{-3 Ch[-2], 3 Ch[-1]\}\}
}

\mathematica{3}{0.965}{ScattDiagLV[LiTrees, 0]
}{
\raisebox{-\totalheight/2}{\includegraphics[width=6cm]{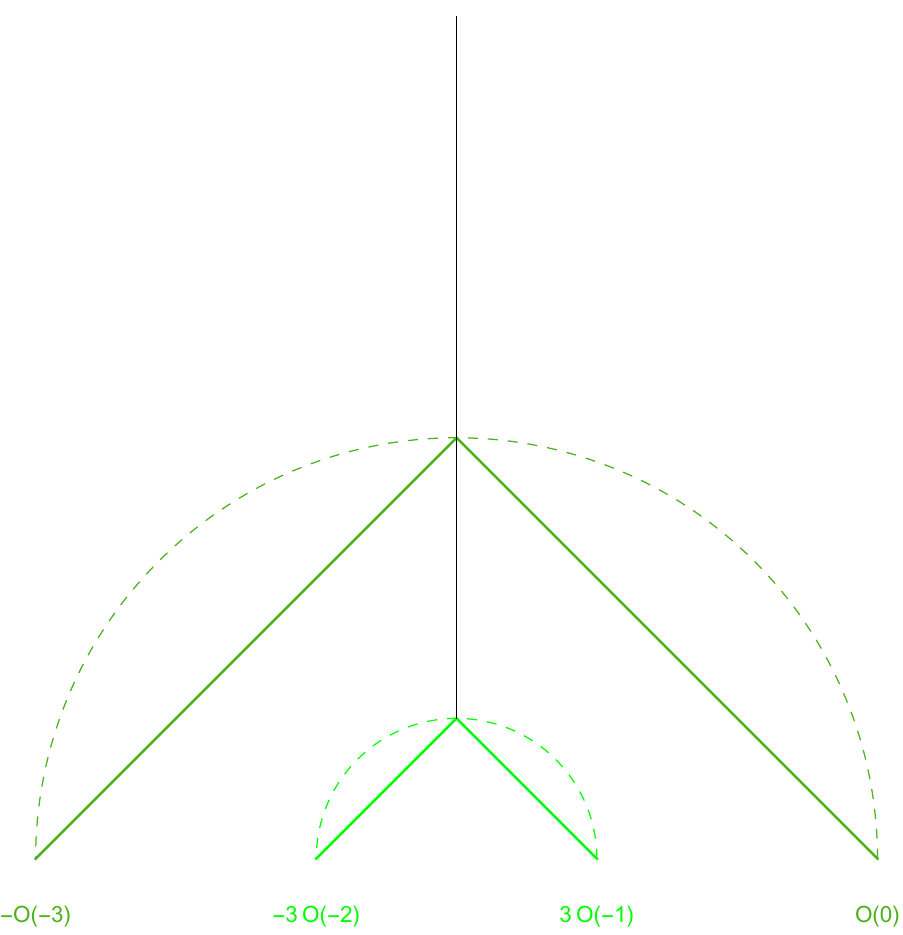}}
}

\mathematica{4}{0.965}{Limit[EvaluateKronecker[ScattIndex[LiTrees]], y -> 1]
}{
\{9,18\}
}
reproducing the GV invariant $N_3^{(0)}=27$ (compare with \S\ref{sec_D2D0}). 
Note that the current implementation of the routine \fun{ScattIndex}
assumes that the index associated to each scattering sequence
 is a product of Kronecker indices associated to each
vertex, and may give the wrong result if some of the edges carry non-primitive charges
(see \eqref{Om221} for an example).
 In the case above, it does produce the correct results for both scattering sequences, $\Omega_\infty(\gamma)=K_9(1,1)+K_3(3,3)$. More generally, the routine \fun{IndexFromSequences$[\{\text{\it trees}\},\{s,t\}]$}
 computes the total rational index $\bOm_{s,t}(\gamma)$
 by decomposing each scattering
 sequence into attractor flow trees as explained at the end of \S\ref{sec_inidata}, and perturbing
 the charges of the constituents $\gamma_i\to\gamma_i+\epsilon_i \delta$ such that only binary
 splittings remain:
 
 \mathematica{5}{0.965}{Limit[Plus@@Flatten[IndexFromSequences[LiTrees, \{-3/2, 2\}]], 
 y -> 1]
}{
82/3
}
consistent with $\bOm_{s,t}(3\gamma)=
\Omega_{s,t}(3\gamma)+\frac19 \Omega_{s,t}(\gamma)=27+\frac13$.

Similarly, one can find the scattering sequences contributing near
the orbifold point using \fun{McKayScanAllTrees}: for the same charge, corresponding to dimension vector $(0,3,6)$, a single scattering sequence contributes in the anti-attractor chamber, with index 18,

\mathematica{6}{0.965}{
LiTrees = McKayScanAllTrees[chiton[\{0, 3, 0\}]]; 
LiTrees /. McKayrep
}{
$\{\{ 3\gamma_2,6 \gamma_3\}\}$
}

\mathematica{7}{0.965}{
Limit[EvaluateKronecker[McKayScattIndex[LiTrees]], y ->1]
}{
\{18\}
}

\mathematica{8}{0.965}{
Show[McKayInitialRays[2], McKayScattDiag[LiTrees]]
}{
\raisebox{-\totalheight/2}{\includegraphics[width=6cm]{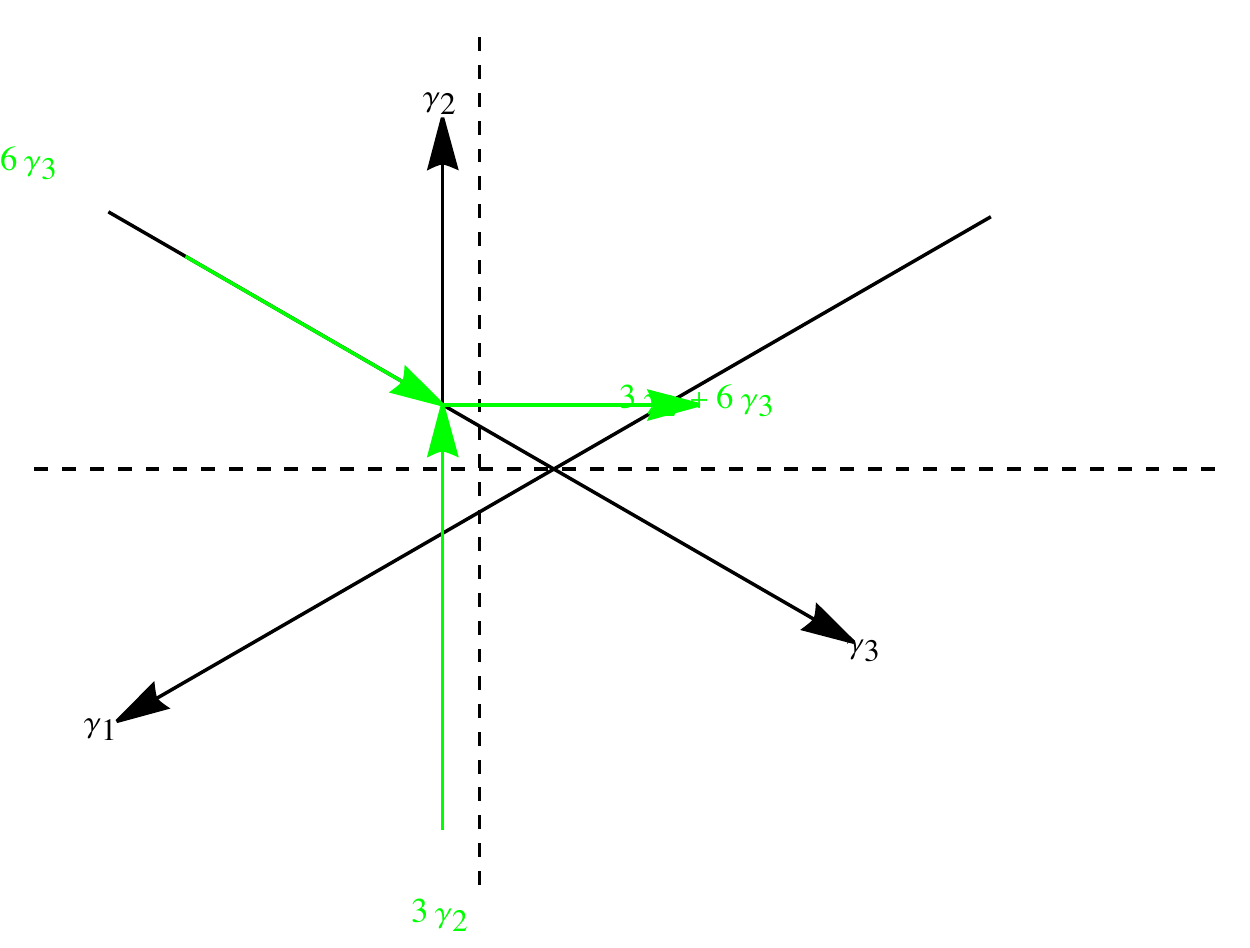}}
}
In this example, the index in the anti-attractor chamber differs from the one at large volume,
due to wall-crossing along the circle $\cC(-\frac32,\frac12)$.

%\bibliography{combined}
%\bibliographystyle{utphys}

\providecommand{\href}[2]{#2}\begingroup\raggedright\endgroup

\end{document}